\documentclass[11pt,letter]{article}

\usepackage[margin=1in]{geometry}

\usepackage[utf8]{inputenc}
\usepackage[T1]{fontenc}
\usepackage{xspace}

\usepackage{amsmath,amssymb,amsfonts,amsthm}
\usepackage{mathtools}
\usepackage{hyperref}
\usepackage{stmaryrd}
\usepackage{etoolbox}
\usepackage{float}
\usepackage{bm}
\usepackage{cite}
\usepackage{graphicx}
\usepackage{textcomp}
\usepackage{cleveref}
\usepackage{enumitem}
\usepackage{longtable}
\usepackage{bigdelim} 

\usepackage{tabto}
\usepackage{thm-restate}
\makeatletter
\pretocmd{\thmt@rst@storecounters}{\Hy@SaveLastskip}{}{}
\apptocmd{\thmt@rst@storecounters}{\Hy@RestoreLastskip}{}{}
\makeatother

\usepackage{blkarray}
\usepackage{setspace}
\usepackage{refcount}

\usepackage{lineno}

\usepackage{caption}

\usepackage{xcolor}
\def\BibTeX{{\rm B\kern-.05em{\sc i\kern-.025em b}\kern-.08em
    T\kern-.1667em\lower.7ex\hbox{E}\kern-.125emX}}

\usepackage{tikz}
\usetikzlibrary{
  decorations.pathreplacing,
  calc,
  arrows,
  positioning,
  backgrounds,
  shapes.geometric
}
\usepackage{pgfplots}
\usepackage{svg}

\newcommand{\tikzmark}[1]{\tikz[overlay,remember picture] \node (#1) {};}

\newcommand*{\AddNote}[4]{%
    \begin{tikzpicture}[overlay, remember picture]
        \draw [decoration={brace,amplitude=0.4em},decorate,ultra thick,gray]
          ($(#3)!(#2.south)!($(#3)-(0,1)$)$) --  
          ($(#3)!(#1.south)!($(#3)-(0,1)$)$)  
                node [align=center, text width=2.5cm, pos=0.5, right=-1.9cm] {\rotatebox{90}{\textsc{#4}}};
    \end{tikzpicture}
}%

\usepackage{colortbl}

\usepackage{algorithm}
\usepackage[noend]{algpseudocode}

\algnewcommand\algorithmiccontext{\textbf{Context:}}
\algnewcommand\Context{\item[\algorithmiccontext]}
\algnewcommand\algorithmicbranchoutput{\textbf{Branch Output:}}
\algnewcommand\BranchOutput{\item[\algorithmicbranchoutput]}
\algnewcommand\algorithmicndbranchoutput{\textbf{Output of each branch ($\beta$):}}
\algnewcommand\NDBranchOutput{\item[\algorithmicndbranchoutput]}
\algnewcommand\algorithmicglobaloutput{\textbf{Global Output:}}
\algnewcommand\GlobalOutput{\item[\algorithmicglobaloutput]}
\algnewcommand\algorithmicglobalspec{\textbf{Ensuring:}}
\algnewcommand\GlobalSpec{\item[\algorithmicglobalspec]}
\algnewcommand\algorithmicswitch{\textbf{switch}}
\algnewcommand\algorithmiccase{\textbf{case}}
\algnewcommand\algorithmicforeach{\textbf{foreach}}
\algnewcommand\algorithmicnondet{\textbf{nondet}}
\algnewcommand\algorithmicor{\textbf{or}}
\algnewcommand\algorithmicassert{\textbf{assert}}
\algnewcommand\algorithmiclet{\textbf{let}}
\algnewcommand\Assert[1]{\State \algorithmicassert(#1)}
\algdef{SE}[SWITCH]{Switch}{EndSwitch}[1]{\algorithmicswitch\ #1\ \algorithmicdo}{\algorithmicend\ \algorithmicswitch}%
\algtext*{EndSwitch}%
\algdef{SE}[CASE]{Case}{EndCase}[1]{\algorithmiccase\ #1\ \textbf{:}}{\algorithmicend\ \algorithmiccase}%
\algtext*{EndCase}%
\algdef{S}[FOR]{ForEach}[1]{\algorithmicforeach\ #1\ \algorithmicdo}%
\algdef{SE}[NONDET]{Nondet}{EndNondet}{\algorithmicnondet\ }{\algorithmicend\ \algorithmicnondet}%
\algtext*{EndNondet}%
\algdef{SE}[Or]{Or}{EndOr}{\algorithmicor\ }{\algorithmicend\ \algorithmicor}%
\algtext*{EndOr}%
\algdef{SE}[LET]{Let}{EndLet}{\algorithmiclet\ }{\algorithmicend\ \algorithmiclet}%
\algtext*{EndLet}%

\AtBeginEnvironment{algorithm}{\linespread{1.05}}

\newcounter{stepnum}
\usepackage{todonotes}
\colorlet{dmitry}{green!50!black!80}
\colorlet{alessio}{blue!70!black!80}
\colorlet{hitarth}{red!70!blue!80}
\definecolor{mikhail}{HTML}{9F2D20}
\definecolor{guru}{HTML}{9F2D20}



\makeatletter
\@ifpackageloaded{stix}{%
}{%
  \DeclareFontEncoding{LS2}{}{\noaccents@}
  \DeclareFontSubstitution{LS2}{stix}{m}{n}
  \DeclareSymbolFont{stix@largesymbols}{LS2}{stixex}{m}{n}
  \SetSymbolFont{stix@largesymbols}{bold}{LS2}{stixex}{b}{n}
  \DeclareMathDelimiter{\lBrace}{\mathopen} {stix@largesymbols}{"E8}%
                                            {stix@largesymbols}{"0E}
  \DeclareMathDelimiter{\rBrace}{\mathclose}{stix@largesymbols}{"E9}%
                                            {stix@largesymbols}{"0F}
}
\makeatother

\theoremstyle{plain}
\newtheorem{example}{Example}
\newtheorem{definition}{Definition}
\newtheorem{namedtheorem}{Theorem}
\newtheorem{namedlemma}{Lemma}

\newtheoremstyle{noparens}
    {}{}{\itshape}{}%
    {\bfseries}{.}{ }%
    {\thmname{#1}\thmnumber{ #2}\thmnote{ {\mdseries #3}}}

\theoremstyle{noparens}
\newtheorem{theorem}[namedtheorem]{Theorem}
\newtheorem{lemma}[namedlemma]{Lemma}
\newtheorem{proposition}{Proposition}

\newtheorem{claim}{Claim}
\newtheorem{open}{Open problem}
\newtheorem{conjecture}{Conjecture}
\newtheorem*{fact*}{Fact}
\newtheorem{remark}{Remark}
\newtheorem*{remark*}{Remark}

\newcommand{\cand}{&\hspace{-5pt}}

\newcommand*\TOCAppendices
{\vspace{25pt}\noindent\parbox[t]{\textwidth}{\textbf{Appendices}}\par
 \global\let\TOCAppendices\empty }

\usepackage{tocloft}
\usepackage{appendix}
\newboolean{appendix}
\setboolean{appendix}{false}

\usepackage{marginnote}
\newcommand{\proofnote}[1]{\ifthenelse{\boolean{appendix}}{}{\marginnote{\footnotesize{\color{gray}Proof in page~\pageref{#1}}}}}

\newcommand{\superlabel}[2]{\label{#1}\marginnote{\ifthenelse{\boolean{appendix}}{\footnotesize{\color{gray}Statement in~page~\pageref{#1}}}{\footnotesize{\color{gray}Proof in page~\pageref{#2}}}}}
\newcommand{\newextmathcommand}[2]{%
    \newcommand{#1}{\ensuremath{#2}\xspace}
}

\makeatletter
\newcommand{\labeltext}[2]{%
  #1%
  \@bsphack%
  \csname phantomsection\endcsname 
  \def\@currentlabel{#1}{\label{#2}}%
  \@esphack%
}
\makeatother

\makeatletter
\newcommand{\customlabel}[2]{%
  \@bsphack%
  \csname phantomsection\endcsname 
  \def\@currentlabel{#1}{\label{#2}}%
  \@esphack%
}
\makeatother

\makeatletter
               {\list{}{\leftmargin=0pt
                        \labelwidth\z@ \itemindent-\leftmargin
                        }}%
               {\endlist}
\makeatother

\newextmathcommand{\N}{\mathbb{N}}
\newextmathcommand{\Np}{\Nat_+}
\newextmathcommand{\Z}{\mathbb{Z}}
\newextmathcommand{\Q}{\mathbb{Q}}
\newextmathcommand{\R}{\mathbb{R}}
\newextmathcommand{\PP}{\mathbb{P}}
\newextmathcommand{\X}{\mathbb{X}}

\newcommand{\sem}[1]{\ensuremath{\left\llbracket#1\right\rrbracket}\xspace}
\newcommand{\semlast}[1]{\ensuremath{\left\llbracket#1\right\rrbracket}_{\bullet}\xspace}

\newextmathcommand{\ptime}{\textup{\textsc{P}}\xspace}
\newextmathcommand{\fptime}{\textup{\textsc{FP}}\xspace}
\newextmathcommand{\bpp}{\textup{\textsc{BPP}}\xspace}
\newextmathcommand{\np}{\textup{\textsc{NP}}\xspace}
\newextmathcommand{\fnp}{\textup{\textsc{FNP}}\xspace}
\newextmathcommand{\pspace}{\textup{\textsc{PSpace}}\xspace}
\newextmathcommand{\nexptime}{\textup{\textsc{NExpTime}}\xspace}
\newextmathcommand{\expspace}{\textup{\textsc{ExpSpace}}\xspace}
\newextmathcommand{\twonexptime}{\textup{\textsc{2NExpTime}}\xspace}
\newextmathcommand{\threeexptime}{\textup{\textsc{3ExpTime}}\xspace}
\newextmathcommand{\tower}{\textup{\textsc{Tower}}\xspace}
\newextmathcommand{\npo}{\textup{\textsc{NPO}}\xspace}
\newextmathcommand{\npocmp}{\textup{\textsc{NPO-cmp}}\xspace}
\newextmathcommand{\factoring}{\textup{\textsc{factoring}}\xspace}

\let\temp\phi
\let\phi\varphi
\let\varphi\temp
\newextmathcommand{\totient}{\varphi}
\renewcommand{\vec}{\bm}

\newextmathcommand{\lcm}{{\rm lcm}}

\newcommand{\abs}[1]{\ensuremath{\left|#1\right|}\xspace}
\newcommand{\ceil}[1]{\ensuremath{\left\lceil#1\right\rceil}\xspace}
\newcommand{\floor}[1]{\ensuremath{\left\lfloor#1\right\rfloor}\xspace}

\DeclareFontEncoding{LS2}{}{\noaccents@}
  \DeclareFontSubstitution{LS2}{stix}{m}{n}
  \DeclareSymbolFont{stix@largesymbols}{LS2}{stixex}{m}{n}
  \SetSymbolFont{stix@largesymbols}{bold}{LS2}{stixex}{b}{n}
  \DeclareMathDelimiter{\lBrace}{\mathopen} {stix@largesymbols}{"E8}%
                                            {stix@largesymbols}{"0E}
  \DeclareMathDelimiter{\rBrace}{\mathclose}{stix@largesymbols}{"E9}%
                                            {stix@largesymbols}{"0F}

\newcommand{\defeq}{\coloneqq}
\newcommand{\eqdef}{\defeq}

\newcommand{\sub}[2]{\ensuremath{[#1\,/\,#2]}\xspace}
\newcommand{\elimdisctxt}{\ensuremath{\textit{Elim}}\xspace}
\newcommand{\elimdisc}[3]{\ensuremath{\elimdisctxt(#1,#2,#3)}\xspace}

\newcommand{\onenorm}[1]{\ensuremath{\lVert{#1}\rVert_{1}}\xspace}

\newextmathcommand{\card}{\#}

\newcommand{\linnorm}[1]{\ensuremath{\lVert{#1}\rVert_{\!\mathfrak{L}}}\xspace}

\newextmathcommand{\V}{V_2}

\newextmathcommand{\fterms}{\textup{terms}}
\newextmathcommand{\fmod}{\textit{mod}}
\newextmathcommand{\divides}{\mathrel{|}}
\newextmathcommand{\fdiv}{\textit{div}}
\newextmathcommand{\lst}{\textit{lst}}
\newextmathcommand{\coeff}{\textup{coeff}}
\newextmathcommand{\const}{\textup{const}}
\newextmathcommand{\vars}{\textup{vars}}

\newextmathcommand{\lead}{\ell}
\newextmathcommand{\prevlead}{p}

\newcommand{\diag}{\mathrm{diag}}

\definecolor{light-gray}{gray}{0.95}
\newcolumntype{g}{>{\columncolor{light-gray}}r}

\newcommand{\myguess}{\textbf{guess}\xspace}

\newextmathcommand{\GaussQE}{\textsc{GaussQE}}
\newextmathcommand{\OptILEP}{\textsc{OptILEP}}
\newextmathcommand{\GaussOpt}{\textsc{ElimVars}}
\newextmathcommand{\BTP}{\hyperref[algo:btp]{\textsc{BTP}}}
\newextmathcommand{\MonotonicRefinement}{\textsc{MonotonicRefinement}}
\newextmathcommand{\Master}{\hyperref[algo:master]{\textsc{LinExpOpt}}}
\newextmathcommand{\ElimMaxVar}{\hyperref[algo:elimmaxvar]{\textsc{ElimMaxVar}}}
\newextmathcommand{\SolvePrimitive}{\hyperref[algo:linearize]{\textsc{SolvePrimitive}}}
\newextmathcommand{\SplitLSP}{\hyperref[algo:splitlsp]{\textsc{SplitLSP}}}

\newextmathcommand{\kl}{(k,\ell)}
\newextmathcommand{\leac}{LEAC\xspace}

\newextmathcommand{\tests}{\textit{TP}}
\newextmathcommand{\testsilep}{\tests_{\textup{ILEP}}}
\newextmathcommand{\disciplineilep}{\abssub{\cdot}{\cdot}_{\textup{ILEP}}}

\newextmathcommand{\constraints}{\mathcal{C}}
\newextmathcommand{\objectives}{\mathcal{F}}
\newextmathcommand{\objcons}{\mathcal{I}}

\newextmathcommand{\interiorpoints}{\textup{interior}}

\newextmathcommand{\goal}{\mathrm{goal}}
\newextmathcommand{\opt}{\mathrm{opt}}
\newextmathcommand{\short}{\mathrm{short}}
\newextmathcommand{\sol}{\mathrm{sol}}

\newextmathcommand{\lin}{\mathrm{lin}}
\newextmathcommand{\context}{\mathrm{ctx}}

\newextmathcommand{\odd}{\textit{odd}}

\newextmathcommand{\map}{\Lambda}

\newextmathcommand{\trunc}{T}
\newextmathcommand{\indic}{\mathbf{1}}

\newextmathcommand{\posileslp}{\textup{\textsc{Nat}}_{\scalebox{0.7}{\textup{\textsc{ILESLP}}}}}
\newextmathcommand{\modileslp}{\textup{\textsc{Div}}_{\scalebox{0.7}{\textup{\textsc{ILESLP}}}}}

\newcommand{\objfun}[2]{\ensuremath{{#1}[{#2}]}\xspace}
\newcommand{\inst}[2]{\ensuremath{\langle{#1} \mathbin{;} {#2}\rangle}\xspace}

\newextmathcommand{\preleac}{\text{PreLEAC}}
\newextmathcommand{\aux}{\text{aux}}

\newextmathcommand{\poly}{\text{poly}}

\newextmathcommand{\bit}{\textit{bit}}

\usepackage{fancyhdr}
\renewcommand{\subsectionmark}[1]{\markright{Section~\arabic{section}.\arabic{subsection}: #1}{}}
\renewcommand{\sectionmark}[1]{\markright{Section~\arabic{section}: #1}{}}
\newcommand{\RestoreHeader}{\fancyhead[R]{{\color{gray}\rightmark}}}
\RestoreHeader

\renewcommand{\headrulewidth}{1pt}
\renewcommand{\headrule}{\hbox to\headwidth{\color{gray!50}\leaders\hrule height \headrulewidth\hfill}}

\usepackage{authblk}
\title{Optimization Modulo Integer Linear-Exponential~Programs}
\author[1]{S Hitarth}
\author[2]{Alessio Mansutti}
\author[3]{Guruprerana Shabadi}
\affil[1]{Hong Kong University of Science and Technology, Hong Kong}
\affil[2]{IMDEA Software Institute, Spain}
\affil[3]{University of Pennsylvania, USA}

\date{}

\begin{document}

\maketitle

\begin{abstract}
  This paper presents 
  the first study of the complexity
  of the optimization problem for \emph{integer linear-exponential programs} 
  which extend classical integer linear programs with the exponential function $x \mapsto 2^x$ and the remainder function ${(x,y) \mapsto (x \bmod 2^y)}$. 
  The problem of deciding if such a program has a solution was recently shown to be NP-complete in~[\mbox{Chistikov et~al.,} ICALP'24]. 
  The optimization problem instead asks for a solution that maximizes (or minimizes) a linear-exponential objective function, 
  subject to the constraints of an integer linear-exponential program.
  We establish the following results:
  \begin{itemize}
    \item If an optimal solution exists, then one of them can be succinctly represented as an \emph{integer linear-exponential straight-line program (ILESLP)}: an arithmetic circuit whose gates always output an integer value (by construction) and implement the operations of addition, exponentiation, and multiplication by rational numbers.
    \item There is an algorithm that runs in polynomial time, given access to an integer factoring oracle, which determines whether an ILESLP encodes a solution to an integer linear-exponential program. This algorithm 
    can also be used to compare the values taken by the objective function 
    on two given solutions. 
  \end{itemize}
  Building on these results, we place the optimization problem for integer linear-exponential programs within an extension of the optimization class \npo that lies within $\fnp^{\np}$. In essence, this extension forgoes determining the optimal solution via binary search.

  \vspace{0.5cm}
  \noindent
  {\footnotesize(Page~\pageref{toctoc} includes a table of contents.)}
\end{abstract}

\vfill

\vfill 
\paragraph*{Acknowledgements.} We would like to thank Dmitry Chistikov for suggesting to 
switch from partial derivatives to finite differences to simplify certain arguments in the proof of~\Cref{prop:monotone-decomposition}, Dario Fiore for pointing us to the work of Rivest, Shamir and Wagner on time-lock puzzles~\cite{Rivest96}, 
and Christoph Haase for pointing us to the work of Myasnikov, Ushakov and Won~\cite{MyasnikovUW12}. 
This work began while S~Hitarth and Guruprerana Shabadi were interns 
at IMDEA Software Institute.

\vspace{\baselineskip}
\noindent
\begin{minipage}{0.85\linewidth}
\paragraph*{Funding.}
This work is part of a project co-funded by the European Union (GA 101154447) 
and by MCIN/AEI (GA PID2022-138072OB-I00).
Views and opinions expressed are however those of the authors only and do not necessarily reflect those of the European Union or European Commission. Neither the European Union nor the granting authority can be held responsible for them.
\end{minipage}%
\begin{minipage}{0.15\linewidth}
\flushright
\vspace{-14pt}
\includegraphics[scale=0.25]{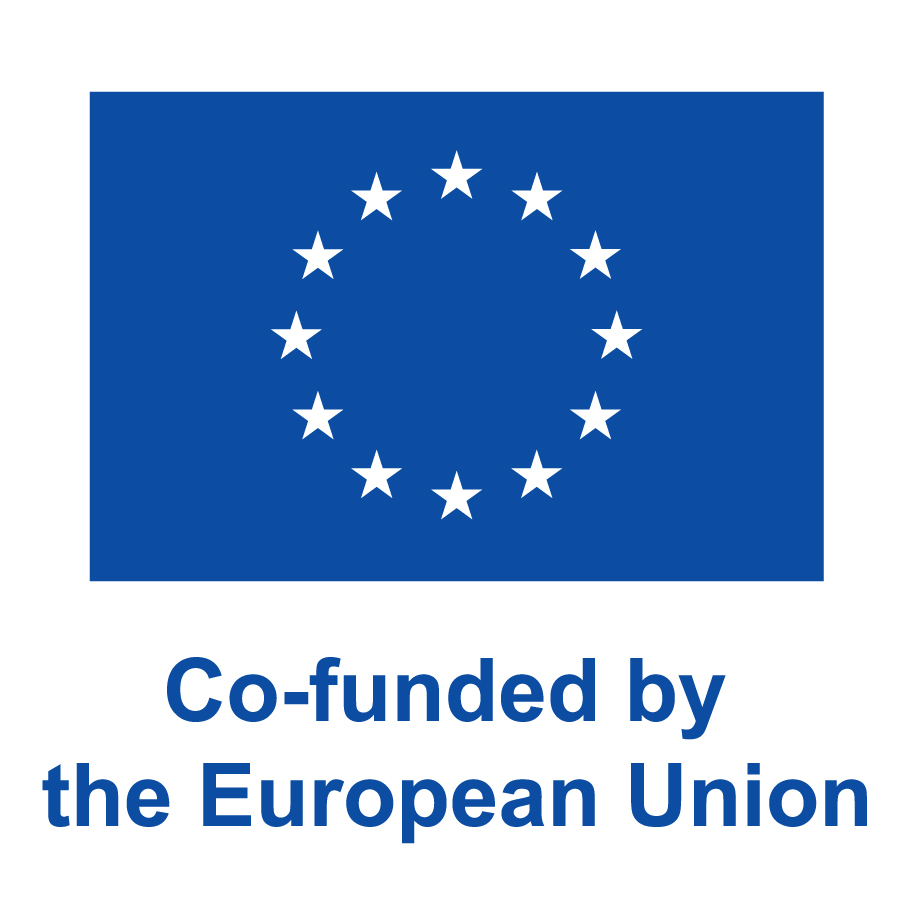}
\end{minipage}

\clearpage
\addtocontents{toc}{\protect\setcounter{tocdepth}{0}}

\section{Introduction}
\label{sec:intro}

\emph{Integer Linear Programming (ILP)}, the problem of determining an optimal (maximal or minimal)
value of a multivariate linear polynomial evaluated over the
integer solutions to a system of linear inequalities $A \cdot \vec x \leq \vec b$, 
offers one of the most versatile frameworks  for solving computational problems in operations research and computer science.
Summarizing the preface of~\emph{``50 Years of Integer Programming 1958--2008''}~\cite{FiftyYears}, over decades, a rich collection of methods 
for solving ILP have been developed, such as cutting-plane methods, branch-and-bound algorithms, and techniques from~polyhedral geometry. 
These developments have not only deepened our understanding of the structure of the problem and its complexity, but also have been translated into powerful solvers (e.g., \texttt{SCIP}, \texttt{CPLEX}, \texttt{Gurobi}) that can handle large-scale real-world instances very efficiently.

In this paper, we study the optimization problem of \emph{Integer Linear-Exponential Programming (ILEP)},
which extends ILP with the \emph{exponential function} $x \mapsto 2^x$ 
and the \emph{remainder function} $(x,y) \mapsto (x \bmod 2^y)$.
An instance of ILEP is a maximization (or minimization) problem%
\begin{align*}
    \text{maximize }& \tau(\vec x)\\  
    \text{subject to }& \tau_i(\vec x) \leq 0 \text{ for each } i \in \{1,\dots,k\}\\
        &\tau_i(\vec x) = 0 \text{ for each } i \in \{k+1,\dots,m\},
\end{align*}
where $\vec x$ is a vector of variables over the non-negative integers~$\N$, and $\tau, \tau_1,\dots,\tau_m$ are \emph{linear-exponential terms} of the form 
\begin{equation}
    \label{eq:exp-lin-term}
    \sum\nolimits_{i=1}^n \big(a_i \cdot x_i + b_i \cdot 2^{x_i} + \sum\nolimits_{j=1}^n c_{i,j} \cdot (x_i \bmod 2^{x_j})\big) + d,
\end{equation}
in which all \emph{coefficients} $a_i,b_i,c_{i,j}$ and the \emph{constant} $d$ are integers.
The system of constraints defined by the inequalities $\tau_i(\vec x) \leq 0$ 
and equalities $\tau_i(\vec x) = 0$
is an \emph{integer linear-exponential program}.

\begin{example}
    \label{example-1}
    To get a feel for this optimization problem, let us look at the instance
    \begin{center}
        \begin{minipage}{0.5\linewidth}
        \begin{align*}
            \textup{maximize } \tau(x,y) \,\coloneqq\,{}&\ 8x + 4y -(2^x+2^y)\\[5pt]
            \textup{subject to } \phi(x,y,z) 
            \,\coloneqq\,{}&\  y \leq 5\\
            &\ y \leq 2^x\\ 
            &\ 2^{z} \leq 2^{16} y\\ 
            &\ z = 3 \cdot x.
            \notag
        \end{align*}
        \end{minipage}%
        \begin{minipage}{0.5\linewidth}
            \centering
            \includegraphics[scale=0.28]{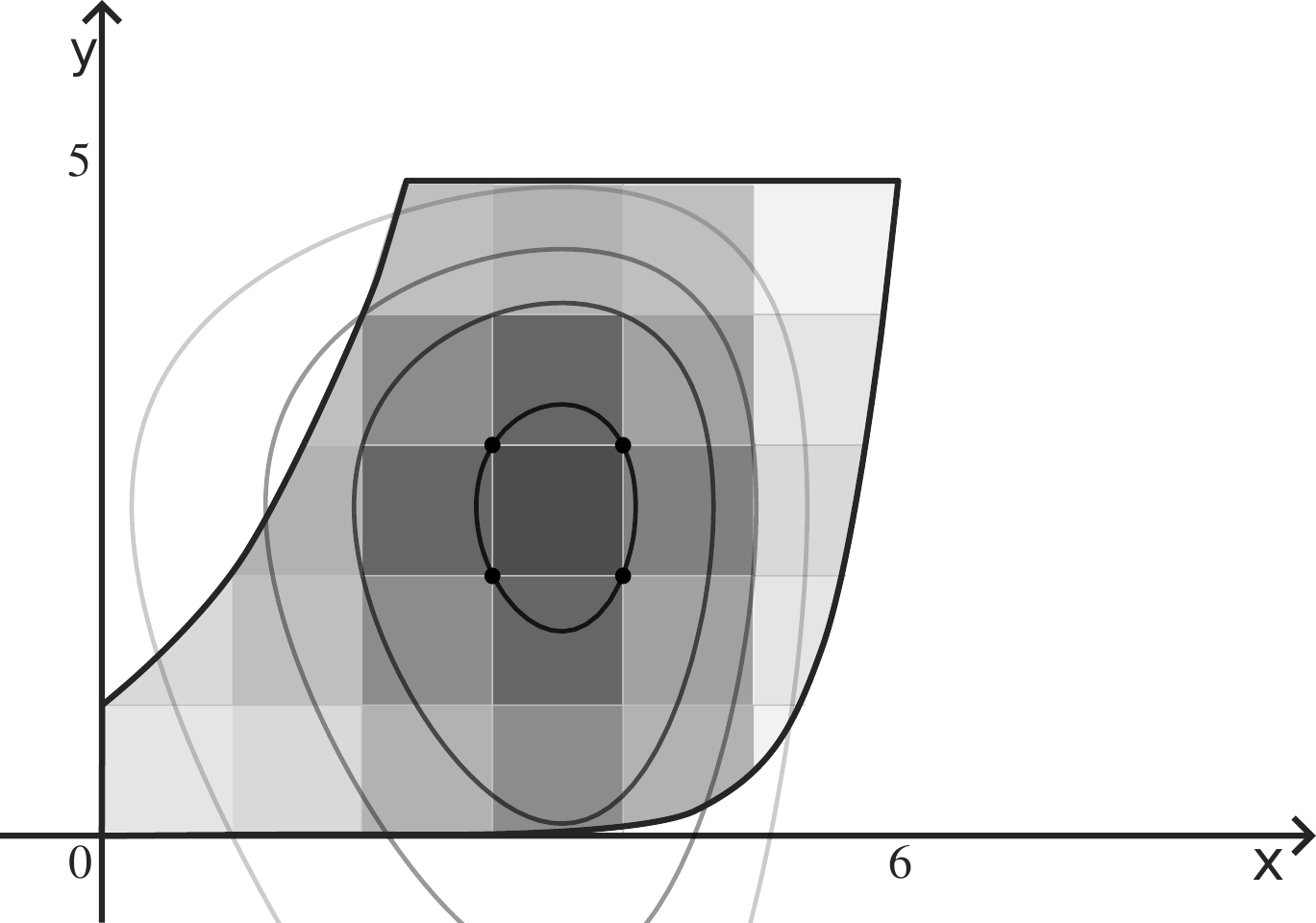}
        \end{minipage}
    \end{center}
    After projecting away the variable $z$, which is just a proxy for $3 \cdot x$, 
    the plot on the right shows a heat map of the objective function $\tau$ over the feasible region defined by~$\phi$
    (we only consider non-negative integer solutions). Visually, we see that any point 
    in $\{3,4\} \times \{2,3\}$ is optimal. 
\end{example}

Integer linear-exponential programming lacks two of the key properties that are central to ILP:
\begin{enumerate}
    \item In ILP, it is a classical fact that if an optimal solution exists, 
        then there is one whose bit size is polynomially bounded by the bit size of the input~\cite{BoroshT76,vonzurGathenS78}. This is not the case for integer linear-exponential programs, 
        where solutions may demand a non-elementary number of bits when represented in binary: 
        by setting $x_0=1$ and writing a sequence of constraints of the form $x_{i+1} = 2^{x_i}$, one can force~$x_i$ to be equal to the tower of $2$s of height~$i$. 
    \item In ILP, whenever (optimal) solutions exist, at least one lies near the boundary of the feasible region defined by the system of linear inequalities. (This fact is made more precise in~\Cref{example:ILP}.) 
    In ILEP, this geometric property no longer holds. 
    Intuitively, this can already be seen in the instance from~\Cref{example-1}, 
    where all optimal solutions lie near the center of the feasible region rather than near its boundary.
    We will revisit this observation in~\Cref{example:no-go-ILP}. 
    Finding optimal solutions despite this fact is a central difficulty addressed in the paper.
\end{enumerate}

Although solutions to integer linear-exponential programs may require astronomically 
large binary representations, the complexity of the \emph{feasibility problem} for ILEP, that is, the problem of checking if an instance has
a (not necessarily optimal) solution, is comparable to that of ILP.
Indeed, this problem was recently shown 
to be~\np-complete by Chistikov, Mansutti and \mbox{Starchak} in~\cite{ChistikovMS24}, who developed a 
non-deterministic polynomial-time procedure based on quantifier elimination. 
This implies that a short and polytime-time checkable certificate exists for at least one solution of a feasible system\footnote{While these certificates are not discussed explicitly in~\cite{ChistikovMS24}, they can be extracted from the accepting paths of the non-deterministic procedure.}.
In contrast, it is not known whether \emph{optimal} solutions 
can be represented efficiently, namely, by polynomial-size 
objects that can be verified as valid solutions in polynomial time.
This leads to the central question we explore in this paper:
\begin{center}
    \emph{Are there efficient representations for the optimal solutions to ILEP?}
\end{center}
As this introduction hopes to convey, answering this question yields a distinctive perspective on integer programming. The algebraic techniques we employ in this paper are, as far as we know, non-standard in the context of optimization, and they appear to be applicable to other extensions of ILP, such as quadratic~\cite{PiaDM17,Lokshtanov15,EibenGKO19} and parametric versions of integer programming~\cite{Shen18,BogartGW2017}. Furthermore, 
the representation 
we consider is quite natural 
and can be viewed as an extension of the class of~\emph{power circuits} introduced 
by Myasnikov, Ushakov and Won in~\cite{MyasnikovUW12}, which played 
a crucial role in resolving several questions in algorithmic 
group theory, most notably in establishing that the word problem for 
the one-relator Baumslag group lies in \ptime~\cite{MyasnikovUW11}. 
To our knowledge, this is the first application 
of power circuits within the context of integer programming.

\subsection{Succinct encoding of optimal solutions}
\label{subsec:succinct-encoding-optimal-solutions}

To address the central question posed above, we introduce a new representation of solutions
called \emph{Integer Linear-Exponential Straight-Line Programs}.
Let us begin by defining a \emph{Linear-Exponential Straight-Line Program}~(\emph{LESLP}) as a sequence~$\sigma \coloneqq {(x_0 \gets \rho_0,\, \dots\, ,\, x_n \gets \rho_n)}$
of variable assignments such that each expression $\rho_i$ ($i \in [0..n]$) has one of the 
following forms: $0$, ${x_j + x_k}$, 
$2^{x_j}$, or \emph{scaling expressions} $a \cdot x_j$, where the indices~$j,k \in [0..i-1]$
refer to previous assignments in the program,
and $a \in \Q$. 
The \emph{bit size} of $\sigma$ is defined as the number of symbols required to write it down, which includes encoding the indices $0,\dots,n$ in unary, and the rational coefficients in scaling expressions as pairs of integers $\frac{m}{g}$ with $g \geq 1$, encoded in binary.


We define $\sem{\sigma} \colon \{x_0,\dots,x_n\} \to \R$ as the map that assigns to each variable $x_i$ the value that the expression $\rho_i$ takes when evaluated using standard arithmetic. Note that $x_0$ always takes the value $0$.  
We call $\sigma$ an \emph{Integer Linear-Exponential Straight-Line Program} (\emph{ILESLP}) if all of its variables evaluate to integers.
For example, the following LESLP $\sigma$
\[ 
    \textstyle x_0 \gets 0,\ \
    x_1 \gets 2^{x_0},\ \ 
    x_2 \gets -1 \cdot x_1,\ \  
    x_3 \gets 2^{x_2},\ \
    x_4 \gets 2^{x_3},
\]
is not an ILESLP as~$\sem{\sigma}(x_3) = \frac{1}{2}$ and $\sem{\sigma}(x_4) = \sqrt{2}$.

Consider an instance $(\tau,\phi)$ of ILEP, where $\tau$ is the objective function (to be maximized or minimized)
and $\phi$ is an integer linear-exponential program.
An ILESLP $\sigma$ is a solution to $(\tau,\phi)$ whenever
\textit{(i)}~the set $\{x_0,\dots,x_n\}$ contains (at least) all the variables of $\tau$ and $\phi$, 
and \textit{(ii)} the variable assignment $\sem{\sigma}$ satisfies all constraints in~$\phi$. 
We now state our main theorem:

\begin{restatable}{theorem}{ThOneIntroduction}
    \label{theorem:small-optimum}
    If an instance of integer linear-exponential programming has an optimal solution, then it has one
    representable with a polynomial-size {\rm ILESLP}.
\end{restatable}

We defer giving an overview of the proof of~\Cref{theorem:small-optimum} to~\Cref{subsection:overview-theorem-one}.
Let us stress that, for the sake of a simpler exposition, we solely focus on integer linear-exponential programs 
with \emph{variables ranging over} $\N$. Our results can however be easily adapted to variables ranging over~$\Z$ 
by similar arguments as the ones given in~\cite[Sec.~8]{ChistikovMS24} for the feasibility problem.
That being said, the variables in ILESLPs must still range over $\Z$, 
and auxiliary variables not occurring in the instance of ILEP are necessary to succinctly encode a solution.
Consider for example the linear-exponential program $\phi(x,y,z) \coloneqq {x = k \land y = 2^x \land z = 2^{y} - 1}$, where~$k$ is a positive integer encoded in binary.
A (short) ILESLP~$\sigma$ representing the only solution to~$\phi$ is
\[ 
    x_0 \gets 0,\ \
    x_1 \gets 2^{x_0},\ \ 
    x \gets k \cdot x_1,\ \  
    y \gets 2^{x},\ \
    x_2 \gets 2^{y},\ \
    x_3 \gets -1 \cdot x_1,\ \ 
    z \gets x_2 + x_3,
\]
where $x_0,\dots,x_3$ are auxiliary variables, and $\sem{\sigma}(x_3)$ is negative.
Intuitively, it is not possible to have a short ILESLP in which all variables evaluate to non-negative integers, 
because the binary expansion of $2^{2^k}-1$ has doubly exponentially many $1$s with respect to the bit size of $\phi$.

\paragraph*{Power circuits.} 
In~\cite{MyasnikovUW12}, Myasnikov, Ushakov and Won 
consider a class of straight-line programs, which they refer to as~\emph{(constant) power circuits}, that feature the operations~$x+y$, $x-y$ and $x \cdot 2^y$, and the constant~$1$. In the paper, the authors 
develop several polynomial-time algorithms for manipulating such circuits. 
The main one is a normalization procedure that, among other things, reduces the operation $x \cdot 2^y$ 
to the simpler exponential function~$2^y$.
Given that power circuits are semantically restricted to integer-valued variables, ILESLPs 
thus represent a natural generalization that introduces scaling by rational constants via the expressions $a \cdot x$, with $a \in \Q$.
The need for rational coefficients is, in fact, already discussed in~\cite[Section~9.1]{MyasnikovUW12}, 
as we explain next.

Consider an integer linear-exponential program~$\phi$ whose constraints imply $3 \cdot x = 2^{2y}-1$ and require $y$ to be a positive integer of exponential magnitude in the size of $\phi$.
To see why any polynomial-size ILESLP $\sigma$ encoding a solution to $\phi$ must have a scaling expression with a non-integer coefficient, observe that for every $k \geq 1$, the number $\frac{2^{2k}-1}{3}$ is a positive integer, and moreover its binary representation is $1(01)^{k-1}$.
Since $\sem{\sigma}(y)$ is large, the binary expansion of $\sem{\sigma}(x)$ must then alternate between~$0$s and $1$s 
exponentially many times relative to the size of $\phi$.
However, one can show that an ILESLP with only integer coefficients (alternatively,~a power circuit) can only encode numbers whose binary expansion alternates between~$0$s 
and~$1$s at most polynomially many times in the bit size of the ILESLP. 
Therefore, $\sigma$ must either feature some non-integer coefficient, or be exponentially larger than $\phi$.

\subsection{Recognizing ILESLPs and when they encode solutions}
\label{subsection:intro:recognizing}

\Cref{theorem:small-optimum} indicates that the optimization problems of ILP and ILEP are close: 
while integers must be encoded more succinctly in the case of ILEP,
both problems admit short representations for optimal solutions. The first difference arises 
when we consider the problems of recognizing the set of ILESLPs, and of checking whether an ILESLP is a solution 
to an instance of ILEP.  

Consider a LESLP $\sigma \coloneqq (x_0 \gets \rho_0,\, \dots\, ,\, x_n \gets \rho_n)$. 
The snippet of code below decides whether~$\sigma$ is an ILESLP 
by testing whether $\sem{\sigma}(x_i) \in \Z$ 
iteratively on $i$ from $1$ to $n$. Such a test comes for free for additions: 
if $\sem{\sigma}(x_j)$ and $\sem{\sigma}(x_k)$ are integers, 
and $\sigma$ features $x_i \gets x_j + x_k$, then $\sem{\sigma}(x_i) \in \Z$. 
\begin{algorithmic}[1]
    \For{$i = 1$ to $n$}
    \Comment{during the $i$th iteration, we already know that $x_0,\dots,x_{i-1}$ are integers}
    \Statex \Comment{in the next two lines, ``$\textbf{assert}\, \phi$'' stands for ``$\textbf{if}\, \neg \phi\, \textbf{then}\, \textbf{return}\, \text{false}$''}
        \If{$\rho_i$ is of the form $2^x$}\label{snippet:line2}
            \textbf{assert} $\sem{\sigma}(x) \geq 0$
            \Comment{recall: $\sem{\sigma}(x) \in \Z$}
        \EndIf
        \If{$\rho_i$ is of the form $\frac{m}{g} \cdot x$}\label{snippet:line3}
            \textbf{assert} $\frac{g}{\gcd(m,g)}$ divides $\sem{\sigma}(x)$ 
        \EndIf 
    \EndFor
    \vspace{-5pt}
    \State \textbf{return} true
\end{algorithmic}

Given an LESLP $\sigma \coloneqq (x_0 \gets \rho_0, \dots, x_n \gets \rho_n)$, 
let us write $\semlast{\sigma}$ as a shorthand for $\sem{\sigma}(x_n)$.
Lines~\ref{snippet:line2} and~\ref{snippet:line3} of the above code only verify properties of variables whose values were already established to be integers
in earlier iterations of the \textbf{for} loop.
Therefore, the problem of checking if an LESLP is an ILESLP reduces to deciding the following two properties
of an input \underline{I}LESLP $\sigma$: 
\begin{description}
    \item[\posileslp:] Is $\semlast{\sigma} \geq 0$? \vspace{3pt}
    \item[\modileslp:] Is $\semlast{\sigma}$ divisible by $g$, for $g \in \N_{\geq 1}$ given in binary?
\end{description}

The problem \posileslp is the ``linear-exponential analogue'' of the well-known \textsc{PosSLP} problem, which involves straight-line programs featuring assignments~$x_i \gets x_j \cdot x_k$ in place of exponentiation, and whose complexity is still wide open~\cite{BurgisserJ24,BlaserDJ24}.
In contrast, the corresponding decision problem for power circuits is known to be decidable in polynomial time~\cite[Sec.~7.5]{MyasnikovUW12}. We show that this result carries over to the more general setting of ILESLPs:

\begin{restatable}{lemma}{TheoremPosInPtime}
    \label{theorem:pos-in-ptime}
    \posileslp can be decided in polynomial time.
\end{restatable}
In~\cite{MyasnikovUW12}, one notable feature of the previously-mentioned 
normalization procedure for power circuits is that it makes checking 
the sign trivial: once a circuit is in normal form, the sign 
of the encoded number is immediately evident from the structure 
of the circuit.
(Another key property is that power circuits representing the same 
number have the same normal form.)
While we believe that a similar normal form exists for ILESLPs, 
in this paper we instead provide a direct procedure for solving~\posileslp. 
Setting aside complexity considerations for now, the procedure originates 
from a simple idea.
Given an ILESLP (or power circuit) $\sigma$ where $\sem{\sigma}(z) = a \cdot 2^{\sem{\sigma}(x)} - b \cdot 2^{\sem{\sigma}(y)}$ for three variables $x,y,z$ and positive integers $a,b$, look at the distance $k \coloneqq \abs{\sem{\sigma}(x) - \sem{\sigma}(y)}$.
One possibility is for~$k$ to be at least $c \coloneqq \ceil{\log_2(\max(a,b))}$: 
the sign of $\sem{\sigma}(z)$ 
is then the sign of the coefficient~$a$ or $b$ corresponding to the larger variable among~$x$ and $y$. 
We can check $k \geq c$ by opportunely modifying $\sigma$ 
so as to be able to test $\sem{\sigma}(x) - \sem{\sigma}(y) - c \geq 0$ and $\sem{\sigma}(y) - \sem{\sigma}(x) - c \geq 0$ with two recursive calls to the algorithm for~\posileslp. 
If $k < c$ instead, $k$ is logarithmic in the bit size of~$\sigma$.
We can then compute~$k$: a na\"ive solution is to perform binary search on a suitable interval, repeatedly invoking the algorithm for~\posileslp on a modified ILESLP. 
Then, $\sem{\sigma}(z)$ has the same sign as either $a \cdot 2^{k} - b$ or $a - b \cdot 2^{k}$, depending on which of the two variables, $x$ or $y$, is larger. 
While our final polynomial-time procedure differs from this outline, the distinction between ``large distance'' and ``short distance'' remains central.

Turning to the problem~\modileslp, we show that it can be decided in polynomial time when having access to an integer factorization oracle. This is arguably the best we can hope for, as solving \modileslp in \ptime (in fact, even in~\bpp) would refute the \emph{Sequential Squaring Assumption}, a well-known cryptographic assumption 
put forward by Rivest, Shamir and Wagner~in~\cite{Rivest96}.\footnote{A~proof of this hardness result is provided for completeness in~Appendix~\ref{appendix:ssa}; 
see also~\cite{ChvojkaJSS21} for a further reference. It is worth noting that rational constants do not have any role in this proof: the problem is unlikely to be in \ptime even in the more restricted setting of power circuits.}

\begin{restatable}{lemma}{TheoremModInPFactoring}
    \label{theorem:mod-in-p-factoring}
    \modileslp is in $\ptime^{\factoring}$.
\end{restatable}

The main step in establishing~\Cref{theorem:mod-in-p-factoring} is showing that, even though $\sem{\sigma}(x)$ can be astronomically large,
we can still compute $\sem{\sigma}(x) \bmod \totient(g)$ in polynomial time using the factoring oracle, where $\totient$ stands for Euler's totient function. This allows us to then efficiently compute~$2^{\sem{\sigma}(x)} \bmod g$ using the exponentiation-by-squaring method~\cite[Ch.~1.4]{BressoudW08}.

As per all $\ptime^{\factoring}$ algorithms, given an input ILESLP~$\sigma$ and $g \in \N_{\geq 1}$,
there is a polynomial-sized set of small primes that, when provided as an advice,
enables running the algorithm deciding \modileslp in polynomial time, avoiding all calls to the factorization oracle. We will explicitly construct this set in Section~\ref{sec:deciding-mod}. Looking back at line~\ref{snippet:line3} of the above snippet of code,
note that the algorithm deciding~\modileslp has to be invoked only on divisors of the denominators~$g$ appearing in the rational coefficients of the LESLP~$\sigma$. This allows us to define a common set $\PP(\sigma)$ of polynomially-many small primes that suffices to decide in polynomial time all instances of \modileslp that are relevant when determining if a LESLP is an ILESLP. From~\Cref{theorem:pos-in-ptime,theorem:mod-in-p-factoring}, along with the fact that primality testing is in $\ptime$~\cite{AgrawalKS04}, we then establish the following result:

\begin{restatable}{proposition}{TheoremURecognition}
    \label{theorem:U-recognition}
    Given an {\rm{LESLP}}~$\sigma$ and~$\PP(\sigma)$, 
    one can decide in polynomial time if~$\sigma$ is an {\rm{ILESLP}}. 
    In order words, the set $U \coloneqq \{(\sigma,\PP(\sigma)) : \sigma \text{ is an \textup{ILESLP}}\}$ 
    is recognizable in polynomial time.
\end{restatable}

The set $U$ in~\Cref{theorem:U-recognition} represents the \emph{universe} of all certificates for ILEP. Since $\PP(\sigma)$ can be encoded using polynomially many bits relative to the size of $\sigma$, \Cref{theorem:small-optimum} 
implies that any instance of ILEP with an optimal solution 
has one representable by a polynomial-size element of~$U$.
\Cref{theorem:U-recognition} highlights a nuanced distinction between ILP and ILEP: certificates for the latter problem require some external objects (the sets $\PP(\sigma)$) which are introduced to achieve polynomial-time recognizability of the certificates, but are not inherently required to encode~solutions. 

Let us now consider the problem of checking whether a given $(\sigma,\PP(\sigma)) \in U$ is 
a solution to an instance of ILEP.
To verify if $\sigma$ satisfies an inequality of the form ${\sum_{i=1}^n \big(a_i \cdot x_i + b_i \cdot 2^{x_i} \big) + d \le 0}$, 
we first check that $\sem{\sigma}(x_i) \geq 0$ for all~${i \in [1..n]}$; 
as solutions are over~$\N$. We then append new assignments to $\sigma$, 
to obtain an ILESLP $\sigma'$ such that $\semlast{\sigma'} = {\sum_{i=1}^n \big(a_i \cdot \sem{\sigma}(x_i) + b_i \cdot 2^{\sem{\sigma}(x_i)}\big) + d - 1}$.
The ILESLP~$\sigma$ satisfies the inequality if and only if the algorithm for~\posileslp returns false when applied to~$\sigma'$. 
For the more general case of the linear-exponential terms from~\Cref{eq:exp-lin-term}, we must also account for the expressions $(x_j \bmod 2^{x_k})$ involving the remainder function. 
We show that these expressions are unproblematic (\Cref{computing-ileslp-xmod2y}): starting from $(\sigma,\PP(\sigma))$, we can compute in polynomial time an ILESLP $\sigma''$ such that $\semlast{\sigma''} = \sem{\sigma}(x_j) \bmod 2^{\sem{\sigma}(x_k)}$. 
Consequently, after appropriately updating the ILESLP, the verification proceeds similarly to the case without remainder functions. 
We emphasize that computing~$\sigma''$ 
requires access either to~$\PP(\sigma)$ or to an integer factoring oracle.

\begin{restatable}{proposition}{TheoremFastChecking}
    \label{theorem:fast-checking}
    Checking whether $(\sigma,\PP(\sigma)) \in U$ encodes a solution to an instance $(\tau,\phi)$ of ILEP 
    can be done in polynomial time in the bit sizes of $\sigma$ and $\phi$.
\end{restatable}

\subsection{Comparing values of the objective function without computing them}
\label{subsection:intro:comparing}

Continuing our comparison between ILP and ILEP, we need to address one last problem: the evaluation of the objective function. In ILP, the objective function~$\tau(\vec x)$, being a linear term, is trivial to evaluate: it suffices to perform a few additions and multiplications, and return the resulting integer, which is guaranteed to be of polynomial bit size with respect to the bit size of the solution and of~$\tau$. The property of~$\tau$ being polynomial-time computable is a common feature of all optimization problems belonging to the complexity class~\npo from~\cite{Ausiello99}. 
This property implies that maximizing (or minimizing) $\tau$ subject to an integer linear program~$\phi(\vec x)$ 
can be achieved through \emph{binary search} over a suitable interval~$[a..b] \subseteq \Z$ containing the optimal value of~$\tau$; repeatedly solving an instance of the \emph{feasibility} problem of ILP at each step of the search.
For example, the first query checks whether~${\phi(\vec x) \land \tau(\vec x) \geq \frac{b-a}{2}}$ 
is satisfiable, and updates the interval to $[a..\floor{\frac{b-a}{2}}]$ or $[\ceil{\frac{b-a}{2}}..b]$ accordingly to the answer.
When $a$ and $b$ are encoded in binary, polynomially many feasibility queries suffice to locate an optimal solution; that is, $\npo \subseteq \fptime^{\np}$.

In ILEP there seems to be no easy way to perform binary search over the set of numbers encoded by polynomial-size ILESLPs (\Cref{open:binary-search} in~\Cref{subsec:future-work} formalizes this issue).
However, given an instance $(\tau,\phi)$ of ILEP, 
we can still \emph{compare} the values of $\tau$ at two solutions $\vec s_1$ and $\vec s_2$, each encoded as an ILESLP, in polynomial time relative to the sizes of $\tau$, $\vec s_1$ and $\vec s_2$.
This is a direct consequence of the fact that~\posileslp is in $\ptime$ (\Cref{theorem:pos-in-ptime}):
to perform the comparison $\tau(\vec s_1) \leq \tau(\vec s_2)$, we construct an ILESLP $\sigma$ such that $\semlast{\sigma} = \tau(\vec s_2) - \tau(\vec s_1)$, 
and then use the algorithm for \posileslp to determine the sign of this difference.

As a way of summarizing our comparison between ILP and ILEP, we introduce an adequate complexity class, which we denote by~\npocmp. In this class, the requirement ``the objective function is computable in polynomial-time'' of~\npo is weakened to ``comparisons between values taken by the objective function can be performed in polynomial time''; see~\Cref{section:npocmp} for the formal definition of~\npocmp.
This relaxation forgoes the ability to search for the optimum via binary search;
and so instead of an inclusion with~$\fptime^{\np}$,~we have ${\npocmp \subseteq \fnp^{\np}}$.
From the above discussion, and~\Cref{theorem:small-optimum} and~\Cref{theorem:U-recognition,theorem:fast-checking}, 
we obtain: 
\begin{restatable}{corollary}{CorrILEPinNPOCMP}
    \label{corollary:ILEP-in-npocmp}
    The optimization problem for integer linear-exponential programs is in~\npocmp.
\end{restatable}

Of course, whether \npocmp should be considered a ``natural'' complexity class is open for debate 
and lies beyond the scope of this paper. Echoing Goldreich~\cite[Chapter~2.1.1.1]{Goldreich08}, understanding the true content of this class is challenging because, like~\npo, it is defined solely in terms of the ``external behavior'' (algorithmic properties)
instead of the ``internal structure'' of its problems. Nonetheless, at an intuitive level, \npocmp seems ``natural'' 
in the context of optimization problems whose solutions must be encoded succinctly, 
and where it is therefore unreasonable to require the objective function to produce, in polynomial time, an integer encoded in binary.

\subsection{Overview of the proof of Theorem~\ref{theorem:small-optimum}}
\label{subsection:overview-theorem-one}

To establish~\Cref{theorem:small-optimum}, the starting point is given by the
non-deterministic polynomial-time algorithm designed in~\cite{ChistikovMS24} for
solving the feasibility problem for ILEP (we give an overview of this procedure
in~\Cref{section:summary-procedure}). In a nutshell, this algorithm solves the
linear integer-exponential program by progressively obtaining linearly occurring
variables, which are eliminated with a procedure that combines 
Bareiss's algorithm for Gaussian elimination~\cite{Bareiss68}
with a quantifier elimination procedure for Presburger arithmetic~\cite{Pre29} 
(that is, the first-order theory of the structure~$\langle \N; 0,1,+,\leq \rangle$). 
This ``variable elimination
step'' only preserves the equisatisfiability of the formula; consequently, in
the setting of optimization, the algorithm may miss all optimal solutions. We
look closely at this issue, and show that the variable elimination step can be
strengthened to ensure that at least one optimal solution is preserved (provided
one exists). Furthermore, each
non-deterministic branch of execution can be associated with an ILESLP whose
size is polynomial in the sizes of the intermediate formulae produced during the
run. 
When the execution terminates successfully, this ILESLP encodes the computed solution.
Then, the final component of the proof involves analyzing the running time of the
algorithm.

\begin{algorithm}[t]
  \caption{\GaussOpt: A template for variable elimination.}
  \label{algo:gaussopt}
  \begin{algorithmic}[1]
    \Require $\vec x \colon$variables; \ $f\colon$an objective function; $\phi \colon$a system of constraints.
    \vspace{3pt}
    \While{some variable from $\vec x$ appears in $f$ or $\phi$} \label{metagauss:mainloop}
    \State $(a \cdot x = \tau) \gets \myguess \text{ an element in } \tests(\vec x, f,\phi)$
    \Comment{guesses an equality with $a \neq 0$}
    \label{metagauss:oracle-call}
    \State $(f,\phi) \gets \elimdisc{f}{\phi}{a \cdot x = \tau}$
    \Comment{subproblem in which $x = \frac{\tau}{a}$}
    \label{metagauss:new-sate}
    \EndWhile \label{metagauss:endmainloop}
    \State \textbf{return} $(f,\phi)$
  \end{algorithmic}
\end{algorithm}%

\paragraph*{Variable elimination.}
Without going into full-details, one can abstract the ``variable elimination step'' we seek to define into the template given in~\Cref{algo:gaussopt} (\GaussOpt). It describes a procedure that, 
given in input a vector of variables~$\vec x$ to be eliminated, an objective function $f$, and some system of constraints $\phi$, iteratively performs the following operations:
\begin{enumerate}
    \item Guess an equality ${a \cdot x = \tau}$ from a finite set~$\tests(\vec x, f, \phi)$ (line~\ref{metagauss:oracle-call}), 
    where $a \in \Z \setminus \{0\}$, $x$ is a variable in~$\vec x$ occurring in $\phi$ or~$f$, and $\tau$ is an expression over variables in $f$ or $\phi$ other than~$x$. 

    \item Apply an \emph{elimination discipline}~$\elimdisctxt$ (line~\ref{metagauss:new-sate}). 
    This operator updates~$f$ and $\phi$ to a new objective function and constraint system, representing the subproblem obtained by narrowing the search space to only those solutions where $x$ is set to $\frac{\tau}{a}$. 
\end{enumerate}
Slightly overloading terminology from computer algebra, we refer to elements $a \cdot x = \tau$ of~$\tests(\vec x, f, \phi)$ as~\emph{test points}, emphasizing that~\GaussOpt \emph{tests} the case where $x$ is set to $\frac{\tau}{a}$. The algorithm only explores solutions corresponding to such tests. Hence, if too few test points are used, the algorithm may fail to find any solution to some satisfiable formula, i.e., it might be \emph{incomplete}. Even when it is complete, it may still miss all optimal solutions, if none of them corresponds to some test point. Given a specific class of objective functions and constraint systems, one can therefore ask: \emph{how should the test points be chosen to ensure that the algorithm runs in non-deterministic polynomial time and explores at least one optimal solution?}

\begin{example}[ILP with divisibility constraints] 
    \label{example:ILP}
    Consider the optimization problem:
    \begin{equation}
        \label{eq:ip-instance} 
        \text{maximize } f(\vec y) 
        \text{ subject to } \phi(\vec y) \coloneqq \big(A \cdot \vec y \leq \vec b \land \textstyle\bigwedge_{i=1}^k m_i \divides \tau_i(\vec y)\big)\,,
    \end{equation}
    where $f$ is a linear polynomial, $A$ is an integer matrix,~$\vec b$ is an integer vector, and each $m_i \divides \tau_i$ is a \emph{divisibility constraint} featuring a non-zero divisor $m_i \in \Z$ and a linear polynomial $\tau_i$. Given $a,b \in \Z$,~$a \divides b$ is true whenever $a$ is a divisor of $b$. 
    From quantifier elimination procedures for Presburger arithmetic (see, e.g.,~\cite{Weispfenning90}), we know that defining the (finite) set of test points as
    \begin{align*}
        \tests(\vec y, f,\phi) &\coloneqq \!
        \left\{ \begin{aligned} 
            &a \cdot x = \tau - s : &\text{ the variable $x$ appears in $\phi$ or $f$},\\ 
            & & (a \cdot x - \tau) \text{ is either $-x$ or a row of $A \cdot \vec y - \vec b$,}\\
            & & a \neq 0 \text{ and } s \in [0..\abs{a} \cdot \lcm(m_1,\dots,m_k)-1]\\
        \end{aligned}
        \right\}
    \end{align*}
    ensures that a solution over $\Z$ is explored. 
    In essence, this set shows that a solution can always be found 
    by shifting the hyperplanes describing the feasible region defined by~$\phi$ or, 
    when $x$ appears only in $f$, by shifting the constraint $-x = 0$.
    In fact, in~\Cref{lemma:monotone-gaussian-elimination} (\Cref{section:ILEP-in-npocmp}) we will see that this set also guarantees exploration of an optimal solution, 
    due to the monotonicity of the linear objective $f$.

    Given $f$, $\phi$ and an equality $a \cdot x = \rho$ from $\tests(\vec y, f,\phi)$, 
    we can define $\elimdisctxt$ as the operator that
    replaces $x$ with $\frac{\tau}{a}$ in both $f$ and $\phi$ (performing basic manipulations to preserve the integrality of the coefficients in $\phi$), 
    and appends the divisibility constraint~$a \divides \tau$ to $\phi$.
    These updates mirror those performed by quantifier elimination procedures for Presburger arithmetic. Complexity-wise, this elimination discipline is suboptimal, as it causes the bit sizes of the integers in $\phi$ to grow exponentially in the number of eliminated variables. The results in~\cite{ChistikovMS24} show how to fix this issue by relying on Bareiss algorithm.
    We will rely on similar arguments in~\Cref{sec:efficient-variable-elimination}.
\end{example}

In the case of non-monotone objective functions, the instantiation of~\GaussOpt 
given in the above example fails to explore optimal solutions.   

\begin{example}
    \label{example:no-go-ILP}
    Consider the optimization problem
    \begin{center}
        \begin{minipage}{0.5\linewidth}
        \begin{align*}
            \text{maximize } f(x,y) \,\coloneqq\,{}&\ 8x + 4y -(2^x+2^y)\\[5pt]
            \text{subject to } \phi(x,y,z) 
            \,\coloneqq\,{}&\  0 \leq x \leq 6\\
            &\  0 \leq y \leq 5
            \notag
        \end{align*}
        \end{minipage}%
        \begin{minipage}{0.5\linewidth}
            \vspace{-18pt}
            \centering
            \includegraphics[scale=0.27]{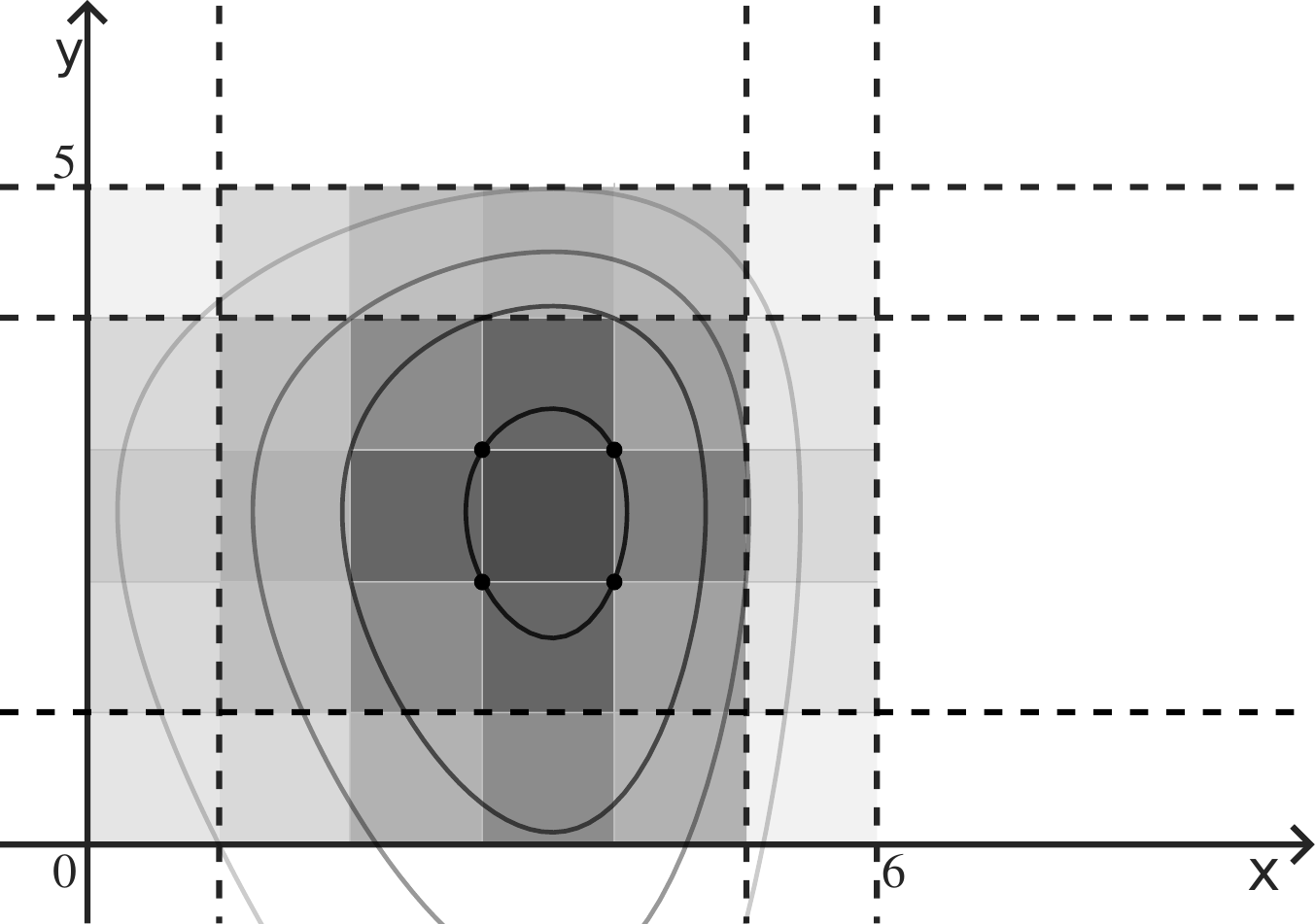}
        \end{minipage}
    \end{center}
    In the figure, vertical and horizontal lines represent the test points in the set $T \coloneqq \tests(\{x,y\},f,\phi)$ from \Cref{example:ILP}. 
    None intersect an optimal solution, and~$T$ is therefore insufficient to solve the problem of maximizing a linear-exponential~term subject to an integer linear program.
\end{example}

In order to instantiate~\GaussOpt to the context of ILEP, we must consider 
a class of objective functions represented as \emph{Linear-Exponential Arithmetic Circuits} (\emph{LEACs}). Informally, a LEAC~$C$ is an ILESLP that includes some \emph{free variables}~$\vec y$, that is, variables that appear in arithmetic expressions but are not themselves assigned any expression within the straight-line program. For a given \emph{output variable}~$z$ in $C$, 
the function represented by $C$ takes values for the free variables~$\vec y$ as input, 
evaluates all expressions in the circuit, and returns the integer corresponding to the expression assigned to~$z$. (LEACs are formally defined in~\Cref{subsec:setup-ilep}; see~\Cref{def:LEAC}.)
The function~$f$ from~\Cref{example:no-go-ILP} can be represented with a LEAC.

\paragraph*{Exploring optimal solutions.}
Returning to~\Cref{example:no-go-ILP}, we can ensure an optimal solution is explored by adding the equalities $x = 3$ and $x = 4$ to the set~$T$. One way of interpreting this addition is by looking 
at two subproblems: one where $x$ ranges over $[0..3]$, and another where it ranges over $[4..6]$. Within each of these intervals, the function $f$ is monotone in $x$ as both $x = 3$ and $x = 4$ are near a zero of the partial derivative $\frac{\partial f}{\partial x} = 8 - \ln(2) \cdot 2^x$ of $f$ in~$x$. Because of monotonicity, each subproblem can be tackled using the test points from~\Cref{example:ILP}, and the union of the test points of the two subproblems is exactly the set $T \cup \{x=3,x=4\}$. 

In essence, our instantiation of~\GaussOpt for ILEP adapts the above observation to the setting of LEACs. We show how to decompose the search space in such a way that the objective function encoded by the LEAC exhibits a form of monotonicity within each region of the decomposition, to then rely on the idea from~\Cref{example:ILP} that, for monotone functions, an optimal solution must occur near the boundary of the feasible region.
We refer to these decompositions as \emph{monotone decompositions}.
Since variables range over $\N$ instead of $\R$, we use \emph{finite differences} instead of derivatives: for a function $f(x, \vec y)$ in $1+d$ variables (in our case, a LEAC) and $p \in \N$, the \emph{$p$-spaced partial finite difference of $f$ with respect to~$x$}, denoted $\Delta_x^p[f]$, is the function $f(x + p, \vec y) - f(x, \vec y)$.
The function~$f$ is said to be \emph{$(x,p)$-monotone locally to a set $S \subseteq \N^{1+d}$} if there is a sign ${{\sim} \in \{<,=,>\}}$ 
such that, for every $(u,\vec v) \in S$ with $(u+p,\vec v) \in S$, we have $\Delta_x^p[f](u, \vec v) \sim 0$.
(Similarly to~\Cref{example:ILP}, our instantiation of~\GaussOpt adds divisibility constraints. The integer~$p$ in the finite difference corresponds to the least common multiple of the divisors in these constraints.)

\begin{example}
    Let $f$ and $\phi$ be as in~\Cref{example:no-go-ILP}. 
    The $1$-spaced partial finite difference in $x$ of~$f$ is $\Delta_x^1[f] = {8-2^x}$.
    This function is positive for $x \leq 2$, zero at ${x = 3}$, and negative for ${x \geq 4}$.
    Accordingly, the monotone decomposition of the search space features three regions, given by the sets of solutions to $\phi \land (x \leq 2)$, ${\phi \land (x = 3)}$, and $\phi \land (x \geq 4)$. The function $f$ is $(x,1)$-monotone locally to each region, 
    and we define the set~$\tests(\{x,y\},f,\phi)$ to include~$x =2$, $x=3$ and $x=4$.
\end{example}

\paragraph*{Complexity.} After defining the set of test points by relying on 
monotone decompositions, most of the technical effort required to
prove~\Cref{theorem:small-optimum} is devoted to ensuring that no exponential
blow-up occurs during the procedure. (In fact, this effort starts when
defining the monotone decompositions, as doing so uncarefully would already cause such a blow-up; see the discussion on
page~\pageref{lemma:rewriting-monotone-hyperplane-2}.) 
As already mentioned
in~\Cref{example:ILP}, an important step in avoiding exponential blow-ups is the
design of an efficient elimination discipline, which we base on a variation of
Bareiss algorithm. Once the elimination discipline is in place, a careful complexity
analysis, tracking several parameters of both the integer linear-exponential
programs and the LEACs, is required to show that the entire procedure runs
(non-deterministically) in polynomial time.

\begin{remark}
    As noted in page~\pageref{subsec:succinct-encoding-optimal-solutions}, 
    the techniques in this paper also appear applicable to 
    quadratic and parametric versions of integer programming. In a nutshell, this is because it is relatively simple 
    to define monotone decompositions in those contexts.
\end{remark}

\subsection{Open problems and future directions}
\label{subsec:future-work}

The results presented in this paper provide a positive answer to the question of whether optimal solutions to ILEP admit efficient representations, 
and offer what we believe to be a first satisfactory perspective on the computational differences between ILP and ILEP. Yet this perspective gives rise to several open problems, some of the most interesting of which we outline below.

Among the problems related to the complexity of ILEP, a fundamental question is whether our $\fnp^\np$ upper bound can be improved to $\fptime^\np$.

\begin{open}
    \label{open:ILEP-in-FPNP}
    Is the optimization problem for integer linear-exponential programs in $\fptime^\np$?
\end{open}

Based on our discussion in~\Cref{subsection:intro:comparing}, 
this problem can be settled with
an algorithm for performing binary search on a large set of ILESLPs. 
We formalize this objective in the following open problem
(here, $\card{S}$ stands for the cardinality of a set $S$):  

\begin{open}
    \label{open:binary-search}
    Let $S$ be the set of all ILESLPs of size at most $k$. 
    Is there an algorithm with runtime polynomial in $k$ 
    that, given as input~$\sigma_1,\sigma_3 \in S$, 
    computes $\sigma_2 \in S$ 
    such that the size of each of the sets 
    $S_1 \coloneqq \{ \sigma \in S: \semlast{\sigma_1} \leq \semlast{\sigma} \leq \semlast{\sigma_2}\}$ 
    and 
    $S_2 \coloneqq \{ \sigma \in S: \semlast{\sigma_2}\leq \semlast{\sigma} \leq \semlast{\sigma_3}\}$ 
    belongs~to~$\Omega(\card{S_1} + \card{S_2})$?
\end{open}

\noindent
Although missing a formal connection, 
the fact that~\modileslp is unlikely to lie in~\ptime (\Cref{appendix:ssa})
suggests that the above open problem may need to be relaxed 
to also allow for algorithms that run in polynomial time with access to an integer factoring oracle.
For example, this would apply to algorithms that first construct an LESLP~$\sigma_2$ 
of size at most $k$, to then check that $\sigma_2$ is an ILESLP.
Efforts to address~\Cref{open:binary-search} might begin by focusing on non-trivial subsets
of $S$. For instance, one could consider the problem 
of performing binary search on power circuits of size at most $k$, 
hence avoiding rational constants. 
A closely related open problem is 
the \emph{successor problem}: 
given~$\sigma_1 \in S$, find (if it exists)~$\sigma_2 \in S$ 
satisfying $\semlast{\sigma_2} = \min\{\semlast{\sigma} : \semlast{\sigma_1} < \semlast{\sigma}\}$.

The connection between~\modileslp and the Sequential Squaring Assumption (\Cref{appendix:ssa}) 
suggests that it is unlikely that integer linear-exponential programs with at least three variables 
can be solved in polynomial time. In contrast, ILP can be solved in polynomial time for any fixed number of variables~\cite{Lenstra83}. Can we say more about the complexity of ILEP in fixed~dimension?

\begin{open}
    \label{open:fixed-dimension}
    When the number of variables is fixed, can integer linear-exponential programming be solved in polynomial time with access to an integer factoring oracle?
\end{open}

\subsection{ILEP in context}
\label{intro:subsec:context}

For the interested reader, we conclude this overview by providing 
a broader perspective on~ILEP.
A notable trend in computer science sees integer linear programming 
being used not only in its classical applications (such as
scheduling, logistics, and finance) but also in automated reasoning and program
analysis. This is due in large part to the advances 
in \emph{Satisfiability Modulo Theory (SMT)} solvers~\cite{BarrettT18}.
These solvers bootstrap general (semi-)decision procedures for full first-order logical theories starting from tools that solve the so-called ``conjunctive fragment'' of these theories. 
For instance, ILP is the conjunctive fragment of Presburger arithmetic, 
and SMT solvers rely on tools for ILP to decide the feasibility problem of Presburger arithmetic~\cite{BarrettKT14}.

One challenge in applying Presburger arithmetic (and thus ILP) to areas such as
program analysis stems from limitations in its expressive power. The simplest
example of this comes from bit-vector analysis. Let us see a bit-vector $b$ of
length $n$ as the non-negative integer ${\sum_{i=0}^n b[i] \cdot 2^i}$, where
$b[i]$ denotes the $i$th entry of $b$. Presburger arithmetic lacks the ability
to express even the simple two-variable formula $\bit(b,y)$ asserting that $b[y]
= 1$, i.e., that the bit in position $y$ is set. Owing to these limitations,
recent research focuses on extending Presburger arithmetic and ILP with
additional predicates and functions while retaining
decidability~\cite{KarimovLNO025,BlanchardH24,DefossezHMP24}. A~prominent
extension is given by~\emph{Semenov arithmetic}~\cite{Semenov84}, which adds to
Presburger arithmetic the \emph{exponential function}~$x \mapsto 2^x$. Although
it still cannot express the formula~$\bit(b,y)$, Semenov arithmetic can reason
about \emph{bit sizes}: the formula~$2^y \leq x \land x < 2 \cdot 2^y$ binds $y$
to be the bit size of~$x$. Because of this, Semenov arithmetic has recently
found applications in the analysis of worst-case runtime complexity 
and program non-termination. 
The Loop Acceleration Tool
(\href{https://loat-developers.github.io/LoAT/}{\texttt{LoAT}},~\cite{FrohnG22})
relies on a procedure for the existential fragment of Semenov arithmetic,
implemented within the SMT
solver~\href{https://ffrohn.github.io/swine/}{\texttt{SwInE}}~\cite{FrohnG24}, to
support such analyses. The implementation in SwInE is based on the procedure
proposed in~\cite{BenediktCM23}, which was later improved
in~\cite{ChistikovMS24}. 

Extending Semenov arithmetic with the remainder function~$(x,y) \mapsto (x \bmod
2^y)$ yields a first-order theory known as~\emph{\mbox{B\"uchi-Semenov}
arithmetic}. From a logic viewpoint, ILEP is the conjunctive fragment of
B\"uchi-Semenov arithmetic. This theory is more expressive than Semenov
arithmetic: back to our toy example,~$\bit(b,y)$ is definable simply as $\exists
z: {z = y+1} \land {(b \bmod 2^z) - (b \bmod 2^y) \geq 1}$. Recent work shows
that B\"uchi-Semenov arithmetic has practical applications in solving string
constraints~\cite{ZhanEtAl23,DraghiciHM24}. For instance,~\cite{ZhanEtAl23}
studies string constraints with string-to-integer conversions and variables over
flat regular languages (STR$_{\text{S2I}}$ constraints). These constraints
naturally arise in symbolic execution of string manipulating
programs~\cite{AbdullaACDDJHLW20}. The authors of~\cite{ZhanEtAl23} show that
STR$_{\text{S2I}}$ constraints can be encoded in ILEP.  
To the best of our knowledge, this provides the only known proof that solving
STR$_{\text{S2I}}$ constraints~is~in~\np. 
The connection between ILEP and string solving also prompted the
study of extensions of ILEP featuring \emph{regular predicates} (constraints $x
\in R$ where $R$ is a regular expression), though the complexity of the
feasibility problem for these extensions ceases to be in~\np and becomes
\pspace-hard~\cite{DraghiciHM24,Starchak25}. 

It is worth noting that, at the
time of writing this paper, all existing tools for Semenov and B\"uchi-Semenov
arithmetic, such as those stemming from~\cite{ZhanEtAl23,FrohnG24}, are limited
to providing yes/no answers or binary-encoded solutions. In this setting, the
ILESLPs studied in this paper offer what is arguably the most natural
certificate format these tools~could~use.

\clearpage
\fancyhead[R]{{\color{gray}Table of contents (with an overview)}}
{\renewcommand{\contentsname}{Table of contents}
\label{toctoc}
\tableofcontents
\addtocontents{toc}{\protect\setcounter{tocdepth}{1}}}

\clearpage
\newcommand{\TableNotationTitle}{Table of notation}
\fancyhead[R]{\color{gray}\TableNotationTitle}
\section*{\TableNotationTitle}
\addcontentsline{toc}{section}{\TableNotationTitle}
This list is non-exhaustive and includes only symbols that are not local 
to a particular context, such as a single proof. Entries without a page number appear
on the same page as the preceding entry.
\bgroup
\def\arraystretch{1.1}
\begin{longtable}{p{0.18\textwidth}p{0.68\textwidth}>{\raggedleft\arraybackslash}p{0.061\textwidth}}
\endfirsthead
\endhead
\endfoot
\endlastfoot
\multicolumn{3}{@{}l}{\textbf{Basic mathematical notation}} \\[0.5ex]
$[a..b]$ & Set $\{n \in \Z : a \leq n \leq b\}$ of integers between $a$ and $b$ &\textbf{\pageref{prelim:basic-notation}}\\
$\card{S}$ & Cardinality of a finite set $S$ \\
$\card{\vec x}$ & Dimension (number of entries) in the vector $\vec x$\\
$\X$ & Countable set of variables\\ 
$\nu$ & Often a map from a subset of $\X$ to $\N$\\
$\nu_1 + \nu_2$ & Pointwise addition of maps $\nu_1 \colon X_1 \to \N$ and $\nu_2 \colon X_2 \to \N$
&\textbf{\pageref{prelim:ILEP}}\\
$\vec e_i^d$ & $i$-th vector of canonical basis of $\R^d$ &\textbf{\pageref{par:some-notation-secthree}}\\
$\totient(n)$ & Euler's totient function &\textbf{\pageref{equation:compute-totient-via-factorization}} \\
$\odd(a)$ & Largest odd factor of $a \in \N_{\geq 1}$ &\textbf{\pageref{sec:deciding-mod}}\\
\\[-2ex]
\multicolumn{2}{@{}l}{\textbf{Integer linear-exponential terms}}\\[0.5ex]
$\tau(\vec x)$ & Linear-exponential term with integer coefficient over variables $\vec x$\\
& General form: $\tau = \sum_{i=1}^n (a_i \cdot x_i + b_i \cdot 2^{x_i} + \sum_{j=1}^n c_{i,j} \cdot (x_i \bmod 2^{x_j})) + d$&\textbf{\pageref{part:small-ILESLP}}\\
$\nu(\tau)$ & Evaluation of $\tau$ on a map $\nu \colon X \to \N$\\
$\onenorm{\tau}$ & 1-norm: $\sum_i (\abs{a_i} + \abs{b_i} + \sum_j \abs{c_{i,j}}) + \abs{d}$\\
$\linnorm{\tau}$ & Linear norm: $\max\{ \abs{a_i}, \abs{c_{i,j}} : i,j \in [1..n] \}$ 
&\textbf{\pageref{item:linearnorm}}\\
\\[-2ex]
\multicolumn{2}{@{}l}{\textbf{ILEPs: Integer Linear-Exponential Programs}}\\[0.5ex]
$\phi(\vec x)$ & Conjunction of constraints of the form \(\tau(\vec x) \leq 0\) or $\tau(\vec x) = 0$ &\textbf{\pageref{ref:parameters-for-complexity}}\\
$d \divides \tau$ & Divisibility constraint: $d \in \N$ is a divisor of $\tau$ \\
$\card{\phi}$ & Number of constraints in $\phi$ \\
$\vars(\phi)$ & Set of all variables occurring in $\phi$ \\
$\fterms(\phi)$ & Set of all terms $\tau$ in inequalities~$\tau \leq 0$ or equalities $\tau = 0$ of $\phi$ \\
$\onenorm{\phi}$ & $\max\{\onenorm{\tau} : \tau \in \fterms(\phi)\}$ \\
$\fmod(\vec x, \phi)$ & LCM of all divisors of divisibility constraints with variables from~$\vec x$ \\
$\fmod(\phi)$ & Same as $\fmod(\vec x, \phi)$ when assuming all variables in $\phi$ to be from $\vec x$ \\
$\tau\sub{\frac{\tau'}{a}}{b \cdot x}$ & Ad-hoc substitution of $b \cdot x$ by $\frac{\tau'}{a}$ in a term $\tau$ &\textbf{\pageref{enum:ad-hoc-sub:i1}}\\
$\phi\sub{\frac{\tau'}{a}}{b \cdot x}$ & Ad-hoc substitution applied to all terms in~$\phi$\\
$\lst(\phi,\theta)$ & Set of least significant terms of $\phi$, for a variable ordering $\theta$ &\textbf{\pageref{item:lst}}\\
\\[-2ex]
\multicolumn{2}{@{}l}{\textbf{ILESLPs: Integer Linear-Exponential Straight-Line Programs}}\\[0.5ex]
$\sigma$ & ILESLP: a sequence of assignments  $(x_0 \gets \rho_0, \ldots, x_n \gets \rho_n)$ where each $\rho_i$ is of the form $0$, $\rho_j + \rho_k$, $\frac{m}{g}\rho_j$, or~$2^{\rho_j}$, with $j,k \in [0..i-1]$. &\textbf{\pageref{subsec:succinct-encoding-optimal-solutions}}\\
$\sem{\sigma}$ & Map from $\{x_0, \dots, x_n\}$ to $\R$ assigning values to variables\\
$\semlast{\sigma}$ & Shorthand for $\sem{\sigma}(x_n)$ &\textbf{\pageref{theorem:pos-in-ptime}}\\
$\sem{\sigma}(E)$ & Evaluation of expression $E = \sum_{j \in J} a_j \cdot 2^{x_j}$ on the map~$\sem{\sigma}$ &\textbf{\pageref{part:deciding-properties-ILESLP}}\\
$e(\sigma)$ (resp., $d(\sigma)$) & Absolute value of the product of all non-zero numerators (resp., denominators) occurring in rationals $\frac{m}{g}$ of expressions $\frac{m}{g}\rho_j$ of~$\sigma$ \\
$\nu_\sigma(x)$ & The function $\varphi(\odd(x \cdot d(\sigma)))$ &\textbf{\pageref{sec:deciding-mod}} \\
$\nu_\sigma^k$ & $k$-th iterate of the function~$\nu_\sigma$ \\
$\PP(\sigma,g)$ & $\{ p \text{ prime} : p \text{ divides either $d(\sigma)$ or $\nu_\sigma^k(g)$, for some $k \in [0..n-2]$}\}$ &\textbf{\pageref{equation:PPsigmag}}\\
$\PP(\sigma)$ & The set of primes $\PP(\sigma,d(\sigma) \cdot \nu_\sigma(1))$ &\textbf{\pageref{def-npocmp:universe}}\\
\\[-2ex]
\multicolumn{2}{@{}l}{\textbf{LEACs: Linear-Exponential Arithmetic Circuits}}\\[0.5ex]
$C$ & LEAC: similar to an ILESLP, but variables are allowed to occur free
&\textbf{\pageref{def:LEAC}}\\
$\vars(C)$ & Set of free variables in (the LEAC) $C$ &\textbf{\pageref{para:instances-for-opt-ilep}}\\
$\mu_C$, $\eta_C$ & Denominators in $C$ \\
$\xi_C$ & Sum of absolute values of coefficients in $C$ \\
$\objfun{C}{x_m}$ & Objective function given by $C$ with respect to a variable~$x_m$\\
$\Psi(C)$ & Formula implied by constraints built from $C$\\
\\[-2ex]
\multicolumn{2}{@{}l}{\textbf{Finite differences and monotonicity}}\\[0.5ex]
$\Delta_i^p[f]$ or $\Delta_x^p[f]$ & $i$-th $p$-spaced partial finite difference: $f(\vec x + p \cdot \vec e_i) - f(\vec x)$&\textbf{\pageref{section:ILEP-in-npocmp}}\\
$(i,p)$-periodic & Property of a set satisfying certain periodicity conditions \\
$(i,p)$-monotone & A function $f$ having a consistent sign of $\Delta_i^p[f]$ on a given set\\
\\[-2ex]
\multicolumn{3}{@{}l}{\textbf{Algorithms and notation in algorithms}} \\[0.5ex]
$\vec r$ & Vectors of remainder variables &\textbf{\pageref{cms:summary:step1}}\\ 
$\vec q$ & Vector of quotient variables &\textbf{\pageref{lemma:split:inequalities}}\\
$\ast$ & Non-deterministic choice &\textbf{\pageref{algo:btp:line-shift-tau2}}\\
\OptILEP & Procedure for the optimization problem of ILEP &
\textbf{\pageref{section:proof-monotone-decomposition},\pageref{pseudocode:opt-ilep}}\\
$\theta$ & Ordering $2^{x_n} \geq \dots \geq 2^{x_1} \geq 2^{x_0} = 1$ \\
\preleac & Similar to a LEAC, used for loop invariant of \OptILEP &\textbf{\pageref{def:pre-LEAC}}\\
\GaussOpt & Eliminates linearly occurring variables &\textbf{\pageref{algo:gaussopt},\pageref{algo:gaussopt-instantiated}}\\
$\tests(\cdot, \cdot, \cdot)$ & Returns the set of test points used for variable elimination\\
$\elimdisctxt(\cdot, \cdot, \cdot)$ & Elimination discipline operator\\
\\[-2ex]
\multicolumn{2}{@{}l}{\textbf{Abbreviations used in the context of monotone decompositions for ILEP}}\\[0.5ex]
$\theta_k$ & Ordering $(2^{x_{n-k}} \geq \dots \geq 2^{x_0} = 1)$&\textbf{\pageref{section:proof-monotone-decomposition}}\\
$\vec x_k$ & Vector $(x_{n-k}, \dots, x_n)$ of previously eliminated variables, plus $x_{n-k}$\\
$\vec y_k$ & Vector $(x_0, \dots, x_{n-k})$ of variables that are yet to be eliminated\\
$\vec r_k$ & Vector $(r_{n-k},\dots,r_n)$ of remainder variables\\
$\vec q_k$ & Vector $(q_{n-k},\dots,q_n)$ of quotient variables\\
$\vec q_{[\ell,k]}$ & Vector $(q_{n-k}, \dots, q_{n-\ell})$ of quotient variables\\
$u$ & Often a proxy for $2^{x_{n-k}-x_{n-k-1}}$ &\textbf{\pageref{objcons:i1}}\\
$\inst{\gamma}{\psi}$ & Conjunction of a linear program with divisions $\gamma$ and a linear-exponential program with divisions $\psi$; it satisfies further properties\\
$\objcons_k^\ell$ & Family of pairs $(C,\inst{\gamma}{\psi})$ considered for the monotone decompositio; it satisfies many technical properties\\ 
\end{longtable}
\egroup

\clearpage
\newcommand{\PartITitle}{Polynomial-size {ILESLPs} for optimal solutions}
\fancyhead[R]{{\color{gray}Part I: \PartITitle}}
\part{\PartITitle}\label{part:small-ILESLP}
\addtocontents{toc}{This part of the paper establishes~\Cref{theorem:small-optimum}.
This is the longest part of the paper, due to the many technical details that
must be resolved in order to obtain a proof of the theorem. \emph{An Advice:}
The reader should consider skipping the proofs on a first reading; the
surrounding text should suffice to convey the intuition behind the most of the
constructions involved in the proofs. An exception to this is the proof
of~\Cref{prop:monotone-decomposition}, which we recommend skimming during the
first pass.\par}

This first part of the paper is fully devoted to proving~\Cref{theorem:small-optimum}.
After introducing some preliminary definitions and notation (see below), 
we begin (in~\Cref{section:summary-procedure}) with a high-level overview 
of the algorithm from~\cite{ChistikovMS24} for solving the feasibility problem for ILEP. 
In particular, we expand on the description given in~\Cref{subsection:overview-theorem-one}, 
identifying the specific step ---which we refer to as the ``variable elimination step''--- where the non-deterministic executions of this algorithm may fail 
to cover optimal solutions. We also explain how each execution is ultimately constructing an ILESLP.  

In~\Cref{section:ILEP-in-npocmp}, we present a framework for deriving a variable elimination step tailored for optimization. The framework relies on splitting the search space into regions within which the objective function is (in some sense) monotone. An optimal solution can then be found by examining points that are close to the boundary of these regions. 
\Cref{section:proof-monotone-decomposition} instantiates this framework to ILEP. This instantiation reveals a set of additional constraints, beyond the ones required to solve the feasibility problem, that are required to characterize the regions of the decomposition. 

The results in \Cref{section:proof-monotone-decomposition} carry over 
to~\Cref{sec:efficient-variable-elimination}, where we implement the optimum-preserving variable elimination step. Ensuring that the overall procedure runs (non-deterministically) in polynomial time requires great care. To this end, we revisit the arguments from~\cite{ChistikovMS24} concerning the integration of Bareiss algorithm into the quantifier elimination procedure of Presburger arithmetic, and show that the constraints added by the monotone decomposition retain enough structure to allow a suitable variation of Bareiss algorithm to be successfully implemented.

Finally, \Cref{sec:putting-all-together} presents the complete optimization procedure for integer linear-exponential programming. From the correctness and complexity analysis of this procedure, we conclude that its output is a polynomial-size ILESLP, thereby proving~\Cref{theorem:small-optimum}.

\medskip
We now present the preliminaries for this part of the paper.
Some of the concepts introduced here reiterate those from~\Cref{sec:intro}, albeit given in a slightly more formal manner.

\paragraph*{Basic notation.}\label{prelim:basic-notation}
For $a \in \R$, we write $\abs{a}$, $\ceil{a}$, and $\log a$ for the \emph{absolute value}, \emph{ceiling},
and (if $a > 0$) the \emph{binary logarithm} of $a$.
All numbers encountered by our algorithm are encoded in binary;
assuming that $n \in \Z$ is represented using $\ceil{\log(\abs{n}+1)}+1$ bits.
For $a,b \in \R$, we write $[a..b]$ to denote the set $\{n \in \Z : a \leq n \leq b\}$.
Vectors are denoted using boldface letters, as in $\vec x$ or~$\vec y$. 
We write $\card \vec x$ for the number of entries in~$\vec x$; 
similarly, $\card S$ stands for the cardinality of a finite set $S$.


\paragraph*{Integer Linear-Exponential Terms.}\label{prelm:lin-exp-terms}
A \emph{linear-exponential term}~$\tau$ is an expression
\begin{equation*}
  \label{eq:a-term}
  \sum\nolimits_{i=1}^n \big(a_i \cdot x_i + b_i \cdot 2^{x_i} + \sum\nolimits_{j=1}^n c_{i,j} \cdot (x_i \bmod 2^{x_j})\big) + d,
\end{equation*}
where $a_i,b_i,c_{i,j} \in \Z$ are the \emph{coefficients} of the term and $d \in \Z$ is its \emph{constant}. 
If all $b_i$~and~$c_{i,j}$ are zero then the term is said to be \emph{linear}.
If $a_i \neq 0$, we call $a_i \cdot x_i$ a \emph{linear occurrence} of $x_i$.
If $b_i = 0$, we say that $x_i$ \emph{does not occur in exponentials};
this is weaker than saying that $x_i$ \emph{only occurs linearly}, as in this case we also have $c_{i,j} = 0$ for all $j \in [1..n]$.
We assume all variables used in linear-exponential terms to belong to a totally-ordered countable set $\X$, 
and write $\tau(\vec x)$ if all variables in the term~$\tau$ are from the vector (or set)~$\vec x$. 
The $1$-norm of $\tau$ is defined as $\onenorm{\tau} \coloneqq \sum_{i=1}^n (\abs{a_i} + \abs{b_i} + \sum_{j=1}^n \abs{c_{i,j}}) + \abs{d}$. 
The size of $\tau$ is defined as the number of symbols needed to write down the term, assuming that integers are encoded in binary, and that the $k$th variable in the ordering of~$\X$ requires $k$ bits.
Given a map $\nu \colon X \to \N$, where $X$ is a subset of $\X$ including the variables in $\vec x$, we write $\nu(\tau)$ for the integer obtained by \emph{evaluating} $\tau$ \emph{on} $\nu$, that is, replacing every variable $x$ occurring in $\tau$ with $\nu(x)$, 
and evaluating all operations in the resulting term.

For a variable $y \in \X$ and $b \in \N$,
we write $[y \mapsto b]$ for the map $\nu$ with domain $\{y\}$ 
and such that $\nu(y) = b$.
Let $\nu_1 \colon X_1 \to \N$ and $\nu_2 \colon X_2 \to \N$ be two maps.
The expression $\nu_1 + \nu_2$ defines the map $(\nu_1 + \nu_2) \colon X_1 \cup X_2 \to \N$ assigning $\nu_1(x) + \nu_2(x)$ to every $x \in X_1 \cup X_2$, 
where we assume $\nu_i(x) = 0$ whenever~${x \not \in X_i}$.
Therefore, $\nu + [y \mapsto b]$ stands for the map obtained from $\nu$ by adding $b$ 
to the value given to $y$ (again, assuming $\nu(y) = 0$ if $y \not\in X$).

\paragraph*{Integer Linear-Exponential Programs.}\label{prelim:ILEP}
A \emph{(integer) linear-exponential program}~$\phi$ is a conjunction of constraints $\tau = 0$ and $\tau \leq 0$, where $\tau$ is a linear-exponential term. 
If all terms are linear, then~$\phi$ is an \emph{(integer) linear program}.
We sometimes diverge from this syntax, but the intended meaning of the constraints should always be clear from the context. For instance, 
we sometimes write $\tau_1 \leq \tau_2$ as a shorthand for $\tau_1 - \tau_2 \leq 0$, 
and $\tau_1 < \tau_2$ as a shorthand for $\tau_1 - \tau_2 +1 \leq 0$. 
We write $\phi(\vec x)$ when the free variables of~$\phi$ are from the vector~$\vec x$.

While linear-exponential programs only feature equalities and inequalities, 
symbolic procedures for ILEP, such as the one developed in~\cite{ChistikovMS24}, require the introduction of additional \emph{divisibility constraints}~${d \divides \tau}$, 
where $\tau$ is a linear-exponential term, $d \in \N$ is non-zero, and $\divides$ is the \emph{divisibility predicate},
$\{(d,n) \in \Z \times \Z : \text{$n = k \cdot d$ for some $k \in \Z$} \}$. 
Without loss of generality, we assume all integers in the term $\tau$ to belong to $[0..d-1]$; 
our procedures will tacitly enforce this assumption by reducing all integers modulo $d$.
We say that the linear-exponential program is \emph{with divisions} if we allow divisibility constraints to occur in it.
For simplicity of the presentation, we also sometimes consider arbitrary \emph{formulae} 
from B\"uchi-Semenov arithmetic. In this theory, linear-exponential programs with divisions 
are extended to include the standard features of first-order logic, such as 
conjunction~$(\land)$, disjunction~$(\lor)$, negation~$(\lnot)$, implication~$(\!{\implies}\!)$ and first-order quantification~$(\forall$~and~$\exists)$.
For example, in the forthcoming sections we will often write equalities $u = 2^{x-y}$, which should be seen 
as shortcuts for formulae $\exists z\,(u = 2^z \land z = x - y)$, where $z$ is a fresh variable.
Note that, since we are only interested in non-negative \emph{integer} solutions (see below), $u = 2^{x-y}$ implies $x \geq y$.

Let $\phi$ be a linear-exponential program with divisions.
We write: 
\begin{itemize}
    \item\label{ref:parameters-for-complexity} $\card{\phi}$ for the number of constraints (inequalities, equalities and divisibility constraints) in $\phi$;
    \item $\vars(\phi)$ for the set of all variables occurring in $\phi$;
    \item $\fterms(\phi)$ for the set of all terms $\tau$ occurring in inequalities $\tau \leq 0$ or equalities $\tau = 0$ of $\phi$;
    \item $\onenorm{\phi} \coloneqq \max\{\onenorm{\tau} : \tau \in \fterms(\phi)\}$;
    \item given a vector $\vec x$ of variables, $\fmod(\vec x, \phi)$ for the least common multiple
    of the divisors $d$ of the divisibility constraints $d \divides \tau$ of $\phi$ in which 
    at least one variable from $\vec x$ occur (with a non-zero coefficient).
    We omit $\vec x$, and simply write $\fmod(\phi)$, when considering all variables in~$\phi$. 
\end{itemize}
The size of $\phi$ is defined as the number of symbols required to write it down (following the same assumptions used for defining the size of a term).

A map $\nu \colon X \to \N$, where $X$ is a finite subset of $\X$, 
is a \emph{solution} to a linear-exponential program with divisions~$\phi$~whenever (i) $X$ includes all variables occurring in $\phi$, 
and (ii) replacing each variable $x$ in $\phi$ with $\nu(x)$ 
lead to all constraints (inequalities, equalities and divisibilities) being satisfied.
For convenience, we sometimes see the set of solutions to $\phi$ not as a set of maps but as a subset $S \subseteq \N^d$, 
where $d$ is the number of variables in $\phi$. The $i$th entry of each vector in $S$ corresponds 
to the $i$th variable occurring in $\phi$, with respect to the total order of the set~$\X$. 

\paragraph*{Integer Linear-Exponential Programming (ILEP).}
By \emph{Integer Linear-Exponential Programming} 
we mean solving the maximization  
problem (or the analogous minimization problem)
\begin{center}
  \text{maximize} $\tau(\vec x)$ \text{subject to} $\phi(\vec x)$,
\end{center}
where $\tau$ is a linear-exponential term (the \emph{objective function}) 
and $\phi$ is a linear-exponential program (without divisions).
Unless otherwise stated, we stress that all the variables 
in an instance of integer linear-exponential programming range over the natural numbers. 

A map $\nu \colon X \to \N$,
is a solution to an instance of integer linear-exponential programming whenever $X$ includes all variables occurring in $\tau$ and $\phi$, and $\nu$ is a solution to $\phi$.
The \emph{value} of the objective function $\tau$ for the solution $\nu$ is 
the integer $\nu(\tau)$.

\paragraph*{Substitutions.} For technical reasons, we need an ad-hoc form of term substitution. We denote such a substitution with $\sub{\frac{\tau}{a}}{b \cdot x}$, where $\tau$ is a linear-exponential term, $x$ is a variable, and $a$ and $b$ are two non-zero integers. 
(This substitution can be interpreted as enforcing the equality $a \cdot b \cdot x = \tau$.)
When applied to a linear-exponential term $\rho$, the resulting term $\rho\sub{\frac{\tau}{a}}{b \cdot x}$ is constructed as follows:
\begin{enumerate}
  \item\label{enum:ad-hoc-sub:i1} Multiply every integer in $\rho$ by $\abs{a}$.
  \item\label{enum:ad-hoc-sub:i2} Consider the linear occurrence of $x$ in $\rho$ (if there is one). Try to factorize its coefficient as $a \cdot b \cdot c$, for some non-zero $c \in \Z$. If successful, 
  replace $a \cdot b \cdot c \cdot x$ with $c \cdot \tau$.
\end{enumerate}
Observe that, to eliminate $x$ using this substitution, we need to ensure that it only occurs linearly in~$\rho$, and that its coefficient is divisible by $b$.
We omit $a$ and/or $b$ from $\sub{\frac{\tau}{a}}{b \cdot x}$ when they are equal to one, writing for instance $\sub{\tau}{x}$ instead of $\sub{\frac{\tau}{1}}{1 \cdot x}$.

We will also need to simultaneously apply multiple substitutions to terms. Consider distinct variables $x_1,\dots,x_n$, terms $\tau_1,\dots,\tau_n$ not featuring these variables, and two non-zero integers~$a$ and~$b$. 
By \emph{simultaneously applying the substitutions } $\sub{\frac{\tau_1}{a}}{b \cdot x_1},\dots,\sub{\frac{\tau_n}{a}}{b \cdot x_n}$ to the term $\rho$
we mean the process of first multiplying every integer in $\rho$ by $\abs{a}$, 
to then apply to the resulting term the substitutions $\sub{\tau_1}{a \cdot b \cdot x_1},\dots,\sub{\tau_n}{a \cdot b \cdot x_n}$ (in any order). So, differently from sequentially applying~$\sub{\frac{\tau_1}{a}}{b \cdot x_1},\dots,\sub{\frac{\tau_n}{a}}{b \cdot x_n}$, simultaneous substitutions multiply by $\abs{a}$ only once.

When applying a substitution $\sub{\frac{\tau}{a}}{b \cdot x}$ to a linear-exponential program with divisions~$\phi$, the resulting program $\phi\sub{\frac{\tau}{a}}{b \cdot x}$ is constructed as follows: 
\begin{itemize}
    \item For every equality~$\rho = 0$ or inequality~$\rho \leq 0$ occurring in~$\phi$, replace $\rho$ with $\rho\sub{\frac{\tau}{a}}{b \cdot x}$.
    \item Replace every divisibility constraint $d \divides \rho$ occurring in $\phi$ with $(\abs{a \cdot b} \cdot d) \divides \rho\sub{\frac{\tau}{a \cdot b}}{x}$ .
\end{itemize}

\RestoreHeader
\section{The algorithm for deciding ILEP feasibility, briefly}%
\label{section:summary-procedure}%



We present a high-level overview of the procedure from~\cite{ChistikovMS24} for deciding the feasibility problem of ILEP, 
highlighting its properties in the context of optimization. 
As we will see, the main loop of the procedure can be divided in four steps (Steps I--IV). Steps I and III preserve optimal solutions; we can thus use them as black-boxes when designing our optimization procedure. 
\Cref{section:analysis-step-i-and-iii} gives more information on these two steps, as well as their pseudocode. In contrast, Step II and IV may discard all optimal solutions. Step II is the main ``variable elimination step'', which we will focus on in the upcoming sections of this part of the paper.
Step IV is a simplified variant of Step II, and will be handled directly when presenting the full optimization procedure in~\Cref{sec:putting-all-together}.

Let $\phi$ be an input ILEP. 
As a preliminary step, the procedure in~\cite{ChistikovMS24} 
non-deterministically guesses an ordering $\theta$ of the form $2^{x_n} \geq \dots \geq 2^{x_1} \geq 2^{x_0} = 1$. 
Here, $x_1,\dots,x_n$ is a permutation of the variables in $\phi$, whereas $x_0$ is a fresh variable introduced to handle the termination of the algorithm. Note that~$\phi \land (2^{x_0} = 1)$ is equivalent to the disjunction~$\bigvee_{\theta \in \Theta} (\phi \land \theta)$ ranging over the set of all orderings~$\Theta$. In the context of optimization, no optimal solution is
lost in this step of the procedure:  
it suffices to optimize locally to each disjunct $\phi \land \theta$, and then take the maximum (or minimum) of the resulting optimal solutions.

After guessing the ordering~$\theta$, the algorithm enters its main loop, where it iteratively eliminates from $\phi$ and~$\theta$ all variables $x_1,\dots,x_n$, starting from the largest one in~$\theta$. These eliminations introduce new \emph{remainder variables} $\vec r$, variables that never occur in exponentials, and are always smaller than the largest term in~$\theta$.
After eliminating $x_n,\dots,x_i$, the ordering~$\theta$ is updated to $2^{x_{i-1}} \geq \dots \geq 2^{x_0} = 1$, and all remainder variables are constrained to be smaller than $2^{x_{i-1}}$.
After $n$ iterations of the main loop, $\theta$ reduces to just~$2^{x_0} = 1$, and $\phi$ becomes a formula~$\phi'(x_0,\vec r)$ that implies $\vec r < 2^{x_0}$. The main loop terminates. 
The only possible solution for $\phi' \land (2^{x_0} = 1)$ is $(x_0,\vec r) = \vec 0$; and if this is a solution, then the original formula~$\phi$ is  satisfiable.

We now describe an iteration of the main loop, dividing it in the aforementioned Steps~I--IV.\footnote{Our division of the procedure into steps differs from that used by the authors of~\cite{ChistikovMS24} to describe the algorithm. Specifically, we have included lines 4--14 of Algorithm~2 of~\cite{ChistikovMS24} as part of the Step I, instead of considering them separately. 
This adjustment is made solely for the sake of presentation clarity.}
To aid in following the interactions between these steps, \Cref{fig:flowchart} is provided alongside the description.

\tikzset{
    invisible/.style={opacity=0},
    emph/.style={color=magenta},
    alert/.style={color=red},
    anotherhl/.style={color=blue},
    bold/.style={very thick},
    emph on/.style={alt={#1{emph}{}}},
    alert on/.style={alt={#1{alert}{}}},
    anotherhl on/.style={alt={#1{anotherhl}{}}},
    arrow on/.style={alt={#1{->}{}}},
    bold on/.style={alt={#1{bold}{}}},
    visible on/.style={alt={#1{}{invisible}}},
    }

\pgfdeclaredecoration{dashsoliddouble}{initial}{
  \state{initial}[width=\pgfdecoratedinputsegmentlength]{
    \pgfmathsetlengthmacro\lw{.5pt+.5\pgflinewidth}
    \begin{pgfscope}
      \pgfpathmoveto{\pgfpoint{0pt}{\lw}}%
      \pgfpathlineto{\pgfpoint{\pgfdecoratedinputsegmentlength}{\lw}}%
      \pgfmathtruncatemacro\dashnum{%
        round((\pgfdecoratedinputsegmentlength-3pt)/6pt)
      }
      \pgfmathsetmacro\dashscale{%
        \pgfdecoratedinputsegmentlength/(\dashnum*6pt + 3pt)
      }
      \pgfmathsetlengthmacro\dashunit{3pt*\dashscale}
      \pgfsetdash{{\dashunit}{\dashunit}}{0pt}
      \pgfusepath{stroke}
      \pgfsetdash{}{0pt}
      \pgfpathmoveto{\pgfpoint{0pt}{-\lw}}%
      \pgfpathlineto{\pgfpoint{\pgfdecoratedinputsegmentlength}{-\lw}}%
      \pgfusepath{stroke}
    \end{pgfscope}
  }
}

\begin{figure}
    \begin{tikzpicture}[
      >=stealth,
      a/.style={align=left, node font=\itshape},
      s/.style={align=left},
      es/.style={align=right, anchor=north east},
      f/.style={draw=gray,fill=white},
      new/.style={draw=blue,text=blue},
      fornew/.style={draw=blue},
      node distance=1.05cm and 0.6cm,
      aster/.style={circle,draw=black,fill=white,inner sep=1pt},
      every node/.style={inner sep=0pt}
    ]

      \node (input-after-ordering-1) 
        {\addtolength{\tabcolsep}{-0.2em}\renewcommand{\arraystretch}{1.3}\begin{tabular}{|p{0.11\textwidth}p{0.43\textwidth}|}  
        \hline
        $\theta(\vec x)$&: ordering $2^x \geq 2^y \geq \dots \geq 2^{x_0} = 1$\\
        \hline
        \end{tabular}};

      \node (input-after-ordering-2) [below = 0.1cm of input-after-ordering-1]
        {\addtolength{\tabcolsep}{-0.2em}\renewcommand{\arraystretch}{1.3}\begin{tabular}{|p{0.11\textwidth}p{0.43\textwidth}|}
        \hline
        $\phi(\vec x,\vec r)$&: linear-exponential program with divisions\\
        \hline
        \end{tabular}};

      \node (after-step-i-1) [below = 1cm of input-after-ordering-2]
        {\addtolength{\tabcolsep}{-0.2em}\renewcommand{\arraystretch}{1.3}\begin{tabular}{|p{0.11\textwidth}p{0.43\textwidth}|}
        \hline 
        $\gamma(q_x,\vec q,u)$&: linear program with divisions\\
        \hline
        \end{tabular}};

      \node (after-step-i-2) [below = 0.1cm of after-step-i-1]
        {\addtolength{\tabcolsep}{-0.2em}\renewcommand{\arraystretch}{1.3}\begin{tabular}{|p{0.11\textwidth}p{0.43\textwidth}|}
        \hline
        $\psi(\vec y, r_x, \vec r')$&: linear-exponential program with divisions\\
        \hline
        \end{tabular}};

      \node (after-step-ii) [below = 1cm of after-step-i-2]
        {\addtolength{\tabcolsep}{-0.2em}\renewcommand{\arraystretch}{1.3}\begin{tabular}{|p{0.11\textwidth}p{0.43\textwidth}|} 
        \hline
        $\gamma'(q_x,u)$&: linear program with divisions\\
        \hline
        \end{tabular}};

      \node (after-step-iii-1) [below = 1cm of after-step-ii]
        {\addtolength{\tabcolsep}{-0.2em}\renewcommand{\arraystretch}{1.3}\begin{tabular}{|p{0.11\textwidth}p{0.43\textwidth}|}
        \hline  
        $\gamma''(q_x)$&: linear program with divisions\\
        \hline
        \end{tabular}};

      \node (after-step-iii-2) [below = 0.1cm of after-step-iii-1]
        {\addtolength{\tabcolsep}{-0.2em}\renewcommand{\arraystretch}{1.3}\begin{tabular}{|p{0.11\textwidth}p{0.43\textwidth}|}
        \hline
        $\psi''(y,r_x)$&: linear-exponential program with divisions\\
        \hline
        \end{tabular}};

      \node (step-iv-text) [below left = 0.5cm and -2.9cm of after-step-iii-2]
        {\textit{Step IV:}};
      \node (step-iv) [below = 0.1cm of step-iv-text]
        {\ is $\gamma''$ satisfiable?\ };

      \node[inner sep=2pt] (step-iv-no) [below = 0.6 of step-iv] {\textbf{reject}};

      \node (next-iteration) [below right = 0.8cm and -4.61cm of after-step-iii-2]
        {\addtolength{\tabcolsep}{-0.5em}\renewcommand{\arraystretch}{1.3}\begin{tabular}{|rl|}
        \hline
        $\ \phi \gets{}$&$\psi \land \psi''$\\ 
        $\ \theta \gets{}$&$ (2^y \geq \dots \geq 2^{x_0} = 1)\ $\\ 
        \hline
        \end{tabular}};


      \draw [->] (input-after-ordering-2.south) -- node[a,right] {\ Step~I (\Cref{lemma:CMS:first-step})} (after-step-i-1.north);

      \draw[->, rounded corners]
          (after-step-i-1.west)
          -- ++(-0.25,0) coordinate (turn1)
          -- ++(0,-1.6) coordinate (turn2)
          -- ++(5,0) coordinate (turn3) node[a,right]{\ Step~II (\Cref{lemma:CMS:second-step})}
          -- (turn3|-after-step-ii.north);

      \draw [->] (after-step-ii.south) -- node[a,right] {\ Step~III (\Cref{lemma:CMS:third-step})} (after-step-iii-1.north);

      \draw[->,rounded corners]
          (after-step-iii-1.west)
          -- ++(-0.25,0) coordinate (turn4)
          -- (turn4|-step-iv.west) coordinate (turn5)
          -- (step-iv.west);

      \draw [-implies,double equal sign distance] (step-iv.south) -- node[a,right] {\ no} (step-iv-no.north);

      \draw [-implies,double equal sign distance] (step-iv.east) -- node[a,above] {\ \raisebox{5pt}{yes}} (step-iv.east-|next-iteration.west);

      \draw[->, rounded corners]
          (after-step-i-2.east)
          -- ++(0.25,0) coordinate (turn6)
          -- ++(0,-5.6) coordinate (turn7)
          -- (turn7-|next-iteration.east);

      \draw[->, rounded corners]
          (after-step-iii-2.south-|next-iteration.north)
          -- (next-iteration.north);
    

      \node (notes-placement) [above right = 0cm and 4cm of input-after-ordering-1] {};
      \node[align=left] (notes-1) [below = 0cm of notes-placement] 
        {\renewcommand{\arraystretch}{0.8}\begin{tabular}{p{0.3\textwidth}}  
        ~\\[0.2cm]
        \footnotesize $\phi$ implies $\vec r < 2^y$.\\ 
        \footnotesize Variables~$\vec r$ not in exponentials.
        \end{tabular}};

      \node[align=left] (notes-2) [below = 0.55cm of notes-1] 
        {\renewcommand{\arraystretch}{0.8}\begin{tabular}{p{0.3\textwidth}}  
        \footnotesize $\psi$ implies $r_x < 2^y \land \vec r' < 2^y$.\\ 
        \footnotesize $\vec r'$ and $r_x$ not in exponentials.\\
        \footnotesize $\vec y$ are the variables $\vec x$, excluding $x$.\\
        \footnotesize \textit{Key equations connecting $\phi$ with $\gamma$ and~$\psi$:} $u = 2^{x-y}$, ${x = q_x \cdot 2^y \cdot r_x}$ and $\vec r = \vec q \cdot 2^y + \vec r'$. 
        \end{tabular}};
      \node[align=left] (notes-3) [below = 0.45cm of notes-2] 
        {\renewcommand{\arraystretch}{0.8}\begin{tabular}{p{0.3\textwidth}}  
        \footnotesize Step II eliminates the variables $\vec q$.\\
        \footnotesize \textbf{Main problem:} Step II preserves equisatisfiability, but optimal solutions may be lost.
        \end{tabular}};
      \node[align=left] (notes-4) [below = 0.5cm of notes-3] 
        {\renewcommand{\arraystretch}{0.8}\begin{tabular}{p{0.3\textwidth}}  
        \footnotesize Step III eliminates $u$ and, following the equations $x = q_x \cdot 2^y + r_x$ and $u = 2^{x-y}$, 
        also elminates $x$.\\
        \end{tabular}};

      \node[align=left] (notes-5) [below = 1cm of notes-4] 
        {\renewcommand{\arraystretch}{0.8}\begin{tabular}{p{0.3\textwidth}}  
        \footnotesize If $(2^y \geq \dots \geq 2^{x_0} = 1)$ is $(2^{x_0} = 1)$ then the loop exists, and the algorithm checks if $\phi(\vec 0)$ is a solution.\\
        \end{tabular}};

    \end{tikzpicture}
    \caption{Flowchart of the main loop of~\cite{ChistikovMS24}.}
    \label{fig:flowchart}
\end{figure}
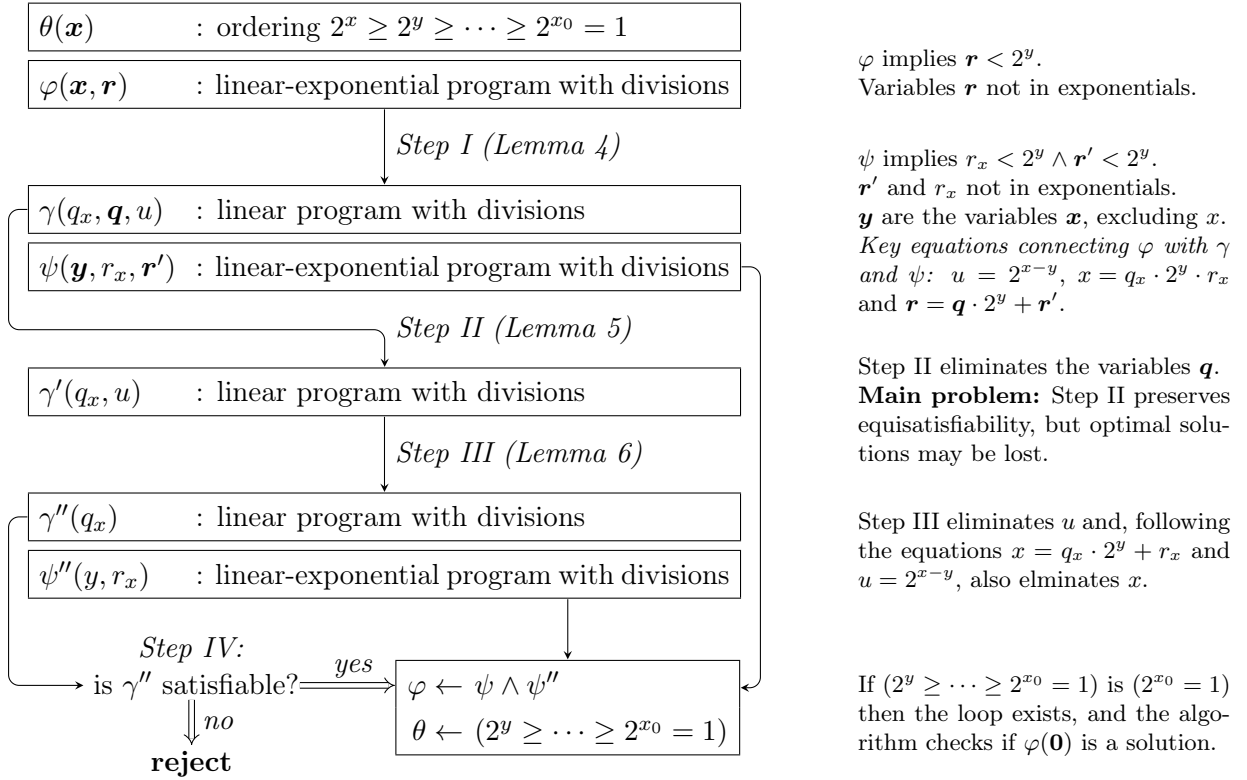

\paragraph*{Step I (division by $2^y$).}\label{cms:summary:step1} 
Let $2^x$ and~$2^y$ be the largest and second-largest terms in $\theta$ 
(so,~${x \neq x_0}$). 
The first step to eliminate $x$ is to symbolically divide all linear occurrences of this variable, as well as all remainder variables~$\vec r$, by~$2^y$. 
That is, the algorithm rewrites the linear occurrences of $x$ as $q_x \cdot 2^y + r_x$, and $\vec r$ as $\vec q \cdot 2^y + \vec r'$, where the fresh variables $(q_x,\vec q)$ and $(r_x, \vec r)$ represent the \emph{quotient} and \emph{remainder} of the division by $2^y$, respectively. 
The variables $r_x$ and $\vec r'$ are remainder variables in the next iteration of the main loop;
and indeed the algorithm adds the constraints $r_x < 2^y$ and $\vec r' < 2^y$ to the system. The variables $q_x$ and $\vec q$ are called \emph{quotient variables}.
The procedure introduces a further expression $u = 2^{x-y}$, with $u$ fresh, 
and through several manipulations decouples the quotient variables from all other variables except $u$ (this is similar to a \emph{monadic decomposition}~\cite{Libkin03}). The key equivalences enabling 
this decoupling are given in the next lemma. In the context of the algorithm, the integer~$t$ in this lemma corresponds to a linear term featuring the variable $u$ and the quotient variables $\vec q$, whereas the integer $s$ corresponds to a linear-exponential term involving the remainder variables $r_x$ 
and~$\vec r'$, linear occurrences of~$y$, and all the variables in $\theta$ that are distinct from~$x$ and~$y$. The key point is that, in the right-hand side of the equivalences in the lemma, $t$ and $s$ are decoupled (that is, they never appear together within a single (in)equality).

\begin{lemma}[\cite{ChistikovMS24}]
    \label{lemma:split:inequalities}
    Let $C,D \in \Z$, with $C \leq D$. For $y \in \N$, $t \in \Z$, and ${s \in [C \cdot 2^y..D \cdot 2^y]}$, the following equivalences hold: 
    \begin{enumerate} 
        \item\label{step1:fund-equiv-1} $t \cdot 2^y + s = 0
        \iff \bigvee_{r = C}^{D} \big(t + r = 0 \land s = r \cdot 2^y\big)$,
        \item\label{step1:fund-equiv-2} $t \cdot 2^y + s \leq 0 
        \iff \bigvee_{r = C}^{D} \big(t + r \leq 0 \land (r-1) \cdot 2^y < s \leq r \cdot 2^y\big)$,
        \item\label{step1:fund-equiv-3} $t \cdot 2^y + s < 0 
        \iff \bigvee_{r = C}^{D} \big(t + r+1 \leq 0 \land  s = r \cdot 2^y\big) \lor \big(t + r \leq 0 \land (r-1) \cdot 2^y < s < r \cdot 2^y\big)$.
    \end{enumerate}
    \vspace{0pt}
\end{lemma}

To see~\Cref{lemma:split:inequalities} in action, 
consider the equality $2^x - 2^y + y - z = 0$. 
Assuming $\theta = (2^x \geq 2^y \geq 2^z)$ and $u = 2^{x-y}$, we can rewrite 
this equality as $(u - 1) \cdot 2^y + y - z = 0$.
Moreover, we see that~$\theta$ implies that $y-z$ belongs to $[0 \cdot 2^y..1 \cdot 2^y]$.
Then, \Cref{lemma:split:inequalities}.\ref{step1:fund-equiv-1} 
tells us that the equality can 
be rewritten as $\bigvee_{r = 0}^1 ((u-1) + r = 0 \land (y-z) = r \cdot 2^y)$.
Here, the equation $(u-1) + r = 0$ is in a sense the ``quotient'' of the division by $2^y$, 
whereas $y-z = r \cdot 2^y$ indicates properties of the ``remainder'' of the division (in this case, that $y-z$ has remainder zero when divided by $2^y$).

The effects of Step~I of the main loop are formalized in the next lemma, 
where the output formula $\gamma_\beta$ contains 
the~``quotients'' of the divisions by $2^y$, and $\psi_\beta$ contains constraints on the~``remainders''.

\begin{restatable}[{\cite{ChistikovMS24}}]{lemma}{CMSFirstStep}
  \label{lemma:CMS:first-step}
  There is a non-deterministic procedure with the following specification:
  \begin{description}
    \setlength{\tabcolsep}{2pt}
    \item[\textbf{\textit{Input:}}] 
      \hspace{-14pt}\begin{minipage}[t]{\linewidth}
        \begin{tabular}[t]{rcp{0.9\linewidth}}
        $\theta(\vec x)$&:& ordering of exponentiated variables;\\
          &&[Below, let~$2^x$ and $2^y$ be the largest and second-largest terms in this ordering, 
          and\\ 
          && let $\vec y$ be the vector obtained by removing $x$ from $\vec x$.]\\ 
        $\phi(\vec x, \vec r)$&:& linear-exponential program with divisions, implying $\vec r < 2^x$.\\ 
        &&Variables $\vec r$ do not occur in exponentials.
        \end{tabular}
      \end{minipage}
    \item[\textbf{\textit{Output of each branch ($\beta$):}}]\,
      
      \begin{minipage}[t]{\linewidth}
        \hspace{-12pt}
        \begin{tabular}[t]{rcp{0.85\linewidth}}
        $\gamma_{\beta}(q_x,\vec q, u)$&:& linear program with divisions;\\ 
        $\psi_{\beta}(\vec y, r_x,\vec r')$&:&linear-exponential program with divisions, 
        implying $r_x < 2^y \land \vec r' < 2^y$.\\ 
        &&Variables $r_x$ and $\vec r'$ do not occur in exponentials.
        \end{tabular}
      \end{minipage}
  \end{description}
  The variables $q_x$, $\vec q$, $u$, $\vec y$, $r_x$ and $\vec r'$
  are common to all outputs, across all non-deterministic branches.
  The procedure ensures that the system
  \begin{equation}
    \label{eq:CMS:first-step}
    \left[\begin{matrix}
      x\\ 
      \vec r
    \end{matrix}\right]
    = \left[\begin{matrix}
      q_x\\ 
      \vec q
    \end{matrix}\right] \cdot 2^y + \left[\begin{matrix}
      r_x\\ 
      \vec r'
    \end{matrix}\right],
  \end{equation}
  yields a one-to-one correspondence between the solutions of
  $\phi \land \theta$ and the solutions of the formula
  ${\bigvee_{\beta} \big(\gamma_{\beta} \land \psi_{\beta} \land (u = 2^{x-y}) \land (x = q_x \cdot 2^y + r_x) \land \theta\big)}$.
  This correspondence is the identity for the variables these two formulae share (that is, the variables in $\vec x$).
\end{restatable}

The one-to-one correspondence 
described in~\Cref{lemma:CMS:first-step} implies that 
no solution is lost in this step of the procedure.
We can therefore use Step I also in the context of optimization.

\paragraph*{Step II (variable elimination: the problematic step).} 
The procedure now considers the \emph{linear} program with divisions $\gamma(q_x,\vec q, u)$ in output of~Step~I, 
and applies a quantifier elimination procedure to remove all the quotient variables in~$\vec q$. 
(Not $q_x$, this quotient variable cannot be eliminated yet, because the equalities $u = 2^{x-y}$ and $x = q_x \cdot 2^y + r_x$ shown in~\Cref{lemma:CMS:first-step} make the variable $u$ depend exponentially on $q_x$.)
This elimination step mixes ingredients from the quantifier elimination procedure for Presburger arithmetic~\cite{Weispfenning90}
with Bareiss' version of Gaussian elimination~\cite{Bareiss68}. 
As in the case of the former of these two procedures, 
this step introduces new divisibility constraints.
Here is the specification of Step~II:

\begin{restatable}[\cite{ChistikovMS24}]{lemma}{CMSSecondStep}
  \label{lemma:CMS:second-step}
  There is a non-deterministic procedure with the following specification:
  \begin{description}
    \setlength{\tabcolsep}{2pt}
    \item[\textbf{\textit{Input:}}] 
      \begin{minipage}[t]{0.94\linewidth}
        \hspace{3pt}
        \begin{tabular}[t]{rcp{0.75\linewidth}}
        $\gamma(q_x,\vec q, u)$&:& linear program with divisions.
        \end{tabular}
      \end{minipage}
    \item[\textbf{\textit{Output of each branch ($\beta$):}}]
    
    \begin{minipage}[t]{\linewidth}
      \hspace{3pt}
      \begin{tabular}[t]{rcp{0.75\linewidth}}
        $\gamma_{\beta}'(q_x,u)$&:& linear program with divisions.
        \end{tabular}
      \end{minipage}
  \end{description}
  The procedure ensures that the formulae $\exists \vec q \, \gamma$ 
  and $\bigvee_{\beta} \gamma_{\beta}'$ are equivalent.
  Let $\vec q = (q_1,\dots,q_k)$.
  For~every branch $\beta$, there is a system of equalities 
  \begin{align}
    a_{1} \cdot q_{1} = \tau_{1}(u,q_x)\,,\, \dots\, ,\,    
    a_{k} \cdot q_{k} = \tau_{k}(u,q_x)\,,\label{eq:CMS:second-step}
  \end{align}
  where each $a_i \in \Z$ is non-zero 
  and each $\tau_i$ is a linear term, with the following property. 
  The formula $\gamma_\beta'$ is obtained from $\gamma$ 
  by performing the sequence of substitutions $\sub{\frac{\tau_1}{a_1}}{q_1},\dots,\sub{\frac{\tau_k}{a_k}}{q_k}$ and conjoining the system of divisibilities
  $(a_1 \divides \tau_{1}(u,q_x)) \land \dots \land (a_k \divides \tau_{k}(u,q_x))$.
\end{restatable}

Concerning optimization, the guarantees achieved by this crucial step of the procedure are too weak.
Rather than establishing a one-to-one correspondence between 
the solutions of the input and those of the outputs, it only achieves an equivalence with respect to the variables~$q_x$ and~$u$. 
Notably, if some variables from $\vec q$ appear in the objective function, then optimal solutions may be lost. 

\paragraph*{Step III (elimination of $x$ and $u$).} 
The third step of the main loop is somewhat similar to the first one. 
We start with the formula $\gamma'(q_x,u)$ obtained from Step II, 
add the constraints $x = q_x \cdot 2^y + r_x$ and $u = 2^{x-y}$, 
and decouple $q_x$ from all other variables. 
By using machinery developed by Semenov in~\cite{Semenov84}, 
this decoupling makes it possible to eliminate the variables~$x$ and~$u$.

Here is the specification of Step III:

\begin{restatable}[\cite{ChistikovMS24}]{lemma}{CMSThirdStep}
  \label{lemma:CMS:third-step}
  There is a non-deterministic procedure with the following specification:
  \begin{description}
    \setlength{\tabcolsep}{2pt}
    \item[\textbf{\textit{Input:}}] 
      \begin{minipage}[t]{0.94\linewidth}
        \hspace{3pt}
        \begin{tabular}[t]{rcp{0.75\linewidth}}
        $\gamma'(q_x,u)$&:& linear program with divisions.
        \end{tabular}
      \end{minipage}
    \item[\textbf{\textit{Output of each branch ($\beta$):}}]

      \begin{minipage}[t]{0.94\linewidth}
        \hspace{3pt}
        \begin{tabular}[t]{rcp{0.8\linewidth}}
        $\gamma_{\beta}''(q_x)$&:& linear program with divisions;\\ 
        $\psi_{\beta}''(y, r_x)$&:& linear-exponential program with divisions.
        \end{tabular}
      \end{minipage}
  \end{description}
  The procedure ensures that the equation 
  \begin{equation}
    \label{eq:CMS:third-step}
    x = q_x \cdot 2^y + r_x
  \end{equation}
  yields a one-to-one correspondence between the solutions 
  of~$\gamma' \land (u = 2^{x-y}) \land (x = q_x \cdot 2^y + r_x)$
  and the solutions 
  of~$\bigvee_{\beta} \big(\gamma_{\beta}'' \land \psi_{\beta}''\big)$.
  This correspondence is the identity for the variables these two
  formulae share (that is, $y$, $q_x$ and $r_x$).
\end{restatable}

As in the case of Step I, the one-to-one correspondence described in~\Cref{lemma:CMS:third-step} ensures that optimal solutions are preserved 
during this step of the main loop.

\paragraph{Step~IV (elimination of $q_x$).}
After Step III completes, we are left with its output formulae $\gamma''(q_x)$ and $\psi''(y,r_x)$, and
the formula $\psi(\vec y, r_x, \vec r')$ computed in Step I.
The main loop of now performs one last operation: it checks whether~$\gamma''$ (a univariate linear program with divisions) is satisfiable. If it is, $\gamma''$ can be replaced with $\top$, effectively eliminating the variable $q_x$. (Otherwise, the non-deterministic branch of the program rejects.)
An alternative way of implementing Step IV is to apply the variable elimination procedure underlying~\Cref{lemma:CMS:first-step}, but again this may cause the algorithm to lose all optimal solutions. In this particular case, since $\gamma''$ only features the variable~$q_x$, 
the formula constructed by the variable elimination procedure simply replaces~$q_x$ with an integer $c$; i.e., the system of equalities analogous to the one in~\Cref{eq:CMS:second-step} simplifies in this case to just~$q_x = c$.

This concludes the current iteration of the main loop of the procedure. 
If $y$ is not the variable~$x_0$,
the loop performs another iteration. In that iteration, the input to Step~I becomes the ordering $\theta'$ obtained from $\theta$ by removing the term $2^x$ ($2^y$ is now the largest term), together with the linear-exponential program with divisions~${\psi(y, r_x) \land \psi''(\vec y, r_x, \vec r')}$.%

\subsection{Where are the ILESLPs?}
\label{subsection:where-are-the-ILESLP}

As emphasized in~\Cref{lemma:CMS:first-step,lemma:CMS:second-step,lemma:CMS:third-step}, the procedure in~\cite{ChistikovMS24} is in a sense guided by~\Cref{eq:CMS:first-step,eq:CMS:second-step,eq:CMS:third-step}. 
Upon closer inspection, we see that these equations are constructing an ILESLP. Let us reason bottom-up and suppose that we have constructed an ILESLP~$\sigma$ that is a solution to the formula $\psi(y, r_x) \land \psi''(\vec y, r_x, \vec r') \land \theta'$ described above. We construct an ILESLP that is a solution for $\phi \land \theta$ by appending further assignments to $\sigma$. 
The first three assignments are 
\[ 
  z_1 \gets 2^y,\ \ 
  z_2 \gets c \cdot z_1,\ \ 
  x   \gets z_2 + r_x,
\]
where $z_1$ and $z_2$ are auxiliary fresh variables, and $c$ is 
the integer in the equation $q_x = c$. 
We are essentially performing the assignment~$x \gets c \cdot 2^y + r_x$, 
accordingly to~\Cref{eq:CMS:third-step}.
Observe that $\sigma$ already assigns expressions to $y$ and $r_x$. 
Next, we add assignments to represent each variable in~$\vec r$. 
For each variable~$v$ belonging to $\vec r$ we have
\begin{align*}
  v &= q_v \cdot 2^y + v' 
  &\Lbag\text{\Cref{eq:CMS:first-step}, where $v'$ is some variable in $\vec r'$}\Rbag\\
  &= \frac{\tau(u,q_x)}{a} \cdot 2^y + v'
  &\Lbag\text{\Cref{eq:CMS:second-step}}\Rbag\\
  &= \frac{b \cdot 2^{x-y} + d}{a} \cdot 2^y + v'
  &\Lbag\text{using $u = 2^{x-y}$ and $q_x = c$}\Rbag\\
  &= \frac{b \cdot 2^{x} + d \cdot 2^y}{a} + v',
\end{align*}
for some integers $a,b,d$, with $a \neq 0$.
We can easily add assignments to~$\sigma$ to obtain ${v \gets \frac{b \cdot 2^{x} + d \cdot 2^y}{a} + v'}$. The resulting ILESLP is guaranteed to be a solution to $\phi \land \theta$.

\subsection{\OptILEP: from feasibility to optimization}
\label{subsection:OptILEP}

Following the above description of the procedure from~\cite{ChistikovMS24}, 
it should be now clear that a way to obtain a procedure for the optimization 
problem of ILEP is to focus on the ``variable elimination step'' (Step II),
strengthening it into a procedure that is guaranteed to explore an optimal solution. 
In the remainder of the paper, we refer to the resulting procedure as~\OptILEP.
We will present the pseudocode of this procedure in~\Cref{sec:putting-all-together} (see~\Cref{pseudocode:opt-ilep}). For now, the specific details of the procedure are unimportant; 
the only features to keep in mind are the following: 
\begin{itemize}
  \item The procedure begins by guessing an ordering $\theta$ of the form $2^{x_n} \geq \dots \geq 2^{x_1} \geq 2^{x_0} = 1$, where $x_1,\dots,x_n$ are the variables appearing input instance of ILEP. 
  \item It then iteratively eliminates the variables~$x_n,\dots,x_1$ (in this order), updating both the linear-exponential program and the objective function.
  Every iteration of this ``main loop'' appeals to Step~I and Step~III of~\cite{ChistikovMS24}, interposed with an optimum-preserving ``variable elimination step''. This elimination step instantiates the template given by~\Cref{algo:gaussopt} (\GaussOpt), 
  introduced in~\Cref{subsection:overview-theorem-one}.
  The main loop concludes with a step analogous to Step IV, 
  modified along the same lines as Step II to ensure preservation
  of optimal solutions.
\end{itemize}

\section{Exploring optimal solutions through monotone decompositions}
\label{section:ILEP-in-npocmp}

As explained in the overview given in~\Cref{subsection:overview-theorem-one}, the template for the ``variable elimination step'' given by~\Cref{algo:gaussopt} (\GaussOpt) eliminates some variables~$\vec x$ from a system of constraints $\phi$ (in which the variables~$\vec x$ occur linearly) and an objective function $f$, by iteratively

\begin{enumerate}
    \item Guessing an equation ${a \cdot x = \tau}$ from a finite set of
    \emph{test points}~$\tests(\vec x, f, \phi)$, where $a$ is a non-zero integer, and $x$ is a variable from~$\vec x$ that still occurs in $f$ or $\phi$.
    \item Appealing to an \emph{elimination discipline}~$\elimdisctxt$, which updates $f$ and $\phi$ 
    by ``replacing~$x$ with $\frac{\tau}{a}$''.
    (Intuitively, this means that we will only be searching for solutions 
    lying inside the hyperplane described by the equation ${a \cdot x = \tau}$.)
\end{enumerate}

In this section, we describe a method, based on \emph{monotone decompositions} of the search space, for constructing sets of test points that are guaranteed to preserve at least one optimal solution. We develop the approach in a general setting, where the objective function is treated as a black box. 

\paragraph*{Some notation.}\label{par:some-notation-secthree}
Given $d \in \N$ and $i \in [1..d]$, we write~$\vec e_i^d$ for the $i$-th vector of the canonical basis of $\R^d$ (i.e., $\vec e_i^d$ has a $1$ at the $i$-th component and $0$ elsewhere); 
omitting the superscript $d$ when clear from the context.
A set $S \subseteq \N^d$ is said to be \emph{$(i,p)$-periodic} if
for any $\vec v \in S$ and $\vec v + m \cdot \vec e_i \in S$ with $m \geq p$, 
we also have $\vec v + p \cdot \vec e_i \in S$.
We write $\Delta_i^p [f]$ for the \emph{$i$-th $p$-spaced partial finite difference} of a function $f \colon \R^d \to \R$.  
It is defined as $\Delta_i^p[f](\vec x) \coloneqq f(\vec x + p \cdot \vec e_i) - f(\vec x)$, for every $\vec x \in \R^d$.

A function $f \colon \R^d \to \R$ is said to be \emph{$(i,p)$-monotone locally to a set $S$} if there is a sign ${{\sim} \in \{<,=,>\}}$ 
such that for every $\vec v \in S$ with $\vec v + p \cdot \vec e_i \in S$, we have $\Delta_i^p[f](\vec v) \sim 0$.
When considering logical formulae, 
we will abuse this notation and, in the above definitions, replace the indices of the entries of vectors by variable names.
For instances, we will say that the set of solution of a formula $\phi(\vec x)$ is $(x,p)$-periodic, with $x$ being a variable from $\vec x$, 
and we will write $\Delta_x^p [f]$ for the $p$-spaced partial finite difference of~$f$ with respect to (the coordinate corresponding to) $x$.

\paragraph*{Locating optimal solutions for monotone functions.} 
We start by adapting a folklore result from Presburger arithmetic to ILEP: with respect to a linearly occurring variable, the set of solutions of an integer linear-exponential program is periodic. (See~\cite[Theorem~4.10]{Smorynski91} for a similar result.)

\begin{restatable}{lemma}{LemmaModPeriodicity}
    \superlabel{lemma:mod-periodicity}{proof:LemmaModPeriodicity}
    Let $\phi$ be a linear-exponential program with divisions, and let $x$ be a variable occurring linearly in $\phi$.
    The set of solutions of $\phi$ is $(x,\fmod(x,\phi))$-periodic.
\end{restatable}

The periodic behavior described in the above lemma implies that, 
for monotone functions, optimal solutions are located near 
the boundary of the feasible region described by the integer program. 
This boundary is determined by equalities of~$\tau = 0$
that are derived from 
(in)equalities ${\tau \sim 0}$ appearing in the program (where ${\sim} \in \{{=},{\leq}\}$). To ensure we find an optimal solution, it suffices to examine 
shifted versions of these boundary equalities, 
that is equalities of the form~${\tau + s = 0}$, 
where $s$ ranges over a small set of integer offsets.

\begin{restatable}{lemma}{LemmaMonotoneGaussianElimination}
    \superlabel{lemma:monotone-gaussian-elimination}{proof:LemmaMonotoneGaussianElimination}
    Let $\phi(\vec x)$ be a linear-exponential program with divisions, and let $x$ be a variable occurring linearly in $\phi$.
    Let~$p \coloneqq \fmod(x,\phi)$,
    and let $f(\vec x)$ be a $(x,p)$-monotone function locally to the set of solutions to $\phi$.
    If the instance $(f,\phi)$ has a maximum (analogously, a minimum), 
    then it has one satisfying an equation $a \cdot x + \tau + r = 0$, 
    where $(a \cdot x + \tau) \in \fterms(\phi \land x  \geq 0)$, $a \neq 0$, 
    and $r \in [0..\abs{a} \cdot p-1]$.
\end{restatable}

By considering~$f(x) = x$ and considering the minimization problem, 
\Cref{lemma:monotone-gaussian-elimination} simplifies to the following corollary. 
This corollary, in fact, captures the core argument found in nearly all proofs 
of quantifier elimination in Presburger arithmetic (see, e.g.,~\cite[Lemma 2.6]{Weispfenning90}).

\begin{restatable}{corollary}{CorrBasicFactFromPresburger}
    \superlabel{corr:basic-fact-from-presburger}{proof:CorrBasicFactFromPresburger}
    Let $\phi$ be a linear-exponential program with divisions, and let $x$ be a variable occurring linearly in $\phi$. If $\phi$ has a solution, then it has one satisfying an equation $a \cdot x + \tau + r = 0$, 
    where $(a \cdot x + \tau) \in \fterms(\phi \land x  \geq 0)$, $a \neq 0$, 
    and $r \in [0..\abs{a} \cdot \fmod(x,\phi)-1]$.
\end{restatable}

\paragraph*{Locating optimal solutions for non-monotone functions.}
\Cref{lemma:monotone-gaussian-elimination} suggests a natural strategy for tackling optimization problems involving non-monotone functions: partition the search space into multiple regions where the function becomes monotone.
This idea is formalized with the subsequent definition of \emph{monotone decomposition} and~\Cref{lemma:monotone-only-x-matters}.
It is important to note that, depending on the objective function, constructing such a  decomposition with only regions that can be characterized with integer linear-exponential programs (or other desired classes of constraints) can be highly non-trivial or even impossible.
In the next section, we show how to achieve this only in the specific setting needed to solve the ILEP optimization problem.

\begin{definition}[Monotone decomposition]
    A \emph{$(i,p)$-monotone decomposition of ${S \subseteq \N^d}$ for a function $f \colon \R^d \to \R$}
    is a finite family $R_1,\dots,R_t \subseteq \N^d$ such that 
    (i) $S = \bigcup_{j=1}^{t} R_j$ and (ii) for every $j \in [1..t]$, 
    $R_j$ is \mbox{$(i,p)$-periodic} and $f$ is $(i,p)$-monotone locally to $R_j$.
\end{definition}

\begin{restatable}{lemma}{LemmaMonotoneOnlyXMatters}
    \superlabel{lemma:monotone-only-x-matters}{proof:LemmaMonotoneOnlyXMatters}
    Let $\phi$ be a linear-exponential program with divisions, and $x$ be a variable occurring linearly in~$\phi$.
    Suppose that the set of solutions to $\phi$ has a $(x,\fmod(x,\phi))$-monotone 
    decomposition $R_1,\dots,R_t$ for a function $f$, 
    where each $R_i$ is the set of solutions of an integer linear-exponential program with divisions~$\psi_i$ in which $x$ occurs linearly.
    If the instance~$(f,\phi)$ has a maximum (analogously, a minimum), 
    then it has one satisfying an equation 
    $a \cdot x + \tau + r = 0$ 
    such that $a \neq 0$, $(a \cdot x + \tau) \in \fterms(\psi_i \land x \geq 0)$ 
    and 
    $r \in [0..\abs{a} \cdot \fmod(x,\psi_i)-1]$, for some $i \in [1..t]$.
\end{restatable}

\begin{remark}
    \label{remark:complex-regions-unproblematic}
    Consider $\phi$, $x$, and $f$ as in~\Cref{lemma:monotone-only-x-matters}. 
    By defining $\tests(x,f,\phi)$ as the set of all equalities $a \cdot x + \tau + r = 0$ 
    specified in that lemma, we are guaranteed to explore at least one optimal solution (if optimal solutions exists). Notably, when designing a non-deterministic polynomial-time procedure, neither the number of regions nor the number of constraints required to describe each region is inherently restrictive. The key requirement is instead the ability to guess \emph{a single equation} involving the variable~$x$ targeted for elimination. This means that, while the total number of distinct constraints containing $x$ can be at most exponential across regions, the number of constraints independent of $x$ can be arbitrarily large (and these need not even be expressible as linear-exponential programs).
\end{remark}

\begin{remark}
    All lemmas in this section were formulated for linear-exponential programs with divisions. 
    However, upon inspecting their proofs, it should be evident that the only crucial assumption is the \emph{linear occurrence} of $x$.
    Indeed, these lemmas could have been stated for any expansion of integer linear programming augmented with arbitrary functions, provided that ``occurring linearly'' is interpreted as ``occurring only within the scope of addition''. 
\end{remark}

\section{Monotone decompositions for ILEP}
\label{section:proof-monotone-decomposition}

We now specialize the setting from the previous section, constructing monotone decompositions for the instances that~\GaussOpt must be able to handle 
in order for~\OptILEP to explore optimal solutions. In doing so, great attention must be paid to ensure that the formulae and ILESLPs computed throughout the procedure remain of polynomial size. 
Proving that this is the case will occupy us through the end of~\Cref{sec:putting-all-together}.

To simplify the exposition, we work under the following assumptions:
\begin{enumerate}
    \item The linear-exponential program $\phi$ in input to~$\OptILEP$ features the variables~$x_1,\dots,x_n$. 

    \item The goal is to maximize a single variable $x_m$, with $m \in [1..n]$. (\Cref{sec:putting-all-together} will relax this assumption to handle both maximization and minimization of general linear-exponential terms.)
    
    \item The ordering~$\theta$ guessed at the very beginning of~\OptILEP is $\theta \coloneqq (2^{x_n} \geq \dots \geq 2^{x_1} \geq 2^{x_0} = 1)$.
\end{enumerate}
As stated in~\Cref{subsection:OptILEP}, \OptILEP 
iteratively eliminates $x_n,\dots,x_1$.
Throughout the section, we assume that the variables $x_{n},\dots,x_{n-k+1}$ 
have already been successfully eliminated while preserving optimal solutions, 
for some~$k \in [0..n-1]$. 
We focus on the elimination of~$x_{n-k}$, 
which occurs during the $(k+1)$th iteration of the main loop of~\OptILEP. 
More precisely, 
we consider the appeal to~\GaussOpt during 
this $(k+1)$th iteration.
According to~\Cref{subsection:OptILEP}, 
this appeal is preceded by an application of Step~I 
of the procedure from~\cite{ChistikovMS24} (\Cref{lemma:CMS:first-step}),
and is followed by~Step~III of the same procedure (\Cref{lemma:CMS:third-step}).
With respect to this $(k+1)$th iteration, we define:

\begin{enumerate}
    \setcounter{enumi}{3}
    \item The vector~$\vec x_k \coloneqq (x_{n-k},\dots,x_n)$ containing all variables that have been eliminated, plus $x_{n-k}$.
    
    \item The ordering $\theta_k \coloneqq {(2^{x_{n-k}} \geq \dots \geq 2^{x_0} = 1)}$ obtained from $\theta$ by removing 
    the variables that have been eliminated. We also 
    define $\vec y_k \coloneqq (x_{0},\dots,x_{n-k})$ for the variables in this ordering.
    Note that $\vec y_k$ and $\vec x_k$ only share the variable $x_{n-k}$.

    \item The vectors ${\vec q_k \coloneqq (q_{n-k},\dots,q_n)}$ and ${\vec r_k \coloneqq (r_{n-k},\dots,r_n)}$ of the \emph{quotient variables} and \emph{remainder variables} introduced during the $(k+1)$th iteration the main loop of~\OptILEP (accordingly to~\Cref{lemma:CMS:first-step}). 
    We assume the variables $q_1,\dots,q_n,r_1,\dots,r_n$ to be reused at each iteration; e.g., $\vec q_{k-1} = (q_{n-(k-1)},\dots,q_n)$ is the vector of quotient variables used at the $k$th iteration.
    Given $\ell \in [0..k]$, we also write $\vec q_{[\ell,k]}$ for the vector $(q_{n-k},\dots,q_{n-\ell})$.
\end{enumerate}

\subsection{Setup}
\label{subsec:setup-ilep}
We now introduce a class of circuits that model the evolution of the objective function during the execution of~\OptILEP, and characterize the class of objective functions and systems of constraints for which we design our monotone decomposition.

\paragraph*{Linear-Exponential arithmetic circuits.}
In~\Cref{subsection:where-are-the-ILESLP} we gave a bottom-up argument of how the procedure from~\cite{ChistikovMS24} 
constructs ILESLPs. An analogous top-down perspective shows the construction of ILESLPs by progressively 
building \emph{Linear-Exponential Arithmetic Circuits} (\emph{LEACs}):

\begin{definition}[LEAC]
    \label{def:LEAC}
    Let $k \in [0..n-1]$ and $\ell \in [0..k]$.
    An $(k,\ell)$-linear-exponential arithmetic circuit~$C$ ---a $(k,\ell)$-LEAC, in short--- is a sequence of assignments 
    of the form
    \begin{align*}
        \hspace{0.4cm} q_{n-i} & \gets \frac{\tau_{n-i}(u, \vec q_{[\ell,k]})}{\eta} 
            &\text{for $i$ from $\ell-1$ to $0$},\\
        \hspace{0.4cm} x_{n-i} & \gets \frac{\textstyle\sum_{j={i+1}}^{k} a_{i,j} \cdot
        2^{x_{n-j}}}{\mu} +
        q_{n-i} \cdot 2^{x_{n-k-1}} + r_{n-i} &\hspace{-3cm}\text{for $i$ from $k$ to $0$},
    \end{align*}
    where each $\tau_{n-i}$ is a linear term (with integer coefficients), 
    every $a_{i,j}$ is in $\Z$, 
    and the denominators $\eta$ and $\mu$ are positive integers. 
    We refer to these denominators as $\mu_C$ and $\eta_C$, respectively.
    When $k = 0$, the sum $\textstyle\sum_{j={i+1}}^{k} a_{i,j} \cdot
    2^{x_{n-j}}$ equals $0$, and so we are free to update $\mu_C$ to any positive integer; 
    we postulate $\mu_C \coloneqq 1$ in this case.
    If $\ell = 0$, then $\eta_C$ is undefined; for technical reasons we postulate $\eta_C \coloneqq \mu_C$ in this case.
    Moreover, we define $\xi_C \coloneqq \sum\{\abs{a_{i,j}}: \text{$i \in [0..k]$, $j \in [i+1..k]$}\}$, and write $\vars(C)$ for the set of \emph{free variables} 
    of $C$, i.e., $u$, $x_{n-k-1}$, and those in the vectors $\vec q_{[\ell,k]}$ and $\vec r_k$. 
\end{definition}

\paragraph*{Objective functions.}
Consider a $(k,\ell)$-LEAC $C$, and a variable $x_m$ with $m \in [1..n]$. We denote by $\objfun{C}{x_m}$ the (objective) function defined as follows: If~${n-k > m}$, then $\objfun{C}{x_m}$ takes as input maps $\nu \colon X \to \N$ where $X$ is a set of variables featuring $x_m$, and 
returns $\nu(x_m)$. Else, $\objfun{C}{x_m}$ takes as input maps $\nu \colon X \to \N$ such that $\vars(C) \subseteq X$, and outputs 
the number $\objfun{C}{x_m}(\nu)$ computed as follows: 
{\setstretch{1.1}
    \begin{algorithmic}[1]
        \State update $C$: replace each variable $z \in \vars(C)$ with $\nu(z)$ 
        \State evaluate $C$ 
        \Comment{each assignment becomes $y \gets a$ where $a$ is a number}
        \State \textbf{return} the number assigned to~$x_m$ in $C$
    \end{algorithmic}
}

\label{para:instances-for-opt-ilep}
\paragraph*{The instances.} For our purposes, it suffices to define the monotone decomposition for elements of a set $\bigcup_{k=0}^{n-1}\bigcup_{\ell=0}^{k-1}  \objcons_{k}^\ell$, 
where $\objcons_{k}^\ell$ (defined also for $\ell = k$) is the set of all pairs $(C,\inst{\gamma}{\psi})$ such~that:
\begin{enumerate}[label=(\roman*)]
    \item\label{objcons:i1} $C = (y_1 \gets \rho_1,\dots,y_t \gets \rho_t)$ is a $(k,\ell)$-LEAC such that $\mu_C$ divides $\eta_C$, 
    as well as all coefficients of the variables $\vec q_{[\ell,k]}$ occurring in the term $\tau_{n-i}$
    featured in assignments $q_{n-i} \gets \frac{\tau_{n-i}}{\eta_{C}}$ of $C$, with $i \in [0..\ell-1]$. 
    (Recall that $\eta_{C} \coloneqq \mu_{C}$ for $\ell = 0$.)
    \item\label{objcons:i2} The formula $\inst{\gamma}{\psi}$ is a \emph{conjunction} of a linear program with divisions $\gamma(u,\vec q_{[\ell,k]})$ 
    and a linear-exponential program with divisions $\psi$. Inequalities and equalities in $\gamma$ are such that all the coefficients of the variables~$\vec q_{[\ell,k]}$ 
    are divisible by $\mu_C$. Moreover, for every $q$ in $\vec q_{[\ell,k]}$, $\gamma$~contains an inequality $a \cdot q \geq 0$, for some $a \geq 1$ (divisible by $\mu_C$).
    The system $\psi$ is of the form $\chi(\vec y_{k-1}, \vec r_k) \land \theta_k \land (x_{n-k} = q_{n-k} \cdot 2^{x_{n-k-1}} + r_{n-k}) \land 
    (u = 2^{x_{n-k}-x_{n-k-1}})$.

    (We prefer writing $\inst{\gamma}{\psi}$ instead of $\gamma \land \psi$ as it emphasize more the distinction between $\gamma$ and~$\psi$. \GaussOpt only updates $\gamma$, treating $\psi$ as an invariant used to ensure correctness.)
    \item\label{objcons:i3} The formula $\inst{\gamma}{\psi}$ implies the formula $\Psi(C)$ defined as 
    \begin{align*}
        \vec 0 \leq \vec r_k < 2^{x_{n-k-1}} \land \exists \vec q_{[0,\ell-1]}
        \Big(\vec 0 \leq \vec q_k \cdot 2^{x_{n-k-1}}+ \vec r_k < 2^{x_{n-k}}  \land
        \exists \vec x_{k-1} \big( \theta \land \textstyle\bigwedge_{i=1}^{t} (y_i = \rho_i) \big)\Big).
    \end{align*}
    (For $\ell = 0$, simply conjoin $\vec 0 \leq \vec r_k < 2^{x_{n-k-1}}$ with the formula in the scope of $\exists \vec q_{[0,\ell-1]}$.
    Analogously, for $k = 0$, the subformula $\exists \vec x_{k-1} \big( \theta \land \textstyle\bigwedge_{i=1}^{t} (y_i = \rho_i) \big)$ becomes $\theta \land \textstyle\bigwedge_{i=1}^{t} (y_i = \rho_i)$.)
\end{enumerate}
We remark that the variables that are quantified in $\Psi(C)$ are those assigned to some expression in $C$, excluding~$x_{n-k}$. 
In essence, elements of $\objcons_k^\ell$ satisfy certain basic properties that are sufficient to obtain a monotone decomposition. (While our proof relies on all of these properties, it remains unclear whether they are truly necessary for achieving a monotone decomposition.)
For example, the subformula $\exists \vec x_{k-1} \big( \theta \land \textstyle\bigwedge_{i=1}^{t} (y_i = \rho_i) \big)$
appearing in $\Psi(C)$
ensures that from any solution of $\inst{\gamma}{\psi}$, we can assign values to the eliminated variables $x_n,\dots,x_{n-k+1}$ 
that preserve the ordering~$\theta$, simply by following the assignments defined in the LEAC $C$.
This is consistent with the goal of~\OptILEP of finding a solution to the input integer linear-exponential program respecting~$\theta$.

\subsection{The monotone decomposition}

We are now ready to formalize our monotone decomposition:

\begin{restatable}{proposition}{PropMonotoneDecomposition}\label{prop:monotone-decomposition}
    Consider $(C,\inst{\gamma}{\psi}) \in \objcons_k^{\ell}$, 
    with $\ell < k$,
    and~$p \coloneqq \fmod(q_{n-\ell},\gamma)$.
    The set of solutions to~${\inst{\gamma}{\psi}}$  
    has a $(q_{n-\ell},p)$-monotone decomposition 
    $R_1,\dots,R_t$ for the function $\objfun{C}{x_m}$.
    Each~$R_i$ is the set of solutions of a linear-exponential program with divisions~$\phi_i$ 
    satisfying ${\fmod(q_{n-\ell},\phi_i) = p}$,
    and in which all constraints featuring~$q_{n-\ell}$ 
    are either from $\gamma \land \gamma\sub{q_{n-\ell}+p}{q_{n-\ell}}$,
    or they are inequalities~${\tau \leq 0}$,
    where $\tau$ is a term (non-deterministically) returned by~\Cref{algo:additional-hyperplanes}.
\end{restatable}

\begin{algorithm}[t]
    \setstretch{1.1}
    \caption{Additional terms for the monotone decomposition.}
    \label{algo:additional-hyperplanes}
    \begin{algorithmic}[1]
        \Statex \Comment{see~\Cref{prop:monotone-decomposition} for the definitions of $C$, $\gamma$ and $p$}
        \State $(a,d)$ $\gets$ \myguess an element from $[-L..L]^2$, where $L \coloneqq 3 \cdot \mu_C \cdot (4 \cdot \ceil{\log_2(2 \cdot \xi_C + \mu_C)}+8)$\label{algo:btp:guess-d}
        \State $\lambda \gets \frac{\eta_C}{\mu_C}$\label{algo:btp:lambda}
        \State $q',q''$ $\gets$ \myguess two elements in $\vec q_{k-1}$ (they can be equal)\label{algo:btp:guessqq}
        \State $\tau' \gets$ \textbf{if} $C$ assigns an expression $\frac{\tau}{\eta_C}$ to $q'$ \textbf{then} $\tau$ \textbf{else} $\eta_C \cdot  q'$ \label{algo:btp:line-tau1}
        \State $\tau'' \gets$ \textbf{if} $C$ assigns an expression $\frac{\tau}{\eta_C}$ to $q''$ \textbf{then} $\tau$ \textbf{else} $\eta_C \cdot  q''$ \label{algo:btp:line-tau2}
        \State \textbf{if} $\ast$ \textbf{then} $\tau' \gets \tau'\sub{q_{n-\ell} + p}{q_{n-\ell}}$ \label{algo:btp:line-shift-tau1}
        \Comment{$\ast$ stands for non-deterministic choice}
        \State \textbf{if} $\ast$ \textbf{then} $\tau'' \gets \tau''\sub{q_{n-\ell} + p}{q_{n-\ell}}$ \label{algo:btp:line-shift-tau2}
        \State $(b \cdot q_{n-\ell} - \rho)$ $\gets$\label{algo:btp:term-before-shift}
        term obtained from $(a \cdot u + \mu_C \cdot (q' - q'') + d)$ by simultaneously applying\label{algo:btp:border} 
        \Statex \hphantom{$(b \cdot q_{n-\ell} - \rho)$ $\gets$} the substitutions $\sub{\frac{\tau'}{\lambda}}{\mu_C \cdot q'}$ and $\sub{\frac{\tau''}{\lambda}}{\mu_C \cdot q''}$
        \State \textbf{assert}($b \neq 0$)\label{algo:btp:assert} 
        \Comment{else, reject this non-deterministic branch}
        \State \textbf{return} $(b \cdot q_{n-\ell} - \rho)$\label{algo:btp:return}
    \end{algorithmic}
\end{algorithm}

The remainder of this section is dedicated to proving~\Cref{prop:monotone-decomposition}. At this stage, we are unable to motivate why the formulae $\phi_i$ given in \Cref{prop:monotone-decomposition} suffice for obtaining a monotone decomposition; rather, these the formulae that naturally emerge when trying to build such a decomposition.

\paragraph*{Preliminary results.}
Before proceeding with the proof of~\Cref{prop:monotone-decomposition}, we need a few lemmas.
The first is a small technical result giving sufficient conditions under which an expression  
of the form $2^{C} - 2^{C/2} - d \cdot C$ is non-negative (where $d,C \in \R$).
Ultimately, this lemma plays a role in the definition of the quantity~$L$ 
defined in line~\ref{algo:btp:guess-d} of~\Cref{algo:additional-hyperplanes}.

\begin{restatable}{lemma}{LemmaPosAnalysisConstant}
    \superlabel{lemma:pos-analysis-constant}{proof:LemmaPosAnalysisConstant}
    Let $d,C \in \R$ with \(d \geq 1\) and \(C \geq 4 \cdot \log_2(d) + 8\). Then,
    \(2^C - 2^{C/2} - d \cdot C \geq 0\). 
\end{restatable}

The next lemma echoes some ideas firstly used by Semenov for proving the decidability of Presburger arithmetic enriched with the exponential function~\cite{Semenov84}. In a nutshell, it establishes that, under appropriate hypotheses, any inequality of the form $\sum_{i=1}^\ell a_i \cdot 2^{x_i} + \mu \cdot y + \mu \cdot d \leq 0$ can be reduced to true (in the lemma, $0 \leq 0$), false ($1 \leq 0$), or a simplified inequality where 
the sum $\sum_{i=1}^\ell a_i \cdot 2^{x_i}$ is replaced with a single exponential term $a \cdot 2^{x_1}$. In our case, this lemma will play a central role in characterizing the regions of our monotone decomposition.

\begin{lemma}
    \label{lemma:rewriting-monotone-hyperplane}
    Let $E$ be an expression $\sum_{i=1}^\ell a_i \cdot 2^{x_i} + \mu \cdot y + \mu \cdot d$, where each $a_i$ is in $\Z$, and $\mu,d \in \N$.
    Let $M,k \in \N$ such that 
    $M \geq \max(1+2 \cdot \log_2(\sum_{i=1}^{\ell} \abs{a_i}+k \cdot \mu),\ 4 \cdot \log_2(\mu) + 8,\ d)$.
    Also consider a formula $\psi(x_1,\dots,x_\ell,y)$, with $y$ ranging over $\Z$ and $x_1,\dots,x_\ell$ ranging over $\N$, of the form
    \[ 
        \psi(x_1,\dots,x_\ell,y) \ \coloneqq\  -k \cdot 2^{x_1} \leq y \leq k \cdot 2^{x_1} \land \bigwedge\nolimits_{i = 1}^{\ell-1} (x_{i+1} \sim_i x_i + d_i),
    \]
    where each pair $(\sim_i,\,d_i)$ is either $(\geq,\,M)$ or is of the form $(=,\,g)$, with $g \in [0..M-1]$. 
    Let~\({{\sim} \in \{\leq, =\}}\).
    There is an expression $E'$ 
    from the set ${\{0,1\} \cup \{ a \cdot 2^{x_{1}} + \mu \cdot y + \mu \cdot d \,:\, a \in [-4^{\ell \cdot M}..4^{\ell \cdot M}]\}}$ such that the formula $\psi$ 
    implies $(E \sim 0 \iff E' \sim 0)$.
\end{lemma}

\begin{proof}
    First, observe that if $\ell = 0$, then we can take ${E' \coloneqq E}$. 
    Hence, below let us assume $\ell \geq 1$.
    Note that $\psi$ implies the ordering $x_\ell \geq x_{\ell-1} \geq \dots \geq x_1$.
    By induction on $j$ from $1$ to $\ell$, 
    we show that for every expression 
    $E_j$ of the form ${h_j \cdot 2^{x_{j}} + \sum_{i=1}^{j-1} a_i \cdot 2^{x_i} + \mu \cdot y + \mu \cdot d}$, 
    with $\abs{h_j} \leq 4^{(\ell-j+1) M}$, 
    there is an expression $E_j'$ 
    such that $\psi$ implies ${(E_j \sim 0 \iff E_j' \sim 0)}$, 
    and $E_j'$ belongs to the set 
     \(\{0,1\} \cup \{ a \cdot 2^{x_1} + \mu \cdot y + \mu \cdot d \,:\, a \in [-\abs{h_j}\cdot 4^{(j-1)  M}..\abs{h_j}\cdot 4^{(j-1)  M}]\}\).
    The lemma then follows from the fact that 
    $E$ is an expression of the form of $E_m$, 
    setting $a_m = h_m$; 
    as indeed $\abs{a_m} \leq 4^{M}$.


    \begin{description}
        \item[base case: $j = 1$.] For every expression $E_1$ of the form $h_{1} \cdot 2^{x_1} + \mu \cdot y + \mu \cdot d$, we can take $E_1' \coloneqq E_1$.
        \item[induction hypothesis.] Given $j > 1$, from every $E_{j-1} = h_{j-1} \cdot 2^{x_{j-1}} + {\sum_{i=1}^{j-2} a_i \cdot 2^{x_i} + \mu \cdot y + \mu \cdot d}$, 
        with $\abs{h_{j-1}} \leq 4^{(\ell - j + 2)  M}$, 
        there is an expression $E_{j-1}'$ from the set 
        \(\{0,1\} \cup \{ a \cdot 2^{x_1} + \mu \cdot y + \mu \cdot d \,:\, a \in [-\abs{h_{j-1}}\cdot 4^{(j-1)  M}..\abs{h_{j-1}}\cdot 4^{(j-1)  M}]\}\),
        such that $\psi$ implies $(E_{j-1} \sim 0 \iff E_{j-1}' \sim 0)$.
        \item[induction step: $j > 1$.]  
            Consider an expression $E_j$ of the form $h_j \cdot 2^{x_{j}} + \sum_{i=1}^{j-1} a_i \cdot 2^{x_i} + \mu \cdot y + \mu \cdot d$, with $\abs{h_j} \leq 4^{(\ell-j+1)M}$. 
            If $h_j = 0$, then we directly obtain $E_j'$ by 
            applying the induction hypothesis on the expression~$\sum_{i=1}^{j-1} a_i \cdot 2^{x_i} + \mu \cdot y + \mu \cdot d$, 
            since from the assumption on $M$ in 
            the statement of the lemma, we have $\abs{a_{j-1}} \leq 4^M$.
            Below, let us assume then that $h_j \neq 0$. We distinguish two cases, depending on whether the constraint $x_{j} \sim_{j-1} x_{j-1} + d_i$ 
            occurring in $\psi$
            is 
            of the form $x_{j} \geq x_{j-1} + M$ or $x_{j} = x_{j-1} + g$ 
            for some $g \in [0..M-1]$.

            \begin{description}
                \item[case: $x_{j} \geq x_{j-1} + M$ {\rm occurs in} $\psi$.]
                    Intuitively, in this case $\psi$ is constraining $x_j$ to be so large comparatively to $x_{j-1}$ that, in any solution to $\psi$, $E_j \neq 0$ and the sign of $E_j$ 
                    is solely dictated by the sign of $h_j$.
                    When~$\sim$ from $E \sim 0$
                    is the equality symbol, we can pick $E_j' = 1$, making $E_j' \sim 0$ unsatisfiable. 
                    When $\sim$ is instead $\leq$, we set $E_j' = 0$ if $h_j$ is negative, 
                    and $E_j' = 1$ otherwise. 
                    To show that $h_j$ dictates the sign of $E_j$,
                    it suffices to establish that $\psi$ implies 
                    $2^{x_j} > \abs{\sum_{i=1}^{j-1} a_i \cdot 2^{x_i} + \mu \cdot y + \mu \cdot d}$.
                    First, note that $\psi$ implies $2^{x_j} \geq 2^M \cdot 2^{x_{j-1}}$. 
                    As $M \geq d$ and $\psi$ implies 
                    both $x_{j-1} \geq \dots \geq x_1$ and $\abs{y} \leq k \cdot 2^{x_1}$, 
                    we~have
                    \begin{align*} 
                        \abs{\,\sum\nolimits_{i=1}^{j-1} a_i\cdot 2^{x_i} + \mu \cdot y + \mu \cdot d\,} 
                        &\,\leq\, \sum\nolimits_{i=1}^{j-1} \abs{a_i} 2^{x_i} + \abs{\mu \cdot y} + \mu \cdot d\\ 
                        &\,\leq\,
                        \left(\sum\nolimits_{i=1}^{j-1} \abs{a_i} + k\cdot \mu + M\cdot \mu\right)\cdot 2^{x_{j-1}}.
                    \end{align*}
                    Therefore, it suffices to show that $2^M > \sum_{i=1}^{j-1} \abs{a_i} + k\cdot \mu + M\cdot \mu$; 
                    or equivalently that $2^M - M\cdot \mu > \sum_{i=1}^{j-1} \abs{a_i} + k\cdot \mu$.
                    Since \(M \geq 4\log_2(\mu) + 8\),
                    by~\Cref{lemma:pos-analysis-constant}
                    we have \(2^M - M \cdot \mu \geq 2^{M/2}\). Then, 
                    by $M > 2 \cdot \log_2(\sum_{i=1}^{\ell} \abs{a_i} + k \cdot \mu)$,
                    \begin{equation*}
                        2^M - \mu \cdot M \geq 2^{M/2} > 2^{\lceil\log_2(\sum \abs{a_i} + k \cdot \mu)\rceil} \geq 
                        \sum\nolimits_{i=1}^{j-1} \abs{a_i} + k \cdot \mu.
                    \end{equation*}
                    
                \item[case: $x_{j} = x_{j-1} + g$ {\rm occurs in} $\psi$.] 
                    Let $E_{j-1}$ be the expression obtained from $E_j$ 
                    by replacing $2^{x_j}$ by $2^g \cdot 2^{x_{j-1}}$; 
                    that is, $E_{j-1} = (h_j \cdot 2^g + a_{j-1}) \cdot 2^{x_{j-1}} + \sum_{i=1}^{j-2} a_i \cdot 2^{x_i} + \mu \cdot y + \mu \cdot d$.
                    We have $\psi$ implies $(E_j \sim 0 \iff E_{j-1} \sim 0)$.
                    To conclude the proof, 
                    it suffices to prove that the induction hypothesis can be applied to $E_{j-1}$. 
                    This is the case as soon as $\abs{h_j \cdot 2^g + a_{j-1}} \leq 4^{(\ell-j+2) M}$ holds, which we show below:
                    \begin{align*}
                        \abs{h_j \cdot 2^g + a_{j-1}} 
                        &\leq
                        \abs{h_j} \cdot 2^g + \abs{a_{j-1}}\\
                        &\leq 4^{(\ell-j+1)M} \cdot 2^M + 2^M 
                            &\Lbag \text{bounds on $\abs{h_j}$, $g$ and $M$}\Rbag\\
                        &\leq 2^{2(\ell-j+1)M + 2M}
                            &\Lbag \text{recall: $M \geq 8$}\Rbag\\
                        &\leq 4^{(\ell-j+2)M}.
                        &&\qedhere
                    \end{align*} 
            \end{description}
    \end{description}
\end{proof}

The set of possible expressions $E'$ appearing in the conclusion of~\Cref{lemma:rewriting-monotone-hyperplane} can be significantly reduced by noticing 
that if the absolute value of the integer $a$ exceeds $\mu \cdot (k + d)$, then the truth of $E' \sim 0$ becomes independent of the value given to~$x_1$. 
This observation allows us to eliminate the polynomial dependence of $a$ 
on the magnitude of the coefficients $a_1,\dots,a_\ell$ in the original expression $E$.
Although not obvious, it turns out that this plays a critical point in the proof of~\Cref{theorem:small-optimum}. Retaining values of $a$ with polynomial dependence on $a_1,\dots,a_\ell$ would cause the integers in the LEAC constructed by~\OptILEP to grow polynomially within each variable elimination step. As a result, their bit sizes would become exponential by the end of the procedure.
(In the proof of~\Cref{prop:monotone-decomposition}, we will apply~\Cref{lemma:rewriting-monotone-hyperplane} using a value for the integer $d$ that depends only logarithmically on $a_1,\dots,a_\ell$; hence, the resulting values of $a$ do in fact retain a logarithmic dependence on these coefficients.)
The next lemma gives the refined set of expressions~$E'$. 

\begin{lemma}
    \label{lemma:rewriting-monotone-hyperplane-2}
    Let the expression $E$, the non-negative integers $\mu$, $d$, $M$ and $k$, the formula $\psi$, and the symbol $\sim$ be defined as in~\Cref{lemma:rewriting-monotone-hyperplane}.
    Let $b \coloneqq \mu \cdot (k + d)$.
    There is an expression $E'$ 
    from the set $\{0,1\} \cup \{ a \cdot 2^{x_{1}} + \mu \cdot y + \mu \cdot d \,:\, a \in [-b..b]\}$ such that the formula $\psi$ 
    implies $(E \sim 0 \iff E' \sim 0)$.
\end{lemma}

\begin{proof}
    Following
    \Cref{lemma:rewriting-monotone-hyperplane},
    it suffices to take $E'$ to be an expression of the form $a \cdot 2^{x_{1}} + \mu \cdot y + \mu \cdot d$, where $a \in [-4^{\ell \cdot M}..4^{\ell \cdot M}]$, 
    and show that if $a$ lies outside $[-b..b]$,
    then $E'$ can be rewritten to $0$ or~$1$.

    From its definition, $\psi$ implies 
    $\abs{\mu \cdot y + \mu \cdot d} 
    \leq \mu \cdot (k + d) \cdot 2^{x_1}$. 
    Hence, as soon as $\abs{a} > b$, the truth of $E' \sim 0$ 
    is determined by the sign of $a$. 
    In particular, whenever $\sim$ is the symbol $=$, or when $a > 0$,  
    the expression $E'$ can be replaced with $1$; that is, in this case $\psi$ implies $\lnot (E' \sim 0)$. 
    Otherwise, $E'$ can be replaced with $0$;
    and in this case $\psi$ implies $E' \sim 0$.
\end{proof}

\paragraph*{Proof of~\Cref{prop:monotone-decomposition}.} 
We are now ready to prove~\Cref{prop:monotone-decomposition}. 
While long, the proof is divided in several steps and claims.
An intuition of the construction is given after defining various objects 
required for the proof; see the paragraph titled~``Construction of \(\psi_1, \dots, \psi_s\): some intuition''.








\fancyhead[R]{{\color{gray}\rightmark\ (proof of~\Cref{prop:monotone-decomposition})}}%
\begin{proof}
    Throughout the proof, we let $\phi \coloneqq \inst{\gamma}{\psi}$.
    Let $X$ be the set of all variables appearing in $\vec q_{[\ell,k]}$, $\vec r_k$, $x_0,\dots,x_{n-k}$ and~$u$. 
    Remark that all variables occurring in $\phi$ are among the set $X$.
    The value~$\objfun{C}{x_m}(\nu)$ of objective function $\objfun{C}{x_m}$ is defined for every map $\nu \colon X \to \N$, and corresponds to the value taken by $x_m$ when evaluating $C$ on $\nu$.
    During the proof, we refer to the variable \(q_{n-\ell}\) simply as \(q\).
    By the definition of $\objcons_k^\ell$, the variable $q$ appears linearly 
    in~$\gamma$ and does not appear in $\psi$. Thus, by~\Cref{lemma:mod-periodicity}
    the set of solutions to $\phi$ is $(q,p)$-periodic, 
    where $p = \fmod(q,\gamma) = \fmod(q,\phi)$.
    
    Let us begin by considering the (corner) case where $x_m$ satisfies $n-k \geq m$. 
    When $n-k > m$, the variable $x_m$ does not occur in $C$. Consequently, the function $\objfun{C}{x_m}$ is constant in the variable~$q$ (that is,
    for any solution $\nu \colon X \to \N$ 
    to $\phi$, the function $\objfun{C}{x_m}$ is constant on maps 
    of the form $\nu + [q \mapsto j]$ for $j \in \Z$).
    Similarly, when $m = n-k$, the circuit $C$ assigns to the variable $x_m$ the expression $q_{n-k} \cdot 2^{x_{n-k-1}} + r_{n-k}$. As all involved variables ($q_{n-k}$, $x_{n-k-1}$ and $r_{n-k}$) belong to $X$ and are distinct from~$q$, once again we obtain that the function $\objfun{C}{x_m}$ is constant in the variable~$q$. 
    We conclude that if $m \leq n-k$, 
    then $\objfun{C}{x_m}$ is $(q,p)$-monotone locally to $\phi$. 
    By~\Cref{lemma:monotone-gaussian-elimination},
    for a monotone decomposition it thus suffices to take a single set $R_1$ given by the set of solutions to $\phi$.

    In the remaining of the proof, 
    we assume that $x_m$ satisfies $m > n-k$.
    This means that $C$ contains an assignment to $x_m$, 
    and that $x_m$ is not $x_{n-k}$. 
    We will construct a sequence of linear-exponential programs $\psi_1, \dots, \psi_s$
    satisfying the following conditions:
    \begin{enumerate}
        \item[\labeltext{I}{proof:monotone-decomp:item:interior-cover}.] The formula $\phi \land \phi\sub{q+p}{q}$ implies $\psi_1 \lor \dots \lor \psi_s$.
        \item[\labeltext{II}{proof:monotone-decomp:item:regions-monotone}.] $\objfun{C}{x_m}$ is $(q,p)$-monotone locally to (the solutions of) $\phi \land \phi\sub{q+p}{q} \land \psi_i$, for every $i \in [1..s]$.
        \item[\labeltext{III}{proof:monotone-decomp:item:interior-decomp-form}.] Each $\psi_i$ is of the form $\bigwedge_{j=1}^{n_i} \psi_{i,j}$, where every $\psi_{i,j}$ 
        is an inequality with the following property. 
        Let $\psi_{i,j}$ be $\tau \leq 0$ (and note that then $- \tau - 1 \leq 0$ is equivalent to $\lnot \psi_{i,j}$).
        If $q$ occurs in $\psi_{i,j}$, 
        both $\tau$ and $- \tau -1$ are terms 
        returned by a non-deterministic branch of 
        the execution of~\Cref{algo:additional-hyperplanes} with respect to~$C$, $\gamma$ and $p$. (This is as required in the statement of the proposition; remark also that in $\phi \land \phi\sub{q+p}{q}$, all constraints featuring $q$ are from $\gamma \land \gamma\sub{q+p}{q}$.)
    \end{enumerate}
    Then, the desired $(q,p)$-monotone decomposition $R_1,\dots,R_t$ 
    is given by the formulae in the set 
    \[
        \Big\{\phi \land \phi\sub{q+p}{q} \land \psi_i : i \in [1..s]\Big\} \cup \Big\{\phi \land \bigwedge\nolimits_{i=1}^s \lnot \psi_{i,f(i)} : f \in \mathcal{G}\Big\},
    \] 
    where $\mathcal{G}$ is the set of all function $f \colon [1..s] \to \N$ such that $f(i) \in [1..n_i]$ for every $i \in [1..s]$. 
    Indeed, the function $\objfun{C}{x_m}$ is $(q,p)$-monotone locally to the formula $\phi \land \lnot \phi\sub{q+p}{q}$, 
    and so also locally to any formula $\phi \land \bigwedge_{i=1}^s \lnot \psi_{i,f(i)}$; as the latter implies the former by Item~\ref{proof:monotone-decomp:item:interior-cover}.
    Additionally, the sets of solutions of all the formulae $\phi \land \phi\sub{q+p}{q} \land \psi_i$ and $\phi \land \bigwedge_{i=1}^s \lnot \psi_{i,f(i)}$ are $(q,p)$-periodic, since~$q$ occurs linearly in
    these formulae, and none of the $\psi_i$ contains any divisibility constraints.
    (A side remark: the number of formulae $\psi_1,\dots,\psi_s$ will be exponential in the size of $\phi$, making the number of constraints in each formula $\bigwedge_{i=1}^s \lnot \psi_{i,f(i)}$ also exponential. As explained in~\Cref{remark:complex-regions-unproblematic},  
    this is unproblematic: to eliminate $q$, it suffices to guess a \emph{single} constraint from any of these formulae.)




    \paragraph{\textit{Construction of \(\psi_1, \dots, \psi_s\): a preliminary step.}} We begin by introducing an $\ell$-LEAC $C^{+p}$ that will be associated to the formula $\phi\sub{q+p}{q}$. We define $C^{+p}$ as the $\ell$-LEAC obtained from $C$ by replacing $q$ with $q+p$, and renaming to $\overline{v}$ every variable $v$ among $q_{n-\ell+1},\dots,q_{n},x_{n-k+1},\dots,x_n$. These are the variables to which $C$ assigns an expression, and whose value may depend on the value given to $q$. To clarify, if $C$ is defined as 
    \begin{align*}
            &q_{n-i} \gets \frac{\tau_{n-i}}{\eta} 
                &\text{for $i$ from $\ell-1$ to $0$},\\[3pt]
            &x_{n-i} \gets \frac{\textstyle\sum_{j={i+1}}^{k} a_{i,j} \cdot
            2^{x_{n-j}}}{\mu} +
            q_{n-i} \cdot 2^{x_{n-k-1}} + r_{n-i}
                &\text{for $i$ from $k$ to $0$},
    \end{align*}
    then $C^{+p}$ is defined as 
    \begin{align*}  
            &\overline{q}_{n-i} \gets \frac{\tau_{n-i}\sub{q+p}{q}}{\eta} 
                &\text{for $i$ from $\ell-1$ to $0$},\\[3pt]
            &x_{n-k} \gets q_{n-k} \cdot 2^{x_{n-k-1}} + r_{n-k},\\[3pt]
            &\overline{x}_{n-i} \gets \frac{\textstyle a_{i,k} \cdot 2^{x_{n-k}} + \sum_{j={i+1}}^{k-1} a_{i,j} \cdot 2^{\overline{x}_{n-j}}}{\mu} + q_{n-i} \cdot 2^{x_{n-k-1}} + r_{n-i}
                &\text{for $i$ from $k-1$ to $\ell+1$}\\[3pt]
            &\overline{x}_{n-\ell} \gets \frac{\textstyle a_{\ell,k} \cdot 2^{x_{n-k}} + \sum_{j={i+1}}^{k-1} a_{\ell,j} \cdot 2^{\overline{x}_{n-j}}}{\mu} + (q_{n-\ell} + p) \cdot 2^{x_{n-k-1}} + r_{n-i}\\[3pt]
            &\overline{x}_{n-i} \gets \frac{\textstyle a_{i,k} \cdot 2^{x_{n-k}} +  \sum_{j={i+1}}^{k} a_{i,j} \cdot 2^{\overline{x}_{n-j}}}{\mu} + \overline{q}_{n-i} \cdot 2^{x_{n-k-1}} + r_{n-i}
                &\text{for $i$ from $\ell-1$ to $0$}.
    \end{align*}
    To simplify the presentation, we introduce 
    the symbolic aliases:
    \begin{align*}
        (z_1,\dots,z_{2k+1}) &\coloneqq (x_n,x_{n-1},\dots,x_{n-k+1},\overline{x}_n,\overline{x}_{n-1},\dots,\overline{x}_{n-k+1},x_{n-k}),\\ 
        (y_1,\dots,y_{2k+1}) &\coloneqq (q_n,q_{n-1},\dots,q_{n-k+1},\overline{q}_n,\overline{q}_{n-1},\dots,\overline{q}_{n-\ell+1},(q_{n-\ell}+p),q_{n-\ell-1}\dots,q_{n-k}),\\ 
        (s_1,\dots,s_{2k+1}) &\coloneqq (r_n,r_{n-1},\dots,r_{n-k+1},r_n,r_{n-1},\dots,r_{n-k+1},r_{n-k}).
    \end{align*}

    \noindent
    For every $j \in [1..k]$ (resp.~$j \in [k+1..2k]$), the variables $y_j$ and $s_j$ represent the quotient and remainder variables occurring in the expression assigned to $z_j$ in $C$ (resp.~$C^{+p}$).
    Note that both $C$ and $C^{+p}$ include the assignment $x_{n-k} \gets q_{n-k} \cdot 2^{x_{n-k-1}} + r_{n-k}$, which, using of our aliases, is expressed as $z_{2k+1} \gets y_{2k+1} \cdot 2^{x_{n-k-1}} + s_{2k+1}$.
    Additionally, we introduce symbolic aliases for all variables to which $C$ or $C^{+p}$ assign an expression:
    \[ 
        (w_1,\dots,w_{2(k+\ell)+1}) \coloneqq (\underbrace{z_1,\dots,z_{k},y_1,\dots,y_{\ell}}_{\text{assigned in $C$}},\underbrace{z_{k+1},\dots,z_{2k},y_{k+1},\dots,y_{k+\ell}}_{\text{assigned in $C^{+p}$}},\underbrace{z_{2k+1}}_{\hspace{-5pt}\text{a.k.a.}~x_{n-k}\hspace{-5pt}}).
    \]
    For $j \in [1..2(k+\ell)+1]$, let $\rho_j$ denote the expression assigned to the variable $w_{j}$ 
    in either $C$ or $C^{+p}$. Since $x_{n-k}$ is the only variable shared between the two circuits, the definition of $\rho_j$ is unambiguous. In particular, $\rho_{2(k+\ell)+1}$ corresponds to the expression $q_{n-k} \cdot 2^{x_{n-k-1}} + r_{n-k}$.

    Recall that $\phi$ implies the formula $\Psi(C)$ defined as
    \begin{align*}
        \Psi(C) \coloneqq \vec 0 \leq \vec r_k < 2^{x_{n-k-1}} \land \exists \vec q_{[0,\ell-1]}: 
        \Big(&\vec 0 \leq \vec q_k \cdot 2^{x_{n-k-1}}+ \vec r_k < 2^{x_{n-k}}  \land{}\\
        &\hspace{-1cm}\exists x_{n-k+1} \dots \exists x_n \,  \, \big( \theta \land \textstyle\bigwedge_{i=1}^{k+\ell} (w_i = \rho_i) \land (w_{2(k+\ell)+1} = \rho_{2(k+\ell)+1})\big)\!\Big).
    \end{align*}
    Similarly, it is simple to see that $\phi\sub{q+p}{q}$ 
    implies the formula $\Psi(C^{+p})$ defined as 
    \begin{align*}
        \Psi(C^{+p}) \coloneqq \vec 0 \leq \vec r_k < 2^{x_{n-k-1}} \land \exists \overline{\vec q}_{[0,\ell-1]}: 
        \Big(&
        \bigwedge\nolimits_{j=0}^{\ell-1} \big(0 \leq \overline{q}_{n-j} \cdot 2^{x_{n-k-1}}+ r_{n-j} < 2^{x_{n-k}}\big) \land{}\\
        & \big(0 \leq q \cdot 2^{x_{n-k-1}} + p \cdot 2^{x_{n-k-1}} + r_{n-\ell} < 2^{n-k}\big) \land{}
        \\ 
        & \bigwedge\nolimits_{j=\ell+1}^{k}
        \big(0 \leq q_{n-j} \cdot 2^{x_{n-k-1}} + r_{n-j} < 2^{x_{n-k}}\big) \land{}\\
        & \exists\,\overline{x}_{n-k+1} \dots \exists\,\overline{x}_n \,  \, \big( \overline{\theta} \land \textstyle\bigwedge_{i=k+\ell+1}^{2(k+\ell)+1} (w_i = \rho_i) \big)\!\Big),
    \end{align*}
    where $\overline{\theta} \coloneqq 2^{\overline{x}_n} \geq \dots \geq 2^{\overline{x}_{n-k+1}} \geq 2^{x_{n-k}} \geq \dots \geq 2^{x_0} = 1$.
    To show that $\phi\sub{q+p}{q}$ implies $\Psi(C^{+p})$, 
    observe that ${(\phi\sub{q+p}{q} \implies \Psi(C^{+p}))}$ is syntactically equal 
    to ${(\phi \implies \Psi(C))\sub{q+p}{q}}$, except for the names used for the existentially quantified variable. Specifically, every such variable $v$ in $\Psi(C)$ is replaced with $\overline{v}$ in $\Psi(C^{+p})$. Since ${(\phi \implies \Psi(C))}$ is a valid formula, and validity is preserved under substitution and variable renaming, it follows that ${(\phi\sub{q+p}{q} \implies \Psi(C^{+p}))}$
    is~also~valid. 

    The following claim establishing further properties of~$\objfun{C^{+p}}{\overline{x}_m}$ 
    follows directly from the fact that $C^{+p}$ is obtained from $C$ by replacing with $q + p$ all occurrences of the variable~$q$.

    \begin{claim}
        \label{claim:prop:monotone-decomp:2}
        Let $\nu$ and $\nu + [q \mapsto p]$ be two solutions to $\phi$. 
        Then, $\objfun{C}{x_m}(\nu + [q \mapsto p]) = \objfun{C^{+p}}{\overline{x}_m}(\nu)$.  
    \end{claim}

    \paragraph{\textit{Construction of \(\psi_1, \dots, \psi_s\): mixing $x$s and $\overline{x}s$ into a single ordering.}}
    Before giving some more intuition on the construction of $\psi_1,\dots,\psi_s$, we need some additional formulae manipulations.
    Let $\vec w$ denote the vector of all variables that appear quantified in the formulae $\Psi(C)$ or $\Psi(C^{+p})$.
    Specifically, these are the variables $w_1,\dots,w_{2(k+\ell)}$ (note that $w_{2(k+\ell)+1} = x_{n-k}$ is a free variable in both formulae).
    Observe that, by the definition of $\Psi(C)$ and $\Psi(C^{+p})$, 
    the linear-exponential program $\phi \land \phi\sub{q+p}{q}$ 
    implies 
    \begin{equation}
        \label{proof:monodone-decomp:eq1}
        \exists \vec w : 
        (x_n \geq x_{n-1} \geq \cdots \geq x_{n-k}) \land 
        (\overline{x}_n \geq \overline{x}_{n-1} \geq \cdots \geq \overline{x}_{n-k+1} \geq x_{n-k})
        \land \bigwedge\nolimits_{i=1}^{2(k+\ell)+1} (w_i = \rho_i).
    \end{equation}
    
    Let us denote with $\mathcal{P}$ the set of all permutations $\sigma \colon [1..2k+1] \to [1..2k+1]$ 
    (on the indices of the variables $z_1,\dots,z_{2k+1}$)
    that satisfy: 
    \begin{itemize}
        \item $\sigma^{-1}(1) \leq \sigma^{-1}(2) \leq \dots \leq \sigma^{-1}(k)$, i.e., $\sigma$ respects the ordering $x_n \geq  \cdots \geq x_{n-k+1}$.
        \item $\sigma^{-1}(k+1) \leq \sigma^{-1}(k+2) \leq \dots \leq \sigma^{-1}(2k)$, i.e., $\sigma$ respects the ordering $\overline{x}_n \geq \cdots \geq \overline{x}_{n-k+1}$.
        \item $\sigma(2k+1) = 2k+1$, i.e., the variable $x_{n-k}$ is smaller or equal than any other variable.
    \end{itemize}
    The formula in~\Cref{proof:monodone-decomp:eq1} is equivalent to $\bigvee_{\sigma \in \mathcal{P}} \chi_{\sigma}$, where $\chi_{\sigma}$ is defined as 
    \begin{equation}
        \label{proof:monodone-decomp:eq2}
        \chi_{\sigma} \coloneqq \exists \vec w : 
        (z_{\sigma(1)} \geq z_{\sigma(2)} \geq \cdots \geq z_{\sigma(2k)} \geq z_{\sigma(2k+1)}) \land \bigwedge\nolimits_{i=1}^{2(k+\ell)+1} (w_i = \rho_i).
    \end{equation}

    \paragraph{\textit{Construction of \(\psi_1, \dots, \psi_s\): some intuition.}}
    We are now ready to provide the promised intuition behind the construction of $\psi_1,\dots,\psi_s$. 
    Since $\chi_{\sigma}$ includes the equations ${\bigwedge_{i=1}^{2(k+\ell)+1} (w_i = \rho_i)}$, which describe the assignments in the circuits $C$ and $C^{+p}$, for every solution to ${\phi \land \phi\sub{q+p}{p}}$ there is one and only one assignment to the variables in $\vec w$ 
    that yields a solution to the  
    quantifier-free part of the formula in~\Cref{proof:monodone-decomp:eq2}. This assignment is determined by the values computed by the circuits~$C$ and $C^{+p}$.
    The order that $\sigma$ induces on the variables $z_1,\dots,z_{2k+1}$ 
    implies either $x_m \geq \overline{x}_m$ or $\overline{x}_m \geq x_m$.
    Therefore, by relying on~\Cref{claim:prop:monotone-decomp:2},
    one concludes that
    the function $\objfun{C}{x_m}$ is $(q,p)$-monotone locally to
    $\phi \land \phi\sub{q+p}{q} \land \chi_\sigma$:
    
    \begin{restatable}{claim}{ClaimPropMonotoneDecompThree}
        \superlabel{claim:prop:monotone-decomp:3}{proof:ClaimPropMonotoneDecompThree}
        For every $\sigma \in \mathcal{P}$,
        the function $\objfun{C}{x_m}$ is $(q,p)$-monotone locally to 
        $\phi \land \phi\sub{q+p}{q} \land \chi_\sigma$.
    \end{restatable}

    \noindent
    Clearly, we cannot use the formulae $\chi_\sigma$ as the desired formulae $\psi_1,\dots,\psi_s$, 
    as these formulae are quantified over variables not occurring in $\phi$ (we will thus go against 
    our goal of achieving variable elimination).  
    Instead, we will further refine and manipulate the ordering in~\Cref{proof:monodone-decomp:eq2} and, by appealing to~\Cref{lemma:rewriting-monotone-hyperplane}, 
    restate it solely in terms of variables that occur (free) in~$\phi$. 
    This will allow us to push the ordering outside the scope of the quantifiers $\exists \vec w$, leading to formulae of the form $\psi \land \exists \vec w \bigwedge\nolimits_{i=1}^{2(k+\ell)+1} (w_i = \rho_i)$.
    Next, we will show that the function $\objfun{C}{x_m}$ is $(q,p)$-monotone locally to $\phi \land \phi\sub{q+p}{q} \land \psi \land \exists \vec w \bigwedge\nolimits_{i=1}^{2(k+\ell)+1} (w_i = \rho_i)$.
    Since $\exists \vec w \bigwedge\nolimits_{i=1}^{2(k+\ell)+1} (w_i = \rho_i)$ is implied by $\phi \land \phi\sub{q+p}{q}$, this ensures that $\objfun{C}{x_m}$ is $(q,p)$-monotone locally to $\phi \land \phi\sub{q+p}{q} \land \psi$, 
    as required by Item~\ref{proof:monotone-decomp:item:regions-monotone}.
    Moreover, because of how these formulae are constructed, the disjunction over all such formulae~$\psi$ will still be implied by~${\phi \land \phi\sub{q+p}{q}}$, fulfilling~Item~\ref{proof:monotone-decomp:item:interior-cover}.
    Finally, we will manipulate~$\psi$ to meet the structural requirements of Item~\ref{proof:monotone-decomp:item:interior-decomp-form}.

    \paragraph{\textit{Construction of \(\psi_1, \dots, \psi_s\).}}
    We start by further refining the orderings induced by the permutations in $\mathcal{P}$ 
    by quantifying the gaps between variables. 
    Let $\mathcal{D}$ be the set of all functions ${d \colon [1..2k] \to [0..M]}$, 
    where $M \coloneqq 4 \cdot \ceil{\log_2(2 \cdot \xi_C + \mu_C)} + 8$. 
    For $g \in [0..M-1]$, we write $\sim_g$ as an alias for $=$, 
    and $\sim_M$ as an alias for $\geq$. 
    For $\sigma \in \mathcal{P}$, the formula $\chi_{\sigma}$ 
    in~\Cref{proof:monodone-decomp:eq2} is equivalent to the formula~$\bigvee_{d \in \mathcal{D}} \chi_{\sigma,d}$ where
    \begin{equation}
        \label{proof:monodone-decomp:eq3}
        \chi_{\sigma,d} \coloneqq \exists \vec w : 
        \bigwedge\nolimits_{j=1}^{2k}\big(z_{\sigma(j)} \sim_{d(j)} z_{\sigma(j+1)} + d(j)\big)
        \land \bigwedge\nolimits_{i=1}^{2(k+\ell)+1} (w_i = \rho_i).
    \end{equation} 

    Essentially, every constraint $z_{\sigma(j)} \geq z_{\sigma(j+1)}$ from $\chi_{\sigma}$ 
    is refined in $\chi_{\sigma,d}$ to either $z_{\sigma(j)} \geq z_{\sigma(j+1)} + M$ or $z_{\sigma(j)} = z_{\sigma(j+1)} + g$, for some $g \in [0..M-1]$. 
    Note that each $\chi_{\sigma,d}$ implies $\chi_\sigma$. Therefore, 
    by Claim~\ref{claim:prop:monotone-decomp:3}, the function $\objfun{C}{x_m}$ is $(q,p)$-monotone locally to $\phi \land \phi\sub{q+p}{q} \land \chi_{\sigma,d}$.
    Furthermore, the formula $\phi \land \phi\sub{q+p}{q}$ 
    implies $\bigvee_{\sigma \in \mathcal{P}} \bigvee_{d \in \mathcal{D}} \chi_{\sigma,d}$.

    Let $\sigma \in \mathcal{P}$ and $d \in \mathcal{D}$.
    The next step of the proof involves manipulating the constraints $z_{\sigma(j)} \sim_{d(j)} z_{\sigma(j+1)} + d(j)$ from the formula $\chi_{\sigma,d}$. 
    This manipulation produces a formula $\chi_{\sigma,d}'$ in which these constraints only feature variables from $\phi$.
    Moreover, $\phi \land \phi\sub{q+p}{q}$ 
    implies $(\chi_{\sigma,d} \iff \chi_{\sigma,d}')$.
    The value of $M$ introduced when defining $\chi_{\sigma,d}$ was chosen 
    to make this manipulation possible by appealing to~\Cref{lemma:rewriting-monotone-hyperplane}. 
    The core step in this manipulation is given in the next claim.%
    \begin{claim}
        \label{claim:prop:monotone-decomp:4}
        Let $\sigma \in \mathcal{P}$ and $d \in \mathcal{D}$. 
        The formula 
        $\phi \land \phi\sub{q+p}{q}$  
        implies
        \begin{equation*}
            \forall \vec w \,\Big( 
            \bigwedge\nolimits_{i=1}^{2(k+\ell)+1} (w_i = \rho_i) 
            \implies
            \Big(\bigwedge\nolimits_{j=1}^{2k}\big(z_{\sigma(j)} \sim_{d(j)} z_{\sigma(j+1)} + d(j)\big)  \iff 
            \bigwedge\nolimits_{j=1}^{2k}(0 \sim_{d(j)} E_j)
            \Big) \Big),
        \end{equation*}
        where each $E_j$ (with $j \in [1..2k]$) 
        is $0$, $1$, or an expression 
        of the form 
        \begin{equation}
            a \cdot 2^{x_{n-k}} + \mu_C \cdot \big((y_{\sigma(j+1)} - y_{\sigma(j)})\cdot 2^{x_{n-k-1}} + (s_{\sigma(j+1)} - s_{\sigma(j)})\big) + \mu_C \cdot d(j),
        \end{equation}
        for some $a \in [-b..b]$, where $b \coloneqq  \mu_C \cdot (1 + M)$.
    \end{claim}

    \begin{proof}[Proof of~\Cref{claim:prop:monotone-decomp:4}]
        The proof is by induction on $t$ from $2k+1$ to $1$, 
        with induction hypothesis stating that  $\phi \land \phi\sub{q+p}{q}$  
        implies $\forall \vec w \big( \bigwedge\nolimits_{i=1}^{2(k+\ell)+1} (w_i = \rho_i) \implies \Gamma_{\sigma,d}^{(t)}\, \big)$, 
        where $\Gamma_{\sigma,d}^{(t)}$ is defined as
        \begin{equation}
            \label{claim:prop:monotone-decomp:4:eq1}
                \bigwedge\nolimits_{j=t}^{2k}\big(z_{\sigma(j)} \sim_{d(j)} z_{\sigma(j+1)} + d(j)\big)
                \iff 
                \bigwedge\nolimits_{j=t}^{2k}(0 \sim_{d(j)} E_j),
        \end{equation}
        for some suitable expressions $E_j$ having the form described in the statement of the claim.

        \begin{description}
            \item[base case: $t = 2k+1$.] 
                This case is trivial, as~$\Gamma_{\sigma,d}^{(2k+1)}$ is defined as the tautology $(\top \iff \top)$.
            \item[induction hypothesis.] For $t \in [1..2k]$,   $\phi \land \phi\sub{q+p}{q}$ implies $\forall \vec w (\bigwedge\nolimits_{i=1}^{2(k+\ell)+1} (w_i = \rho_i) \implies \Gamma_{\sigma,d}^{(t+1)})$. 
            \item[induction step: $t < 2k+1$.] From the induction hypothesis, 
                $\phi \land \phi\sub{q+p}{q}$  
                implies 
                \begin{align} 
                    \forall \vec w \,\Big( \bigwedge\nolimits_{i=1}^{2(k+\ell)+1} (w_i = \rho_i) \implies 
                    \Big(
                    &\bigwedge\nolimits_{j=t}^{2k}\big(z_{\sigma(j)} \sim_{d(j)} z_{\sigma(j+1)} + d(j)\big)
                    \iff{}\notag\\ 
                    \label{claim:prop:monotone-decomp:4:eq2}
                    &\bigwedge\nolimits_{j={t+1}}^{2k}\big(0 \sim_{d(j)} E_j\big)
                    \land \boxed{z_{\sigma(t)} \sim_{d(t)} z_{\sigma(t+1)} + d(t)}\,
                    \Big) \Big),
                \end{align}
                for some suitable expressions $E_j$ adhering to the form described in the statement of the claim.

                We construct $E_t$ by modifying the formula in~\Cref{claim:prop:monotone-decomp:4:eq2},
                specifically by only updating the \emph{boxed} occurrence of~${(z_{\sigma(t)} \sim_{d(t)} z_{\sigma(t+1)} + d(t))}$ 
                that appears on the right-hand side of the double implication.
                After all manipulations, the resulting formula is still implied by~${\phi \land \phi\sub{q+p}{q}}$.

                We rewrite the \emph{boxed} constraint $(z_{\sigma(t)} \sim_{d(t)} z_{\sigma(t+1)} + d(t))$ by replacing the variables $z_{\sigma(t)}$ and $z_{\sigma(t+1)}$ with the
                corresponding expressions featured in the antecedent $\bigwedge\nolimits_{i=1}^{2(k+\ell)+1} (w_i = \rho_i)$ of the implication in~\Cref{claim:prop:monotone-decomp:4:eq2}. Specifically, if $w_{i_1}$ is an alias of $z_{\sigma(t)}$ and $w_{i_2}$ is an alias of $z_{\sigma(t+1)}$, then we replace $z_{\sigma(t)}$ by $\rho_{i_1}$ and $z_{\sigma(t+1)}$ by $\rho_{i_2}$.
                By the definition of $C$ and $C^{+p}$, the result is a constraint $(0 \sim_{d(t)} E)$ where $E$ is of the form 
                \begin{equation}
                    \label{claim:prop:monotone-decomp:4:long-expression}
                    \sum_{i=t+1}^{2k+1} a_i \cdot 2^{z_{\sigma(i)}} + \mu_C \cdot \big((y_{\sigma(t+1)} - y_{\sigma(t)})\cdot 2^{x_{n-k-1}} + (s_{\sigma(t+1)} - s_{\sigma(t)})\big) + \mu_C \cdot d(t),
                \end{equation}
                where every $\sum_{i=t+1}^{2k+1} \abs{a_i}$ is bounded, in absolute value, by $2 \cdot \xi_C$.

                After the updates, the formula in~\Cref{claim:prop:monotone-decomp:4:eq2} 
                is still implied by $\phi \land \phi\sub{q+p}{q}$. 
                Because of the induction hypothesis, 
                to show this it suffices to see that 
                \begin{itemize}
                    \item ${w_{i_1} = \rho_{i_1} \land w_{i_2} = \rho_{i_2} \land (z_{\sigma(t)} \sim_{d(t)} z_{\sigma(t+1)} + d(t))}$ implies $(0 \sim_{d(t)} E)$, and
                    \item ${w_{i_1} = \rho_{i_1} \land w_{i_2} = \rho_{i_2} \land \lnot (z_{\sigma(t)} \sim_{d(t)} z_{\sigma(t+1)} + d(t))}$ implies $\lnot (0 \sim_{d(t)} E)$.
                \end{itemize}
                Both these implications follows trivially from how $E$ is constructed.

                We further manipulate the expression $E$ from~\Cref{claim:prop:monotone-decomp:4:long-expression}.
                The following three facts hold: 
                \begin{itemize}
                    \item 
                    The right-hand side of the double implication of the updated formula from~\Cref{claim:prop:monotone-decomp:4:eq2}, 
                    now featuring $(0 \sim_{d(t)} E)$, 
                    implies (because of the double implication)
                    the constraints ${\big(z_{\sigma(i)} \sim_{d(i)} z_{\sigma(i+1)} + d(i)\big)}$
                    for every $i \in [t+1..2k]$.
                    \item Since $\phi \land \phi\sub{q+p}{q}$ implies the formula in~\Cref{claim:prop:monotone-decomp:4:eq2}, we can bound the term $\big((y_{\sigma(t+1)} - y_{\sigma(t)})\cdot 2^{x_{n-k-1}} + (s_{\sigma(t+1)} - s_{\sigma(t)})\big)$ 
                    occurring in $E$ as follows: 
                    \begin{equation} 
                        \label{claim:prop:monotone-decomp:4:eq3}
                        -2^{x_{n-k}} 
                        < \big((y_{\sigma(t+1)} - y_{\sigma(t)})\cdot 2^{x_{n-k-1}} + (s_{\sigma(t+1)} - s_{\sigma(t)})\big) 
                        < 2^{x_{n-k}}.
                    \end{equation}
                    Indeed, from the definition of $\Psi(C)$ and $\Psi(C^{+p})$, 
                    for both $j \in \{\sigma(t),\sigma(t+1)\}$, we have:
                    \begin{itemize}
                        \item 
                            If $C$ and $C^{+p}$ do not assign an expression to $y$,
                            then the formula $\phi \land \phi\sub{q+p}{q}$ implies $0 \leq y_j \cdot 2^{x_{n-k-1}} + s_j < 2^{x_{n-k}}$. 
                        \item 
                            If $C$ or $C^{+p}$ assign an expression~$\rho$ to $y_j$, then 
                            the formula $\phi \land \phi\sub{q+p}{q}$ implies $\exists y_j\, (0 \leq y_j \cdot 2^{x_{n-k-1}} + s_j < 2^{x_{n-k}} \land y_j = \rho)$. In this case, 
                            observe that since $y_j$ is (an alias of) a quotient variable among 
                            $q_n,\dots,q_{n-\ell+1},\overline{q}_n,\dots,\overline{q}_{n-\ell+1}$, 
                            the expression $\rho$ only contains variables occurring in $\phi$ (and it does not contain~$y_j$).
                            Consequently, the value of $y_j$ is uniquely determined given a solution to $\phi \land \phi\sub{q+p}{q}$.
                            This means that $\phi \land \phi\sub{q+p}{q}$ also implies $\forall y_j\, (y_j = \rho \implies 0 \leq y_j \cdot 2^{x_{n-k-1}} + s_j < 2^{x_{n-k}})$. 
                    \end{itemize}
                    We conclude that $\phi \land \phi\sub{q+p}{q}$ implies 
                    \[
                        \forall \vec w \,\Big(\bigwedge\nolimits_{i=1}^{2(k+\ell)+1} (w_i = \rho_i) \implies \bigwedge\nolimits_{j \in \{\sigma(t),\sigma(t+1)\}} 0 \leq y_j \cdot 2^{x_{n-k-1}} + s_j < 2^{x_{n-k}}\Big), 
                    \]
                    which allows us to bound $(y_{\sigma(t+1)} - y_{\sigma(t)})\cdot 2^{x_{n-k-1}} + (s_{\sigma(t+1)} - s_{\sigma(t)})$ 
                    as in~\Cref{claim:prop:monotone-decomp:4:eq3}. 

                    \item The integer $M$ used to define the map $d$ satisfies 
                    \[
                        M \geq \max\big(1+2 \cdot \log_2\big(\sum\nolimits_{i=t+1}^{2k+1} \abs{a_i} + \mu_C\big),\ 4 \cdot \log_2(\mu_C) + 8,\ d(t)\big).
                    \]
                \end{itemize}
                Due to the three items above, 
                we can invoke~\Cref{lemma:rewriting-monotone-hyperplane} 
                and~\Cref{lemma:rewriting-monotone-hyperplane-2} 
                to rewrite $E$ as an expression that is either $0$ or $1$, or has the form
                \begin{equation*}
                    a \cdot 2^{x_{n-k}} + \mu_C \cdot \big((y_{\sigma(t+1)} - y_{\sigma(t)})\cdot 2^{x_{n-k-1}} + (s_{\sigma(t+1)} - s_{\sigma(t)})\big) + \mu_C \cdot d(t),
                \end{equation*}
                for some $a \in [-b..b]$, where $b \coloneqq  \mu_C \cdot (1 + M)$.
                This completes the proof of the claim.
                \qedhere
        \end{description}
    \end{proof}

    Resuming the proof of~\Cref{prop:monotone-decomposition}, 
    from the chain of equivalences involving~\Cref{proof:monodone-decomp:eq1,proof:monodone-decomp:eq2,proof:monodone-decomp:eq3} and by applying~\Cref{claim:prop:monotone-decomp:4},
    we see that $\phi \land \phi\sub{q+p}{q}$ implies $\bigvee_{\sigma \in \mathcal{P}}\bigvee_{d \in \mathcal{D}} \chi_{\sigma,d}'$ where%
    \begin{equation}
        \label{proof:monodone-decomp:eq4}
        \chi_{\sigma,d}' \coloneqq \exists \vec w : 
        \bigwedge\nolimits_{j=1}^{2k}(0 \sim_{d(j)} E_{\sigma,d,j})
        \land \bigwedge\nolimits_{i=1}^{2(k+\ell)+1} (w_i = \rho_i),
    \end{equation}
    and each $E_{\sigma,d,j}$ is an expression having one of the forms given in~\Cref{claim:prop:monotone-decomp:4}.
    (Below, we simply write $E_j$ instead of $E_{\sigma,d,j}$, as $\sigma$ and $d$ will be clear from the context.)
    Moreover, as a consequence of~Claim~\ref{claim:prop:monotone-decomp:3}, the function $\objfun{C}{x_m}$ is $(q,p)$-monotone locally to $\phi \land \phi\sub{q+p}{q} \land \chi_{\sigma,d}'$.
    
    Given $\sigma \in \mathcal{P}$ and $d \in \mathcal{D}$,
    we further manipulate the formula $\chi_{\sigma,d}'$. 
    We consider every expression $E_j$ of the form 
    $a \cdot 2^{x_{n-k}} + \mu_C \cdot \big((y_{\sigma(j+1)} - y_{\sigma(j)})\cdot 2^{x_{n-k-1}} + (s_{\sigma(j+1)} - s_{\sigma(j)})\big) + \mu_C \cdot d(j)$. 
    Let us write $L_{\sigma,d,j}$ and $R_{\sigma,d,j}$ (again, we omit $\sigma$ and $d$ for simplicity) for the two expressions: 
    \[
        L_j \coloneqq a \cdot u + \mu_C \cdot (y_{\sigma(j+1)} - y_{\sigma(j)}), 
        \qquad
        R_j \coloneqq \mu_C \cdot (s_{\sigma(j+1)} - s_{\sigma(j)}) + \mu_C \cdot d(j).
    \]
    Recall that $\phi$ occurs in $\objcons_k$, 
    and therefore it features the equality $u = 2^{x_{n-k}-x_{n-k-1}}$; 
    so $\phi \land \phi\sub{q+p}{q}$ implies $\forall \vec w : {E_j = L_j \cdot 2^{x_{n-k-1}} + R_j}$.
    Moreover, from the definition of $\Psi(C)$, 
    we see that $\phi$ implies $(0 \leq s_{\sigma(j)} < 2^{x_{n-k-1}})$ and $(0 \leq s_{\sigma(j+1)} < 2^{x_{n-k-1}})$. We then conclude that $\phi \land \phi\sub{q+p}{q}$ implies
    $-\mu_C \cdot 2^{x_{n-k-1}} < R_j < (\mu_C + 1) \cdot M \cdot 2^{x_{n-k-1}}$.
    Given $r \in [-\mu_C.. (\mu_C + 1) \cdot M]$,
    we write~$\gamma_{\sigma,d,j,r}$ (often omitting $\sigma$ and $d$, as they will be clear from the context)
    for the formula given by:
    \begin{itemize}
        \item if $\sim_{d(j)}$ stands for $=$, then $\gamma_{j,r} \coloneqq \big( L_j +r = 0 \land R_j = r \cdot 2^{x_{n-k-1}}\big)$, and,
        \item if $\sim_{d(j)}$ stands for $\geq$, then $\gamma_{j,r} \coloneqq \big(L_j + r \leq 0 \land (r-1) \cdot 2^y < R_j \leq r \cdot 2^y\big)$.
    \end{itemize}
    From~\Cref{lemma:split:inequalities}, $\phi \land \phi\sub{q+p}{q}$ implies $\forall \vec w \big((0 \sim_{d(j)} E_j) \iff \bigvee_{r=-\mu_C}^{(\mu_C + 1) \cdot M} \gamma_{j,r}\big)$. 

    We further update the expressions $L_j + r$ above
    by updating $y_{\sigma(j)}$ and $y_{\sigma(j+1)}$ 
    using the corresponding expressions featured in the formula $\bigwedge\nolimits_{i=1}^{2(k+\ell)+1} (w_i = \rho_i)$, if such expressions exist. 
    That is, if a variable $w_i$ in this formula is an alias for $y_{\sigma(j)}$, then we consider the substitution $\sub{\frac{\rho_i}{\lambda}}{\mu_C \cdot y_{\sigma(j)}}$,
    where $\lambda \coloneqq \frac{\eta_C}{\mu_C}$.
    Since $(C,\phi) \in \objcons_k^\ell$, we know that $\mu_C$ divides $\eta_C$, 
    hence $\lambda$ is a positive integer.
    Observe moreover that in $L_j$, 
    the variables $y_{\sigma(j)}$ and $y_{\sigma(j+1)}$ have a coefficient of $\pm \mu_C$;
    this means that applying the substitution $\sub{\frac{\rho_i}{\lambda}}{\mu_C \cdot y_{\sigma(j)}}$ eliminates $y_{\sigma(j)}$. 
    If no variable $w_i$ is an alias for $y_{\sigma(j)}$, 
    we consider instead the substitution $\sub{\frac{\eta_C \cdot y_{\sigma(j)}}{\lambda}}{\mu_C \cdot y_{\sigma(j)}}$ (applying this substitution to $L_j+r$ 
    simply scales all coefficients by $\lambda$). 
    Observe that these substitutions are among those constructed in lines~\ref{algo:btp:line-tau1} and~\ref{algo:btp:line-shift-tau1} of~\Cref{algo:btp}.
    We handle $y_{\sigma(j+1)}$ analogously, 
    and simultaneously apply the resulting substitutions to $L_j + r$. 
    Let $L_{j,r}'$ be the resulting expression. It is of the form
    \begin{equation}
        \label{eq:monotone-decomposition:post-substitution}
         \lambda \cdot a \cdot u + \tau(u, \vec q_{[\ell,k]}) + \lambda \cdot r
    \end{equation}
    where $\tau(u, \vec q_{[\ell,k]})$ is a difference $\tau' - \tau''$ of two terms $\tau'$ and $\tau''$ having forms among the following: 
    \begin{itemize}
        \item a variable $\eta_C \cdot q_{n-i}$ with $i \in [\ell..k]$; note that these variables occurs free in $\Psi(C)$ and~$\Psi(C^{+p})$,%
        \item the expression $\eta_C \cdot (q+p)$; note that $(q+p)$ is aliased by one of the variables $y_1,\dots,y_{2k+1}$, and that no expression in $C$ and $C^{+p}$ 
        is assigned to $q$.
        \item the term $\tau_{n-i}$, for some $i \in [0..\ell-1]$, which occurs in $C$ in the expression assigned to~$q_{n-i}$,%
        \item the term $\tau_{n-i}\sub{q+p}{q}$, for some $i \in [0..\ell-1]$, which occurs in~$C^{+p}$ in the expression for~$\overline{q}_{n-i}$.
    \end{itemize}
    Lastly, let $\gamma_{j,r}'$ be the formula obtained from $\gamma_{j,r}$ by replacing $L_j+r$ with $L_{j,r}'$, and rewriting $L_{j,r}' = 0$ as $L_{j,r}' \leq 0 \land -L_{j,r}' \leq 0$. Note that no variable from $\gamma_{j,r}'$ occurs in the vector~$\vec w$.
    We are now in the position of establishing the following two results, 
    which almost complete the proof of~\Cref{prop:monotone-decomposition}.

    \begin{claim}\label{claim:prop:monotone-decomp:5}
        The formula~$\phi \land \phi\sub{q+p}{q}$ implies $\bigvee_{\sigma \in \mathcal{P}}\bigvee_{d \in \mathcal{D}} \bigwedge_{j=1}^{2k}  \bigvee_{r=-\mu_C}^{(\mu_C + 1) \cdot M} \gamma_{\sigma,d,j,r}'$.
    \end{claim}

    \begin{claim}\label{claim:prop:monotone-decomp:6}
        For every $\sigma \in \mathcal{P}$ and every $d \in \mathcal{D}$,
        the function $\objfun{C}{x_m}$ is $(q,p)$-monotone locally to the set of solutions of the formula~$\phi \land \phi\sub{q+p}{q} \land \bigwedge_{j=1}^{2k}  \bigvee_{r=-\mu_C}^{(\mu_C + 1) \cdot M} \gamma_{\sigma,d,j,r}'$.
    \end{claim}

    \begin{proof}[Proof of Claims~\ref{claim:prop:monotone-decomp:5}--\ref{claim:prop:monotone-decomp:6}]
        We have already 
        established 
        that: 
        \begin{enumerate}
            \item\label{claim:prop:monotone-decomp:proof56:i1} The formula $\phi \land \phi\sub{q+p}{q}$ implies $\bigvee_{\sigma \in \mathcal{P}}\bigvee_{d \in \mathcal{D}} \chi_{\sigma,d}'$, where 
            $\chi_{\sigma,d}'$ is as in~\Cref{proof:monodone-decomp:eq4}. 
            \item\label{claim:prop:monotone-decomp:proof56:i2} For every $\sigma \in \mathcal{P}$ and $d \in \mathcal{D}$, $\phi \land \phi\sub{q+p}{q}$ implies $\forall \vec w \big((0 \sim_{d(j)} E_j) \iff \bigvee_{r=-\mu_C}^{(\mu_C + 1) \cdot M} \gamma_{j,r}\big)$.
            \item\label{claim:prop:monotone-decomp:proof56:i3} The function $\objfun{C}{x_m}$ is $(q,p)$-monotone locally to ${\phi \land \phi\sub{q+p}{q} \land \chi_{\sigma,d}'}$.
        \end{enumerate}
        Every formula $\gamma_{j,r}'$ is 
        obtained from $\gamma_{j,r}$ by ``applying''
        the equalities from ${\bigwedge\nolimits_{i=1}^{2(k+\ell)+1} (w_i = \rho_i)}$ as substitutions. 
        Together with Item~\ref{claim:prop:monotone-decomp:proof56:i2} above, 
        we thus conclude that $\chi_{\sigma,d}'$ is equivalent to 
        \[ 
                \exists \vec w : 
        \Big(\Big(\bigwedge\nolimits_{j=1}^{2k}\bigvee\nolimits_{r=-\mu_C}^{(\mu_C + 1) \cdot M} \gamma_{j,r}'\Big)
        \land \bigwedge\nolimits_{i=1}^{2(k+\ell)+1} (w_i = \rho_i)\Big).
        \]
        The variables $\vec w$ do not occur in any formula $\gamma_{j,r}'$, and so the above formula is equivalent to 
        \[ 
            \Big(\bigwedge\nolimits_{j=1}^{2k}\bigvee\nolimits_{r=-\mu_C}^{(\mu_C + 1) \cdot M} \gamma_{j,r}'\Big)
            \land 
            \exists \vec w : \bigwedge\nolimits_{i=1}^{2(k+\ell)+1} (w_i = \rho_i).
        \]
        The two claims then follows from Items~\ref{claim:prop:monotone-decomp:proof56:i1} and~\ref{claim:prop:monotone-decomp:proof56:i3}, together with the fact that
        the formula $\phi \land \phi\sub{q+p}{q}$ implies $\Psi(C) \land \Psi(C^{+p})$, 
        which in turn implies ${\exists \vec w\bigwedge\nolimits_{i=1}^{2(k+\ell)+1} (w_i = \rho_i)}$ by definition.
    \end{proof}


    At last, let us define the formulae $\psi_1,\dots,\psi_s$. 
    They correspond to all the linear-exponential systems occurring as disjuncts 
    of the disjunctive normal form of the formulae $\bigwedge_{j=1}^{2k}  \bigvee_{r=-\mu_C}^{(\mu_C + 1) \cdot M} \gamma_{\sigma,d,j,r}'$, for every $\sigma \in \mathcal{P}$ and $d \in \mathcal{D}$.
    Directly from~\Cref{claim:prop:monotone-decomp:5} and~\Cref{claim:prop:monotone-decomp:6}, we conclude that these formulae satisfy the desired conditions in Items~\ref{proof:monotone-decomp:item:interior-cover} and~\ref{proof:monotone-decomp:item:regions-monotone}. 
    To complete the proof, it suffices to show that the condition in Item~\ref{proof:monotone-decomp:item:interior-decomp-form} is also satisfied.

    Every constraint occurring in the formulae $\psi_1,\dots,\psi_s$ 
    that features the variable $q$ is of the form $\pm L_{\sigma,d,j,r}' \leq 0$, 
    where $L_{\sigma,d,j,r}'$ is an expression as in~\Cref{eq:monotone-decomposition:post-substitution}. 
    Therefore, the term $\pm L_{\sigma,d,j,r}'$
    is of the form $\lambda \cdot a' \cdot u + \tau'(u, \vec q_{[\ell,k]}) + \lambda \cdot r'$ where
    \begin{itemize}
        \item $a' \in \Z$ belongs to the set $[-\mu_C \cdot (1 + M)..\mu_C \cdot (1 + M)]$. 
        (Recall: $M \coloneqq 4 \cdot \ceil{\log_2(2 \cdot \xi_C + \mu_C)} + 8$.)
         
        \item $\tau'(u, \vec q_{[\ell,k]})$ is the term obtained from $(y_{\sigma(j+1)} - y_{\sigma(j)})$ or $(y_{\sigma(j)} - y_{\sigma(j+1)})$ by applying suitable substitutions. 
        These are among the substitutions considered in lines~\ref{algo:btp:line-tau1}--\ref{algo:btp:line-shift-tau2} of~\Cref{algo:additional-hyperplanes}.
        
        \item $r' \in \Z$ belongs to $[-(\mu_C + 1) \cdot M..(\mu_C + 1) \cdot M]$.
    \end{itemize}
    Note that $\lnot (\pm L_{\sigma,d,j,r}' \leq 0)$ 
    is equivalent to $\mp L_{\sigma,d,j,r}' + 1 \leq 0$.
    Then, in order to cover $\lnot (\pm L_{\sigma,d,j,r}' \leq 0)$  
    with terms $\lambda \cdot a' \cdot u + \tau'(u, \vec q_{[\ell,k]}) + \lambda \cdot r'$ 
    as above, it suffices to increase the interval for the integers $r'$ to $[-(\mu_C + 1) \cdot M-1..(\mu_C + 1) \cdot M+1]$.
    It is then easy to see that Item~\ref{proof:monotone-decomp:item:interior-decomp-form} holds. 
    In fact, for simplicity of the presentation,~\Cref{algo:additional-hyperplanes} uses slightly larger ranges for $a'$ and $r'$ (these integers are called $a$ and~$d$ in the pseudocode, respectively, see~line~\ref{algo:btp:guess-d}).
    Indeed, since $\mu_C \geq 1$ and $M \geq 8$, both $\mu_C \cdot (1 + M)$ and $(\mu_C + 1) \cdot M+1$ are bounded by $3 \cdot \mu_C \cdot M$. 
\end{proof}

\section{An efficient variable elimination that preserves optimal solutions}
\label{sec:efficient-variable-elimination}
\RestoreHeader

Building on our monotone decomposition, 
this section instantiates~\Cref{algo:gaussopt} (\GaussOpt) into the optima-preserving variable elimination procedure that was promised in~\Cref{subsection:OptILEP}. We also provide the proofs of correctness and complexity 
of this algorithm.

The pseudocode of the instantiation of~\GaussOpt is given on page~\pageref{algo:gaussopt-instantiated}. The instantiation is obtained by \textit{(i)} defining 
the inputs of~\GaussOpt, 
\textit{(ii)} providing an algorithm for computing the test points 
and \textit{(iii)} implementing the elimination discipline. 
Following the arguments in~\Cref{section:proof-monotone-decomposition}, achieving the first two points is simple. In particular, appealing to the notation from~\Cref{section:proof-monotone-decomposition}:
\begin{itemize}
    \item The \emph{inputs} of~\GaussOpt 
    are triples $(\vec q_{k-1}, \objfun{C}{x_m}, \inst{\gamma}{\psi})$, for every~$k \in [0..n-1]$, where ${(C,\inst{\gamma}{\psi})}$ belongs to $\objcons_k^0$.
    (The procedure will thus eliminate the variables~$q_{n-k+1},\dots,q_n$, but not the variable~$q_{n-k}$. This is consistent with the sketch of the procedure from~\cite{ChistikovMS24} given in~\Cref{section:summary-procedure}, where the latter variable is eliminated only after~Step~III.) 
    \item The pseudocode of the procedure for computing the test points is given in~\Cref{algo:btp}. Briefly, the procedure first (non-deterministically) computes a term~$(a \cdot q - \tau)$ stemming from the monotone decomposition of~\Cref{prop:monotone-decomposition}, 
    and then returns an equality ${(a \cdot q = \tau - s)}$, where $s$ is a non-negative shift 
    that suffices to explore optimal solutions, by~\Cref{lemma:monotone-only-x-matters}.
\end{itemize}  

The implementation of an efficient elimination discipline is a more complex
task. As noted in the introduction (\Cref{example:ILP}), the na\"ive approach of
rewriting a formula $\gamma$ as ${\gamma\sub{\frac{\tau-s}{a}}{q} \land (\abs{a}
\divides \tau - s)}$, where $(a \cdot q = \tau - s)$ is a test point, would lead
to exponential growth in the bit length of integer coefficients, during the
execution of~\GaussOpt. In~\cite{ChistikovMS24}, this problem is
avoided by extending Bareiss' algorithm for Gaussian elimination~\cite{Bareiss68} 
to 
integer linear programs. While our elimination discipline is also based on Bareiss'
algorithm, it is different from the one in~\cite{ChistikovMS24}. In
particular, we do not introduce ``slack variables'', and add 
some technical machinery to handle the additional test points required by the monotone
decomposition (those computed by~\Cref{algo:additional-hyperplanes}).

\subsection{Efficient elimination discipline: the high-level idea}
\label{subsec:efficient-elim-high-level}
\begin{figure}
\input{procedure-Bareiss-two.tex}
\vspace{0.3cm}

\begin{algorithm}[H]
  \caption{$\tests$: Non-deterministic generation of the test points for ILEP.}
  \label{algo:btp} 
  \setstretch{1.1}
  \begin{algorithmic}[1]
    \Require 
      \begin{minipage}[t]{0.92\linewidth}
        \setlength{\tabcolsep}{2pt}
        \begin{tabular}[t]{rcp{0.75\linewidth}}
        $(\vec q_{k-1}, \objfun{C}{x_m},\inst{\gamma}{\psi})\colon$ a triple such that $(C,\inst{\gamma}{\psi})$ belongs to $\objcons_k^{\ell}$, with $\ell < k$.
        \end{tabular}
      \end{minipage} 
    \medskip
    \State $p \gets \fmod(q_{n-\ell},\gamma)$\label{algo:true-tp:mod} 
    \State $(a \cdot q_{n-\ell} - \tau) \gets{}$ \textbf{guess} a term with $a \neq 0$ that is either from $\fterms(\gamma \land \gamma\sub{q_{n-\ell}+ p}{q_{n-\ell}})$,\label{algo:true-tp:guess} 
    \Statex \hphantom{$(a \cdot q_{n-\ell} - \tau) \gets{}$ \textbf{guess}} or computed using~\Cref{algo:additional-hyperplanes} with respect to $(C,\gamma,p)$.  
    \State $s \gets \text{\textbf{guess} an element in } [0..\abs{a} \cdot p-1]$\label{algo:true-tp:guess-2}
    \State \textbf{return} $(a \cdot q_{n-\ell} = \tau - s)$\label{algo:true-tp:return}
  \end{algorithmic}
\end{algorithm}%

\vspace{0.3cm}

\begin{algorithm}[H]
  \caption{\elimdisctxt : An efficient elimination discipline for ILEP.}
  \label{algo:sub-disc}
  \setstretch{1.1}
  \begin{algorithmic}[1]
    \Require 
      \begin{minipage}[t]{0.92\linewidth}
        \setlength{\tabcolsep}{2pt}
        \begin{tabular}[t]{rcp{0.85\linewidth}}
          $(\objfun{C}{x_m},\inst{\gamma}{\psi})$&:& 
          a pair such that $(C,\inst{\gamma}{\psi}) \in \objcons_k^\ell$, with $\ell < k$;\\
          $a \cdot q_{n-\ell} = \tau$&:& an equality returned by~\Cref{algo:btp} on input $(\vec q_{k-1},\objfun{C}{x_m},\inst{\gamma}{\psi})$.
        \end{tabular}
      \end{minipage}
    \medskip
    \If{$a < 0$} $(a,\tau) \gets (-a,-\tau)$ 
      \Comment{consider~$-a \cdot q_{n-\ell} = -\tau$ instead}
      \label{algo:sub-disc:make-a-positive}
    \EndIf
    \State $\lambda \gets \frac{\eta_C}{\mu_C}$; \ \ $\alpha \gets \frac{a}{\mu_C}$\label{algo:sub-disc:lambda}
    \vspace{1pt}
    \State $\gamma \gets \gamma\sub{\frac{\tau}{\alpha}}{\mu_C \cdot q_{n-\ell}} \land (a \divides \tau)$\label{algo:sub-disc:eliminate}
    \State \textbf{assert}(in every equality $\tau = 0$ of $\gamma$, the constant of the term $\tau$ is divisible by $\lambda$)\label{algo:sub-disc:assert-equality}
    \Statex \Comment{here and below, \textbf{assert}(\text{false}) causes the non-deterministic branch to reject}
    \State update each equality $\tau = 0$ in $\gamma$ :\label{algo:sub-disc:simplify}
    \Statex  \ \ $\bullet$ divide all integers appearing in the term $\tau$ by $\lambda$
    \State update each inequality $\tau \leq 0$ in $\gamma$ :\label{algo:sub-disc:simplify-2}
    \Statex  \ \ $\bullet$ divide all variable coefficients in the term $\tau$ by $\lambda$
    \Statex  \ \ $\bullet$ replace the constant $c$ of the term $\tau$ with $\lceil{\frac{c}{\lambda}}\rceil$
    \State update $C$ :\label{algo:sub-disc:update-C}
    \Statex \ \ $\bullet$ for every $i \in [0..\ell-1]$, consider the assignment $q_{n-i} \gets \frac{\tau_{n-i}}{\eta_C}$ in $C$
    \Statex \ \ \ \ $-$ \textbf{assert}(the constant of the term $\tau_{n-i}\sub{\frac{\tau}{\alpha}}{ \mu_C \cdot q_{n-\ell}}$ is divisible by $\lambda$)\label{algo:sub-disc:assert-circuit}
    \Statex \ \ \ \ $-$ replace $q_{n-i} \gets \frac{\tau_{n-i}}{\eta_C}$ with $q_{n-i} \gets \frac{\tau_{n-i}'}{a}$,  where the term $\tau_{n-i}'$ is obtained 
    \Statex \hphantom{\ \ \ \ $-$} from the term $\tau_{n-i}\sub{\frac{\tau}{\alpha}}{ \mu_C \cdot q_{n-\ell}}$ by dividing all integers by $\lambda$
    \Statex \ \ $\bullet$ prepend the assignment $q_{n-\ell} \gets \frac{\tau}{a}$ 
    \State \textbf{return} $\left(\objfun{C}{x_m}, \inst{\gamma}{\psi}\right)$ 
    \label{algo:sub-disc:return}
    \Comment{$(C,\inst{\gamma}{\psi})$ belongs to $\objcons_k^{\ell+1}$}
  \end{algorithmic}
\end{algorithm}%

\end{figure}

The pseudocode of our elimination discipline is given in~\Cref{algo:sub-disc}. 
We now discuss the overall idea behind this procedure. 
Consider one of its inputs: a pair $(\objfun{C}{x_m},\inst{\gamma}{\psi})$, 
where $(C,\inst{\gamma}{\psi}) \in \objcons_k^\ell$ with $\ell < k$, 
and an equality $a \cdot q_{n-\ell} = \tau$ returned by~\Cref{algo:btp} on input $(\vec q_{k-1},\objfun{C}{x_m},\inst{\gamma}{\psi})$.

Let us briefly explain why the na\"ive approach of eliminating the variable $q_{n-\ell}$ from $\gamma$ by performing the substitution $\sub{\frac{\tau}{a}}{q_{n-\ell}}$ leads to an exponential growth.
Assuming $a > 0$, this substitution rewrites an inequality $b \cdot q_{n-\ell} \leq \tau'$ as $b \cdot \tau \leq a \cdot \tau'$. Let $c$ and $d$ denote the coefficients of another variable, say $y$, in $\tau$ and $\tau'$, respectively.  
Then, in the term $b \cdot \tau - a \cdot \tau'$, the coefficient of $y$ is $b \cdot c - a \cdot d$. In essence, this shows that the coefficients of the variables in $\gamma$ 
may grow quadratically within each elimination of one of the variables in~$\vec q_{k-1}$. As a result, by the end of~\GaussOpt, the na\"ive approach may produce a linear program with coefficients of exponential bit size.

The above explosion can be avoided by observing that the variable coefficient $b \cdot c - a \cdot d$ is exactly the one we would obtain when performing a na\"ive version of the Gaussian elimination procedure for putting a matrix with entries over $\Z$ in echelon form. Building on Bareiss's observation~\cite{Bareiss68}, we see that these growing variable coefficients accumulate common factors as we iteratively eliminate variables. Divisions by these common factors after each variable elimination keep variable coefficients of polynomial bit size. (These divisions are performed in lines~\ref{algo:sub-disc:simplify}--\ref{algo:sub-disc:simplify-2} of~\Cref{algo:sub-disc}.) 
While variable coefficients evolve as in Gaussian elimination, 
the \emph{constants} of the terms do not. In particular, the shift performed in line~\ref{algo:true-tp:guess-2} of~\Cref{algo:true-tp:guess-2} disrupt any structure in the constants. This however does not pose a problem.
Inequalities of the form $\lambda \cdot \rho + c \leq 0$ (where $\lambda \geq 1$ is the common factor, $\rho$ a term, and $c$ is the integer constant) can be rewritten as $\rho + \ceil{\frac{c}{\lambda}} \leq 0$. 
For equalities $\lambda \cdot \rho + c = 0$ instead, observe that they are unsatisfiable when $c$ is not a multiple of $\lambda$, and otherwise they can be rewritten as $\rho + \frac{c}{\lambda} = 0$. 
\Cref{algo:sub-disc:assert-equality} implements this reasoning in lines~\ref{algo:sub-disc:assert-equality}--\ref{algo:sub-disc:simplify-2}, where the positive integer $\lambda$ defined in line~\ref{algo:sub-disc:lambda} represents the common factor.
(We will clarify why this common factor is exactly the ratio $\frac{\eta_C}{\mu_C}$ of the two denominators of the $(k,\ell)$-LEAC $C$ in~\Cref{subsec:evolution-integers-elim-var}.)
A similar argument can be made for updates performed to the circuit $C$, by seeing each assignment $q \gets \frac{\tau}{\eta_C}$ as the equality $\eta_C \cdot q = \tau$; 
see line~\ref{algo:sub-disc:update-C} of~\Cref{algo:sub-disc:update-C}.

Further technical details must be added to the above picture to keep the complexity of~\OptILEP in check. The main problem arises when~\Cref{algo:btp} produces a test point by appealing to~\Cref{algo:additional-hyperplanes}. Such test points depend on the assignments featured in the circuit $C$. 
This dependence makes it challenging to maintain the delicate ``Gaussian-elimination-style evolution'' that variable coefficients must have throughout~\GaussOpt.
The solution~\Cref{algo:sub-disc} implements starts from the observation that all non-zero coefficients of the quotient variables $\vec q_{k}$ in 
the term in line~\ref{algo:btp:border} of~\Cref{algo:btp:border} are (before substitutions) only $\pm \mu_C$. Together with the constraints imposed by~$\objcons_k^\ell$, which ensure that all coefficients of the variables $\vec q_{[\ell,k]}$ are divisible by $\mu_C$, 
this observation will enable us to give a variation of Bareiss algorithm,
from which we can prove that~\GaussOpt, and later~\OptILEP, run in non-deterministic polynomial time. 

Before moving to a more technical analysis of~\GaussOpt, 
let us note that the overwhelming presence of $\mu_C$ in both~\Cref{algo:btp:border} and elements of~$\objcons_k^\ell$ implies the following property of~\Cref{algo:btp}:

\begin{restatable}{lemma}{LemmaWhatBTPCoples}
    \superlabel{lemma:what-btp-comples}{proof:LemmaWhatBTPCoples}
    Given in input a triple~$(\vec q_{k-1}, \objfun{C}{x_m}, \inst{\gamma}{\psi})$ with $(C,\inst{\gamma}{\psi}) \in \objcons_k^\ell$
    \Cref{algo:btp} guesses in line~\ref{algo:true-tp:guess}  
    a linear term $a \cdot q_{n-\ell} - \tau(u,\vec q_{[\ell+1,k]})$ in which all coefficients of~$\vec q_{[\ell,k]}$ 
    are divisible~by~$\mu_C$.
\end{restatable}

\subsection{Correctness of~\GaussOpt}
\label{subsec:correctness-gauss}

The integration of the machinery from Bareiss algorithm to keep the growth of the coefficients in check has a ``presentational drawback'': the arguments for establishing 
the complexity of the algorithm now mix with those needed to prove its correctness, as we must ensure that this machinery is implemented correctly. 

To ease the presentation, we structure this and the next three (\Cref{subsec:variation-bareiss-body,subsec:evolution-integers-elim-var,subsec:complexity-elimvars}) as follows. In the current section, we isolate the key property necessary for the correct implementation of Bareiss's machinery.
This property is formalized in~\Cref{claim:divisions-without-remainder}. 
Assuming this claim to hold, we then establish the correctness of~\GaussOpt.
\Cref{subsec:variation-bareiss-body,subsec:evolution-integers-elim-var} 
develop the arguments needed to prove~\Cref{claim:divisions-without-remainder}, 
while also setting up the properties required for the complexity analysis.
More precisely, \Cref{subsec:variation-bareiss-body} presents a variation of Bareiss algorithm in which the evolution of variable coefficients precisely mirrors that of~\GaussOpt. We state a series of results characterizing this evolution; 
their proofs are deferred to~\Cref{appendix:gaussian-elimination} 
---these proofs involve a detour to linear algebra and Bareiss algorithm,
and we prefer to keep the focus of this section on~\GaussOpt. In~\Cref{subsec:evolution-integers-elim-var} we formalize the connection between this variation of Bareiss algorithm and~\GaussOpt, and use it to prove~\Cref{claim:divisions-without-remainder}. Finally, in~\Cref{subsec:complexity-elimvars} 
we leverage this connection to analyze the complexity of~\GaussOpt.

Here is the aforementioned key property related to Bareiss algorithm: 

\begin{restatable}{claim}{ClaimDivisionsWithoutRemainder}
    \label{claim:divisions-without-remainder}
    The following property is true across all the executions of~\Cref{algo:sub-disc} performed in all non-deterministic branches of~\GaussOpt, on any of its inputs. 
    In all equalities and inequalities of the formula~${\gamma\sub{\frac{\tau}{\alpha}}{\mu_C \cdot q_{n-\ell}}}$ computed in line~\ref{algo:sub-disc:eliminate}, and in terms ${\tau_{n-i}\sub{\frac{\tau}{\alpha}}{ \mu_C \cdot q_{n-\ell}}}$ computed in line~\ref{algo:sub-disc:update-C}, 
    all coefficients of the variables~$\vec q_{[\ell+1,k]}$ are divisible by $\eta_C$, and all coefficients of $u$ are divisible~by~$\frac{\eta_C}{\mu_C}$.
\end{restatable}

The following lemma establishes the correctness of~\GaussOpt.

\begin{restatable}{lemma}{SecondStepOpt}
  \label{lemma:second-step-opt}
  There is a non-deterministic procedure with the following specification:
  \begin{description}
    \setlength{\tabcolsep}{2pt}
    \item[\textbf{\textit{Input:}}] 
      \begin{minipage}[t]{0.94\linewidth}
        \hspace{3pt}
        \begin{tabular}[t]{rcp{0.87\linewidth}}
        $\vec q_{k-1}$&:& the vector of quotient variables $q_{n-(k-1)},\dots,q_n$; \hfill (for any $k$)\\
        $\objfun{C}{x_m}$&:& objective function, where $C$ is a $(k,0)$-LEAC;\\ 
        $\inst{\gamma}{\psi}$&:& linear exponential program with divisions,\\
        && such that the pair $(C,\inst{\gamma}{\psi})$ belongs to $\objcons_k^{0}$.
        \end{tabular}
      \end{minipage}
    \item[\textbf{\textit{Output of each branch ($\beta$):}}]
    
    \begin{minipage}[t]{\linewidth}
      \hspace{3pt}
      \begin{tabular}[t]{rcp{0.75\linewidth}}
        $\objfun{C_\beta'}{x_m}$&:& objective function, where $C_\beta'$ is a $(k,k)$-LEAC;\\
        $\inst{\gamma_\beta'}{\psi}$&:& linear exponential program with divisions,\\ 
        &&such that $(C_\beta',\inst{\gamma_\beta'}{\psi})$ belongs to $\objcons_k^k$.
        \end{tabular}
      \end{minipage}
  \end{description}
  The procedure ensures the satisfaction of the following two properties:
  \begin{itemize}
      \item\textbf{Equivalence:} 
      The formulae $\exists \vec q_{k-1} \, \gamma$ 
      and $\bigvee_{\beta} \gamma_{\beta}'$ are equivalent.
      Consider a branch~$\beta$, 
      and let $q_{n-(k-1)} \gets \frac{\tau_{n-(k-1)}}{\eta},\dots,q_{n} \gets \frac{\tau_n}{\eta}$ be the assignments 
      to the variables $\vec q_{k-1}$ occurring in $C_\beta'$.
      Given a solution $\nu \colon \{u,q_{n-k}\} \to \N$ to $\gamma_\beta'$,
      the map $\nu + \sum_{i=0}^{k-1} [q_{n-i} \mapsto \frac{\nu(\tau_{n-i})}{\eta}]$ is a solution~to~$\gamma$.%

      \item\textbf{Preservation of maximum:} if $\max\{\objfun{C}{x_m}(\nu) : \text{$\nu$ is a solution to $\inst{\gamma}{\psi}$}\}$ exists, 
      then it is equal to
      $\max\{\objfun{C_\beta'}{x_m}(\nu) : \text{$\beta$ is a branch, $\nu$ is a solution to $\inst{\gamma_\beta'}{\psi}$}\}$.
  \end{itemize} 
\end{restatable}

\begin{proof}[Proof (assuming~\Cref{claim:divisions-without-remainder}).]
   By induction on $\ell = k,\dots,0$, with induction
   hypothesis:%
   
   \begin{description}
    \item[induction hypothesis.] The specification given
    in~\Cref{lemma:second-step-opt} holds when executing the \textbf{while} loop
    of~\Cref{algo:gaussopt} (\GaussOpt) on a triple $(\vec q_{k-1},
    \objfun{C}{x_m}, \inst{\gamma}{\psi})$ where $(C,\inst{\gamma}{\psi})$
    belongs to $\objcons_{k}^\ell$. That is to say, 
    \begin{itemize}
        \item \textbf{Output:} the output of each branch~$\beta$
        is~$(\objfun{C_\beta'}{x_m},\inst{\gamma_\beta'}{\psi})$, with
        $(C_\beta',\inst{\gamma_\beta'}{\psi}) \in \objcons_k^k$.
        \item \textbf{Equivalence:} 
        The formulae $\exists \vec q_{[\ell,k-1]} \gamma$ and $\bigvee_{\beta} \gamma_\beta'$ are equivalent. Consider a branch~$\beta$, 
        and let $q_{n-(k-1)} \gets \frac{\tau_{n-(k-1)}}{\eta},\dots,q_{n-\ell} \gets \frac{\tau_{n-\ell}}{\eta}$ be the assignments 
        to the variables $\vec q_{[\ell,k-1]}$ in $C_\beta'$.
        Given a solution $\nu \colon \{u,q_{n-k}\} \to \N$ to $\gamma_\beta'$, 
        the map ${\nu + \sum_{i=\ell}^{k-1}[q_{n-i} \mapsto \frac{\nu(\tau_{n-i})}{\eta}]}$ is a solution to $\gamma$.
         

        \item \textbf{Preservation of maximum:} if $\max\{\objfun{C}{x_m}(\nu) :
        \text{$\nu$ is a solution to $\inst{\gamma}{\psi}$}\}$ exists, then it
        is equal to $\max\{\objfun{C_\beta'}{x_m}(\nu) : \text{$\beta$ is a
        branch, $\nu$ is a solution to $\inst{\gamma_\beta'}{\psi}$}\}$.
    \end{itemize}
   \end{description}
   Observe that setting $\ell = 0$ in the above induction hypothesis yields the
   statement of~\Cref{lemma:second-step-opt}.
   
   \begin{description}
    \item[base case: $\ell = k$.] 
        In this case, we have $(C,\inst{\gamma}{\psi}) \in \objcons_k^k$. By
        definition of $\objcons_k^k$ (page~\pageref{objcons:i1}) this means that
        $C$ is a $(k,k)$-LEAC, and $\gamma$ only features the variables
        $q_{n-k}$ and $u$. Therefore, no variable in $\vec q_{k-1}$ occurs in
        these objects, and the condition of the \textbf{while} loop of~\GaussOpt
        fails. The algorithm then returns
        $(\objfun{C}{x_m},\inst{\gamma}{\psi})$, and all requirements stated in
        the induction hypothesis are trivially satisfied.
    \item[induction step: $\ell \in {[0..k-1]}$.] 
        Let $\phi \coloneqq \inst{\gamma}{\psi}$.
        By definition of $\objcons_k^\ell$, the linear program~$\gamma$ features an inequality $b
        \cdot q_{n-\ell} \geq 0$, for some $b \geq 1$. Let $p \coloneqq
        \fmod(q_{n-\ell},\gamma)$. Since $q_{n-\ell}$ belongs to $\vec q_{k-1}$,
        the body of the \textbf{while} loop of~\GaussOpt executes. The call
        to~\Cref{algo:btp} (non-deterministically) returns an equality $a \cdot
        q_{n-\ell} = \tau - s$ such that \textit{(i)} $a \neq 0$, \textit{(ii)}
        $(a \cdot q_{n-\ell}-\tau)$ is either from $\fterms(\gamma \land
        \gamma\sub{q_{n-\ell}+ p}{q_{n-\ell}})$ or it is computed
        using~\Cref{algo:additional-hyperplanes} with respect to $(C,\gamma,p)$,
        and \textit{(iii)} $s \in [0..\abs{a} \cdot p - 1]$. Again by definition
        of~$\objcons_k^\ell$, the variables occurring in~$\tau$ are from the vector $\vec q_{[\ell+1,k]} = (q_{n-k},\dots,q_{n-(\ell+1)})$.
       
        From~\Cref{corr:basic-fact-from-presburger} (which concerns satisfiability)
        and~\Cref{prop:monotone-decomposition}
        and~\Cref{lemma:monotone-only-x-matters} (which concern optimization) we
        conclude that if $\phi$ has a solution (analogously,
        if~$(\objfun{C}{x_m},\phi)$ has a maximum), then it has
        one satisfying an equality $a \cdot q_{n-\ell} = \tau - s$ returned
        by~\Cref{algo:btp}.
        The lemma is therefore implied by the induction hypothesis (applied to elements in $\objcons_{k}^{\ell+1}$) together with the following claim (below, for brevity we write $\rho$ instead of $\tau-s$):

        \begin{claim}\label{claim:conditional-correctness:induction-step}
            Given in input $(\objfun{C}{x_m},\phi)$ and 
            an equality $a \cdot q_{n-\ell} = \rho$ 
            computed by~\Cref{algo:btp}, 
            \Cref{algo:sub-disc} behaves as follows: 
            \begin{enumerate}
                \item\label{claim:conditional-correctness:induction-step:i1} If the statements in the \textbf{assert} commands in lines~\ref{algo:sub-disc:assert-equality} and~\ref{algo:sub-disc:assert-circuit} are not satisfied, then the algorithm rejects.
                In this case the formula $\exists q_{n-\ell} : \phi \land (a \cdot q_{n-\ell} = \rho)$ is unsatisfiable.
                \item\label{claim:conditional-correctness:induction-step:i2} Else,
                    the algorithm 
                    returns $(\objfun{C'}{x_m},\inst{\gamma'}{\psi})$ 
                    such that $(C',\inst{\gamma'}{\psi}) \in \objcons_k^{\ell+1}$. 
                    In~this~case: 
                    \begin{enumerate}
                        \item\label{claim:conditional-correctness:induction-step:i2:a} The formula $\inst{\gamma'}{\psi}$ 
                        is equivalent to $\exists q_{n-\ell} : \phi \land (a
                        \cdot q_{n-\ell} = \rho)$.
                        \item\label{claim:conditional-correctness:induction-step:i2:b} Consider a solution $\nu \colon X \setminus \{q_{n-\ell}\} \to \N$ to $\inst{\gamma'}{\psi}$, 
                        and let $\nu' \coloneqq \nu + [q_{n-\ell} \mapsto \frac{\nu(\rho)}{a}]$. For each $i \in [0..\ell-1]$, given the assignments $q_{n-i} \gets \frac{\tau_{n-i}}{\eta_C}$ and $q_{n-i} \gets \frac{\tau_{n-i}
                        }{\eta_{C'}}$ from $C$ and $C'$, respectively, 
                        we have $\frac{\nu'(\tau_{n-i})}{\eta_C} = \frac{\nu(\tau_{n-i}')}{\eta_{C'}}$. Moreover, $\objfun{C'}{x_m}(\nu) = \objfun{C}{x_m}(\nu')$.
                    \end{enumerate}
            \end{enumerate}
        \end{claim}
        Let us prove~\Cref{claim:conditional-correctness:induction-step}. 
        Observe that line~\ref{algo:sub-disc:make-a-positive}
        rewrites the equality $a \cdot q_{n-\ell} = \rho$ 
        and  $-a \cdot q_{n-\ell} = -\rho$  whenever $a < 0$, hence forcing the coefficient of $q_{n-\ell}$ to be positive. 
        Hence, in what follows we assume without loss of generality that~$a > 0$ (hence $\alpha = \frac{a}{\mu_C}$ in line~\ref{algo:sub-disc:lambda} is positive).

        Let us consider~\Cref{claim:conditional-correctness:induction-step:i1} of the claim. Assume that one of the \textbf{assert} commands in lines~\ref{algo:sub-disc:assert-equality} and~\ref{algo:sub-disc:assert-circuit} 
        is not satisfied. In the case of line~\ref{algo:sub-disc:assert-equality}, 
        this means that an equation $\tau' = 0$ from 
        $\gamma\sub{\frac{\rho}{\alpha}}{\mu_C \cdot q_{n-\ell}}$ 
        is such that the constant~$c$ of $\tau'$ is not divisible by $\lambda$. 
        Since $\lambda = \frac{\eta_C}{\mu_C}$,
        assuming~\Cref{claim:divisions-without-remainder}, 
        $\tau' = 0$ can be written as $\lambda \cdot \tau'' + c = 0$. 
        But then, whenever the variables in $\tau''$ are evaluated to some integers, this equation is asserting that a multiple of $\lambda$ is equal to $-c$; contradicting the fact that $c$ is not divisible by $\lambda$. This implies that $\gamma\sub{\frac{\rho}{\alpha}}{\mu_C \cdot q_{n-\ell}}$ is unsatisfiable, and thus so is $\exists q_{n-\ell} : \phi \land (a \cdot q_{n-\ell} = \rho)$. 
        The argument is similar when the \textbf{assert}  
        command of line~\ref{algo:sub-disc:assert-circuit} 
        is not satisfied. Indeed, consider $i \in [0..\ell-1]$ 
        such that $q_{n-i} \gets \frac{\tau_{n-i}}{\eta_C}$ occurs in $C$, 
        and the constant 
        of the term $\tau_{n-i}\sub{\frac{\tau}{\alpha}}{ \mu_C \cdot q_{n-\ell}}$ 
        is not divisible by $\lambda$. 
        From the definition of $\objcons_{k}^\ell$ (more precisely, the definition of~$\Psi(C)$), $\phi$ implies $\exists q_{n-i} : \eta_C \cdot q_{n-i} = \tau_{n-i}$,
        which in turn means that 
        $\exists q_{n-\ell} : \phi \land (a \cdot q_{n-\ell} = \rho)$ 
        implies $\exists q_{n-i} : \alpha \cdot \eta_C \cdot q_{n-i} = \tau_{n-i}\sub{\frac{\tau}{\alpha}}{ \mu_C \cdot q_{n-\ell}}$. 
        By definition,~$\lambda$ is a divisor of $\eta_C$. 
        Hence, assuming~\Cref{claim:divisions-without-remainder}, 
        in the equality $\alpha \cdot \eta_C \cdot q_{n-i} = \tau_{n-i}\sub{\frac{\tau}{\alpha}}{ \mu_C \cdot q_{n-\ell}}$ all variable coefficients are divisible by $\lambda = \frac{\eta_C}{\mu_C}$, but the constant term is not. 
        As in the previous case, this means that $\alpha \cdot \eta_C \cdot q_{n-i} = \tau_{n-i}\sub{\frac{\tau}{\alpha}}{ \mu_C \cdot q_{n-\ell}}$ 
        is unsatisfiable, and thus so is $\exists q_{n-\ell} : \phi \land (a \cdot q_{n-\ell} = \rho)$.
        This completes the proof of the first of the two items in~\Cref{claim:conditional-correctness:induction-step}.

        We move to~\Cref{claim:conditional-correctness:induction-step:i2}. Assume that the \textbf{assert} commands in lines~\ref{algo:sub-disc:assert-equality} and~\ref{algo:sub-disc:assert-circuit} are satisfied, 
        and therefore that the algorithm returns some pair~${(\objfun{C'}{x_m},\inst{\gamma'}{\psi})}$.
        We first show that~$\inst{\gamma'}{\psi}$ 
        is equivalent to $\exists q_{n-\ell} : \phi \land (a \cdot q_{n-\ell} = \rho)$. 
        From~$(C,\inst{\gamma}{\psi}) \in \objcons_k^\ell$, 
        in all (in)equality from $\gamma$ the coefficients of $q_{n-\ell}$ are divisible by $\mu_C$. 
        This means that the substitution $\sub{\frac{\tau}{\alpha}}{\mu_C \cdot q_{n-\ell}}$ performed in line~\ref{algo:sub-disc:eliminate}
        eliminates $q_{n-\ell}$. 
        Therefore, the formula~$\gamma\sub{\frac{\tau}{\alpha}}{\mu_C \cdot q_{n-\ell}} \land (a \divides \tau)$ is a linear program with divisions 
        over the variables $\vec q_{[\ell+1,k]}$ and $u$. 
        Clearly, this formula is equivalent to $\exists q_{n-\ell} : \gamma \land (a \cdot q_{n-\ell} = \rho)$, proving~\Cref{claim:conditional-correctness:induction-step:i2:a}.
        Since $\psi$ does not contain the variable $q_{n-\ell}$, 
        it thus suffices to show that transformation in lines~\ref{algo:sub-disc:assert-equality}--\ref{algo:sub-disc:simplify-2}, 
        which produces $\gamma'$ from~$\gamma\sub{\frac{\tau}{\alpha}}{\mu_C \cdot q_{n-\ell}} \land (a \divides \tau)$, preserves formula equivalence. 
        These lines rely on the following two equivalences, that hold for any linear term $\tau'$ 
        (since such terms evaluate to integers):
        \begin{align*}
            \lambda \cdot \tau' = 0 &\iff \tau' = 0 & \text{(since $\lambda \neq 0$)}\\ 
            \lambda \cdot \tau' + c \leq 0 &\iff \tau' + \ceil{\frac{c}{\lambda}} \leq 0 & \text{(since $\lambda \geq  1$)}
        \end{align*}
        To ensure that lines~\ref{algo:sub-disc:assert-equality}--\ref{algo:sub-disc:simplify-2} correctly implement these equivalences, we need to verify that all divisions performed in these lines are without remainder. 
        Assuming~\Cref{claim:divisions-without-remainder}, 
        the coefficients of the variables $\vec q_{[\ell+1,k]}$ in 
        (in)equalities of $\gamma\sub{\frac{\tau}{\alpha}}{\mu_C \cdot q_{n-\ell}}$ are divisible by $\eta_C$, 
        while the coefficients of $u$ are divisible by $\lambda = \frac{\eta_C}{\mu_C}$. 
        Since the \textbf{assert} command of line~\ref{algo:sub-disc:assert-equality} is satisfied,
        the constants occurring in equalities are also divisible by $\lambda$. 
        It follows that the divisions performed in lines~\ref{algo:sub-disc:simplify} and~\ref{algo:sub-disc:simplify-2} are without remainder. 
        (Note that this also shows that $\gamma'$ is a linear program with divisions in variables $\vec q_{[\ell+1,k]}$ and $u$.)

        Next, we show that the returned pair ${(\objfun{C'}{x_m},\inst{\gamma'}{\psi})}$ is such that $(C',\inst{\gamma'}{\psi}) \in \objcons_k^{\ell+1}$. 
        Recall that, by~\Cref{lemma:what-btp-comples}, the equality $a \cdot q_{n-\ell} = \rho$ computed by~\Cref{algo:btp} is such that $a \cdot q_{n-\ell} - \rho$ is a linear term featuring variables $u$ and $\vec q_{[\ell,k]}$, in which the coefficients of the variables~$\vec q_{[\ell,k]}$ divisible~by~$\mu_C$.
        Below, \Cref{objcons:i1,objcons:i2,objcons:i3} refers to the 
        items characerizing $\objcons_k^{\ell}$ (page~\pageref{objcons:i1}).

        \begin{itemize}
            \item\textit{\Cref{objcons:i1}:} we must prove that $C'$  
            is a $(k,\ell+1)$-LEAC such that $\mu_{C'}$ divides $\eta_{C'}$, 
            as well as all coefficients of the variables $\vec q_{[\ell+1,k]}$ occurring in the term $\tau_{n-i}'$
            featured in assignments $q_{n-i} \gets \frac{\tau_{n-i}'}{\eta_{C'}}$ of $C'$, with $i \in [0..\ell]$. 

            The circuit $C'$ is constructed in line~\ref{algo:sub-disc:update-C}. 
            This line
            updates all assignments $q_{n-i} \gets \frac{\tau_{n-i}}{\eta_C}$ in~$C$, where $i \in [0..\ell-1]$, 
            and prepends the assignment $q_{n-\ell} \gets \frac{\tau}{a}$. 
            The denominator of all these assignments is $\eta_{C'} = a$, which is divisible by~$\mu_{C'} = \mu_C$.
            Recall that, from~$(C,\inst{\gamma}{\psi}) \in \objcons_k^\ell$, the term $\tau_{n-i}$ 
            is a linear term in variables~$\vec q_{[\ell,k]}$ and $u$, 
            in which the coefficient of the variable $q_{n-\ell}$ is divisible by $\mu_C$.
            From~\Cref{lemma:what-btp-comples}, we conclude that $\tau_{n-i}\sub{\frac{\tau}{\alpha}}{\mu_C \cdot q_{n-\ell}}$ 
            is a linear term in variables $\vec q_{[\ell+1,k]}$ and $u$. 
            Therefore, $C'$ is a $(k,\ell+1)$-LEAC.
            Lastly, assuming~\Cref{claim:divisions-without-remainder}, 
            the coefficients of the variables $\vec q_{[\ell+1,k]}$ in the term $\tau_{n-i}\sub{\frac{\tau}{\alpha}}{\mu_C \cdot q_{n-\ell}}$ are divisible by $\eta_C$, 
            and the coefficient of $u$ is divisible by $\lambda = \frac{\eta_C}{\mu_C}$. 
            Since the \textbf{assert} command of line~\ref{algo:sub-disc:assert-circuit} is satisfied, 
            the constant of this term is also divisible by $\lambda$. 
            We conclude that the divisions by~$\lambda$ 
            performed to compute $\tau_{n-i}'$ are without remainder, 
            and that in this term the coefficients of the variables  in $\vec q_{[\ell+1,k]}$ are divisible by $\mu_{C'}$.

            \item\textit{\Cref{objcons:i2}:} 
            We have already shown that $\gamma'$ is a linear program with divisions in variables $\vec q_{[\ell+1,k]}$ and $u$. 
            It thus suffices to show that $\gamma'$ satisfies the following properties: \textit{(i)} every coefficient of the variables in $\vec q_{[\ell+1,k]}$ is divisible by $\mu_{C'}$, and \textit{(ii)} for each $q$ in $\vec q_{[\ell+1,k]}$, $\gamma'$ contains an inequality of the form~$a \cdot q \geq 0$ for some $a \geq 1$.
            
            For the first property, 
            let us go back to the fact that the coefficients of the variables $\vec q_{[\ell+1,k]}$ in 
            (in)equalities of $\gamma\sub{\frac{\tau}{\alpha}}{\mu_C \cdot q_{n-\ell}}$ are divisible by $\eta_C$ (\Cref{claim:divisions-without-remainder}).
            Following the (remainder-less) divisions by $\lambda = \frac{\eta_C}{\mu_C}$ performed in lines~\ref{algo:sub-disc:simplify} and~\ref{algo:sub-disc:simplify-2}, 
            we conclude that in all (in)equalities of $\gamma'$ the coefficients of the variables in  $\vec q_{[\ell+1,k]}$ are divisible by $\mu_{C'} = \mu_C$. 

            For the second property, recall that
            for every $q$ in $\vec q_{[\ell+1,k]}$, $\gamma$ features an inequality of the form~$b \cdot q \geq 0$ for some $b \geq 1$.
            In $\gamma'$ this inequality is transformed into $\frac{\alpha \cdot b}{\lambda} \cdot q \geq 0$, 
            where $\frac{\alpha \cdot b}{\lambda}$ is positive and, assuming~\Cref{claim:divisions-without-remainder}, 
            an integer.

            \item\textit{\Cref{objcons:i3}:} 
            For brevity, let $C = (y_1 \gets \rho_1, \dots, y_t \gets \rho_t)$ 
            and $C' = {(y_0 \gets \rho_0',\dots, y_t \gets \rho_t')}$, 
            where $y_0 \gets \rho_0'$ is an alias for the assignment $q_{n-\ell} \gets \frac{\rho}{a}$ that the algorithm prepend to $C$ in line~\ref{algo:sub-disc:update-C}, in order to define $C'$.
            We must show that $\inst{\gamma'}{\psi}$ implies~$\Psi(C')$, that is,
            \begin{align*}
                \vec 0 \leq \vec r_k < 2^{x_{n-k-1}} \land \exists \vec q_{[0,\ell]}
                \Big(\vec 0 \leq \vec q_k \cdot 2^{x_{n-k-1}}+ \vec r_k < 2^{x_{n-k}}  \land
                \exists \vec x_{k-1} \big( \theta \land \textstyle\bigwedge_{i=0}^{t} (y_i' = \rho_i') \big)\Big).
            \end{align*}
            We have already established that $\inst{\gamma'}{\psi}$ 
            is equivalent to $\exists q_{n-\ell} : \phi \land (a \cdot q_{n-\ell} = \rho)$. 
            Since $\phi$ implies $\Psi(C)$, we then conclude that $\inst{\gamma'}{\psi}$ 
            implies 
            \begin{align*}
                \vec 0 \leq \vec r_k < 2^{x_{n-k-1}} \land \exists \vec q_{[0,\ell]}
                \Big(&\vec 0 \leq \vec q_k \cdot 2^{x_{n-k-1}}+ \vec r_k < 2^{x_{n-k}}\\ 
                &{} \land
                \exists \vec x_{k-1} \big( \theta \land \textstyle\bigwedge_{i=1}^{t} (y_i = \rho_i) \land (a \cdot q_{n-\ell} = \rho)\big)\Big).
            \end{align*}
            Hence, as line~\ref{algo:sub-disc:update-C} does not modify the assignments in $C$ 
            that feature variables in~$\vec x_{k-1}$, 
            to conclude that $\inst{\gamma'}{\psi}$ implies~$\Psi(C')$
            it suffices to show that, for every $i \in [0..\ell-1]$, 
            \begin{equation}
                \label{eq:correct-elim:update-C}
                a \cdot q_{n-\ell} = \rho 
                \implies 
                \big(\eta_C \cdot q_{n-i} = \tau_{n-i} 
                \iff a \cdot q_{n-i} = \tau_{n-i}'\big),
            \end{equation}
            where $\tau_{n-i}$ is the term such that $q_{n-i} \gets \frac{\tau_{n-i}}{\eta_C}$ occurs in $C$, and $\tau_{n-i}'$ is the term 
            such that $q_{n-i} \gets \frac{\tau_{n-i}'}{a}$ occurs in $C'$. 
            Recall that $\alpha = \frac{a}{\mu_C} \geq 1$, $\lambda = \frac{\eta_C}{\mu_C} \geq 1$, and that all divisions performed in line~\ref{algo:sub-disc:update-C} are without remainder. Then, 
            \Cref{eq:correct-elim:update-C} follows from the equivalences below:
            \begin{align*}
                &(\eta_C \cdot q_{n-i} - \tau_{n-i})\sub{{\textstyle\frac{\rho}{\alpha}}}{\mu_C \cdot q_{n-\ell}}\\ 
                ={}& \alpha \cdot \eta_C \cdot q_{n-i} - (\tau_{n-i}\sub{{\textstyle\frac{\rho}{\alpha}}}{\mu_C \cdot q_{n-\ell}})
                &\Lbag\text{by def.~of substitution}\Rbag\\
                ={}& \frac{\alpha \cdot \eta_C}{\lambda} \cdot q_{n-i} - \tau_{n-i}'
                &\Lbag\text{division by $\lambda$}\Rbag\\
                ={}& a \cdot q_{n-i} - \tau_{n-i}'
                &\Lbag\text{from $\textstyle\alpha = \frac{a}{\mu_C}$, $\textstyle\lambda = \frac{\eta_C}{\mu_C}$}\Rbag
            \end{align*}
        \end{itemize}
        This concludes the proof that $(C',\inst{\gamma'}{\psi})$ belongs to~$\objcons_k^{\ell+1}$. To conclude the 
        proof~\Cref{claim:conditional-correctness:induction-step}, it remains to show~\Cref{claim:conditional-correctness:induction-step:i2:b}.
        Consider a solution $\nu \colon X \setminus \{q_{n-\ell}\} \to \N$ to $\inst{\gamma'}{\psi}$, 
        and let $\nu' \coloneqq \nu + [q_{n-\ell} \mapsto \frac{\nu(\rho)}{a}]$.
        Since~$\inst{\gamma'}{\psi}$ 
        is equivalent to~${\exists q_{n-\ell} : \phi \land (a \cdot q_{n-\ell} = \rho)}$, it follows that $\nu'$ is a valid assignment of variables into~$\N$. 
        By definition, for each $i \in [0..\ell-1]$, 
        evaluating $C'$ on $\nu$ assigns to the variable $q_{n-i}$ the (non-negative) integer~$\frac{\nu(\tau_{n-i}')}{a}$.
        Directly from~\Cref{eq:correct-elim:update-C}, we have $\frac{\nu(\tau_{n-i}')}{a} = \frac{\nu'(\tau_{n-i})}{\eta_C}$. Moreover, since $C$ and $C'$ agree on all the assignments to all variables in $\vec x_{k-1}$, we conclude that  
        $\objfun{C}{x_m}(\nu) = \objfun{C}{x_m}(\nu')$. 
        \qedhere
   \end{description}
\end{proof}

\subsection{A variation of Bareiss algorithm}
\label{subsec:variation-bareiss-body}

We now introduce our variation of Bareiss algorithm. 
As already stated, Bareiss algorithm is commonly used to 
calculate the echelon form of a matrix with integer entries. 
For a formal definition of echelon form, we refer the reader to standard linear algebra textbooks (e.g.,~\cite{HoffmanKunze1971}). However, below we do not rely on this definition, as we characterize every single entry that the manipulated matrix has throughout the procedure, in terms of the original ones.

\paragraph*{Some notation.}
Consider a $m \times d$ integer matrix $B_0$:
\begin{equation*}
    B_0 \coloneqq \begin{pmatrix}
        b_{1,1} & \ldots  & b_{1,d}\\
        \vdots  & \ddots & \vdots\\
        b_{m,1} & \ldots  & b_{m,d}
    \end{pmatrix}.
\end{equation*}
Let $\ell \in [0..\min(m,d)]$. 
We write $b^{(\ell)}_{i,j}$ (with~$(i,j) \in [1..m] \times [1..d]$) and $b^{(\ell)}_{r \gets j}$ (with~$(r,j) \in [1..\ell] \times [1..d]$)
to denote the following sub-determinants of $B_0$:
\begin{equation*}    
    b_{i,j}^{(\ell)} \eqdef 
    \det
    \begin{pmatrix}
    b_{1,1} & \ldots & b_{1,\ell} & b_{1,j} \\
    \vdots  & \ddots & \vdots  & \vdots  \\
    b_{\ell,1} & \ldots & b_{\ell,\ell} & b_{\ell,j} \\
    b_{i,1} & \ldots & b_{i,\ell} & b_{i,j}
    \end{pmatrix}, 
    \qquad 
    b_{r \gets j}^{(\ell)} \eqdef 
    \det
    \begin{pmatrix}
    b_{1,1} & \ldots & b_{1,r-1} & b_{1,j} & b_{1,r+1} & \dots & b_{1,\ell} \\
    \vdots  & \ddots & \vdots  & \vdots & \vdots  & \ddots & \vdots   \\
    b_{\ell,1} & \ldots & b_{\ell,r-1} & b_{\ell,j}  & b_{\ell,r+1} & \dots & b_{\ell,\ell}
    \end{pmatrix}.
\end{equation*}
Let us fix $k \in [0..\min(m,d)]$ (this quantity corresponds to the number of iterations the algorithm will perform).
For every~$\ell \in [1..k]$, 
we define~$\lambda_0 \coloneqq 1$ and ${\lambda_\ell \coloneqq b_{\ell,\ell}^{(\ell-1)}}$.
(This notation is chosen intentionally: we will later show that $\abs{\lambda_\ell}$ corresponds to the integer $\lambda$ computed by~\Cref{algo:btp,algo:sub-disc} when applied to inputs featuring pairs from~$\objcons_{k}^\ell$.) 
Throughout this section, we assume that every $\lambda_\ell$ is non-zero.

Let us moreover fix a positive integer $\mu$ (we will later set $\mu$ to the integer $\mu_C$ of the LEAC $C$ in input of~\Cref{algo:btp,algo:sub-disc}),
fix in $g \in [k..d]$ (an index of a column in the $B_0$), 
and consider the diagonal matrix $U_g \coloneqq \diag(\mu,\dots,\mu,1,\dots,1)$ 
having $\mu$ in all positions $(i,i)$ with $i \in [1..g]$, and $1$ in positions $(i,i)$ with $i \in [g+1..d]$.

\paragraph*{The algorithm.}
The input matrix is of the form
\begin{equation}
    \label{eq:B-matrix-U}
    B_0' \ \coloneqq\ B_0 \cdot U_g \ =\ \begin{pmatrix}
        \mu \cdot b_{1,1} & \ldots  & \mu \cdot b_{1,g} & b_{1,g+1} & \dots & b_{1,d}\\
        \vdots  & \ddots & \vdots\\
        \mu \cdot b_{m,1} & \ldots  & \mu \cdot b_{m,g} & b_{m,g+1} & \dots & b_{m,d}
    \end{pmatrix}.
\end{equation}
The algorithm iteratively constructs a sequence of matrices $B_1',\dots,B_k'$ as follows.
Consider ${\ell \in [0..k-1]}$, and let $B_\ell'$ be the matrix 
\[
    B_\ell' \coloneqq \begin{pmatrix}
        h_{1,1} & \ldots  & h_{1,d}\\
        \vdots  & \ddots & \vdots\\
        h_{m,1} & \ldots  & h_{m,d}
    \end{pmatrix}.
\]
The matrix $B_{\ell+1}'$ is constructed from~$B_\ell'$ 
by applying the following transformation 
{\setstretch{1.3}
\makeatletter%
\def\ALG@step%
   {%
   \addtocounter{ALG@line}{1}%
   \addtocounter{ALG@rem}{1}%
   \ifthenelse{\equal{\arabic{ALG@rem}}{\ALG@numberfreq}}%
      {\setcounter{ALG@rem}{0}\alglinenumber{0\arabic{ALG@line}}}
      {}%
   }%
\makeatletter
\begin{algorithmic}[1]
    \State \textbf{let} $\pm$ be the sign of $h_{\ell+1,\ell+1}$, and $\alpha \coloneqq \frac{\pm h_{\ell+1,\ell+1}}{\mu}$
    \customlabel{01}{apx:algo:gauss:new:sign-pivot}
    \Comment{this division is without remainder}
    \State multiply the row $\ell+1$ of $B_{\ell}$ by $\pm 1$
    \customlabel{02}{apx:algo:gauss:new:multiply-pivoting-row}
    \For{every row $i$ except row $\ell+1$}
        \customlabel{03}{apx:algo:gauss:new:loop}
        \State \textbf{let} $\beta \coloneqq \frac{h_{i,\ell+1}}{\mu}$ 
        \customlabel{04}{apx:algo:gauss:new:line1}
        \Comment{this division is without remainder}
        \State multiply the $i$th row of $B_{\ell}'$ by $\alpha$ 
        \customlabel{05}{apx:algo:gauss:new:line2}
        \Comment{$B_{\ell}'(i,\ell+1)$ is now $\alpha \cdot g_{i,\ell+1}$}
        \State subtract $\pm \beta \cdot (h_{\ell+1,1},\dots,h_{\ell+1,d})$ from the $i$th row of $B_{\ell}'$
        \customlabel{06}{apx:algo:gauss:new:line3}
        \State divide each entry of the $i$th row of $B_{\ell}'$ by $\abs{\lambda_{\ell}}$
        \customlabel{07}{apx:algo:gauss:new:line4}
        \Comment{these divisions are without remainder}
    \EndFor
\end{algorithmic}
}

\paragraph*{Characterization of the entries of the matrices.}
The following three lemmas fully characterize all entries in the matrices $B_1',\dots,B_k'$ in terms of the entries of $B_0'$. 
Their proofs are given in~\Cref{sub-appendix:gaussian-elimination-twist}, 
after introducing the necessary background on the (classical) Bareiss algorithm.

\begin{restatable}{lemma}{LemmaGaussianEliminationNewSameRow}
    \superlabel{lemma:gaussian-elimination:new:same-row}{proof:LemmaGaussianEliminationNewSameRow}
    For all $\ell \in [0..k-1]$, 
    the $(\ell+1)$th row of $B_{\ell+1}'$ is obtained by multiplying the $(\ell+1)$th 
    row of $B_\ell'$ by the sign of its $(\ell+1)$th entry.
\end{restatable}

\begin{restatable}{lemma}{LemmaGaussianEliminationNewBelow}
    \superlabel{lemma:gaussian-elimination:new:below}{proof:LemmaGaussianEliminationNewSameRow}
    Consider $\ell \in [0..k]$ and $i \in [\ell+1..m]$, and let $\pm$ be the sign of $\lambda_\ell$. Then:
    \begin{enumerate}[itemsep=0pt]
        \item\label{lemma:gaussian-elimination:new:below:i1} For every $j \in [1..g]$, 
        the entry in position $(i,j)$ of the matrix $B_\ell'$ is 
        $\pm \mu \cdot b^{(\ell)}_{i,j}$.\\
        (In particular, this entry is zero whenever $j \leq \ell$.)
        \item\label{lemma:gaussian-elimination:new:below:i2} For every $j \in [g+1..d]$, 
        the entry in position $(i,j)$ of the matrix $B_\ell'$ is 
        $\pm b^{(\ell)}_{i,j}$.
    \end{enumerate}
\end{restatable}

\begin{restatable}{lemma}{LemmaGaussianEliminationNewAbove}
    \superlabel{lemma:gaussian-elimination:new:above}{proof:LemmaGaussianEliminationNewAbove}
    Consider $\ell \in [1..k]$
    and $i \in [1..\ell]$, and let $\pm$ be the sign of $\lambda_\ell$. Then:
    \begin{enumerate}[itemsep=0pt]
        \item\label{lemma:gaussian-elimination:new:above:i1} For every~$j \in [1..g]$, 
        the entry in position $(i,j)$ of $B_\ell'$ is 
        $\pm \mu \cdot b^{(\ell)}_{i \gets j}$.\\
        (In particular,
        this entry is zero if $j \leq \ell$ and $i \neq j$,
        and it is instead $\pm \mu \cdot b^{(\ell-1)}_{\ell,\ell}$ when $i = j$.)
        \item\label{lemma:gaussian-elimination:new:above:i2} For every~$j \in [g+1..d]$, 
        the entry in position $(i,j)$ of $B_\ell'$ is 
        $\pm b^{(\ell)}_{i \gets j}$.
    \end{enumerate}
\end{restatable}

The next lemma provides an alternative algorithm to reconstruct the matrix~$B_{\ell}'$ starting from its first $\ell$ rows and from $B_0'$, without constructing $B_1',\dots,B_{\ell-1}'$. In the next section, this lemma will turn out useful when analyzing the terms computed by~\Cref{algo:additional-hyperplanes}.

\begin{restatable}{lemma}{LemmaOneShotReplacementNew}
    \superlabel{lemma:one-shot-replacement:new}{proof:LemmaOneShotReplacementNew}
    Let $\ell \in [0..k]$ and $i \in [\ell+1..m]$. 
    Consider the following transformation applied to~$B_0'$:

    {\setstretch{1.3}
    \begin{algorithmic}[1]
        \State multiply the $i$th row of $B_0'$ by $\abs{\lambda_{\ell}}$
        \label{algo:one-shot-replacement:new:line3}
        \For{$r$ in $[1..\ell]$}
            \ subtract $b_{i,r} \cdot \vec u_r$ to the $i$th row of $B_0'$, where $\vec u_r$ is the $r$th row of $B_{\ell}'$
            \label{algo:one-shot-replacement:new:line4}
        \EndFor
    \end{algorithmic}
    }
    \vspace{5pt}
    \noindent
    After the transformation, the $i$th rows of $B_0'$ and $B_{\ell}'$ are equal.
\end{restatable}

\subsection{How coefficients evolve as~\GaussOpt executes, and proof of~\Cref{claim:divisions-without-remainder}}
\label{subsec:evolution-integers-elim-var}

We now prove that the evolution of the variable coefficients during~\GaussOpt 
mirrors that of the matrix entries in our variation of Bareiss algorithm.
(For brevity, in this section by Bareiss algorithm we always mean this variation.)
This is done by setting up a sequence of matrices~$M_0,M_1,\dots$, where $M_\ell$ snapshots
the coefficients of the variables $\vec q_{k}$ and $u$ in the (in)equalities of $\gamma$ or in the terms of the circuit $C$, after $\ell$ iterations the \textbf{while} loop in~\GaussOpt. We then show that $M_\ell$ can alternatively be obtained from $M_0$ by performing $\ell$ iterations of Bareiss algorithm. In doing so, we also establish~\Cref{claim:divisions-without-remainder}.

Throughout the section, let 
$(\vec q_{k-1}, \objfun{C}{x_m},\inst{\gamma}{\psi})$ be the triple in input to~\GaussOpt, 
with ${(C,\inst{\gamma}{\psi}) \in \objcons_k^{0}}$, for some $k \in [0..n-1]$.
Also, let $\mu \coloneqq \mu_C$.
Since~\Cref{algo:sub-disc} does not modify the formula $\psi$, 
each non-deterministic branch of~\GaussOpt can be identified with a sequence of pairs
\begin{equation}
    \label{eq:sequence-of-configurations}
    (C,\gamma) = (C_0,\gamma_0) \xrightarrow{e_0} (C_1,\gamma_1) \xrightarrow{e_1} (C_2,\gamma_2) \dots \xrightarrow{e_{j-1}} (C_{j},\gamma_{j})
    \xrightarrow{e_j} \dots
\end{equation}
where $e_\ell$ is the equality computed by~\Cref{algo:btp} during the $(\ell+1)$th iteration of~\GaussOpt, and $C_\ell$ and $\gamma_\ell$ are the circuit $C$ and formula $\gamma$ at the completion of the $\ell$th iteration of~\GaussOpt, respectively.
At this point, the length of this sequence is unknown.
It might be short (e.g., if one of the \textbf{assert} commands of~\Cref{algo:sub-disc} fails and the algorithm rejects), or even infinite ---we are not assuming~\Cref{claim:divisions-without-remainder}, and thus cannot guarantee that variables are eliminated correctly. 
Nonetheless, we will prove that in non-rejecting runs, each iteration of the \textbf{while} loop in~\GaussOpt eliminates one of the variables~$\vec q_{k-1}$.
Because of this, we truncate the above sequence to some $j \in [0..k]$, and focus our analysis on this finite prefix. 

\paragraph*{Towards defining the matrices.}
Let us fix an enumeration  
\begin{equation}
    \label{eq:enum-gamma-0}
    (\rho_1 \sim_1 0), (\rho_2 \sim_2 0),\, \dots\,,\, (\rho_t \sim_t 0)
\end{equation}
of all the equalities and inequalities of $\gamma_0$ (that is, each $\sim_i$ belongs to $\{=,\leq\}$). Observe that~\Cref{algo:sub-disc} constructs $\gamma_{\ell+1}$ from $\gamma_{\ell}$ by simply applying a substitution (line~\ref{algo:sub-disc:eliminate}), adding a divisibility constraint (again line~\ref{algo:sub-disc:eliminate}), and performing some integer divisions. Since substitutions are applied locally to each constraints, 
this implies not only that the number of (in)equalities in $\gamma_0,\gamma_1,\dots,\gamma_j$ 
does not change, but that in fact there is a one-to-one mapping between 
(in)equalities of $\gamma_0$ and those in $\gamma_\ell$. That is, there is an enumeration of the (in)equalities of $\gamma_\ell$ such that the $i$th element of the enumeration 
is the inequality obtained from $\rho_i \sim_i 0$ by applying all the substitutions and divisions performed by~\Cref{algo:sub-disc} during the first $\ell$ iterations of \GaussOpt.
We will denote such a one-to-one mapping, from (in)equalities of $\gamma_0$ to those in $\gamma_\ell$, as $\map_\ell$ ($\map_0$ is the identity).

Let us now look at the equality~$e_\ell$ (with $\ell \in [0..j-1]$).
Following~\Cref{algo:btp}, this equality is of the form $a \cdot q = \tau - s$, where 
$s$ is the shift introduced in line~\ref{algo:true-tp:guess-2} of~\Cref{algo:btp}, and 
$a \cdot q - \tau$ is a term of one of the following two types:
\begin{enumerate}
    \item[\labeltext{I}{gaussopt-connection:typeI}.] \textit{A term from $\fterms(\gamma_\ell \land \gamma_\ell\sub{q + p}{q})$.} Note that the substitution $\sub{q + p}{q}$ affects only the constants of (in)equalities. Consequently, there is an (in)equality $\rho' \sim 0$ in $\gamma_\ell$ where the term $\rho$ has the same 
    variable coefficients as $a \cdot q - \tau$.
    In this case, we define the \emph{generator}~$g_\ell$ \emph{of} $e_\ell$ 
    to be the (in)equality $\rho \sim 0$ of~$\gamma_0$ such that $\map_\ell(\rho \sim 0) = (\rho' \sim 0)$.

    \item[\labeltext{II}{gaussopt-connection:typeII}.] \textit{A term computed using~\Cref{algo:additional-hyperplanes} with respect to $(C_\ell,\gamma_\ell,p)$.} This is a term obtained by simultaneously applying two substitutions
    to a term $\rho$ of the form~${a \cdot u + \mu_{C_\ell} \cdot (q' - q'') + d}$, with $a,d \in \Z$ and $q',q''$ from $\vec q_{k-1}$. 
    We define the \emph{generator}~$g_\ell$~\emph{of}~$e_\ell$ to be the equality $\rho = 0$. 
\end{enumerate}
We say that $e_\ell$ is of \emph{Type~\ref{gaussopt-connection:typeI}} or~\emph{Type~\ref{gaussopt-connection:typeII}}, depending on which of the two cases above it falls under.

\paragraph*{The matrix associated to $\gamma_0$.}
Each of the matrices $M_0,\dots,M_j$ we define have $j+t$ rows (where~$t$ is the number of equalities and inequalities in $\gamma_0$) and $k+2$ columns. For every $i \in [0..k]$, the $(i+1)$th column contains coefficients of the variable~$q_{n-i}$. The $(k+2)$th column contains coefficients of~$u$. 

\begin{figure}[t]
\begin{center}
    \scalebox{0.87}{
    $\begin{blockarray}{ccccccl}
    q_n & q_{n-1} & \dots & q_{n-k+1} & q_{n-k} & u \\[5pt]
    \begin{block}{(cccccc)l}
    \ \mu \cdot b_{1,1} & \mu \cdot b_{1,2} & \dots & \mu \cdot b_{1,k} & \mu \cdot b_{1,k+1} & b_{1,k+2} & \text{coefficients of $g_0$} \\[4pt]
    \ \mu \cdot b_{2,1} & \mu \cdot b_{2,2} & \dots & \mu \cdot b_{2,k} & \mu \cdot b_{2,k+1} & b_{2,k+2} & \text{coefficients in $g_1$} \\[4pt]
    \vdots\\[4pt]
    \ \mu \cdot b_{j,1} & \mu \cdot b_{j,2} & \dots & \mu \cdot b_{j,k} & \mu \cdot b_{j,k+1} & b_{j,k+2} & \text{coefficients in $g_{j-1}$} \\[4pt]
    \ \mu \cdot b_{j+1,1} & \mu \cdot b_{j+1,2} & \dots & \mu \cdot b_{j+1,k} & \mu \cdot b_{j+1,k+1} & b_{j+1,k+2} & \text{coefficients in $\rho_1$} \\[4pt]
    \vdots\\[4pt] 
    \ \mu \cdot b_{j+t,1} & \mu \cdot b_{j,2} & \dots & \mu \cdot b_{j+t,k} & \mu \cdot b_{j+t,k+1} & b_{j+t,k+2} & \text{coefficients in $\rho_t$} \\[4pt]
    \end{block}
    \end{blockarray}$}
\end{center}
\caption{Structure of the matrix $M_0$. Observe that all coefficients of the variables in~$\vec q_k$ have $\mu$ as a common factor. 
This is because~$(C,\inst{\gamma}{\psi})$ belongs to~$\objcons_k^{0}$.}
\label{figure:matrix-M0}
\end{figure}

In the matrix $M_0$,
for $i \in [1..j]$, the $i$th row stores the coefficients of the variables $\vec q_{k}$ and $u$ occurring in the generator $g_{i-1}$ of $e_{i-1}$. 
The remaining $t$ rows store the coefficients of the variables $\vec q_{k}$ and $u$ occurring in (in)equalities of $\gamma_0$: following the enumeration in~\Cref{eq:enum-gamma-0}, the $(j+r)$th row stores the variable coefficients of $\rho_r$.
The structure of the matrix $M_0$ is illustrated in~\Cref{figure:matrix-M0}.

\paragraph*{The matrices $M_1,\dots,M_j$.} 
Let $\ell \in [1..j]$. 
At the $\ell$th iteration of the \textbf{while} loop of~\GaussOpt, 
\Cref{algo:sub-disc} \emph{prepends} a single assignment $q \gets \frac{\tau}{a}$ 
to the circuit $C_{\ell-1}$ (where $\pm a \cdot q = \pm \tau$ is $e_{\ell-1}$, with $\pm$ being the sign of $a$), and modifies all other assignments to variables from $\vec q_{k-1}$ so that the denominator of the assigned expression becomes $a$. This means that $\eta_{C_{\ell}} = a$, which we abbreviate as $\eta_\ell$. 
Note that then $C_\ell$ features $\ell$ assignments to variables in $\vec q_{k-1}$, 
and the remaining assignments are the original ones from $C_0$, featuring the variables~$\vec x_k$. In particular, $\mu_{C_{\ell}} = \mu$. 
We define the matrix $M_\ell$ as follows:
\begin{enumerate}
    \item[\labeltext{i}{gaussopt-complexity:matrix-def-1}.] For $i \in [1..\ell]$, let $q \gets \frac{\tau}{\eta_\ell}$ be the $(\ell-(i-1))$th assignment in $C_\ell$. The $i$th row of $M_\ell$ contains the coefficients of the variables $\vec q_k$ and $u$ from the term $\eta_\ell \cdot q - \tau$. 
    \item[\labeltext{ii}{gaussopt-complexity:matrix-def-2}.] For every $i \in [\ell+1..j]$, if $e_{i-1}$ is of Type~\ref{gaussopt-connection:typeI}, then the $i$th row of $M_\ell$ contains the coefficients of the variables $\vec q_k$ and $u$ occurring in the (in)equality $\Lambda_\ell(g_{i-1})$.
    \item[\labeltext{iii}{gaussopt-complexity:matrix-def-3}.] For $i \in [\ell+1..j]$, if $e_{i-1}$ is of Type~\ref{gaussopt-connection:typeII}, then let $\rho = 0$ be $g_{i-1}$.
    The $i$th row of $M_\ell$ contains the coefficients of the variables $\vec q_k$ and $u$ from the term obtained from $\rho$ as follows:
    \begin{algorithmic}[1]
        \State multiply every integer in $\rho$ by the quotient of the division of $\eta_\ell$ by $\mu$
        \For{$r$ in $[1..\ell]$}
            $\rho \gets \rho\sub{\tau}{\eta_\ell \cdot q}$, 
            where $q \gets \frac{\tau}{\eta_\ell}$ is the $r$th assignment in $C_\ell$
        \EndFor
    \end{algorithmic}
    \vspace{-5pt}
    (We will later see that this is in fact a simultaneous substitution.)

    \item[\labeltext{iv}{gaussopt-complexity:matrix-def-4}.] For every $i \in [1..t]$, the $(j+i)$th row of $M_\ell$ contains the coefficients of the variables $\vec q_k$ and $u$ in the term of the (in)equality $\Lambda_\ell(\rho_i \sim_i 0)$.
\end{enumerate}
The following lemma is immediate:
\begin{lemma}
    \label{lemma:matrices-all-constraints-encoded}
    Let $\ell \in [0..j]$. For every assignment $q \gets \frac{\tau}{\eta_\ell}$ in $C_\ell$, with $q$ in $\vec q_k$, 
    there is a row in $M_\ell$ whose entries encode the coefficients that $\vec q_k$ and $u$ have in $\eta_\ell \cdot q - \tau$. 
    Similarly, for every (in)equality $\rho \sim 0$ in $\gamma_\ell$, 
    there is a row of $M_\ell$ whose entries encode the coefficients that $\vec q_k$ and $u$ have in $\rho$.
\end{lemma}
\begin{proof}
    For $\ell = 0$, $C_0$ has no assignments to variables in $\vec q_k$, and 
    the $i$th entry in the enumeration of~\Cref{eq:enum-gamma-0} is in row $j+i$. For $\ell \geq 1$, the lemma follows from Items~\eqref{gaussopt-complexity:matrix-def-1} and~\eqref{gaussopt-complexity:matrix-def-4} above.
\end{proof}

\paragraph*{Correspondence between~\GaussOpt and Bareiss algorithm.} 
The matrix $M_0$ has the form required to run (our) Bareiss algorithm. Let us denote by $B_0$ the ${(j+t) \times (k+2)}$ integer matrix such that $M_0 = B_0 \cdot \diag(\mu,\dots,\mu,1)$, 
and by $b_{ij}$ the entry of $B_0$ in position $(i,j)$; as in~\Cref{figure:matrix-M0}. 
We also use the notation $b^{(\ell)}_{i,j}$ and $b^{(\ell)}_{r \gets j}$ to denote the sub-determinants of $B_0$
analogous to those in~\Cref{subsec:variation-bareiss-body}, and define $\lambda_0 \coloneqq 1$ and $\lambda_\ell \coloneqq b^{(\ell-1)}_{\ell,\ell}$. 
We write $B_1',\dots,B_j'$ for the sequence of matrices iteratively constructed by Bareiss algorithm, starting from the matrix ${B_0' \coloneqq M_0}$.
The next lemma establishes the key correspondence between these matrices and~$M_1,\dots,M_j$.

\begin{lemma}
    \label{lemma:key-correspondence-with-Bareiss}
    Consider $\ell \in [0..j]$. Then, $M_\ell = B_\ell'$, \,$\frac{\eta_\ell}{\mu} = \abs{\lambda_\ell} \neq 0$, and if $\ell \geq 1$, then~\Cref{claim:divisions-without-remainder} holds when
    restricted to~\Cref{algo:sub-disc} having as input 
    $(C_{\ell-1}[x_m], \inst{\gamma_{\ell-1}}{\psi})$ and the equality $e_{\ell-1}$.
\end{lemma}

\begin{proof}
    The proof is by induction on $\ell \in [0..j]$. 

    \begin{description}
        \item[base case: $\ell = 0$.] 
            By definition, $M_0 = B_0'$. Since $C_0$ is a $(k,0)$-LEAC, $\eta_0 = \mu$, and so $\frac{\eta_0}{\mu} = 1 = \lambda_0$. 
        
        \item[induction hypothesis.] 
            For $\ell \geq 1$, $M_{\ell-1} = B_{\ell-1}'$ and 
            $\frac{\eta_{\ell-1}}{\mu} = \abs{\lambda_{\ell-1}} \neq 0$. 
            Moreover, if $\ell \geq 2$, 
            then~\Cref{claim:divisions-without-remainder} holds when
            restricted to~\Cref{algo:sub-disc} having as input 
            $(C_{\ell-2}[x_m], \inst{\gamma_{\ell-2}}{\psi})$ 
            and the equality $e_{\ell-2}$.

        \item[induction step: $\ell \geq 1$.]
            By induction hypothesis,~\Cref{claim:divisions-without-remainder} holds 
            whenever~\Cref{algo:sub-disc} is called on the inputs 
            $(C_{r}[x_m], \inst{\gamma_{r}}{\psi})$ and the equality $e_{r}$, 
            for every $r \in [0..\ell-2]$.
            In particular, following the correctness arguments given in the proof of~\Cref{lemma:second-step-opt}, 
            we conclude that~\GaussOpt is correct for the first $\ell-1$ iterations,  
            and $(C_{\ell-1}[x_m], \inst{\gamma_{r}}{\psi})$ thus belong to $\objcons_{k}^{\ell-1}$. Indeed, 
            to obtain correctness for the first $\ell-1$ iterations, it suffices to restrict~\Cref{claim:divisions-without-remainder} 
            to the first $\ell-1$ calls of~\Cref{algo:sub-disc}, featuring the equalities $e_0,\dots,e_{\ell-2}$. 

            Our first goal is to prove that the variable coefficients of the equality~$e_{\ell-1}$ are 
            stored in the $\ell$th row of~$M_{\ell-1}$.
            Note that since $(C_{\ell-1}[x_m], \inst{\gamma_{r}}{\psi}) \in \objcons_{k}^{\ell-1}$, 
            \Cref{algo:btp} outputs an equality of the form $a \cdot q_{n-(\ell-1)} = \tau$ with $a \neq 0$, and that, by definition, this equality is~$e_{\ell-1}$.

            \begin{claim}
                \label{claim:key-correspondence:substitution}
                The $\ell$th row of $M_{\ell-1}$ contains the variable coefficients of 
                the term $a \cdot q_{n-(\ell-1)} - \tau$.
            \end{claim}

            \begin{proof}
                If $e_{\ell-1}$ is of Type~\ref{gaussopt-connection:typeI}, 
                then by definition the $\ell$th row of $M_{\ell-1}$ contains the coefficients of the variables $\vec q_{k}$ and $u$ occurring in the 
                (in)equality~$\rho' \sim 0$ given by~$\Lambda_{\ell-1}(g_{\ell-1})$.
                (In particular, for $\ell = 1$, we have $\rho' \sim 0$ equal to $g_0$.)
                By definition of $\Lambda_{\ell-1}$, the terms $\rho'$ and $a \cdot q_{n-(\ell-1)} - \tau$ share the same variable coefficients. 

                If $e_{\ell-1}$ is of Type~\ref{gaussopt-connection:typeII}, 
                then $g_{\ell-1}$ is an equality $\rho = 0$ where $\rho$ is of the form ${b \cdot u + \mu \cdot (q' - q'') + d}$, with $b,d \in \Z$ and $q',q''$ from $\vec q_k$. 
                Moreover, inspecting~\Cref{algo:additional-hyperplanes},
                we see that the variable coefficients of $a \cdot q_{n-(\ell-1)} - \tau$ 
                corresponds to the ones obtained from $b \cdot u + \mu \cdot (q' - q'') + d$ by simultaneously applying two substitutions $\nu_1$ and $\nu_2$. 
                The substitution $\nu_1$ has one of the following forms 
                (recall that $\frac{\eta_{\ell-1}}{\mu} = \abs{\lambda_{\ell-1}}$, by  induction hypothesis):
                \begin{enumerate}
                    \item\label{sub-algo2-i1} $\sub{\frac{\eta_{\ell-1} \cdot q'}{\abs{\lambda_{\ell-1}}}}{\mu \cdot q'}$;
                        this is the case when $C_{\ell-1}$ does not assign any expression to $q'$, and the \textbf{guess} in 
                        line~\ref{algo:btp:line-shift-tau1} returns false.
                    \item\label{sub-algo2-i2} $\sub{\frac{\eta_{\ell-1} \cdot q' + \eta_{\ell-1} \cdot p}{\abs{\lambda_{\ell-1}}}}{\mu \cdot q'}$, 
                        with $p \coloneqq \fmod(q_{n-(\ell-1)},\gamma_{\ell-1})$;
                        this is the case when $q' = q_{n-(\ell-1)}$, $C_{\ell-1}$ does not assign any expression to $q'$, and the \textbf{guess} in 
                        line~\ref{algo:btp:line-shift-tau1} returns true.
                    \item\label{sub-algo2-i3} $\sub{\frac{\tau'}{\abs{\lambda_{\ell-1}}}}{\mu \cdot q'}$;
                        in this case $C_{\ell-1}$ features $q' \gets \frac{\tau'}{\eta_{\ell-1}}$ and the \textbf{guess} in 
                        line~\ref{algo:btp:line-shift-tau1} returns false.
                    \item\label{sub-algo2-i4} $\sub{\frac{\tau'\sigma}{\abs{\lambda_{\ell-1}}}}{\mu \cdot q'}$
                    where $\sigma$ is the substitution $\sub{q_{n-(\ell-1)}+p}{q_{n-(\ell-1)}}$; 
                    in this case $C_{\ell-1}$ features $q' \gets \frac{\tau'}{\eta_{\ell-1}}$ and the \textbf{guess} in 
                    line~\ref{algo:btp:line-shift-tau1} returns true.
                \end{enumerate}
                Observe that the coefficients of $\vec q_k$ and $u$ 
                are the same in $\tau'$ and $\tau'\sigma$, and that moreover 
                these terms do not contain variables to which $C_{\ell-1}$ assigns some expressions (because $C_{\ell-1}$ is a $(k,\ell-1)$-LEAC).
                We also note that the only effect that the substitutions in~\Cref{sub-algo2-i1,sub-algo2-i2} have on the variable coefficients of $\rho$ is to multiply them by $\abs{\lambda_{\ell-1}}$, because after this multiplication, $\abs{\lambda_{\ell-1}} \cdot \mu \cdot q'$ is replaced by $\eta_{\ell-1} \cdot q'$ or $\eta_{\ell-1} \cdot q' + \eta_{\ell-1} \cdot p$, but $ \eta_{\ell-1} = \abs{\lambda_{\ell-1}} \cdot \mu$.

                An analysis similar to the one above can be performed for~$\nu_2$ (simply change~$q'$ for $q''$, and line~\ref{algo:btp:line-shift-tau1} for line~\ref{algo:btp:line-shift-tau2}, in the items above). 
                From the definition of simultaneous substitution, we then conclude that 
                the variables coefficients of $a \cdot q_{n-(\ell-1)} - \tau$ 
                are exactly those in the term obtained from $\rho$
                by simultaneously applying all substitutions of the form $\sub{\frac{\rho'}{\abs{\lambda_{\ell-1}}}}{\mu \cdot q}$, 
                where $q \in \{q_{n},\dots,q_{n-(\ell-2)}\}$ and $q \gets \frac{\rho'}{\eta_{\ell-1}}$ is an assignment~in~$C_{\ell-1}$.
                For $\ell = 1$, this list of substitutions is empty, and indeed by definition of $M_0$, the $\ell$th row contains the variable coefficients of $\rho$.
                For $\ell > 1$, these substitutions correspond to the transformation 
                applied to $\rho$ in~Item~\eqref{gaussopt-complexity:matrix-def-3} of the definition of $M_{\ell-1}$, in order to define its $\ell$th column. 
                The claim then holds.
            \end{proof}

            We now analyze $M_\ell$ and $B_{\ell}'$ row by row, showing that the two matrices are equal. In the process, we will also prove the other statements in the lemma. Below, we write $\pm$ for the sign of the coefficient $a \neq 0$ in the term $a \cdot q_{n-(\ell-1)}-\tau$, and define $\alpha \coloneqq \frac{\pm a}{\mu}$
            (by~\Cref{lemma:what-btp-comples} this division is without remainder). Let us also write 
            \begin{align*}
                q_{n-(\ell-2)} \gets \frac{\tau_{n-(\ell-2)}(u,\vec q_{[\ell-1,k]})}{\eta_{\ell-1}},\quad  
                \dots\ ,\quad
                q_{n} \gets \frac{\tau_{n}(u,\vec q_{[\ell-1,k]})}{\eta_{\ell-1}}
            \end{align*}
            for the sequence of all the assignments to variables in $\vec q_k$ featured in $C_{\ell-1}$ (this is a prefix of all the assignments in $C_{\ell-1}$, since this circuit is a $(k,\ell-1)$-LEAC).

            The following sequence of Claims summarizes our analysis. 
            As their proofs are rather similar (each crucially relying on~\Cref{lemma:gaussian-elimination:new:same-row,lemma:gaussian-elimination:new:below,lemma:gaussian-elimination:new:above,lemma:one-shot-replacement:new}) 
            below we only provide a detailed proof of~\Cref{claim:key-correspondence:row-circuit}, 
            deferring the proofs of the remaining claims to~\Cref{subsec:proof-claims-lemma-correspondence-Bareiss}.

            \begin{restatable}{claim}{ClaimKeyCorrespondenceRowL}
                \superlabel{claim:key-correspondence:row-l}{proof:ClaimKeyCorrespondenceRowL}
                The $\ell$th rows of $M_\ell$ and $B_{\ell}'$ are equal. 
                Moreover, $\eta_\ell = \pm a$ and $\alpha = \frac{\eta_\ell}{\mu} = \abs{\lambda_\ell} \neq 0$.
            \end{restatable}

            \begin{claim}
                \label{claim:key-correspondence:row-circuit}
                Let $i \in [1..\ell-1]$. The $i$th rows of $M_\ell$ and $B_\ell'$ are equal. 
                Moreover, in all terms ${\tau_{n-r}\sub{\frac{\pm\tau}{\alpha}}{ \mu \cdot q_{n-{\ell-1}}}}$, with $r \in [0..\ell-2]$, 
                all coefficients of the variables~$\vec q_{[\ell,k]}$ are divisible by $\eta_{\ell-1}$, and all coefficients of $u$ are divisible~by~$\frac{\eta_{\ell-1}}{\mu}$.
            \end{claim}

            \begin{restatable}{claim}{ClaimKeyCorrespondenceRowGamma}
                \superlabel{claim:key-correspondence:row-gamma}{proof:ClaimKeyCorrespondenceRowGamma}
                Let $i \in [j+1..j+t]$. The $i$th rows of $M_\ell$ and $B_\ell'$ are equal. 
                Moreover, in all equalities and inequalities~${\gamma_{\ell-1}\sub{\frac{\pm\tau}{\alpha}}{\mu \cdot q_{n-{\ell-1}}}}$,
                all coefficients of the variables~$\vec q_{[\ell,k]}$ are divisible by $\eta_{\ell-1}$, and all coefficients of $u$ are divisible~by~$\frac{\eta_{\ell-1}}{\mu}$.
            \end{restatable}

            \begin{restatable}{claim}{ClaimKeyCorrespondenceRowEquations}
                \superlabel{claim:key-correspondence:row-equations}{proof:ClaimKeyCorrespondenceRowEquations}
                Let $i \in [\ell+1..j]$. The $i$th rows of $M_\ell$ and $B_\ell'$ are equal.
            \end{restatable}
            \begin{proof}[Proof of~\Cref{claim:key-correspondence:row-circuit}]
                By definition (Item~\eqref{gaussopt-complexity:matrix-def-1}),
                the matrix $M_\ell$ includes, in these rows,
                the variable coefficients of the assignments in $C_\ell$, ranging from the $(\ell-1)$th assignment to the $2$nd one (in this reverse order). The corresponding rows in $M_{\ell-1}$ contain the variable coefficients of the terms ${(\eta_{\ell-1} \cdot q_{n} - \tau_{n})}$ to
                ${(\eta_{\ell-1} \cdot q_{n-(\ell-2)} - \tau_{n-(\ell-2)})}$, which arise from the assignments in~$C_{\ell-1}$. 
                Let $i \in [1..\ell-1]$. 
                Since line~\ref{algo:sub-disc:assert-circuit} of~\Cref{algo:sub-disc} 
                prepends one assignment to $C_{\ell-1}$ and only updates the rest, the $i$th row of $M_\ell$ corresponds precisely 
                to the term obtained by running line~\ref{algo:sub-disc:assert-circuit}
                on the assignment $q_{n-(i-1)} \gets \frac{\tau_{n-(i-1)}}{\eta_{\ell-1}}$. We will analyze this update below.
                At the same time, since the variable coefficients of this assignment are stored in the $i$th row of $B_{\ell-1}'$, we can alternatively track how Bareiss algorithm updates this row when computing~$B_{\ell}'$ from $B_{\ell-1}'$. 
                We will then deduce that the $i$th rows of $M_\ell$ and $B_\ell'$ are equal. 

                Let us examine how line~\ref{algo:sub-disc:assert-circuit} updates
                the assignment $q_{n-(i-1)} \gets \frac{\tau_{n-(i-1)}}{\eta_{\ell-1}}$. Let us write~$\tau_{n-(i-1)}$ as~$\beta \cdot \mu \cdot q_{n-(\ell-1)} + \tau'$.
                Line~\ref{algo:sub-disc:assert-circuit} first constructs the term $\tau_{n-(i-1)}\sub{\frac{\pm\tau}{\alpha}}{\mu \cdot q_{n-(\ell-1)}}$. The substitution multiplies $\tau_{n-(i-1)}$ by $\alpha \geq 1$, to then replace $\alpha \cdot \mu \cdot q_{n-(\ell-1)}$ by $\pm \tau$. The resulting term is
                $\pm \beta \cdot \tau + \alpha \cdot \tau'$
                Note that multiplying the denominator by $\alpha$ yields $\alpha \cdot \eta_{\ell-1} = \frac{\pm a}{\mu} \cdot \eta_{\ell-1} = \pm{a} \cdot \frac{\eta_{\ell-1}}{\mu}$. So, 
                at this intermediate stage of line~\ref{algo:sub-disc:assert-circuit}, the assignment can be viewed as
                \begin{equation}
                    \label{eq:expr-pre-CL}
                    q_{n-(i-1)} \gets \frac{\pm \beta \cdot \tau + \alpha \cdot \tau'}{\pm{a} \cdot \frac{\eta_{\ell-1}}{\mu}}\,.
                \end{equation}
                In the upcoming analysis of the update performed by Bareiss algorithm, we will show that every variable coefficient in the term $\pm \beta \cdot \tau + \alpha \cdot \tau'$ is divisible by $\frac{\eta_{\ell-1}}{\mu}$. This implies that the constant of the term must also be divisible by $\frac{\eta_{\ell-1}}{\mu}$; otherwise the expression in~\Cref{eq:expr-pre-CL} would fail to evaluate to an integer under any variable assignment. 
                This explains the \textbf{assert} command of line~\ref{algo:sub-disc:assert-circuit}. 
                Line~\ref{algo:sub-disc:assert-circuit}
                concludes by dividing every integer in the term $\pm \beta \cdot \tau + \alpha \cdot \tau'$
                by $\frac{\eta_{\ell-1}}{\mu}$, and setting the denominator to $\pm a$. Let us write $(\pm \beta \cdot \tau + \alpha \cdot \tau')/\frac{\eta_{\ell-1}}{\mu}$ for the term resulting from these divisions. The circuit~$C_{\ell}$ thus 
                features the assignment
                \begin{equation*}
                    q_{n-(i-1)} \gets \frac{(\pm \beta \cdot \tau + \alpha \cdot \tau')/\frac{\eta_{\ell-1}}{\mu}}{\pm{a}}\,,
                \end{equation*}
                and the $i$th row of $M_\ell$ stores the variable coefficients of~$\pm a \cdot q_{n-(i-1)} - \big((\pm \beta \cdot \tau + \alpha \cdot \tau')/\frac{\eta_{\ell-1}}{\mu}\big)$. 
                
                We now turn to Bareiss algorithm.
                By induction hypothesis, the $i$th row of $B_{\ell-1}'$ contains the variable coefficients of the term $\eta_{\ell-1} \cdot q_{n-(i-1)} - (\beta \cdot \mu \cdot q_{n-(\ell-1)}+\tau')$.
                The algorithm first multiplies this row by~$\alpha$, 
                resulting in the variable coefficients of~$\alpha \cdot \eta_{\ell-1} \cdot q_{n-(i-1)} - \alpha (\beta \cdot \mu \cdot q_{n-(\ell-1)}+\tau')$.
                Next, the algorithm subtracts to this row the quantity~$\pm (-\beta) \cdot \vec r_\ell$, where $\vec r_\ell$ is the $\ell$th row of $B_{\ell-1}'$. By~\Cref{claim:key-correspondence:substitution}, 
                $\vec r_\ell$ holds the variable coefficients of $a \cdot q_{n-(\ell-1)}-\tau$. Hence, after this subtraction, the $i$th row contains the variable coefficients of the term 
                \begin{align*}
                    &\alpha \cdot \eta_{\ell-1} \cdot q_{n-(i-1)} - \alpha (\beta \cdot \mu \cdot q_{n-(\ell-1)}+\tau') 
                    - \pm (-\beta) \cdot (a \cdot q_{n-(\ell-1)}-\tau)\\
                    ={}&\alpha \cdot \eta_{\ell-1} \cdot q_{n-(i-1)} 
                    - (\pm \beta \cdot \tau+\alpha \cdot \tau').
                \end{align*}
                Lastly, each entry of the $i$th row is divided by $\abs{\lambda_{\ell-1}} = \frac{\eta_{\ell-1}}{\mu}$. 
                Thanks to~\Cref{lemma:gaussian-elimination:new:above}, we know that these divisions are exact, since the results correspond to sub-determinants of the matrix $B_0'$.
                Moreover, because $q_{n-(i-1)}$ does not appear in neither $\tau$ nor $\tau'$, 
                we conclude that every variable coefficient in $\pm \beta \cdot \tau+\alpha \cdot \tau'$ is divisible by $\frac{\eta_{\ell-1}}{\mu}$. Therefore, the divisions performed in line~\ref{algo:sub-disc:assert-circuit} 
                of~\Cref{algo:sub-disc} are also without remainder.
                Since $\alpha \cdot \eta_{\ell-1} = \pm a \cdot \frac{\eta_{\ell-1}}{\mu}$, we conclude that the $i$th row of $B_{\ell}'$ holds the variables coefficients of~$\pm a \cdot q_{n-(i-1)} - \big((\pm \beta \cdot \tau + \alpha \cdot \tau')/\frac{\eta_{\ell-1}}{\mu}\big)$. That is, the $i$th rows of $M_\ell$ and $B_{\ell}'$ coincide.
                
                To complete the proof, let us address the second statement of the claim. 
                Consider once more the term ${\tau_{n-(i-1)}\sub{\frac{\pm\tau}{\alpha}}{\mu \cdot q_{n-(\ell-1)}}}$, that is, 
                $\pm \beta \cdot \tau + \alpha \cdot \tau'$. We have already established that all variable coefficients of this term are divisible by $\frac{\eta_{\ell-1}}{\mu}$; in particular, this shows the second statement of the claim for the variable $u$. As for the remaining variables,~\Cref{lemma:gaussian-elimination:new:above} 
                guarantee that, once divided by~$\frac{\eta_{\ell-1}}{\mu}$, their coefficients are still divisible by $\mu$.
                Therefore, in $\pm \beta \cdot \tau + \alpha \cdot \tau'$, 
                all coefficients of variables other than~$u$ are divisible by~$\eta_{\ell-1}$.
            \end{proof}

            \Cref{lemma:key-correspondence-with-Bareiss} follows:
            Claims~\ref{claim:key-correspondence:row-l} 
            to~\ref{claim:key-correspondence:row-equations}
            imply that $M_\ell = B_\ell'$. 
            \Cref{claim:key-correspondence:row-l} 
            establishes $\frac{\eta_\ell}{\mu} = \abs{\lambda_\ell} \neq 0$. 
            Claims~\ref{claim:key-correspondence:row-circuit}
            and~\ref{claim:key-correspondence:row-gamma} imply that
            \Cref{claim:divisions-without-remainder} holds 
            when
            restricted to~\Cref{algo:sub-disc} having as input 
            the pair~$(C_{\ell-1}[x_m], \inst{\gamma_{\ell-1}}{\psi})$ and the equality $e_{\ell-1}$.
            \qedhere
    \end{description}
\end{proof}

\noindent
\proofnote{proof:ClaimDivisionsWithoutRemainder}%
\Cref{claim:divisions-without-remainder} follows as a corollary of~\Cref{lemma:key-correspondence-with-Bareiss}; thus completing the proof of correctness of~\GaussOpt.

\subsection{Complexity of~\GaussOpt}
\label{subsec:complexity-elimvars}

In addition to being crucial for establishing the correctness of~\GaussOpt, \Cref{lemma:key-correspondence-with-Bareiss} allows us to obtain a refined complexity analysis of the procedure.
The next lemma summarizes this analysis.

\begin{restatable}{lemma}{ILEPGaussOptBoundsNew}
  \superlabel{lemma:ILEP:GaussOptBoundsNew}{proof:ILEPGaussOptBoundsNew}
  The algorithm from~\Cref{lemma:second-step-opt} runs in non-deterministic polynomial time. 
  Consider its execution on an input $(\vec q, \objfun{C}{x_m}, \inst{\gamma}{\psi})$, 
  where $(C, \inst{\gamma}{\psi})$ belongs to $\objcons_k^0$, 
  and define:
  \begin{align*}
      L &\coloneqq 3 \cdot \mu_C \cdot (4 \cdot \ceil{\log_2(2 \cdot \xi_C + \mu_C)}+8),\\
      Q &\coloneqq \max\{\abs{b} : \text{$b \in \Z$ is a coefficient of~$q_{n-k}$ or of  a variable in $\vec q$, in a term from $\fterms(\gamma)$}\},\\
      \qquad
      U &\coloneqq \max\{\abs{a} : \text{$a = L$ or $a \in \Z$ is a coefficient of $u$ in a term from~$\fterms(\gamma)$}\},
      \\
      R &\coloneqq \max\{\abs{d} : \text{$d = L$ or $d \in \Z$ is a constant of a term from~$\fterms(\gamma)$}\}.
  \end{align*}
  In each non-deterministic branch $\beta$, 
  the algorithm returns a pair $(\objfun{C'}{x_m},\inst{\gamma'}{\psi})$ such that:
  \begin{enumerate}[itemsep=1pt]
      \item\label{lemma:ILEP:GaussOptBounds:i1} $\gamma'$ features $k$ constraints more than $\gamma$, they are all divisibility constraints.
      \item\label{lemma:ILEP:GaussOptBounds:i4} The circuits $C$ and $C'$ assign the same expressions to~$x_{n-k},\dots,x_n$ (in particular, $\mu_{C} = \mu_{C'}$).
      \item\label{lemma:ILEP:GaussOptBounds:i5} In terms $\tau$ either from $\fterms(\gamma')$ or in assignments $q_{n-i} \gets \frac{\tau}{\eta_{C'}}$ of $C'$ (where ${i \in [0..k-1]}$), 
      \begin{itemize}[itemsep=1pt]
          \item the coefficient of the variable~$q_{n-k}$ is $\mu_C \cdot c$, for some~$c \in \Z$ with $\abs{c} \leq (k+1)^{k+1} \big(\frac{Q}{\mu_C}\big)^{k+1}$; 
          \item the absolute value of the coefficient of the variable $u$ is bounded by $(k+1)^{k+1} \big(\frac{Q}{\mu_C}\big)^{k} U$;
          \vspace{-2pt}
          \item the absolute value of the constant is bounded by ${\frac{((k+1) \cdot Q)^{2(k+2)^2}}{(\mu_C)^{2k^2}} \cdot \fmod(\gamma) \cdot R}$.
      \end{itemize}
      \vspace{-5pt}
      \item\label{lemma:ILEP:GaussOptBounds:i6} The positive integer $\fmod(\gamma')$ divides $c \cdot \fmod(\gamma)$,
      for some positive integer $c \leq \frac{(k \cdot Q)^{k^2}}{(\mu_C)^{k (k-1)}}$.
      \vspace{-3pt}
      \item\label{lemma:ILEP:GaussOptBounds:i7} We have $\eta_{C'} = \mu_{C} \cdot g$, for some positive integer $g \leq k^k \big(\frac{Q}{\mu_C}\big)^{k}$.
  \end{enumerate}
\end{restatable}

\begin{proof}[Proof idea.]
    The bounds follow by applying~\Cref{lemma:key-correspondence-with-Bareiss} 
    in conjunction with~\Cref{lemma:gaussian-elimination:new:below,lemma:gaussian-elimination:new:above}, 
    and recalling that the Leibniz formula for determinants yields $\abs{\det(A)} \leq d^d \cdot \prod_{i=1}^{d} \alpha_i$ 
    for any $d \times d$ integer matrix $A$ in which the entries of the $i$th column are bounded, in absolute value, by $\alpha_i \in \N$.
\end{proof}

\section{Proof of Theorem~\ref{theorem:small-optimum}}
\label{sec:putting-all-together}

In this section, we complete the proof of Theorem~\ref{theorem:small-optimum}: 
we define the procedure~\OptILEP for solving the integer 
linear-exponential programming optimization problem, 
to then show that the procedure runs in non-deterministic polynomial time 
and returns an ILESLP encoding an (optimal) solution, if one exists.
The section is divided into four parts. We begin with an overview 
of the procedure, expanding on the brief summary provided in~\Cref{subsection:OptILEP}.
As part of this overview, we introduce a slight variant of LEACs, which we refer to as~\preleac{s}. 
In~\Cref{subsec:correctness-optilep} we present the full 
correctness proof of~\OptILEP, followed by its complexity analysis in~\Cref{subsec:complexity-optilep}.
The pseudocode of~\OptILEP considers the setting of maximizing a single variable~$x$ subject to an integer linear-exponential program. 
In~\Cref{subsec:optimize-general-terms} we show (using rather standard arguments), how to extend the procedure to the optimization (maximization or minimization) 
of arbitrary linear-exponential terms, completing the proof of~Theorem~\ref{theorem:small-optimum}.

\subsection{Overview of~\OptILEP}
\label{subsec:overview-optilep}
\begin{algorithm}
  \caption{\OptILEP: Exploration of optimal solutions for ILEP.}
  \label{pseudocode:opt-ilep}

  \tikzmark{left-margin}
  \setstretch{1.1}
  \begin{algorithmic}[1]
    \Require 
      \begin{minipage}[t]{0.92\linewidth}
        \setlength{\tabcolsep}{2pt}
        \begin{tabular}[t]{rcp{0.85\linewidth}}
          $\phi(\vec x)$&:& integer linear-exponential program;\\
          $w$&:& a variable from $\vec x$ (to be maximized)
        \end{tabular}
      \end{minipage}
    \NDBranchOutput An ILESLP $\sigma_\beta$.
    \medskip
    \State $C \gets \emptyset$\label{optilep:line:initialize-C}
    \Comment{the empty~$0$-\preleac. The objective function is $\objfun{C}{w}$}
    \State \textbf{let} $x_0$ be a fresh variable 
    \State $\theta$ $\gets$ \textbf{guess} ordering ${2^{x_n} \geq {\ldots} \geq 2^{x_1} \geq 2^{x_0} = 1}$, where $x_1,\dots,x_n$ is a permutation of $\vec x$\label{optilep:line:guess-ordering}
    \Statex \Comment{below, we write $x_m$ for the variable among $x_1,\dots,x_n$ corresponding to $w$}
      \State $\vec r$ $\gets$ empty vector of (remainder) variables\label{optilep:line:initialize-remainders}
      \While{$\theta$ is not the ordering $2^{x_0} = 1$}\label{optilep:line:while}
          \tikzmark{all-step-i-begin}
          \State $2^x$ $\gets$ leading exponential term of $\theta$\label{optilep:line:stepI:define-2x} 
          \State $2^y$ $\gets$ second-leading exponential term of $\theta$\label{optilep:line:stepI:define-2y}
          \State $(\gamma(q_x,\vec q,u),\psi(\vec y, r_x,\vec r'))$ $\gets$ apply the algorithm from~\Cref{lemma:CMS:first-step} on~$(\phi,\theta)$\vspace{2pt}\label{optilep:line:stepI:run}
          \Statex \Comment{$(q_x,\vec q)$ quotient variables, $(r_x,\vec r')$ new remainder variables, $u$ proxy for $2^{x-y}$}
          \State update $C$: 
          \begin{minipage}[t]{0.85\linewidth}
            add the assignment $x \gets q_x \cdot 2^y + r_x$, and replace each variable in $\vec r$ following the system $\vec r = \vec q \cdot 2^y + \vec r'$ stemming from the above call of the algorithm form~\Cref{lemma:CMS:first-step}
          \end{minipage}\label{line:update-C-stepI}\label{optilep:line:stepI:updateC}
          \tikzmark{all-step-ii-begin}
          \vspace{1pt}
          \State $\gamma \gets \gamma \land \vec q \geq 0 \land q_x \geq 0$\label{optilep:line:stepII:prepare-gamma-1} 
          \Comment{prepare formulae for the call to~\GaussOpt}
          \State update $\gamma$: replace each (in)equality $\tau \sim 0$ with $\mu_C \cdot \tau \sim 0$\label{optilep:line:stepII:prepare-gamma-2} 
          \State $\psi' \gets \psi \land \theta \land (x = q_x \cdot 2^y + r_x) \land (u = 2^{x-y})$\label{optilep:line:stepII:prepare-psi} 
          \State $(\objfun{C}{x_m}, \inst{\gamma'}{\psi'})$ $\gets$ $\GaussOpt(\vec q,\, \objfun{C}{x_m},\, \inst{\gamma}{\psi'})$ 
          \Comment{\Cref{lemma:second-step-opt}}\label{optilep:line:stepII:run}%
          \tikzmark{all-step-ii-end}%
          \State $(\gamma''(q_x),\psi''(y,r_x))$ $\gets$ apply the algorithm from~\Cref{lemma:CMS:third-step} on~$\gamma'$
          \Comment{\textup{\textsc{Step III}}}\label{optilep:line:stepIII}%
          \tikzmark{all-step-iii-end}%
          \State $\ell$ $\gets$ greatest non-negative lower bound of $q_x$ in~$\gamma''$ \Comment{default:~$0$}\label{line:post-step-iii-start}\label{optilep:line:stepIV:lower}%
          \State $h$ $\gets$ least upper bound of $q_x$ in~$\gamma''$ \Comment{default:~$\infty$}\label{optilep:line:stepIV:upper}%
          \If{$h = \infty$} $h \gets \ell + \fmod(\gamma'')$\label{optilep:line:stepIV:upper-infty} 
          \EndIf
          \State $v$ $\gets$ \textbf{guess} a value in $[\ell..h]$ such that $\gamma''(v)$ is true\label{optilep:line:stepIV:guess} 
          \State update $C$: translate into a~\preleac following~\Cref{remark:LEAC-to-PRELEAC}, replacing $u$ for $2^{x-y}$ and $q_{n-k}$~for~$v$\label{optilep:line:stepIV:remove-qx}%
          \tikzmark{all-prepare-next-begin}
          \State $\vec r$ $\gets$ $(r_x,\vec r')$\label{optilep:line:stepV:prepare-r} 
          \State $\phi$ $\gets$ $\psi \land \psi''$\label{optilep:line:stepV:prepare-phi} 
          \State remove $2^x$ from $\theta$\label{line:post-step-iii-end}\label{optilep:line:stepV:parepare-theta} 
          \tikzmark{all-prepare-next-end}
      \EndWhile
      \State \textbf{assert}($\phi(\vec 0)$ is true)\label{optilep:line:check-with-zeros}
      \State update $C$: replace $x_0$ and every variable in $\vec r$ with $0$\label{optilep:line:replace-with-zeros}
      \State \textbf{return} $C$ 
      \Comment{$C$ is an ILESLP encoding a solution to $\phi$}
  \end{algorithmic}
  \AddNote{all-step-i-begin}{all-step-ii-begin}{left-margin}{Step I}%
  \AddNote{all-step-ii-begin}{all-step-ii-end}{left-margin}{Step II}%
  \AddNote{all-step-iii-end}{all-prepare-next-begin}{left-margin}{Step IV}%
\end{algorithm}

The pseudocode of~\OptILEP is given in~\Cref{pseudocode:opt-ilep}.
Echoing~\Cref{subsection:OptILEP}, 
the procedure starts by guessing an ordering $\theta$ of the form $2^{x_n} \geq \dots \geq 2^{x_1} \geq 2^{x_0} = 1$, where $x_1,\dots,x_n$ are the variables appearing input~$(\phi,w)$ of~\OptILEP,  
see lines~\ref{optilep:line:initialize-C}--\ref{optilep:line:initialize-remainders}. 
These lines also initialize the \emph{remainder variables}~$\vec r$ (as described in~\Cref{section:summary-procedure}), 
and the circuit~$C$, which will ultimately become the ILESLP encoding the computed 
solution. After this initialization step, the procedure enters its main loop. 

Let $2^x$ and $2^y$ be the leading and second-leading exponential terms of $\theta$, respectively (as in lines~\ref{optilep:line:stepI:define-2x} and~\ref{optilep:line:stepI:define-2y}).
As mentioned in~\Cref{subsection:OptILEP}, the main loop of the procedure 
eliminates $x$ by mirroring the four steps 
of the procedure from~\cite{ChistikovMS24}, with the key difference that Step II is replaced by our 
optimum-preserving procedure~\GaussOpt. We refer the reader back to~\Cref{section:summary-procedure} 
for a refresher, particularly on the specifications of Steps I and III, 
which we treat here as black boxes.

As discussed in~\Cref{section:summary-procedure},  
Step~I ``divides'' all constraints in~$\phi$ by $2^y$, 
non-deterministically computing from $\phi$ and $\theta$ 
a pair of formulae of the form~$(\gamma(q_x,\vec q,u),\psi(\vec y, r_x, \vec r'))$, where, in particular,~$\gamma$ is a linear program with divisions.
As described in the specification of Step~I given by~\Cref{lemma:CMS:first-step}, $\phi$ and $(\gamma,\psi)$ are ``coupled''
by the system featuring the equalities $x = q_x \cdot 2^y + r_x$ and 
$\vec r = \vec q \cdot 2^y + \vec r'$ (\Cref{eq:CMS:first-step}), with $q_x$ and $\vec q$ \emph{quotient variables}, and $r_x$ and $\vec r'$ fresh remainder variables. 
The change of variables given by this system 
must be applied also to the circuit $C$; 
this is done in line~\ref{optilep:line:stepI:updateC}.

The goal of Step~II is to eliminate the variables~$\vec q$ from~$\gamma$. 
We preserve optimal solutions while eliminating these variables
by appealing to our instantiation of~\GaussOpt.
However, according to~\Cref{lemma:second-step-opt}, 
a correct invocation to this algorithm requires that its input 
belong to~$\objcons_k^0$ (for some~$k$). 
Accordingly, lines~\ref{optilep:line:stepII:prepare-gamma-1}--\ref{optilep:line:stepII:prepare-psi}
perform the necessary manipulations on~$\gamma$ and $\psi$
to ensure this condition is met.

After eliminating the variables~$\vec q$ and appropriately updating the circuit $C$~via~\GaussOpt, line~\ref{optilep:line:stepIII} 
applies Step~III from~\cite{ChistikovMS24}. This eliminates 
the variables $x$ and $u$ (where $u$ is the proxy for $2^{x-y}$). 
According to the specification in~\Cref{lemma:CMS:third-step}, 
this step transforms the linear program with divisions~$\gamma'(q_x,u)$ produced as output of~\GaussOpt 
into a pair consisting of a linear program with divisions~$\gamma''(q_x)$ 
and linear-exponential program with divisions $\psi''(y,r_x)$.

Step IV (lines~\ref{optilep:line:stepIV:lower}--\ref{optilep:line:stepIV:remove-qx}) eliminates $q_x$. In~\cite{ChistikovMS24}, this is achieved by checking
whether the linear program with divisions~$\gamma''$ is satisfiable, and
replacing it with~$\top$ if so. In contrast, \OptILEP instead works by trying
``all'' the solutions to~$\gamma''$. More precisely, since $\gamma''$ is
univariate, all of its inequalities can be rewritten in the form $\ell^* \leq
q_x$ or $q_x \leq h^*$, with $\ell^*,h^* \in \Z$. Then, every solution to
$\gamma''$ must lie in the interval $[\ell..h]$, where $\ell$ is either $0$ or
the largest such integer $\ell^*$, and $h$ is either $\infty$ or the smallest
such integer $h^*$. If $h \neq \infty$, we can (non-deterministically) test all
values in this interval (in the complexity proof we will show that both $\ell$
and $h$ have polynomial bit size). If instead $h = \infty$, in the correctness
proof we will show that either no optimal solution exists, or the 
objective function is independent of~$q_x$. 
To cover the latter case, the algorithm updates
$h$ from $\infty$ to $\ell + \fmod(\gamma'')$ (line~\ref{optilep:line:stepIV:upper-infty}),
ensuring that at least one solution of $\gamma''$ is explored.
After eliminating $q_x$, the body of the loop terminates with a small ``Step V'' (lines~\ref{optilep:line:stepV:prepare-r}--\ref{optilep:line:stepV:parepare-theta}), 
which prepares $\phi$ for the next iteration, 
and updates~$\theta$ by removing $2^x$ (making $2^y$ the new leading exponential term).

In the above overview, we have not elaborated on 
the structure of the circuit $C$ during the procedure. 
According to~\Cref{lemma:second-step-opt}, 
$C$ must be a $(k,0)$-LEAC when~\GaussOpt 
is called, and it evolves into a $(k,k)$-LEAC 
by the time this algorithm terminates.
From this information, we know that $C$ must become 
a $(k,0)$-LEAC precisely when line~\ref{optilep:line:stepI:updateC} executes. 
Prior to this line, however, $C$ contains no quotient variable, 
and instead has the structure given in the following definition:

\begin{definition}[\preleac]
    \label{def:pre-LEAC}
    Let $k \in [0..n]$. 
    A $k$-\preleac~$C$ is a sequence of assignments
    \begin{align*}
        \hspace{0.4cm} x_{n-i} & \gets \frac{\textstyle\sum_{j={i+1}}^{k} a_{i,j} \cdot
        2^{x_{n-j}}}{\mu}
         + r_{n-i} &\hspace{-1cm}\text{for $i$ from $k-1$ to $0$},
    \end{align*}
    where every $a_{i,j}$ is in $\Z$, 
    and the denominator $\mu$ is a positive integer.
\end{definition}

We transfer the notation used for LEACs also to~\preleac{s}. 
In particular, we refer to the denominator $\mu$ as $\mu_C$, postulating $\mu_C \coloneqq 1$ when $k = 0$ (note that $C$ is the empty sequence in this case). We also define $\xi_C \coloneqq \sum\{\abs{a_{i,j}}: \text{$i \in [0..k-1]$, $j \in [i+1..k]$}\}$, and write $\vars(C)$ for the set of \emph{free variables} of $C$, that is, $x_{n-k}$ and the variables $r_{n-i}$. 
Lastly, for a variable $x_m$ with $m \in [0..n]$, 
we write $\objfun{C}{x_m}$ for the function analogous 
to the one defined for LEACs on~page~\pageref{def:LEAC}. 

It is easy to verify that, starting from $C$ being a $k$-\preleac, 
line~\ref{optilep:line:stepI:updateC} of~\OptILEP produces a $(k,0)$-LEAC. 
More interesting is the transformation that occurs in line~\ref{optilep:line:stepIV:remove-qx}, 
where the variables $u$ and~$q_{n-k}$ are removed from the $(k,k)$-LEAC 
returned by~\OptILEP. 
The following remark describes this transformation, 
which yields a $(k+1)$-\preleac.

\begin{remark}[(From $(k,k)$-LEACs to $(k+1)$-\preleac{s})]
  \label{remark:LEAC-to-PRELEAC}
  Let $k \in [0..n-1]$. Let $C$ be a $(k,k)$-LEAC
  \begin{align*}
      \hspace{0.4cm} q_{n-i} & \gets \frac{b_{n-i} \cdot u + c_{n-i} \cdot q_{n-k} + d_{n-i}}{\eta} 
          &\text{for $i$ from $k-1$ to $0$},\\
      \hspace{0.4cm} x_{n-i} & \gets \frac{\textstyle\sum_{j={i+1}}^{k} a_{i,j} \cdot
      2^{x_{n-j}}}{\mu} +
      q_{n-i} \cdot 2^{x_{n-k-1}} + r_{n-i} &\hspace{-3cm}\text{for $i$ from $k$ to $0$},
  \end{align*}
  such that $\mu$ divides $\eta$. Let~$\lambda = \frac{\eta}{\mu}$.
  By replacing $u$ for $2^{x_{n-k}-x_{n-k-1}}$, assigning an integer $v$ to~$q_{n-k}$, and substituting the expressions for $q_{n-(k-1)},\dots,q_n$ into the expressions for $x_{n-k},\dots,x_n$, one transforms $C$ into the following $(k+1)$-\preleac~$C'$:
  \begin{align*}
      \hspace{0.4cm} x_{n-k} & \gets \frac{\eta \cdot v \cdot 2^{x_{n-k-1}}}{\eta} + r_{n-k}\\
      \hspace{0.4cm} x_{n-i} & \gets \frac{\textstyle  (c_{n-i} \cdot v + d_{n-i}) \cdot 2^{x_{n-k-1}} + (\lambda \cdot a_{i,k} + b_{n-i}) \cdot 2^{x_{n-k}} + \sum_{j={i+1}}^{k-1} \lambda \cdot a_{i,j} \cdot
      2^{x_{n-j}}}{\eta} + r_{n-i}\\ 
      &&\hspace{-3cm}\text{for $i$ from $k-1$ to $0$}.
  \end{align*}
  Note that $\vars(C') = \vars(C) \setminus \{u,q_{n-k}\}$, and when evaluating $C$ and $C'$ on a map $\nu \colon \vars(C) \to \N$ satisfying ${\nu(q_{n-k}) = v}$ and $\nu(u) = 2^{\nu(x_{n-k})-\nu(x_{n-k-1})}$, the values taken by~$x_{n-k},\dots,x_n$ coincide, 
  that is, $\objfun{C}{x_{n-i}}(\nu) = \objfun{C'}{x_{n-i}}(\nu)$ for every~$i \in [0..k]$.
\end{remark}

\subsection{Correctness of~\OptILEP}
\label{subsec:correctness-optilep}

The next proposition states that~\OptILEP correctly solves ILEP.

\begin{proposition}
  \label{theorem:correctness-full-procedure}
  There is a non-deterministic procedure with the following specification:
  \begin{description}
    \setlength{\tabcolsep}{2pt}
    \item[\textbf{\textit{Input:}}] 
      \begin{minipage}[t]{0.94\linewidth}
        \hspace{3pt}
        \begin{tabular}[t]{rcp{0.87\linewidth}}
        $\phi(\vec x)$&:& an integer linear-exponential program;\\
        $w$&:& a variable occurring in $\vec x$ (to be maximized).
        \end{tabular}
      \end{minipage}
    \item[\textbf{\textit{Output of each branch ($\beta$):}}]
    
    \begin{minipage}[t]{\linewidth}
      \hspace{3pt}
      \begin{tabular}[t]{rcp{0.75\linewidth}}
        $\sigma_\beta$&:& an ILESLP.
        \end{tabular}
      \end{minipage}
  \end{description}
  The algorithm ensures that,
  if $\phi$ is satisfiable (resp., the problem of maximizing $x$ subject to~$\phi$ has a solution), 
  then there a branch~$\beta$ such that $\sem{\sigma_\beta}$ is a solution 
  (resp.,~an optimal solution) to $\phi$. 
\end{proposition}

\begin{proof}
  Let $\phi_0(\vec x)$ be the linear-exponential program in input of~\OptILEP, and let $n$ denote the number of variables in $\phi_0$.
  In line~\ref{optilep:line:guess-ordering}, the algorithm guesses an ordering $2^{x_n} \geq {\ldots} \geq 2^{x_1} \geq 2^{x_0} = 1$, where $x_1,\dots,x_n$ is a permutation of $\vec x$, and $x_0$ is a fresh variable.
  Let $\Theta$ be the set of all such ordering. Clearly, $\phi_0$ is equivalent to $\bigvee_{\theta \in \Theta} (\phi_0 \land \theta)$. To prove the proposition, it suffices to show, for a given $\theta \in \Theta$, that if $\phi_0 \land \theta$ is satisfiable (resp., the problem of maximizing $x$ subject to~$\phi_0 \land \theta$ has a solution), 
  then, in a non-deterministic branch~$\beta$, the algorithm returns an ILESLP~$\sigma_\beta$ such that $\sem{\sigma_\beta}$ is a solution 
  (resp.,~an optimal solution) to $\phi_0 \land \theta$.
  Therefore, throughout the proof we fix the ordering guessed in line~\ref{optilep:line:guess-ordering}
  to be some $\theta_0 \coloneqq {(2^{x_n} \geq {\ldots} \geq 2^{x_1} \geq 2^{x_0} = 1)}$ from $\Theta$. Moreover, let $x_m$ (for some $m \in [1..n]$) denote the variable to be maximized, with respect to the order $\theta_0$.

  Throughout the proof, we write $S_k$ for the set of all triples $(C,\theta,\phi)$ that represent the state of~\OptILEP in any of its non-deterministic branches at the point when the execution reaches line~\ref{optilep:line:while} (the condition of the only \textbf{while} loop of the procedure) for the $(k+1)$th time.
  Since we fix the ordering $\theta_0$, in particular $S_0 = \{(\emptyset,\theta_0,\phi_0)\}$, where $\emptyset$ is the empty sequence assigned to $C$ in line~\ref{optilep:line:initialize-C}.

  We show that the \textbf{while} loop of~\OptILEP enjoys the following loop invariant:

  \begin{description}
    \item[loop invariant.] Let~$k \in \N$. For every $(C,\theta,\phi) \in S_k$: 
    \begin{enumerate}
      \item[$I_1$.]\customlabel{$I_1$}{optilep:inv:i1} $\theta$ is the ordering $\theta_k \coloneqq (2^{x_{n-k}} \geq \dots \geq 2^{x_1} \geq 2^{x_0} = 1)$.
      \item[$I_2$.]\customlabel{$I_2$}{optilep:inv:i2} $\phi$ is a linear-exponential program with divisions featuring variables $\vec y_k \coloneqq (x_0,\dots,x_{n-k})$ and \emph{remainder variables} $\vec r_{k-1} \coloneqq (r_{n-(k-1)},\dots,r_n)$. The remainder variables do not occur in exponentials, and for every $i \in [0..k-1]$, $\phi$ implies $r_{n-i} < 2^{x_{n-k}}$.
      \item[$I_3$.]\customlabel{$I_3$}{optilep:inv:i3} $C$ is a $k$-\preleac of the form $(x_{n-(k-1)} \gets \frac{\tau_{n-(k-1)}}{\mu_C} + r_{n-(k-1)} \,,\, \dots \,,\, x_n \gets \frac{\tau_{n}}{\mu_C} + r_{n})$. 
      \item[$I_4$.]\customlabel{$I_4$}{optilep:inv:i4} Given a solution $\nu \colon \vars(\phi \land \theta) \to \N$ to $\phi \land \theta$, the map 
      $\nu + \textstyle\sum_{i=0}^{k-1} [x_{n-i} \mapsto g_{n-i}]$, 
      where $g_{n-i}$ is the value taken by~$x_{n-i}$ when evaluating $\objfun{C}{x_m}$ on $\nu$, is a solution to $\phi_0 \land \theta_0$.
    \end{enumerate}  
    Moreover: 
    \begin{enumerate}
      \addtocounter{enumi}{4}
      \item[$I_5$.]\customlabel{$I_5$}{optilep:inv:i5} $\exists \vec x_{k-1} (\phi_0 \land \theta_0)$ is equivalent to $\exists \vec r_{k-1} \bigvee_{(C,\theta,\phi) \in S_k} (\phi \land \theta)$, where $\vec x_{k-1} \coloneqq (x_{n-(k-1)},\dots,x_n)$.
      \item[$I_6$.]\customlabel{$I_6$}{optilep:inv:i6} If $\max\{ \nu(x_m) :
        \text{$\nu$ is a solution to $\phi_0 \land \theta_0$}\}$ exists, then it
        is equal to 
        \[
          \max\{\objfun{C}{x_m}(\nu) : \text{$(C,\theta,\phi) \in S_k$ and $\nu$ is a solution to $\phi \land \theta$}\}.
        \]
    \end{enumerate}
  \end{description}
  (For $k \in [0..n]$, the loop invariant assumes a fixed set of variables~$r_0,\dots,r_n$ 
  that are reused across different non-deterministic branches and across iterations of the \textbf{while} loop. 
  This assumption is without loss of generality.)

  Let $(C,\theta,\phi) \in S_n$. 
  The loop invariant implies that $\theta$ is $2^{x_0}=1$. 
  This causes the condition in the \textbf{while} loop to fail, terminating the loop. 
  (In particular, we have $S_k = \emptyset$ for all $k > n$.) 
  Moreover, from Item~\ref{optilep:inv:i2}, we conclude that the only solution to $\phi \land \theta$ 
  is the map assigning~$0$ to every variable. 
  Accordingly, the algorithm checks whether this is indeed a solution (line~\ref
  {optilep:line:check-with-zeros}), and if so, replaces all free variables in $C$ with $0$.
  Since $C$ is a $n$-\preleac, this results in a sequence of the form 
  \begin{align*}
        \hspace{0.4cm} x_{n-i} & \gets \frac{a_{i,n} + \textstyle\sum_{j={i+1}}^{n-1} a_{i,j} \cdot
        2^{x_{n-j}}}{\mu_C} &\hspace{-1cm}\text{for $i$ from $n-1$ to $0$},
    \end{align*}
  which can easily be represented as an ILESLP.
  The proposition then follows from Items~\ref{optilep:inv:i5} and~\ref{optilep:inv:i6}. 
  Therefore, to complete the proof, it suffices to verify that the loop invariant holds.

  The invariant is trivially true for $S_0$. (In particular, $\vec r_{-1}$ and $\vec x_{-1}$ are empty in this case, and formulae like $\exists \vec x_{-1} (\phi_0 \land \theta_0)$ simplify to just $\phi_0 \land \theta_0$.)
  Hence, let us assume that the loop invariant is true when the execution reaches line~\ref{optilep:line:while} for the $(k+1)$th time, 
  with $k \in [0..n-1]$, and show that the invariant still holds when the algorithm comes back to this line 
  for the $(k+2)$th time.

  Consider $(C,\theta,\phi) \in S_k$, and 
  let $T_{k+1}$ be the set of those triples from $S_{k+1}$ that are constructed by (non-deterministially) 
  running the body of the \textbf{while} loop starting from $(C,\theta,\phi)$. More precisely, we will show that each triple in $T_{k+1}$ satisfies Items~\ref{optilep:inv:i1}--\ref{optilep:inv:i4}, and that moreover 
  \begin{enumerate}
      \addtocounter{enumi}{4}
      \item[$I_5'$.]\customlabel{$I_5'$}{optilep:inv:i5p} $\exists x_{n-k} \exists \vec r_{k-1} (\phi \land \theta)$ is equivalent to $\exists \vec r_{k} \bigvee_{(C',\theta',\phi') \in T_{k+1}} (\phi' \land \theta')$.
      \item[$I_6'$.]\customlabel{$I_6'$}{optilep:inv:i6p} If $\max\{ \objfun{C}{x_m}(\nu) :
        \text{$\nu$ is a solution to $\phi \land \theta$}\}$ exists, then it
        is equal to 
        \[
          \max\{\objfun{C'}{x_m}(\nu) : \text{$(C',\theta',\phi') \in T_{k+1}$ and $\nu$ is a solution to $\phi' \land \theta'$}\}.
        \]
    \end{enumerate}
  We divide the proof following the five steps identified in~\Cref{subsec:overview-optilep}:
  Step~I (lines~\ref{optilep:line:stepI:define-2x}--\ref{optilep:line:stepI:updateC}), 
  Step~II (lines~\ref{optilep:line:stepII:prepare-gamma-1}--\ref{optilep:line:stepII:run}), 
  Step~III (line~\ref{optilep:line:stepIII}), 
  Step~IV (lines~\ref{optilep:line:stepIV:lower}--\ref{optilep:line:stepIV:remove-qx}) and 
  Step~V (lines~\ref{optilep:line:stepV:prepare-r} --\ref{optilep:line:stepV:parepare-theta}).
  
  \paragraph{\textit{Step~I (lines~\ref{optilep:line:stepI:define-2x}--\ref{optilep:line:stepI:updateC}).}}
  By Item~\ref{optilep:inv:i1}, $\theta$ is $\theta_k$, and in it the leading and second-leading exponential terms are $2^{x_{n-k}}$ and $2^{x_{n-k-1}}$, respectively. 
  Following the pseudocode of~\OptILEP, throughout the proof we write $x$ for~$x_{n-k}$, and $y$ for $x_{n-k-1}$.
  According to Item~\ref{optilep:inv:i2}, within $\phi(\vec y_k,\vec r_{k-1})$ the variables from $\vec r_{k-1}$ do not occur in exponentials, and moreover $\phi$ implies $\vec r_{k-1} < 2^{x}$. In~line~\ref{optilep:line:stepI:run}, \OptILEP invokes the algorithm from~\Cref{lemma:CMS:first-step} on the pair~$(\phi,\theta)$. Let $E_1$ be the set of pairs $(\gamma,\psi)$ returned by this algorithm across its non-deterministic branches. By~\Cref{lemma:CMS:first-step}, each $\gamma$ is a linear program with division in variables $\vec q_k \coloneqq (q_{n-k},\dots,q_n)$ (called \emph{quotient variables}) and $u$, whereas 
  each $\psi$ is a linear-exponential program with divisions in variables $\vec y_{k+1}$ and $\vec r' \coloneqq (r_{n-k}',\dots,r_n')$, the latter being fresh \emph{remainder variables} not occurring in exponentials, and such that $\psi$ implies $\vec r' < 2^{y}$. Again as in the pseudocode, 
  we often write $q_x$ for $q_{n-k}$, 
  and $r_x$ for $r_{n-k}'$.
  The system
  \begin{equation}
    \label{eq:OptILEP-correctness:CMS:first-step}
    \left[\begin{matrix}
      x\\ 
      r_{n-(k-1)}\\
      \vdots\\ 
      r_n
    \end{matrix}\right]
    = \left[\begin{matrix}
      q_x\\ 
      q_{n-(k-1)}\\
      \vdots\\
      q_n
    \end{matrix}\right] \cdot 2^{y} + 
    \left[\begin{matrix}
      r_x\\ 
      r_{n-(k-1)}'\\
      \vdots\\
      r_n'
    \end{matrix}\right],
  \end{equation}
  yields a one-to-one correspondence between the solutions of
  $\phi \land \theta$ and those of~$\bigvee_{(\gamma,\psi) \in E_1} \Phi(\gamma,\psi)$, 
  where $\Phi(\gamma,\psi) \coloneqq \big(\gamma \land \psi \land (u = 2^{x-y}) \land (x = q_x \cdot 2^{y} + r_x) \land \theta\big)$.
  This correspondence is the identity for the variables these two formulae share (i.e., the variables in~$\vec y_{k}$). 

  In line~\ref{optilep:line:stepI:updateC}, $C$ is updated following the system described in~\Cref{eq:OptILEP-correctness:CMS:first-step}. Let $C'$ be the resulting sequence.
  By Item~\ref{optilep:inv:i3}, 
  $C$ is a $k$-\preleac, and therefore $C'$ takes the form: 
  \begin{align}
    \label{eq:optilep:correctness:Cprime}
          \hspace{0.4cm} x_{n-i} & \gets \frac{\tau_{n-i}}{\mu}
          + q_{n-i} \cdot 2^{y} + r_{n-i}' &\hspace{-1cm}\text{for $i$ from $k$ to $0$},
  \end{align}
  where each $\tau_{n-i}$ is a term of the form $\textstyle\sum_{j={i+1}}^{k} a_{i,j} \cdot 2^{x_{n-j}}$, 
  with each $a_{i,j}$ in $\Z$,
  and $\mu$ is a positive integer.
  In other words, $C'$ is a $(k,0)$-LEAC.
  The claim below follows directly from the one-to-one correspondence given by~\Cref{eq:OptILEP-correctness:CMS:first-step}.
  \begin{claim} 
    \label{claim:optilep:correcntess:stepI}
    The following properties hold: 
    \begin{enumerate}
      \item\label{claim:optilep:correcntess:stepI:i1} 
        The formulae $\exists \vec r_{k-1} (\phi \land \theta)$ and $ \exists \vec q_k \exists u \exists \vec r' \textstyle\bigvee_{(\gamma,\psi) \in E_1} \Phi(\gamma,\psi)$ are equivalent. 
      \item\label{claim:optilep:correcntess:stepI:i2} Consider $(\gamma,\psi) \in E_1$, and let $\nu \colon \vars(\Phi(\gamma,\psi))\to \N$ be a solution to $\Phi(\gamma,\psi)$. Then, the map ${\nu + \textstyle\sum_{i=0}^{k-1} [r_{n-i} \mapsto  \nu(q_{n-i}) \cdot 2^{\nu(y)} + \nu(r_{n-i}')]}$ is a solution to $\phi \land \theta$.
      \item\label{claim:optilep:correcntess:stepI:i3} If $\max\{ \objfun{C}{x_m}(\nu) :
          \text{$\nu$ is a solution to $\phi \land \theta$}\}$ exists, it
          is equal to 
          \[\max\{\objfun{C'}{x_m}(\nu) : \text{$\nu$ is a solution to $\Phi(\gamma,\psi)$, for some $(\gamma,\psi) \in E_1$}\}.\] 
    \end{enumerate}
  \end{claim}
  
  \paragraph{\textit{Step~II (lines~\ref{optilep:line:stepII:prepare-gamma-1}--\ref{optilep:line:stepII:run}).}}
  Fix $(\gamma,\psi) \in E_1$. 
  In a nutshell, Step~II removes the quotient variables $\vec q_{k-1} = (q_{n-(k-1)},\dots,q_n)$ from $\gamma$.
  Let $\widetilde{\gamma}$ be the formula obtained from $\gamma$ by performing the updates in lines~\ref{optilep:line:stepII:prepare-gamma-1} 
  and~\ref{optilep:line:stepII:prepare-gamma-2}; 
  $\widetilde{\gamma}$ and $\gamma$ are equivalent, as all variables range over $\N$ and $\mu_{C'} \geq 1$. 
  Let~$\psi'$ be the formula in line~\ref{optilep:line:stepII:prepare-psi}, 
  that is, $\psi' \coloneqq (\psi \land
    \theta \land ({x = q_x \cdot 2^{y} + r_x}) \land (u =
    2^{x-y}))$.  
  By definition, the three formulae $\inst{\widetilde{\gamma}}{\psi'}$, 
  $\inst{\gamma}{\psi'}$ and $\Phi(\gamma,\psi)$ are equivalent.
  We show that $(C',\inst{\widetilde{\gamma}}{\psi'})$ is in~$\objcons_{k}^0$; and hence that the call to~\GaussOpt 
  performed in line~\ref{optilep:line:stepII:run} adheres to the specification of this algorithm (from~\Cref{lemma:second-step-opt}). 
  Below, we refer to the Items~\ref{objcons:i1}--\ref{objcons:i3} characterizing $\objcons_{k}^0$ (page~\pageref{objcons:i1}): 
  \begin{itemize}
    \item\textit{\Cref{objcons:i1}:} We have already seen that $C'$ is a
    $(k,0)$-LEAC.
    \item\textit{\Cref{objcons:i2}:} By definition, $\widetilde{\gamma}$ is a linear
    exponential program with divisions, in variables $u$ and $\vec q_k$, and in which
    all inequalities and equalities are such that the coefficients of the
    variables $\vec q_k$ are divisible by $\mu_{C'}$. Moreover, $\widetilde{\gamma}$ contains
    an inequality $\mu_{C'} \cdot q \geq 0$ for every $q$ in $\vec q_k$. 
    Lastly, the formula~$\psi'$ trivially satisfies the conditions specified in~\Cref{objcons:i2}.

    \item\textit{\Cref{objcons:i3}:} 
    We need to show that $\inst{\widetilde{\gamma}}{\psi'}$ implies the formula~$\Psi(C')$ given by
    \begin{align*}
        \vec 0 \leq \vec r' < 2^{y} \land
        \vec 0 \leq \vec q_k \cdot 2^{y}+ \vec r' < 2^{x}  \land
        \exists \vec x_{k-1} \big( \theta_0 \land \textstyle\bigwedge_{i=0}^{k} (x_{n-i} = \rho_{n-i}) \big),
    \end{align*}
    where $\rho_{n-i}$ is the expression assigned to $x_{n-i}$ in $C'$
    (following~\Cref{eq:optilep:correctness:Cprime}).
    
    Consider a solution $\nu \colon \vars(\inst{\widetilde{\gamma}}{\psi'}) \to \N$ to $\inst{\widetilde{\gamma}}{\psi'}$.
    This is also a solution to~$\Phi(\gamma,\psi)$. By definition, $\psi$ implies $\vec 0 \leq \vec r' < 2^{y}$.
    From the one-to-one correspondence given by~\Cref{eq:OptILEP-correctness:CMS:first-step}, 
    the map $\nu' \coloneqq \nu + \sum_{i=0}^{k-1} [r_{n-i} \mapsto \nu(q_{n-i}) \cdot 2^{\nu(y)} + \nu(r_{n-i}')]$ is a solution to $\phi \land \theta$. 
    By Item~\ref{optilep:inv:i2}, 
    $\phi$ implies $\vec r_{k-1} < 2^{x}$. 
    Therefore, $\inst{\widetilde{\gamma}}{\psi'}$ 
    implies $\vec 0 \leq \vec q_k \cdot 2^{y} + \vec r' < 2^{x}$.

    Lastly, we see that $\inst{\widetilde{\gamma}}{\psi'}$  
    also implies $\exists \vec x_{k-1} \big( \theta_0 \land \textstyle\bigwedge_{i=0}^{k} (x_{n-i} = \rho_{n-i}) \big)$. 
    Indeed, by Item~\ref{optilep:inv:i4}, 
    the map ${\nu'' \coloneqq \nu' + \sum_{i=0}^{k-1} [x_{n-i} \mapsto \frac{\nu'(\tau_{n-i})}{\mu_C} + \nu'(r_{n-i})]}$ is a solution to $\phi_0 \land \theta_0$. 
    Together with the fact that~$\psi'$ implies ${x = q_x \cdot 2^{y} + r_x}$,
    this means that $\inst{\widetilde{\gamma}}{\psi'}$ 
    implies $\exists \vec x_{k-1} \big( \theta_0 \land \bigwedge_{i=0}^{k} (x_{n-i} = \frac{\tau_{n-i}}{\mu_C} + q_{n-i} \cdot 2^{y} + r_{n-i}')\big)$. 
    By definition, $\rho_{n-i} = (\frac{\tau_{n-i}}{\mu_C} + q_{n-i} \cdot 2^{y} + r_{n-i}')$.
  \end{itemize}
  Therefore, $(C',\inst{\widetilde{\gamma}}{\psi'}) \in \objcons_{k}^0$. In line~\ref{optilep:line:stepII:run}, \OptILEP calls \GaussOpt, eliminating the quotient variables~$\vec q_{k-1}$. 
  Let us denote by $E_2(\gamma,\psi)$ the set of all triples 
  $(C'',\gamma',\psi')$ 
  such that  
  ${(\objfun{C''}{x_m},\inst{\gamma'}{\psi'})}$ is a pair 
  returned by a non-deterministic execution of~\GaussOpt with as input the formulae 
  computed from~$\gamma$ and $\psi$ in \mbox{lines~\ref{optilep:line:stepII:prepare-gamma-1}--\ref{optilep:line:stepII:run}}. 
  Then, by direct application of~\Cref{lemma:second-step-opt}, we obtain:

  \begin{claim}
    \label{claim:optilep:correcntess:stepII}
    Let $(\gamma,\psi) \in E_1$. The following properties hold: 
    \begin{enumerate}
      \item\label{claim:optilep:correcntess:stepII:i1} The formulae $\exists \vec q_{k-1} \Phi(\gamma,\psi)$ and $\bigvee_{(C'',\gamma',\psi') \in E_2(\gamma,\psi)}\inst{\gamma'}{\psi'}$ are equivalent.
      \item\label{claim:optilep:correcntess:stepII:i2} Consider $(C'',\gamma',\psi') \in E_2(\gamma,\psi)$.
      Let $q_{n-(k-1)} \gets \frac{\tau_{n-(k-1)}''}{\eta},\dots,q_{n} \gets \frac{\tau_n''}{\eta}$ be the assignments 
      to the variables $\vec q_{k-1}$ occurring in $C''$.
      Let $\nu \colon \vars(\inst{\gamma'}{\psi'}) \to \N$ be a solution to $\inst{\gamma'}{\psi'}$. Then, the map 
      $\nu + \sum_{i=0}^{k-1}[q_{n-i} \mapsto \frac{\nu(\tau_{n-i}'')}{\eta}]$ is a solution to $\Phi(\gamma,\psi)$. 
      \item\label{claim:optilep:correcntess:stepII:i3} If $\max\{\objfun{C'}{x_m}(\nu) : \text{$\nu$ is a solution to $\Phi(\gamma,\psi)$}\}$ exists, then it is equal to \[\max\{\objfun{C''}{x_m}(\nu) : \text{$\nu$ is a solution to $\inst{\gamma'}{\psi'}$, for some $(C'',\gamma',\psi') \in E_2(\gamma,\psi)$}\}.\]
      \item\label{claim:optilep:correcntess:stepII:i4} For every $(C'',\gamma',\psi') \in E_2(\gamma,\psi)$, the pair $(C'',\inst{\gamma'}{\psi'})$ belongs to $\objcons_k^k$.
    \end{enumerate}
  \end{claim}

  \paragraph{\textit{Step~III (line~\ref{optilep:line:stepIII}).}} Let
  $(C'',\gamma',\psi') \in E_2(\gamma,\psi)$, for some $(\gamma,\psi) \in E_1$.
  Step~III eliminates $x$ and $u$ from $\gamma'$ by applying the algorithm from~\Cref{lemma:CMS:third-step}, which non-deterministically returns a pair~$(\gamma'',\psi'')$. We write~$E_3(\gamma')$ for the set of all such resulting pairs.
  By~\Cref{claim:optilep:correcntess:stepII}.\ref{claim:optilep:correcntess:stepII:i4}, $\gamma'$ is a linear program with divisions in variables $q_x$ and $u$. So, by~\Cref{lemma:CMS:third-step},
  $\gamma''$ is a linear program with divisions in the single variable~$q_x$, 
  and $\psi''$ is a linear-exponential program with divisions in the variables $y$ and $r_x$. 
  Moreover, the equation $x = q_x \cdot 2^{y} + r_x$ 
  yields a one-to-one correspondence between the solutions 
  of~$\gamma' \land (u = 2^{x-y}) \land (x = q_x \cdot 2^{y} + r_x)$
  and those of~$\bigvee_{(\gamma'',\psi'') \in E_3(\gamma')} (\gamma'' \land \psi'')$. 
  This correspondence is the identity for the variables these two
  formulae share (that is, $y$, $q_x$ and~$r_x$). 
  Roughly speaking, this one-to-one correspondence allows us to remove~$x$
  and~$u$ without changing the set of solutions.

  Recall that $\psi' \coloneqq (\psi \land \theta \land ({x = q_x
  \cdot 2^{y} + r_x}) \land (u = 2^{x-y}))$, and observe 
  that $\theta \land (u = 2^{x-y})$ is equal to $\theta_{k+1} \land (u
  = 2^{x-y})$, because $u = 2^{x-y}$ implies
  $2^{x} \geq 2^{y}$ (as $u$ ranges over $\N$).
    Then, the claim below follows immediately from the one-to-one correspondence given by the equation~$x = q_x \cdot 2^{y} + r_x$.
    
  \begin{claim} 
    \label{claim:optilep:correcntess:stepIII}
    Let~$(\gamma,\psi) \in E_1$
    and~$(C'',\gamma',\psi') \in E_2(\gamma,\psi)$.
    The following properties hold: 
    \begin{enumerate}
      \item\label{claim:optilep:correcntess:stepIII:i1} 
        The formulae $\exists x \exists u \inst{\gamma'}{\psi'}$ and $\textstyle\bigvee_{(\gamma'',\psi'') \in E_3(\gamma')} (\gamma'' \land \psi'' \land \psi \land \theta_{k+1})$ are equivalent. 
      \item\label{claim:optilep:correcntess:stepIII:i2} Let $(\gamma'',\psi'') \in E_3(\gamma')$, and $\nu \colon \vars(\gamma'' \land \psi'' \land \psi \land \theta_{k+1})\to \N$ be a solution to ${\gamma'' \land \psi'' \land \psi \land \theta_{k+1}}$. Define~$g \coloneqq \nu(q_x) \cdot 2^{\nu(y)} + \nu(r_x)$.
      The map ${\nu + [x \mapsto g] + [u \mapsto 2^{g - \nu(y)}]}$ is a solution to $\inst{\gamma'}{\psi'}$.
      \item\label{claim:optilep:correcntess:stepIII:i3} If $\max\{ \objfun{C''}{x_m}(\nu) :
          \text{$\nu$ is a solution to $\inst{\gamma'}{\psi'}$}\}$ exists, it
          is equal to 
          \begin{align*}
            \max\{\objfun{C''}{x_m}(\nu) :{}& \text{for some $(\gamma'',\psi'') \in E_3(\gamma')$, the map~$\nu$ is a solution to}\\ 
            &\hspace{0.5cm}\text{the formula~$\gamma'' \land \psi'' \land \psi \land \theta_{k+1} \land (u = 2^{x-y}) \land (x = q_x \cdot 2^{y} + r_x)$}\}.
          \end{align*}
    \end{enumerate}
  \end{claim}
  
  (In~\Cref{claim:optilep:correcntess:stepIII:i3}, the constraints $u = 2^{x-y}$ and $x = q_x \cdot 2^{y} + r_x$ are added to handle the fact that $u$ still appears as a free variable of $C''$. This discrepancy is resolved in line~\ref{optilep:line:stepIV:remove-qx}.)

  \paragraph{\textit{Step~IV (lines~\ref{optilep:line:stepIV:lower}--\ref{optilep:line:stepIV:remove-qx}).}} 
  Let~$(C'',\gamma',\psi') \in E_2(\gamma,\psi)$ and $(\gamma'',\psi'') \in E_3(\gamma')$, for some $(\gamma,\psi) \in E_1$.
  Step~IV removes the quotient variable $q_x$ from the linear program with divisions $\gamma''$, 
  and translates the $(k,k)$-LEAC into a $(k+1)$-\preleac.
  Let us treat each equality $\tau = 0$ in $\gamma''$ as a conjunction of two inequalities: $\tau \leq 0 \land -\tau \leq 0$.
  Since $\gamma''$ contains only the variable~$q_x$, every inequality in it is either of the form $a \leq b \cdot q_x$ or $b \cdot q_x \leq a$, with $b$ non-negative. 
  Let us update $\gamma''$ by rewriting these as $\ceil{\frac{a}{b}} \leq q_x$ and $q_x \leq \floor{\frac{a}{b}}$, respectively. Line~\ref{optilep:line:stepIV:lower} computes the greatest non-negative integer $\ell$ such that $\ell \leq q_x$ occurs in $\gamma''$, while Line~\ref{optilep:line:stepIV:upper} 
  computes the smallest integer $h$ such that $q_x \leq h$ occurs in $\gamma''$. 
  By default, $\ell$ and $h$ are initialized as $0$ and $\infty$, respectively, so in particular we always have $\ell \geq 0$. 
  If $h = \infty$, the algorithm updates it to $\ell + \fmod(\gamma'')$ in line~\ref{optilep:line:stepIV:upper-infty}, ensuring that at least one solution to $\gamma''$ is explored (if one exists). 
  Let us write $B(\gamma'')$ for the set $\{ v \in [\ell..h] : \text{$\gamma''(v)$ holds}\}$.
  Step IV concludes by guessing~$v \in B(\gamma'')$ (line~\ref{optilep:line:stepIV:guess}), to then translating $C''$ into a $(k+1)$-\preleac following~\Cref{remark:LEAC-to-PRELEAC}. If $B(\gamma'')$ is empty then~$\gamma''$ is unsatisfiable; in this case the \textbf{guess} instruction fails, and the non-deterministic branch of the~algorithm~rejects.%

  Let us write $E_4(C'',\gamma'')$ for the set of all circuits obtained from $C''$ when running~line~\ref{optilep:line:stepIV:remove-qx}, with respect to some~$v \in B(\gamma'')$. 
  We show the following claim (whose proof clarifies why it is sufficient to restrict~$q_x$ to values in~$[\ell..h]$ in order to explore optimal solutions).

  \begin{claim}
    \label{claim:optilep:correcntess:stepIV}
    Let~$(\gamma,\psi) \in E_1$, $(C'',\gamma',\psi') \in E_2(\gamma,\psi)$ 
    and $(\gamma'',\psi'') \in E_3(\gamma')$.
    We have: 
    \begin{enumerate}
      \item\label{claim:optilep:correcntess:stepIV:i1} If $B(\gamma'')$ is non-empty, then
        $\exists q_x (\gamma'' \land \psi'' \land \psi \land \theta_{k+1})$ 
        is equivalent to $\psi'' \land \psi \land \theta_{k+1}$.
      \item\label{claim:optilep:correcntess:stepIV:i2} Let $v \in B(\gamma'')$, and let~$\nu \colon \vars(\psi'' \land \psi \land \theta_{k+1})\to \N$ be a solution to ${\psi'' \land \psi \land \theta_{k+1}}$. Then, the map ${\nu + [q_x \mapsto  v]}$ is a solution to $\gamma'' \land \psi'' \land \psi \land \theta_{k+1}$.
      \item\label{claim:optilep:correcntess:stepIV:i3} If
        $M \coloneqq \max\{\objfun{C''}{x_m}(\nu) :{} \text{for some $(\gamma'',\psi'') \in E_3(\gamma')$, the map~$\nu$ is a solution to the formula }$ 
        $\gamma'' \land \psi'' \land \psi \land \theta_{k+1} \land (u = 2^{x-y}) \land (x = q_x \cdot 2^{y} + r_x)\}$ exists, it
        is equal to 
        \[\max\{\objfun{C^*}{x_m}(\nu) : \text{$\nu$ is a solution to $\psi'' \land \psi \land \theta_{k+1}$, for some $C^* \in E_4(C'',\gamma'')$}\}.\] 
    \end{enumerate}
  \end{claim}

  \begin{proof}
    In the formula~$\gamma'' \land \psi'' \land \psi \land \theta_{k+1}$, 
    the variable~$q_x$ only occurs in $\gamma''$. 
    Then, the left-to-right direction of~\Cref{claim:optilep:correcntess:stepIV:i1} follows from~\Cref{corr:basic-fact-from-presburger}, 
    whereas the right-to-left direction holds from~\Cref{claim:optilep:correcntess:stepIV:i2}, 
    which in turn follows directly from the definition of $B(\gamma'')$.

    We prove~\Cref{claim:optilep:correcntess:stepIV:i3}. 
    Let $\Psi \coloneqq (\gamma'' \land \psi'' \land \psi \land \theta_{k+1} \land (u = 2^{x-y}) \land (x = q_x \cdot 2^{y} + r_x))$.
    We divide the proof depending on the variable $x_m$ we are maximizing: 
    \begin{description}
      \item[case: $m < n-k$.] 
        Suppose $M$ exists.
        The circuit $C''$ does not feature an assignment to the variable $x_m$, and the same is true for every~$C^* \in E_4(C'',\gamma'')$.
        For both $C''$ and $C^*$, given a map $\nu$ from their free variables plus $x_m$ to $\N$, we have $\objfun{C''}{x_m}(\nu) = \nu(x_m) = \objfun{C^*}{x_m}(\nu)$. Let~$\nu$ be a solution to $\Psi$.
        Since $\nu$ is also a solution to ${\psi'' \land \psi \land \theta_{k+1}}$, 
        we have $\max\{\objfun{C^*}{x_m}(\nu) : \nu$ is a solution to ${\psi'' \land \psi \land \theta_{k+1}}$, for some $C^* \in E_4(C'',\gamma'')\} \geq M$.

        \emph{Ad absurdum}, assume $\objfun{C^*}{x_m}(\nu^*) > M$, for some solution ${\nu^* \colon \vars(\psi'' \land \psi \land \theta_{k+1}) \to \N}$ to $\psi'' \land \psi \land \theta_{k+1}$, and $C^* \in E_4(C'',\gamma'')$. 
        By definition, $C^*$ is constructed from $C''$ by replacing $q_x$ with some ${v \in B(\gamma'')}$, and $u$ with $2^{x-y}$. 
        By~\Cref{claim:optilep:correcntess:stepIV:i2}, $\nu' \coloneqq \nu^* + [q_x \mapsto v]$ 
        is a solution to $\gamma'' \land \psi'' \land \psi \land \theta_{k+1}$. Define~$g \coloneqq \nu'(q_x) \cdot 2^{\nu'(y)} + \nu'(r_x)$.
        By~\Cref{claim:optilep:correcntess:stepIII}.\ref{claim:optilep:correcntess:stepIII:i2}, 
        the map ${\nu'' \coloneqq \nu' + [x \mapsto g] + [u \mapsto 2^{g - \nu'(y)}]}$ is a solution to $\inst{\gamma'}{\psi'}$. 
        Since $\psi'$ implies $u = 2^{x-y} \land x = q_x \cdot 2^{y} + r_x$, 
        we conclude that $\nu''$ is a solution to~$\Psi$.
        We have $\objfun{C''}{x_m}(\nu'') = \nu''(x_m) = \nu^*(x_m) = \objfun{C^*}{x_m}(\nu^*) > M$, contradicting 
        the fact that~$M$ is maximal. 
        Therefore, \Cref{claim:optilep:correcntess:stepIV:i3} holds.

      \item[case: $m \geq n-k$.] First, let us consider the case where $h \neq \infty$ when defined in line~\ref{optilep:line:stepIV:upper}. Then, all the solutions to $\gamma''$ lie in the interval $[\ell..h]$. The set $E_4(C'',\gamma'')$ is 
      constructed to consider all these solutions, and therefore \Cref{claim:optilep:correcntess:stepIV:i3} follows. 

      Suppose instead that~$h = \infty$ in line~\ref{optilep:line:stepIV:upper}. 
      In this case, we establish~\Cref{claim:optilep:correcntess:stepIV:i3} by showing that $M$ does not exist.
      Consider an arbitrary solution~$\nu$ to $\Psi$.
      Let $p \coloneqq \fmod(\gamma'')$. Since $h = \infty$, we note that increasing $\nu(q_x)$ by any positive multiple of $p$, and increasing $\nu(x)$ and $\nu(u)$ accordingly, still yields a solution to $\Psi$. 
      More precisely, if we want to increase the $\nu(q_x)$ 
      by $i \cdot p$, with $i \in \N$, 
      the resulting map is
      $\nu + [q_x \mapsto i \cdot p] 
        + [x \mapsto i \cdot p \cdot 2^{\nu(y)}] 
        + [u \mapsto (2^{i \cdot p \cdot 2^{\nu(y)}}-1) \cdot 2^{\nu(x)-\nu(y)}]$.
      By definition, $C''$ contains the assignment $x \gets q_x \cdot 2^{y} + r_x$. Therefore, arbitrarily increasing~$\nu(q_x)$ results in arbitrarily large values that~$x$ evaluates to in~$\objfun{C''}{x_m}$. 
      
      Let $\nu$ be a solution to $\Psi$.
      To complete the proof, it suffices to show that the value taken by~$x$ when evaluating $\objfun{C''}{x_m}$ on $\nu$ is at most $\objfun{C''}{x_m}(\nu)$, 
      in other words, that $\objfun{C''}{x_{n-k}}(\nu) \leq \objfun{C''}{x_m}(\nu)$.
      From~\Cref{claim:optilep:correcntess:stepIII}.\ref{claim:optilep:correcntess:stepIII:i2}, $\nu$ is also a solution to $\inst{\gamma'}{\psi'}$. By~\Cref{claim:optilep:correcntess:stepII}.\ref{claim:optilep:correcntess:stepII:i4},~${(C'',\inst{\gamma'}{\psi'})}$ belongs to~$\objcons_k^k$. By definition of~$\objcons_k^k$, 
      $\inst{\gamma'}{\psi'}$ implies ${\exists \vec q_{k-1} \exists \vec x_{k-1} \big( \theta_0 \land \textstyle\bigwedge_{i=1}^{t} (y_{i} = \rho_{i}) \big)}$, 
      where $(y_1 \gets \rho_1,\dots,y_t \gets \rho_t) = C''$.
      The variables in~$\vec q_{k-1}$ and $\vec x_{k-1}$ are all among $y_1,\dots,y_t$, meaning that there is exactly one evaluation for these variables 
      for which~$\theta_0 \land \textstyle\bigwedge_{i=1}^{t} (y_{i} = \rho_{i})$ is satisfied: the values that these variables take when $\objfun{C''}{x_m}$ is evaluated on~$\nu$. 
      Since $m \geq n-k$, the ordering $\theta_0$ implies $x_{n-k} \leq x_{m}$. 
      Hence, when evaluating $\objfun{C''}{x_m}$ on $\nu$, the value taken by~$x$ is at most $\objfun{C''}{x_m}(\nu)$, as required.
      \qedhere
    \end{description}
  \end{proof}

  \paragraph{\textit{Step~V (lines~\ref{optilep:line:stepV:prepare-r}--\ref{optilep:line:stepV:parepare-theta}).}} These lines simply prepare the linear-exponential system for the next loop iteration. 
  Let~$(\gamma,\psi) \in E_1$, $(C'',\gamma',\psi') \in E_2(\gamma,\psi)$, $(\gamma'',\psi'') \in E_3(\gamma')$ and~${C^* \in E_4(C'',\gamma'')}$.
  Line~\ref{optilep:line:stepV:prepare-r} sets~$\vec r'$ as the remainder variables for the next iterations of the loop (in this proof, $\vec r_k$).
  Line~\ref{optilep:line:stepV:prepare-phi} sets $\phi^* \coloneqq \psi \land \psi''$ as the formula for the next iteration, whereas line~\ref{optilep:line:stepV:parepare-theta} updates $\theta$ to $\theta_{k+1}$. This concludes the body of the \textbf{while} loop, and $T_{k+1}$ is the set of all possible triples~$(C^*,\theta_{k+1},\phi^*)$.

  \paragraph*{}
  Let us now complete the proof by showing that all items in the loop invariant are satisfied.
  \begin{itemize}
    \item \textit{Item~\ref{optilep:inv:i1}:} $\theta_{k+1}$ is indeed the ordering required by this item.
    \item \textit{Item~\ref{optilep:inv:i2}:} In $\phi^* = \psi \land \psi''$, both $\psi$ and $\psi''$ are linear-exponential programs with divisions in variables $\vec y_{k+1}$ and $\vec r'$. Moreover, $\psi$ (which was defined in Step~I), implies $\vec r' < 2^{y}$, as required.
    \item \textit{Item~\ref{optilep:inv:i3}:} This follows directly from the manipulation performed in line~\ref{optilep:line:stepIV:remove-qx} to construct $C^*$, recalling that $C''$ is a $(k,k)$-LEAC.
    Below, let $C^* = (x_{n-k} \gets \frac{\tau_{n-k}^*}{\mu^*} + r_{n-k}' \,,\, \dots \,,\, x_n \gets \frac{\tau_{n}^*}{\mu^*} + r_{n}')$.
    \item \textit{Item~\ref{optilep:inv:i4}:} We must show that, given a solution $\nu \colon \vars(\phi^* \land \theta_{k+1}) \to \N$ to $\phi^* \land \theta_{k+1}$, the map $\nu + \textstyle\sum_{i=0}^{k} [x_{n-i} \mapsto g_{n-i}]$,  
    where $g_{n-i}$ is the value taken by~$x_{n-i}$ when evaluating $\objfun{C^*}{x_m}$ on $\nu$, is a solution to $\phi_0 \land \theta_0$. 
    In a nutshell, this is shown by appealing to 
    the second Items in Claims~\ref{claim:optilep:correcntess:stepI}--\ref{claim:optilep:correcntess:stepIV}, to 
    then apply the induction hypothesis.
    Indeed, observe that starting from $\nu$, these Items
    construct a solution $\nu^*$ for $\phi \land \theta$.
    Recall that the initial circuit $C$ 
    is a $k$-\preleac
    \begin{align*}
          \hspace{0.4cm} x_{n-i} & \gets \frac{\tau_{n-i}}{\mu}
          + r_{n-i} &\hspace{-1cm}\text{for $i$ from $k-1$ to $0$}.
    \end{align*}
    This circuit is then manipulated into the circuit~$C'$
    \begin{align*}
          x &\gets q_x \cdot 2^y + r_x,\\
          \hspace{0.4cm} x_{n-i} & \gets \frac{\tau_{n-i}}{\mu}
          + q_{n-i} \cdot 2^{y} + r_{n-i}' &\hspace{-1cm}\text{for $i$ from $k-1$ to $0$},
    \end{align*}
    and $C''$ is constructed from $C'$ 
    by adding assignments $q_{n-i} \gets \frac{\tau_{n-i}''(u,q_x)}{\eta}$ for every $i \in [0..k-1]$.
    Lastly, $C^*$ is essentially obtained from $C''$ 
    by replacing $q_{x}$ the integer $v$ from line~\ref{optilep:line:stepIV:guess} , and $u$ with $2^{x-y}$.
    By induction hypothesis, the map obtained from $\nu^*$ 
    by adding $\sum_{i=0}^{k-1}[x_{n-i} \mapsto t_{n-i}]$, 
    where $t_{n-i}$ is the value taken by $x_{n-i}$ when evaluating $\objfun{C}{x_m}$ on $\nu^*$ 
    is a solution to $\phi_0 \land \theta_0$.
    Then, because of the updates required to obtain $C^*$ from $C$, it suffices to show that 
    \begin{align*}
      \nu^*(x) &= v \cdot 2^{\nu(y)} + \nu(r_x),\\
      \nu^*(r_{n-i}) &= \frac{\tau_{n-i}''(2^{\nu^*(x_{n}) - \nu(y)},v)}{\eta} \cdot 2^{\nu(y)} + \nu(r_{n-i}') & \text{for $i$ from $k-1$ to $0$.}
    \end{align*}
    This follows directly from the identities below:
    \begin{align*}
      \nu' &= \nu + [q_x \mapsto  v]
      &\Lbag\text{\Cref{claim:optilep:correcntess:stepIV}.\ref{claim:optilep:correcntess:stepIV:i2}}\Rbag\\
      \nu'' &= \nu' + [x \mapsto g] + [u \mapsto 2^{g - \nu(y)}],
      \text{\quad where } g \coloneqq v \cdot 2^{\nu(y)} + \nu(r_x)
      &\Lbag\text{\Cref{claim:optilep:correcntess:stepIII}.\ref{claim:optilep:correcntess:stepIII:i2}}\Rbag\\
      \nu''' 
      &= 
      \nu'' + \textstyle\sum_{i=0}^{k-1}[q_{n-i} \mapsto \frac{\nu''(\tau_{n-i}'')}{\eta}]
      &\Lbag\text{\Cref{claim:optilep:correcntess:stepII}.\ref{claim:optilep:correcntess:stepII:i2}}\Rbag\\
      \nu^* &= \nu''' + \textstyle\sum_{i=0}^{k-1} [r_{n-i} \mapsto  \nu'''(q_{n-i}) \cdot 2^{\nu(y)} + \nu(r_{n-i}')]
      &\Lbag\text{\Cref{claim:optilep:correcntess:stepI}.\ref{claim:optilep:correcntess:stepI:i2}}\Rbag
    \end{align*}
  \end{itemize}
  Let us now show Item~\ref{optilep:inv:i5p}; then Item~\ref{optilep:inv:i5} 
  follows directly from the induction hypothesis:
  {\allowdisplaybreaks%
  \begin{align*}
          & \exists x \exists \vec r_{k-1} (\phi \land \theta)\\ 
    \iff{}& \exists x \exists \vec q_k \exists u \exists \vec r' \textstyle\bigvee_{(\gamma,\psi) \in E_1} \Phi(\gamma,\psi)
          &\Lbag\text{\Cref{claim:optilep:correcntess:stepI}.\ref{claim:optilep:correcntess:stepI:i1}}\Rbag\\
    \iff{}& \exists x \exists q_x \exists u \exists \vec r' \textstyle\bigvee_{\substack{(\gamma,\psi) \in E_1,\, (C'',\gamma',\psi') \in E_2(\gamma,\psi)}} \inst{\gamma'}{\psi'}
          &\Lbag\text{\Cref{claim:optilep:correcntess:stepII}.\ref{claim:optilep:correcntess:stepII:i1}}\Rbag\\
    \iff{}& \exists q_x \exists \vec r' \textstyle
    \bigvee_{\substack{(\gamma,\psi) \in E_1,\, (C'',\gamma',\psi') \in E_2(\gamma,\psi),\, (\gamma'',\psi'') \in E_3(\gamma')}}\ (\gamma'' \land \psi'' \land \psi \land \theta_{k+1}) 
          &\Lbag\text{\Cref{claim:optilep:correcntess:stepIII}.\ref{claim:optilep:correcntess:stepIII:i1}}\Rbag\\
    \iff{}& \exists \vec r' \textstyle
    \bigvee_{\substack{(\gamma,\psi) \in E_1,\, (C'',\gamma',\psi') \in E_2(\gamma,\psi),\, (\gamma'',\psi'') \in E_3(\gamma')\, \text{s.t.}\, B(\gamma'') \neq \emptyset}}\ (\psi'' \land \psi \land \theta_{k+1})
          &\Lbag\text{\Cref{claim:optilep:correcntess:stepIV}.\ref{claim:optilep:correcntess:stepIV:i1}}\Rbag\\
    \iff{}& \exists \vec r' \textstyle
    \bigvee_{(C^*,\theta_{k+1},\phi^*) \in T_{k+1}}\ (\phi^* \land \theta_{k+1})
          &\Lbag\text{def.~of $T_{k+1}$}\Rbag
  \end{align*}}%
  Lastly, we show Item~\ref{optilep:inv:i6p}, which implies Item~\ref{optilep:inv:i6}
  by induction hypothesis. Suppose that $M \coloneqq \max\{ \objfun{C}{x_m}(\nu) :
        \text{$\nu$ is a solution to $\phi \land \theta$}\}$ exists. 
  Then,
  {\allowdisplaybreaks%
  \begin{align*}
    M &= \max\{\objfun{C'}{x_m}(\nu) : \text{$\nu$ is a solution to $\Phi(\gamma,\psi)$, for some $(\gamma,\psi) \in E_1$}\} 
      &\hspace{-3cm}\Lbag\text{\Cref{claim:optilep:correcntess:stepI}.\ref{claim:optilep:correcntess:stepI:i3}}\Rbag\\
      &= \max\{\objfun{C''}{x_m}(\nu) : \text{$\nu$ is a solution to $\inst{\gamma'}{\psi'}$, for some $(\gamma,\psi) \in E_1$,}\\ 
      &\hphantom{= \max\{\objfun{C''}{x_m}(\nu) : \text{$\nu$ is a solution to $\inst{\gamma'}{\psi'}$}}\hspace{0.8cm} \text{and $(C'',\gamma',\psi') \in E_2(\gamma,\psi)$}\}
      &\hspace{-3cm}\Lbag\text{\Cref{claim:optilep:correcntess:stepII}.\ref{claim:optilep:correcntess:stepII:i3}}\Rbag\\
      &= \max\{\objfun{C''}{x_m}(\nu) : \text{$\nu$ is a solution to $\gamma'' \land \psi'' \land \psi \land \theta_{k+1} \land (u = 2^{x-y}) \land (x = q_x \cdot 2^{y} + r_x)$}\\ 
      &\hspace{2.1cm}\text{for some~$(\gamma,\psi) \in E_1$, $(C'',\gamma',\psi') \in E_2(\gamma,\psi)$, and $(\gamma'',\psi'') \in E_3(\gamma')$}\}
      &\hspace{-3cm}\Lbag\text{\Cref{claim:optilep:correcntess:stepIII}.\ref{claim:optilep:correcntess:stepIII:i3}}\Rbag\\
      & = \max\{\objfun{C^*}{x_m}(\nu) : \text{$\nu$ is a solution to $\psi'' \land \psi \land \theta_{k+1}$, for some $(\gamma,\psi) \in E_1$,}\\
      &\hspace{2.1cm}\text{ $(C'',\gamma',\psi') \in E_2(\gamma,\psi)$, $(\gamma'',\psi'') \in E_3(\gamma')$, and $C^* \in E_4(C'',\gamma'')$}\}
      &\hspace{-3cm}\Lbag\text{\Cref{claim:optilep:correcntess:stepIV}.\ref{claim:optilep:correcntess:stepIV:i3}}\Rbag\\
      & = \max\{\objfun{C^*}{x_m}(\nu) : \text{$\nu$ is a solution to $\phi^* \land \theta_{k+1}$, for some $(C^*\!,\theta_{k+1},\phi^*) \in T_{k+1}$\}}
      &\hspace{-1.65cm}\Lbag\text{def.~of~$T_{k+1}$}\Rbag
  \end{align*}}%
  Therefore, the loop invariant holds, completing the proof of the proposition.
\end{proof} 

\subsection{Complexity of~\OptILEP}
\label{subsec:complexity-optilep}

We now provide the complexity analysis of~\OptILEP, 
which requires tracking several parameters of linear-exponential systems. 
For a linear-exponential program with divisibilities $\phi$, we track: 
\begin{itemize}
  \item The parameters $\card{\phi}$, $\onenorm{\phi}$ and $\fmod(\phi)$, defined in the preliminaries (page~\pageref{ref:parameters-for-complexity}).
  \item\label{item:linearnorm} The \emph{linear norm} $\linnorm{\phi} \coloneqq \max\{\linnorm{\tau} :  \text{$\tau$ is a term appearing in an equality or inequality of $\phi$} \}$. 
  
  Given a linear-exponential term~$\tau = {\sum\nolimits_{i=1}^n \big(a_i \cdot x_i + b_i \cdot 2^{x_i} + \sum\nolimits_{j=1}^n c_{i,j} \cdot (x_i \bmod 2^{x_j})\big) + d}$, 
  we define its linear norm as~$\linnorm{\tau} \coloneqq \max\{ \abs{a_i}, \abs{c_{i,j}} : i,j \in [1..n] \}$. 

  \item\label{item:lst} Consider an ordering of exponentiated variables $\theta \coloneqq (\theta(\vec x) \coloneqq {2^{x_n} \geq
  2^{x_{n-1}} \geq \dots \geq 2^{x_0} = 1})$. 
  Let $\vec r$ be the variables from $\phi$ that are not in $\vec x$.
  We track the set of the \emph{least significant terms} 
  \begin{align*}
    \lst(\phi,\theta) \coloneqq \big\{ \pm \rho : {}&
      \text{$\rho$ is the least significant part of a term $\tau$ appearing in}\\ 
      &\text{an equality or inequality $\tau \sim 0$ of $\phi$, 
  with respect to $\theta$\,}\big\}. 
  \end{align*}
  The \emph{least significant part} of a term $a \cdot 2^{x_n} + b \cdot x_n + \tau'(x_0,\dots,x_{n-1}, \vec r)$ with respect to $\theta$
  is defined as the term $b \cdot x_n + \tau'$.
\end{itemize}
For a $k$-\preleac $C$, we track 
the growth of the parameters $\mu_C$ and $\xi_C$. 

\begin{remark}
  From the above parameters one can bound the sizes of $\phi$ and $C$: the size of $\phi(\vec x)$ is  in~$\poly(\card{\phi},\card \vec x,\log\onenorm{\phi}, \log\fmod(\phi))$, 
  and the size of a $k$-\preleac $C$ is in $\poly(k,\log\mu_C, \log \xi_C)$.
\end{remark}

\paragraph*{The complexity of Steps~I and~III from~\cite{ChistikovMS24}.}
For the complexity analysis of lines~\ref{optilep:line:stepI:run} and~\ref{optilep:line:stepIII} of~\OptILEP, which correspond to Steps~I and~III of~\cite{ChistikovMS24}, we refer directly to the analysis carried out in~\cite{ChistikovMS24}. (We remind the reader that~\Cref{section:analysis-step-i-and-iii} gives more information on these two steps).
This analysis is reported in the following two lemmas.

\begin{restatable}[\cite{ChistikovMS24}]{lemma}{LemmaComplexityOfStepI}
    \label{lemma:complexity-of-step-i}
    The algorithm from~\Cref{lemma:CMS:first-step} (Step~I) runs in non-deterministic polynomial time.
    Consider its execution on an input $(\theta, \phi)$ where $\theta(\vec x)$ is an ordering of exponentiated variables and $\phi(\vec x, \vec r)$  is a linear exponential program with divisions. In each non-deterministic branch $\beta$, the algorithm returns a pair $(\gamma,\psi)$, where $\gamma(q_x, \vec q, u)$ is a linear program with divisions and $\psi(\vec y, r_x, \vec r')$ a linear-exponential program with divisions, 
    such that (for every $\ell,s,a,c,d \geq 1$): 
    \begin{equation*}
      \text{if\ } 
      \begin{cases}
        \card \lst(\phi,\theta) \cand \leq \ell\\
        \card \phi              \cand \leq s\\
        \linnorm{\phi}          \cand \leq a\\ 
        \onenorm{\phi}          \cand \leq c\\
        \fmod(\phi)             \cand \hspace{3pt}\divides\hspace{2pt} d
      \end{cases}
      \text{\ then\ }  
      \begin{cases}
        \card \lst(\psi,\theta')  \cand \leq \ell + 2 \cdot k\\
        \card \psi                \cand \leq s + 6 \cdot k + 2 \cdot \ell\\
        \linnorm{\psi}            \cand \leq 3 \cdot a\\ 
        \onenorm{\psi}            \cand \leq 4 \cdot c + 5\\
        \fmod(\psi)               \cand \hspace{3pt}\divides\hspace{2pt} d\\
      \end{cases}\text{and\ }
      \begin{cases}
        \card \gamma                  \cand \leq s + 2 \cdot k\\
        \linnorm{\gamma\sub{2^u}{u}}  \cand \leq 3 \cdot a\\ 
        \onenorm{\gamma}              \cand \leq 2 \cdot c + 3\\
        \fmod(\gamma)                 \cand \hspace{3pt}\divides\hspace{2pt} d
      \end{cases}
    \end{equation*}
    where $\theta'$ is the ordering obtained from $\theta$ by removing its largest term $2^x$, and $k \coloneqq 1 + \card{\vec r}$.
\end{restatable}

Below, $\totient$ denotes Euler's totient function. 
Recall that given a positive integer $a$, 
\begin{align}\label{equation:compute-totient-via-factorization}
    \totient(a) \coloneqq \textstyle\prod_{i=1}^k \big((p_i-1) \cdot p_i^{e_i-1}\big), 
    \text{ where $p_1^{e_1} \cdots p_k^{e_k}$ is the prime factorization of $a$.} 
\end{align}

\begin{restatable}[\cite{ChistikovMS24}]{lemma}{LemmaComplexityOfStepIII}
    \label{lemma:complexity-of-step-iii}
    The algorithm from~\Cref{lemma:CMS:third-step} (Step~III) runs in non-deterministic polynomial time. 
      Consider its execution on an input linear program with divisions $\gamma'$.
      In each non-deterministic branch $\beta$, the algorithm returns a pair $(\gamma'',\psi'')$, 
      where $\gamma''$ is a linear program with divisions and $\psi''$ is a linear-exponential program 
      with divisions, such that  (for every $s,a,c,d \geq 1$):  
      \begin{equation*}
        \renewcommand{\cand}{&\hspace{-9pt}}
        \text{if}\, 
        \begin{cases}
          \card \gamma'                 \cand \leq s\\
          \linnorm{\gamma'\sub{2^u}{u}} \cand \leq a\\ 
          \onenorm{\gamma'}             \cand \leq c\\
          \fmod(\gamma')                \cand  \hspace{3pt}\divides\hspace{2pt} d
        \end{cases}
        \text{ \ then}\,  
        \begin{cases}
          \card \gamma''      \cand  \leq s + 2\\
          \linnorm{\gamma''}  \cand  \leq a\\ 
          \onenorm{\gamma''}  \cand  \leq \max(2^{5} c^3,c \cdot d)\\
          \fmod(\gamma'')     \cand \hspace{3pt}\divides\hspace{2pt} \lcm(d,\totient(d))\\
        \end{cases}\text{ \ and}\,
        \begin{cases}
          \card \psi''      \cand \leq 3\\
          \linnorm{\psi''}  \cand \leq 1\\ 
          \onenorm{\psi''}  \cand \leq 12 + 4 \,{\cdot} \log(\max(c,d))\\
          \fmod(\psi'')     \cand \hspace{3pt}\divides\hspace{2pt} \totient(d)
        \end{cases}
      \end{equation*}
\end{restatable}

\paragraph*{The complexity of performing one iteration of the main loop {\rm{is given in the next lemma:}}}

\newcommand{\gfromfmodgammabound}[1][k]{(3 \cdot #1 \cdot \mu \cdot a)^{k^2}}
\newcommand{\gfromMuBound}[1][k]{(3\cdot #1 \cdot a)^{k}}

\newcommand{\onenormgammaNewBound}{3\cdot  ((k+1)\cdot M)^{3(k+2)^2} \cdot d}


\begin{restatable}{lemma}{LemmaPuttingAllTogetherCircuit}
\label{lemma:putting-all-together-one-iteration}
Consider the execution of~\Cref{pseudocode:opt-ilep} on an linear-exponential program~$\phi(x_1,\dots,x_n)$, with $n \geq 1$. Let $(\phi,\theta,C)$ be the system, circuit, and ordering obtained after the $k$th iteration 
of the \textbf{while} loop of line~\ref{optilep:line:while}. 
The $(k+1)$th iteration of the \textbf{while} loop runs in non-deterministic polynomial time in the bit sizes of $\phi$ and $C$. 
Each non-deterministic execution of the loop updates the triple $(\phi,\theta,C)$
into a triple~$(\phi',\theta',C')$ such that (for every $\ell,s,a,c,d \geq 1$):%
    \begin{align*}
      \text{if \ }
      &\begin{cases}
        \card \lst(\phi,\theta) \cand \leq \ell\\
        \card \phi              \cand \leq s\\
        \linnorm{\phi}          \cand \leq a\\ 
        \onenorm{\phi}          \cand \leq c\\
        \fmod(\phi)             \cand \hspace{3pt}\divides\hspace{2pt} d
      \end{cases}
      &\ \text{then}& 
      &\begin{cases}
        \card \lst(\phi',\theta') \cand \leq \ell + 2 \cdot k + 3\\
        \card \phi'              \cand \leq s + 6 \cdot k + 2\cdot \ell + 3  \\
        \linnorm{\phi'}          \cand \leq 3 \cdot a \\ 
        \onenorm{\phi'}          \cand \leq 12 + 4\cdot \max(c,\log \beta)\\
        \fmod(\phi')             \cand \hspace{3pt}\divides\hspace{2pt} \lcm{(d, \totient(\alpha \cdot d))}\\
          \xi_{C'}               \cand \le \xi_C \cdot \gfromMuBound  + 2^6 (k+1) \cdot \beta^4 \\
          \mu_{C'}               \cand \le  \mu_C \cdot \gfromMuBound,
      \end{cases}
    \end{align*}
    with~$\alpha \in [1..(3 \cdot k \cdot \mu_C \cdot a)^{k^2}]$, and $\beta \coloneqq d \cdot \big( 2^7 (k+1) \cdot \mu_C \cdot \max(c, \log(\xi_C + \mu_C))\big)^{3(k+2)^2}$.
\end{restatable}

\begin{proof}
Throughout the proof, $\mu$ is short for $\mu_C$.
We analyze how the parameters evolve over the five steps of the while loop of~\OptILEP.

\paragraph{\textit{Step~I (lines~\ref{optilep:line:stepI:define-2x}--\ref{optilep:line:stepI:updateC}).}} 
Let $\gamma$ and $\psi$ be the systems computed by the algorithm in line~\ref{optilep:line:stepI:run}.
Directly from~\Cref{lemma:complexity-of-step-i}, we derive the following bounds on their parameters:
\begin{equation*}
        \begin{cases}
        \card \lst(\psi,\theta')  \cand \leq \ell + 2 \cdot k\\
        \card \psi                \cand \leq s + 6 \cdot k + 2 \cdot \ell\\
        \linnorm{\psi}            \cand \leq 3 \cdot a\\ 
        \onenorm{\psi}            \cand \leq 4 \cdot c + 5\\
        \fmod(\psi)               \cand \hspace{3pt}\divides\hspace{2pt} d\\
      \end{cases}\text{and\ }
      \begin{cases}
        \card \gamma                  \cand \leq s + 2 \cdot k\\
        \linnorm{\gamma\sub{2^u}{u}}  \cand \leq 3 \cdot a\\ 
        \onenorm{\gamma}              \cand \leq 2 \cdot c + 3\\
        \fmod(\gamma)                 \cand \hspace{3pt}\divides\hspace{2pt} d
      \end{cases}
\end{equation*}
Note that the updates performed to $C$ in line~\ref{optilep:line:stepI:updateC}
do not change the values of $\xi_C$ and $\mu$.

\paragraph{\textit{Step~II (lines~\ref{optilep:line:stepII:prepare-gamma-1}--\ref{optilep:line:stepII:run}).}} 
From~\Cref{lemma:CMS:first-step}, line~\ref{optilep:line:stepII:prepare-gamma-1} adds $k+1$ inequalities of the form $q \geq 0$ to $\gamma$ (one for each quotient variable). The following line~\ref{optilep:line:stepII:prepare-gamma-2} 
multiplies all (in)equalities in $\gamma$ by $\mu$. 
Therefore, when the program reaches line~\ref{optilep:line:stepII:run}, 
the formula $\gamma$ satisfies the following:
\begin{equation*}
    \begin{cases}
    \card \gamma                  &\leq s + 3k + 1 \\
    \linnorm{\gamma\sub{2^u}{u}}  &\leq 3 \cdot \mu \cdot a\\
    \onenorm{\gamma}              &\leq \mu \cdot (2 \cdot c + 3)\\
    \fmod(\gamma)                 &\hspace{3pt}\divides\hspace{2pt} d
  \end{cases}
\end{equation*}
Observe that the formula $\psi'$ constructed in line~\ref{optilep:line:stepII:prepare-psi} 
has bit size polynomial in $\psi'$. The procedure~\GaussOpt does not update this formula (\Cref{lemma:second-step-opt}), moreover $\psi'$ is not used again within the loop.
Therefore, no further analysis on $\psi'$ is necessary.
Together with $\psi'$, the call to~\GaussOpt in line~\ref{optilep:line:stepII:run} 
returns a linear program with divisions $\gamma'$ and a $(k,k)$-LEAC $C^*$.
(More precisely, we have $(C^*, \inst{\gamma'}{\psi'}) \in \objcons_k^k$.)
We bound the parameters of these two objects using~\Cref{lemma:ILEP:GaussOptBoundsNew}. 
In order to simplify the analysis, 
let us define $M \coloneqq 2^6 \cdot \mu \cdot \max(c,\log(\xi_C + \mu))$. 
The values~$L, Q, U$ and~$R$ defined in \Cref{lemma:ILEP:GaussOptBoundsNew} are all bounded by $M$. 
Furthermore, from the bound on $\linnorm{\gamma\sub{2^u}{u}}$, we have~$Q \leq 3 \cdot \mu \cdot a$, 
and therefore $\frac{Q}{\mu} \leq 3 \cdot a$.

\begin{description}
  \item[number of constraints in $\gamma'$:] Directly from~\Cref{lemma:ILEP:GaussOptBoundsNew}.\ref{lemma:ILEP:GaussOptBounds:i1}, $\card{\gamma'} \leq s + 4 \cdot k + 1$.
  \item[linear norm of $\gamma'$:] 
  Since $(C^*, \inst{\gamma'}{\psi'}) \in \objcons_k^k$, the system~$\gamma'$ only features the variable $u$ 
  and the quotient variable $q_{n-k}$. By~\Cref{lemma:ILEP:GaussOptBoundsNew}.\ref{lemma:ILEP:GaussOptBounds:i5}, the linear norm of $\gamma'$ is thus:
  \begin{align*}
    \linnorm{\gamma'} 
    &\leq \max{\Big({\textstyle\mu \cdot (k+1)^{k+1} \Big(\frac{Q}{\mu}\Big)^{k+1}, (k+1)^{k+1} \Big(\frac{Q}{\mu}\Big)^{k} U}\Big)} \\
    &\leq \max{\Big(\mu \cdot (3 \cdot a \cdot(k+1))^{k+1}, (k+1)^{k+1}  (3 \cdot a)^{k} M\Big)} \\
    &\leq (3 \cdot (k+1)\cdot a )^{k+1} M. && \Lbag\text{as $\mu \le M$}\Rbag
  \end{align*}
  \item[$1$-norm of $\gamma'$:] First, by~\Cref{lemma:ILEP:GaussOptBoundsNew}.\ref{lemma:ILEP:GaussOptBounds:i5} the bound on the constants of the terms from $\fterms(\gamma')$ is
  \begin{align*} 
    & {\frac{((k+1) \cdot Q)^{2(k+2)^2}}{\mu^{2k^2}} \cdot \fmod(\gamma) \cdot R}\\
    \leq{}& ((k+1) \cdot M)^{3(k+2)^2} d. &&\Lbag\text{as $R,Q \le M$ and $\fmod(\gamma) \leq d$}\Rbag
  \end{align*}
  Let us define $N \coloneqq 3 \cdot ((k+1) \cdot M)^{3(k+2)^2}$.
  Terms in $\fterms(\gamma')$ are of the form $b_1 \cdot q + b_2 \cdot 2^u + b_3$, where~$\abs{b_1}$ and $\abs{b_2}$ are bounded by $\linnorm{\gamma'} \le (3 \cdot (k+1)\cdot a )^{k+1} M$. Therefore, $\onenorm{\gamma'} \le N \cdot d$.

  \item[modulus of $\gamma'$:] By~\Cref{lemma:ILEP:GaussOptBoundsNew}.\ref{lemma:ILEP:GaussOptBounds:i6}, $\fmod{(\gamma')} \divides \alpha \cdot \fmod(\gamma)$, for some positive integer $\alpha \le \gfromfmodgammabound$.
  (Observe that then $\bmod(\gamma') \leq (k \cdot M)^{k^2} d$; we will silently use this fact when computing the bounds in~\Cref{bounds-one-iteration-step-iii} below.)
  \item[denominator $\eta_{C^*}$:] From \Cref{lemma:ILEP:GaussOptBoundsNew}.\ref{lemma:ILEP:GaussOptBounds:i7}, $\eta_{C^*} = \mu \cdot g$, for some $g \le k^k (\frac{Q}{\mu})^k \le \gfromMuBound$.
  \item[denominator $\mu_{C^*}$ and parameter~$\xi_{C^*}$:] by~\Cref{lemma:ILEP:GaussOptBoundsNew}.\ref{lemma:ILEP:GaussOptBounds:i4}, $\mu_{C^*} = \mu$ and $\xi_{C^*} = \xi_C$.
  \item[numerators in the new assignments of $C^*$:] Let us also observe that the terms $\tau$ occurring in
assignments $q \gets \frac{\tau}{\eta}$ of $C^*$, with $q$ quotient variable,
are linear terms in the variables $u$ and~$q_{n-k}$.
From~\Cref{lemma:ILEP:GaussOptBoundsNew}.\ref{lemma:ILEP:GaussOptBounds:i5}, 
$\linnorm{\tau} \leq (3 \cdot (k+1)\cdot a )^{k+1}\cdot M$, 
whereas the constant of $\tau$ is bounded by $((k+1) \cdot M)^{3(k+2)^2} d$ (same computations as for~$\linnorm{\gamma'}$ and $\onenorm{\gamma'}$).
\end{description}

\paragraph{\textit{Step~III (line~\ref{optilep:line:stepIII}).}} 
Starting from $\gamma'$, line~\ref{optilep:line:stepIII} produces 
two formulae $\gamma''$ and $\psi''$. 
From the bounds we have just obtained from $\gamma''$, 
and by appealing to~\Cref{lemma:complexity-of-step-iii}, we get:
\begin{equation}
  \label{bounds-one-iteration-step-iii}
    \begin{cases}
      \card \gamma''      \cand  \leq s + 4k + 3\\
      \onenorm{\gamma''}  \cand  \leq 2^5 N^3 d^3\\
      \fmod(\gamma'')     \cand \hspace{3pt}\divides\hspace{2pt} \lcm(\alpha \cdot d,\, \totient(\alpha \cdot d))
    \end{cases}
    \text{ \ and }
    \begin{cases}
      \card \psi''      \cand \leq 3\\
      \linnorm{\psi''}  \cand \leq 1\\ 
      \onenorm{\psi''}  \cand \leq 12 + 4 \cdot \log(N \cdot d)\\
      \fmod(\psi'')     \cand \hspace{3pt}\divides\hspace{2pt} \totient(\alpha \cdot d),
    \end{cases}
\end{equation}
for some positive integer $\alpha \le \gfromfmodgammabound$. 


\paragraph{\textit{Step~IV (lines~\ref{optilep:line:stepIV:lower}--\ref{optilep:line:stepIV:remove-qx}).}} 
The integers $\ell$ and $h$ in these lines are bounded by~${\onenorm{\gamma''} + \fmod{(\gamma'')}}$. Therefore, the value $v$ chosen in~line \ref{optilep:line:stepIV:guess} satisfies $0 \le v \le \onenorm{\gamma''} + \fmod{(\gamma'')}$. Observe that $\fmod{(\gamma'')} \le (\alpha \cdot d)^2 \le (k \cdot M)^{2k^2} d^2 \le N \cdot d^2$. Hence, $\abs{v} \le 2^6 N^3 d^3$. 
In line~\ref{optilep:line:stepIV:remove-qx}, the algorithm constructs the $(k+1)$-\preleac $C'$ 
whose bounds we are interested in. Following~\Cref{remark:LEAC-to-PRELEAC}, we obtain the bounds on $\mu_{C'}$ and $\xi_{C'}$ reported in the statement of the lemma:
\begin{description}
  \item[denominator $\mu_{C'}$:] \Cref{remark:LEAC-to-PRELEAC} tells us that $\mu_{C'} = \eta_{C^*} \leq \mu \cdot (3 \cdot k \cdot a)^k$.

  \item[parameter $\xi_{C'}$:] In our case, $\lambda$ from~\Cref{remark:LEAC-to-PRELEAC} is equal to $\frac{\eta_{C^*}}{\mu} = g$. 
  We have:
  \begin{align}
    \xi_{C'} 
    &\coloneqq \!\!\underbrace{\abs{\eta_{C^*}\cdot v}}_{\text{from assignment to $x_{n-k}$}} \!\!+ \sum\nolimits_{i=0}^{k-1} \underbrace{\Big(\abs{c_{n-i} \cdot v + d_{n-i}} + \abs{\lambda \cdot (a_{i,k} + b_{n-i})} + \sum\nolimits_{j=i+1}^{k-1} \abs{\lambda \cdot a_{i,j}}\Big)}_{\text{from assignment to $x_{n-i}$}}\notag\\
    &\le v \cdot \Big(\eta_{C^*} + \sum\nolimits_{i=0}^{k-1}(\abs{c_{n-i}} + \abs{d_{n-i}} + g \cdot \abs{b_{n-i}}) \Big) + g \cdot 
    \xi_{C}.
    \label{complexity-one-iteration:to-be-resumed}
  \end{align}
  Observe now that $\abs{c_{n-i}}$, $\abs{d_{n-i}}$ and $\abs{b_{n-i}}$ are integers occurring in assignments $q \gets \frac{\tau}{\eta}$ of $C^*$, with $q$ quotient variable. 
  From the bounds already deduced for these integers, we have 
  \[ 
    g \cdot \abs{b_{n-i}} \leq (3 \cdot k \cdot a)^k (3 \cdot (k+1)\cdot a )^{k+1} M \leq \textstyle\frac{N}{3},
  \]
  and, similarly, $\abs{c_{n-i}} \leq \frac{N}{3}$ and $\abs{d_{n-i}} \leq \frac{N}{3} \cdot d$. Resuming the computation in~\Cref{complexity-one-iteration:to-be-resumed}:
  \begin{align*}
    \xi_{C'}
    &\le v \cdot (\eta_{C^*} + k \cdot N \cdot d) + g \cdot \xi_C
    \,\le\, v \cdot (\mu \cdot (3 \cdot k \cdot a)^k + k \cdot N \cdot d) + g \cdot \xi_C\\
    &\le 2^6 (k+1) \cdot N^4 d^4 + (3 \cdot k \cdot a)^k \xi_C.
  \end{align*}
  The bound on $\xi_{C'}$ in the statement of the lemma then follows from 
  \[ 
    N \cdot d \leq 3 \cdot d \cdot ((k+1) \cdot M)^{3(k+2)^2}
    \leq 3 \cdot d \cdot ((k+1) \cdot 2^6 \mu \cdot \max(c,\log(\mu + \xi_C)))^{3(k+2)^2} \leq \beta.
  \]
\end{description}

\paragraph{\textit{Step~V (lines~\ref{optilep:line:stepV:prepare-r}--\ref{optilep:line:stepV:parepare-theta}).}}
We have $\phi' \coloneqq \psi \land \psi''$. 
Then, the bounds on the parameters of $\phi'$ 
given in the statement of the lemma 
are obtained by simply combining those computed 
for $\psi$ and $\psi''$.

Moving to the running time of performing one iteration of the body of the \textbf{while} loop, 
the bounds established above show that all operations performed (excluding calls to subprocedures) 
involves objects of polynomial size with respect to the sizes of $\phi$ and $C$. 
It is simple to see that all these operations (e.g., those in lines~\ref{optilep:line:stepI:updateC} or~\ref{optilep:line:stepIV:upper}) can be performed in polynomial 
time. Additionally, the guess in line~\ref{optilep:line:stepIV:guess} ranges over an interval 
of integers with polynomial bit length. Then, the non-deterministic polynomial-time complexity 
of one iteration of the loop follows directly
from~\Cref{lemma:ILEP:GaussOptBoundsNew,lemma:complexity-of-step-i,lemma:complexity-of-step-iii}.
\end{proof}

\paragraph*{The complexity of performing $k$ iterations of the main loop.}

\newcommand{\boundLstk}[1][X]{\ensuremath{\ell + 3 \cdot #1^2}}
\newcommand{\boundCardPhik}[1][X]{\ensuremath{s + 3 \cdot #1^3 + 2\cdot \ell \cdot #1}}
\newcommand{\boundLinNormPhik}[1][X]{\ensuremath{3^{#1} a}}
\newcommand{\boundOneNormPhik}[1][X]{\ensuremath{3^{8 (#1+1)} c}}
\newcommand{\boundXiCk}[1][X]{\ensuremath{3^{8 (#1+2)^8}c^{8 (#1+2)^7}}}

\newcommand{\boundBetak}[1][X]{\ensuremath{3^{2(#1+2)^8} c^{2(#1+2)^7}}}
\newcommand{\boundModPhik}[1][k]{\ensuremath{3^{2 \cdot {#1}^8} a^{2 \cdot {#1}^7}}}
\newcommand{\boundMuk}[1][k]{\ensuremath{3^{#1^3} a^{#1^2}}}

To complete the complexity analysis 
of~\OptILEP it now suffices to iterate 
the bounds computed in~\Cref{lemma:putting-all-together-one-iteration} 
across multiple iterations of the main loop of the algorithm.

\begin{restatable}{lemma}{LemmaPuttingAllTogetherKIterations}%
\superlabel{lemma:putting-all-together-k-iterations}{proof:LemmaPuttingAllTogetherKIterations}
\Cref{pseudocode:opt-ilep} runs in non-deterministic polynomial time.
Consider its execution on an integer linear-exponential program $\phi(x_1, \ldots, x_n)$ with $n\ge 1$. Let  $(\phi_k, \theta_k, C_k)$ the system, circuit, and ordering obtained at the end of $k$th iteration of the \textbf{while} loop of line~\ref{optilep:line:while}, in any non-deterministic branch of the algorithm. Then, the following bounds hold (for every $\ell,s,a,c \geq 1$):
\begin{align*}
  \text{if \ }
  &\begin{cases}
    \card \lst(\phi, \theta) &\leq \ell\\
    \card \phi                &\leq s\\
    \linnorm{\phi}            &\leq a\\ 
    \onenorm{\phi}            &\leq c\\
    \fmod(\phi)               &\hspace{3pt}\divides\hspace{2pt} 1 \\
  \end{cases}
  &\ \text{then}& 
  &\begin{cases}
    \card \lst(\phi_{k}, \theta_{k}) &\leq \boundLstk[k]\\
    \card \phi_{k}                &\leq \boundCardPhik[k]\\
    \linnorm{\phi_{k}}            &\leq \boundLinNormPhik[k]\\ 
    \onenorm{\phi_{k}}            &\leq \boundOneNormPhik[k]\\
    \fmod(\phi_{k})               &\leq \boundModPhik[k] \\
    \xi_{C_k}                   &\le \boundXiCk[k]\\
    \mu_{C_k}                  &\le \boundMuk[k].
  \end{cases}
\end{align*}
\end{restatable}

\begin{proof}[Proof sketch.]
  The proof is by induction on $k$, assuming as the induction hypothesis that the bounds stated hold at the $k$th iteration of the loop. 
  This hypothesis, combined with~\Cref{lemma:putting-all-together-one-iteration}, is sufficient to establish all bounds except for the one given to $\fmod(\phi_{k})$, 
  which we discuss next.
  
    Let $(\phi_0,C_0),\dots,(\phi_{k},C_{k})$ denote the formulae and~\preleac{s}
    constructed by the algorithm during the first $k$ iterations of 
    the \textbf{while} loop. We are looking to bound~$\fmod(\phi_{k+1})$.
    In particular, $\phi_0$ is the linear-exponential program given as input to~\OptILEP, and $C_0$ is the empty $0$-\preleac 
    initialized in line~\ref{optilep:line:initialize-C}.
    By~\Cref{lemma:putting-all-together-one-iteration},
    for every~$i \in [0..k]$,
    there is~$\alpha_{i+1} \in {[1..(3\cdot i \cdot \mu_{C_i} \cdot \linnorm{\phi_i})^{i^2}]}$
    such that $\fmod(\phi_{i+1})$ is a divisor of~$\lcm{(\fmod(\phi_{i}), \totient(\alpha_{i+1} \cdot \fmod(\phi_{i})))}$.
    Let us define $\alpha^* \coloneqq \lcm{(\alpha_1, \alpha_2, \ldots, \alpha_{k+1})}$, and consider the integers $c_0,\dots,c_{k+1}$ such that 
    $c_0 \coloneqq 1$ and $c_{i+1} \coloneqq \lcm(c_i,\totient(\alpha^* \cdot c_i))$ for $i \in [0..k]$.
    \begin{claim}\label{claim:bk-divides-ck:body}
      For every $j \in [0..k+1]$, $\fmod(\phi_{j})$ divides $c_{j}$.
    \end{claim}
    \begin{proof}
    The proof is by induction on $j$.
    \begin{description}
      \item[base case: $j=0$.] We have $\fmod(\phi_0) = 1 = c_0$.
      \item[induction step:] Assume that the claim holds for $j \in [0..k]$. Then,
    \begin{align*}
      \fmod(\phi_{j+1}) \coloneqq{}& 
          \lcm{(\fmod(\phi_j), \totient(\alpha_{j+1} \cdot \fmod(\phi_j)))} \\
          \divides{}& \lcm{(c_j, \totient(\alpha_{j+1} \cdot \fmod(\phi_j)))} 
          &\Lbag\text{$\fmod(\phi_i) \divides c_i$ by induction hypothesis}\Rbag\\
          \divides{}& \lcm{(c_j, \totient(\alpha^* \cdot c_j))} 
          &\Lbag \text{$q \divides r$ implies $\totient(q) \divides \totient(r)$}\Rbag\\
          ={}& c_{j+1}. &&\qedhere
    \end{align*}
    \end{description} 
    \end{proof}
    Given~\Cref{claim:bk-divides-ck:body}, 
    in order to bound $\fmod(\phi_{k+1})$ it suffices to bound $c_{k+1}$. 
    The next lemma from~\cite{ChistikovMS24} 
    will help us analyze this integer.
    \begin{restatable}[{\cite[Lemma 7]{ChistikovMS24}}]{lemma}{LemmaBoundLCMTotient}\label{remark:bound-on-lcm-totient-growth}
      Let $\alpha \geq 1$ be in $\N$. 
      Let
      $b_0, b_1, \dots$ be the integer sequence  given by the recurrence $b_0 := 1$ and $b_{i+1} := \lcm{(b_i, \totient(\alpha \cdot b_i))}$. For every $i \in \mathbb{N}$, $b_i \le \alpha^{2\cdot i^2}$.  
    \end{restatable}
    \noindent
    By~\Cref{remark:bound-on-lcm-totient-growth}, $c_{k+1} \le (\alpha^*)^{2(k+1)^2}$. Therefore, $c_{k+1} \le 3^{2 {(k+1)}^8} a^{2 {(k+1)}^7}$ follows from
    \begin{align*}
    \alpha^*
    &\leq \prod\nolimits_{i=0}^{k} (3 \cdot i \cdot \mu_{i} \cdot \linnorm{\phi_i})^{i^2} \\
    &\leq (3 \cdot k \cdot (\boundMuk[k]) \cdot (\boundLinNormPhik[k]))^{k^2(k+1)}  &\Lbag\text{by~induction~hypothesis}\Rbag\\
    &\leq 3^{(k+1)^6} a^{(k+1)^5}.
    &&\qedhere
    \end{align*}
\end{proof}

\subsection{Maximization and minimization of arbitrary linear-exponential terms}
\label{subsec:optimize-general-terms}

Together, \Cref{subsec:correctness-optilep,subsec:complexity-optilep} 
establish \Cref{theorem:small-optimum} for the case of maximizing a variable~$x$ under a linear-exponential program. We now complete the proof of~\Cref{theorem:small-optimum} by extending the argument to general linear-exponential objective functions, and to include the case of minimization.

Let us consider first the maximization problem 
\begin{equation}
  \label{sec-six-maximization}
  \text{maximize $\tau(\vec x)$ subject to $\phi(\vec x)$},
\end{equation}
where $\tau$ is a linear-exponential term, and $\phi$ an integer linear-exponential program. Let $z$ be a variable not in $\vec x$. 
Since variables in ILEP range over $\N$, a common approach to solving 
this problem is to distinguish two cases based on the sign of $\tau(\vec{x})$ in the optimal solution The variable~$z$ is used to represent the value of $\tau(\vec x)$, adjusting its sign accordingly when assuming $\tau$ to be negative. Here is the corresponding pseudocode:

\begin{algorithmic}[1]
  \If{(\text{maximize} $z$ \text{subject to} $\phi(\vec x) \land z = \tau(\vec x)$) has an optimal solution~$\sigma_1$}
    \textbf{return} $\sigma_1$\label{pseudocode:max:max} 
  \EndIf
  \If{$\phi(\vec x) \land \tau(\vec x) \geq 0$ is satisfiable} \textbf{return} ``no optimal solution exists''
  \EndIf
  \If{(\text{minimize} $z$ \text{subject to} $\phi(\vec x) \land -z = \tau(\vec x)$) has an optimal solution~$\sigma_2$} \textbf{return} $\sigma_2$\label{pseudocode:max:min} 
  \EndIf 
  \State \textbf{return} ``$\phi$ is unsatisfiable''
\end{algorithmic}

From~\Cref{subsec:correctness-optilep,subsec:complexity-optilep}, we know that the maximization problem in line~\ref{pseudocode:max:max} 
admits an optimal solution, then it has 
one representable with a polynomial-size ILESLP. 
Regarding the minimization problem in line~\ref{pseudocode:max:min},
since $z$ ranges over $\N$, a minimal solution 
is guaranteed to exist as soon as $\phi(\vec x) \land -z = \tau(\vec x)$ 
is satisfiable. Again from~\Cref{subsec:correctness-optilep,subsec:complexity-optilep}
when $\phi(\vec x) \land -z = \tau(\vec x)$ is satisfiable, 
then there is \emph{a} solution representable with a polynomial-size ILESLP~$\sigma$.
Then, to solve the minimization problem in line~\ref{pseudocode:max:min}, 
we can consider the equivalent maximization problem ``\text{maximize} $\sem{\sigma}(z) - z$ \text{subject to} $\phi(\vec x) \land -z = \tau(\vec x)$'', 
since $\sem{\sigma}(z) - z$ attains its maximum precisely when~$z$ is minimal.
Given that $\sem{\sigma}(z) - z$ is non-negative, this problem 
can be reformulated as ``\text{maximize} $w$ \text{subject to} $\phi(\vec x) \land (-z = \tau(\vec x)) \land (w = \sem{\sigma}(z) - z)$'', 
where $w$ is a fresh variable.
Of course, $\sem{\sigma}(z)$ may not be representable in binary 
using polynomially many bits.
Instead, we incorporate directly the ILESLP~$\sigma$ 
directly into the constraints of the linear-exponential program. 
To do so, we first rename every variable $y$ occurring in $\sigma$ as $y'$
to avoid conflicts with the variables~$\vec x$, $z$ and $w$.
Let $\sigma$ be now of the form~$(y_0' \gets \rho_1, \dots, y_t \gets \rho_t)$. We then solve the following maximization problem 
\begin{center}
  \text{maximize} $w$ \text{subject to} $\phi(\vec x) \land (-z = \tau(\vec x)) \land (w = z' - z) \land \bigwedge_{i=1}^t (y_i' = \rho_i)$.
\end{center}
From~\Cref{subsec:correctness-optilep,subsec:complexity-optilep}, 
if this problem has an optimal solution, 
then it has one representable with a polynomial-size ILESLP.
We conclude that the same holds for the problem in~\Cref{sec-six-maximization}.

We can treat the minimization problem 
\begin{equation*}
  \text{minimize $\tau(\vec x)$ subject to $\phi(\vec x)$},
\end{equation*}
in a similar way. Again following the sign 
of $\tau$, this problem is solved as follows:

\begin{algorithmic}[1]
  \If{(\text{maximize} $z$ \text{subject to} $\phi(\vec x) \land -z = \tau(\vec x)$) has an optimal solution~$\sigma_1$}
    \textbf{return} $\sigma_1$ 
  \EndIf
  \If{$\phi(\vec x) \land \tau(\vec x) < 0$ is satisfiable} \textbf{return} ``no optimal solution exists''
  \EndIf
  \If{(\text{minimize} $z$ \text{subject to} $\phi(\vec x) \land z = \tau(\vec x)$) has an optimal solution~$\sigma_2$} \textbf{return} $\sigma_2$
  \EndIf 
  \State \textbf{return} ``$\phi$ is unsatisfiable''
\end{algorithmic}
We already know that if one of the optimization problems in the code above has an optimal solution, 
then it has one representable by a polynomial-size ILESLP. 
This concludes the proof of~\Cref{theorem:small-optimum}.

\clearpage
\newcommand{\PartIITitle}{Deciding properties of {ILESLPs}}
\fancyhead[R]{{\color{gray}Part II: \PartIITitle}}
\part{\PartIITitle}\label{part:deciding-properties-ILESLP}
\addtocontents{toc}{This part presents the algorithm for manipulating and deciding properties of ILESLPs. 
In particular, it establishes~\Cref{theorem:pos-in-ptime,theorem:mod-in-p-factoring}, and describe how to compute an ILESLP representing ${\sem{\sigma}(x) \bmod 2^{\sem{\sigma}(y)}}$, which constitutes the main step towards the proof of~\Cref{theorem:fast-checking}. 
\Cref{part:deciding-properties-ILESLP} is completely independent of~\Cref{part:small-ILESLP} (except for the short~\Cref{lemma:pos-analysis-constant}).\par}

In this second part of the paper, we discuss algorithms 
for deciding~\posileslp and~\modileslp (\Cref{sec:deciding-pos,sec:deciding-mod}, respectively),
and for computing ILESLPs representing terms of the form~${(x \bmod 2^y)}$ (\Cref{computing-ileslp-xmod2y}).
This part is almost completely independent of~\Cref{part:small-ILESLP}, 
the sole exception being an appeal to~\Cref{lemma:pos-analysis-constant}
when proving that~\posileslp is in~\ptime. 
We refer the reader to section~\Cref{subsec:succinct-encoding-optimal-solutions} 
for the definition of ILESLPs. 

\paragraph*{Some notation and an auxiliary lemma.}
Let $\sigma \coloneqq (x_0 \gets \rho_0, \dots, x_n \gets \rho_n)$ be an ILESLP.
We write $e(\sigma)$ (respectively, $d(\sigma)$) for the absolute value of the product of all numerators~$m \neq 0$ (respectively, denominators~$g$) occurring in rational constants~$\frac{m}{g}$ of the scaling expressions $\frac{m}{g} \cdot x_j$ in~$\sigma$. By convention, this product is defined to be $1$ when taken over an empty set.
Given an expression ${E \coloneqq \sum_{j \in J} a_j \cdot 2^{x_j}}$, 
where $J \subseteq [0..n]$ and each $a_j$ is an integer,
we write $\sem{\sigma}(E)$ for the number obtained by \emph{evaluating} $E$ on $\sigma$, that is, $\sem{\sigma}(E) \coloneqq \sum_{j \in J}a_j \cdot 2^{\sem{\sigma}(x_j)}$. 
The next auxiliary lemma recasts $\sigma$ into a form that is 
more amenable to our subsequent algorithms.

\begin{restatable}{lemma}{LemmaSimpleExpressions}
    \superlabel{lemma:simple-expressions}{proof:LemmaSimpleExpressions}
    Consider an ILESLP \(\sigma \coloneqq (x_0 \gets \rho_0, \dots, x_n \gets \rho_n)\) and let $i \in [0..n]$. 
    One can compute, in time polynomial in the size of $\sigma$, 
    an expression $E_i$ of the form 
    $\sum_{j=0}^{i-1}a_{i,j} \cdot 2^{x_j}$ 
    such that 
    \(\sem{\sigma}(E_i) = d(\sigma)\cdot\sem{\sigma}(x_i)\). 
    For every $j \in [0..i-1]$, the coefficient $a_{i,j}$ is \emph{(i)} an integer 
    whose absolute value is bounded by $2^i \cdot e(\sigma) \cdot d(\sigma)$, and 
    \emph{(ii)} non-zero only if $\sem{\sigma}(x_j) \geq 0$.
\end{restatable}

\begin{proof}[Proof sketch] 
    Given $i \in [0..n]$,
    let $\sigma_i$ denote the ILESLP $(x_0 \gets \rho_0, \dots, x_i \gets \rho_i)$ obtained by truncating $\sigma$ after $i+1$ assignments.
    We remark that $d(\sigma_i)$ divides $d(\sigma_{j})$ for every $i \leq j$.

    Inductively on $i$, one shows 
    that it is possible to compute
    a vector of rational numbers $\vec b_i = (b_{i,0},\dots,b_{i,i-1})$ satisfying
    \(\sem{\sigma}(x_i) = \sum_{j = 0}^{i-1}b_{i,j} \cdot 2^{\sem{\sigma}(x_j)}\), 
    where each $b_{i,j}$ is of the form $\frac{m}{d(\sigma_i)}$ for some $m \in \Z$ satisfying $\abs{m} \leq 2^{i} \cdot e(\sigma_i) \cdot d(\sigma_i)$, 
    and $m \neq 0$ only if $\sem{\sigma}(x_j) \geq 0$.
    With this result at hand, the expression $E_i$ in the statement of the lemma is computed by multiply all these rational numbers by $d(\sigma)$, as to make them all integers. In particular, if $b_{i,j} = \frac{m}{d(\sigma_i)}$, then in~$E_i$ the coefficient of $2^{x_j}$ is $a_{i,j} \coloneqq m \cdot \frac{d(\sigma)}{d(\sigma_i)}$. 
    Hence, $\abs{a_{i,j}} \leq 2^{i} \cdot e(\sigma_i) \cdot d(\sigma_i) \cdot \frac{d(\sigma)}{d(\sigma_i)} \leq 2^i \cdot e(\sigma) \cdot d(\sigma)$.
    Note that the bit size of each $a_{i,j}$ is thus polynomial in the size of $\sigma$.
    With this in mind, the fact that the whole computation can be performed in polynomial time follows immediately from the inductive~proof.
\end{proof}

\RestoreHeader
\section{Deciding \texorpdfstring{\posileslp}{NatILESLP} in polynomial time} 
\label{sec:deciding-pos}
\begin{algorithm}[t]
    \caption{A polynomial-time algorithm for \posileslp.}
    \label{algo:pos}
    \begin{algorithmic}[1]
      \Require 
      ILESLP $\sigma \coloneqq (x_0 \gets \rho_0, \dots, x_n \gets \rho_n)$.
      \vspace{3pt}
      \State $C \gets 8 \cdot (\text{bit size of $\sigma$}) + 8$\label{algo:pos:line:def-M}
      \vspace{-3pt}
      \State \textbf{let}, for $i \in [0..n]$, $E_i$ be an expression $\sum_{j=0}^{i-1}a_j \cdot 2^{x_j}$, with all $a_j \in \Z$, 
      and~$\sem{\sigma}(x_i) = \frac{\sum_{j=0}^{i-1}a_j \cdot 2^{\sem{\sigma}(x_j)}}{d(\sigma)}$\label{algo:pos:line:def-Ei}
      \State $M$ $\gets$ empty map from $[0..n]^2$ to ${[-C..C] \cup \{-\infty, \infty\}}$\label{algo:pos:line:map-var-diffs}
      \For{$i$ from $0$ to $n$}\label{algo:pos:line:outerloop}
        \State $M(i,i) \gets 0$\label{algo:pos:line:mapV-diag-zero}
        \State \textbf{let} $\{\ell_0 = 0,\dots,\ell_{m}\}$ maximal subset of $[0..i-1]$ such that $M(\ell_k,\ell_{k-1}) \geq 0$ for every $k \in [1..m]$\label{algo:pos:line:def-var-index-ordering}
        \For{$j$ from $0$ to $i-1$}\label{algo:pos:line:loop-diff-var}
          \State $E \gets E_i - E_j$\label{algo:pos:line:expression-diff}
          \Comment{$E$ is of the form~$\sum_{k=0}^{m} c_k 2^{x_{\ell_k}}$}
          \For{$k$ \text{from} $m$ \text{to} $1$}\label{algo:pos:line:innerloop}
            \If{\(M(\ell_k,\ell_{k-1}) \leq C \)}
              \ replace $2^{x_{\ell_k}}$ with $2^{M(\ell_k,\ell_{k-1})} 2^{x_{\ell_{k-1}}}$ in $E$%
              \label{algo:pos:line:sub-var-diff-leq-M}
            \Else
              \State \(a \gets\) coefficient of \(2^{x_{\ell_k}}\) in $E$\label{algo:pos:line:big-case}
              \State \textbf{if} $a > 0$ \textbf{then} $M(i,j) \gets +\infty$\label{algo:pos:line:if-lead-coeff-gt-0}
              \State \textbf{if} $a < 0$ \textbf{then} $M(i,j) \gets -\infty$\label{algo:pos:line:if-lead-coeff-lt-0}
              \State \textbf{if} $a \neq 0$ \textbf{then} \textbf{break}\label{algo:pos:line:break-coeff-neq-zero}
            \EndIf
          \EndFor\label{algo:pos:line:innerloop-end}
          \If{$E$ is of the form $h \cdot 2^{x_0}$ for some $h \in \Z$}\label{algo:pos:line:if-E-integer}\label{algo:pos:line:E-sub-x0}
            \If{\( \frac{h}{d(\sigma)}\in [-C..C]\)}
                \(M(i,j) \gets \frac{h}{d(\sigma)}\)\label{algo:pos:line:if-E-small-set-val}
              \Else\
                \(M(i,j) \gets {\textbf{if}\, \frac{h}{d(\sigma)} \,{>}\, 0 \,\textbf{then}\, {+\infty} \,\textbf{else}\, {-\infty}}\)\label{algo:pos:line:else-val-E-big}
            \EndIf
          \EndIf
          \State \(M(j,i) \gets -M(i,j)\)\label{algo:pos:line:mapV-sym-diff}
        \EndFor\label{algo:pos:line:loop-diff-var-end}
      \EndFor\label{algo:pos:line:outerloopend}
      \State \textbf{return} true \textbf{if }\(M(n,0) \geq 0\) \textbf{else} false
      \label{algo:pos:return}
    \end{algorithmic}
  \end{algorithm}

\begin{center}
    \vspace{5pt}
    {\def\arraystretch{1.2}
    \begin{tabular}{|rl|}
        \hline
        \multicolumn{2}{|c|}{{\posileslp}}\\
        \hline
        \textbf{Input:}& An ILESLP $\sigma$.\\[-2pt]
        \textbf{Question:} & Is $\semlast{\sigma} \geq 0$\,?\\
        \hline
    \end{tabular}}
    \vspace{5pt}
\end{center}
The pseudocode of our procedure for deciding~\posileslp is given in~\Cref{algo:pos}.
In a nutshell, given an ILESLP $\sigma = ({x_0 \gets \rho_0}, \dots, x_n \gets \rho_n)$, 
the algorithm constructs a map $M$ with the following property: for every $i,j \in [0..n]$, the entry $M(i,j)$ stores the value of the difference $\sem{\sigma}(x_i) - \sem{\sigma}(x_j)$ up to a certain threshold $C$ defined in line~\ref{algo:pos:line:def-M}. 
If the absolute value of this difference exceeds the threshold, then $M(i,j)$ is instead equal to $+\infty$ or $-\infty$, 
depending on the sign of the difference.
After constructing $M$, the algorithm checks whether $M(n,0) \geq 0$  to decide if $\sem{\sigma}(x_n) \geq 0$.


Following~\Cref{lemma:simple-expressions}, for every $i \in [0..n]$, the algorithm starts by ``flattening'' the expression $\rho_i$ of $x_i$
into the form $\frac{E_i}{d(\sigma)}$, where $E_i = \sum_{k=0}^{i-1}a_k \cdot 2^{x_k}$ with $a_0,\dots,a_{i-1}$ integers (line~\ref{algo:pos:line:def-Ei}).
The computation of $M(i,j)$ with $j < i$ occurs at the $(i+1)$th iteration of the loop of line~\ref{algo:pos:line:outerloop} and $(j+1)$th iteration of the loop of line~\ref{algo:pos:line:loop-diff-var}. 
To compute $M(i,j)$, the expression $E \coloneqq E_i - E_j$ is considered (line~\ref{algo:pos:line:expression-diff}).
This is again of the form $\sum_{k=0}^{i-1}b_k \cdot 2^{x_k}$, where 
$b_k \ne 0$ only if $\sem{\sigma}(x_k) \geq 0$ (again by~\Cref{lemma:simple-expressions}). 
Since all entries of $M$ involving variables $x_0,\dots,x_{i-1}$ are already computed in the earlier iterations, we can reduce $E$ to an expression $\sum_{k=0}^{m}c_k \cdot 2^{x_{\ell_k}}$, 
where $\ell_0,\dots,\ell_m \in [0..i-1]$ are the indices of the variables $x_\ell$ with $\sem{\sigma}(x_\ell) \geq 0$, in ascending order (see line~\ref{algo:pos:line:def-var-index-ordering}).

Intuitively, if $\sem{\sigma}(x_{\ell_m})$ is large enough compared to $\sem{\sigma}(x_{\ell_{m-1}})$, then the sign of $\sum_{k=0}^{m}c_k \cdot 2^{\sem{\sigma}(x_{\ell_k})}$ is solely determined by the sign of the integer $c_m$ (assuming $c_m \neq 0$). 
The threshold $C$ has been chosen to capture this idea of $x_{\ell_m}$ being ``large enough''. 
In particular, one can show that for every $k \in [1..m]$, if \(\sem{\sigma}(x_{\ell_k}) - \sem{\sigma}(x_{\ell_{k-1}}) > C\) and $c_k \neq 0$,
then \(\abs{c_k 2^{\sem{\sigma}(x_{\ell_k})}} > \abs{\sum_{j=0}^{k-1}c_j 2^{\sem{\sigma}(x_{\ell_j})}} + d(\sigma) \cdot C\).
So, if $M(\ell_{m},\ell_{m-1}) > C$ and $c_m \neq 0$, we have $M(i,j) = \pm \infty$ 
(lines~\ref{algo:pos:line:big-case}--\ref{algo:pos:line:break-coeff-neq-zero}). 
Otherwise, if $\sem{\sigma}(x_{\ell_m})$ is small compared to $\sem{\sigma}(x_{\ell_{m-1}})$, that is, $M(\ell_{m},\ell_{m-1}) \leq C$, 
we replace $2^{x_{\ell_m}}$ with $2^{M(\ell_{m},\ell_{m-1})}\cdot 2^{x_{\ell_{m-1}}}$ in $E$, and iterate the same reasoning; now on variables $x_{\ell_{m-1}}$ and $x_{\ell_{m-2}}$.
At the end of the loop of line~\ref{algo:pos:line:innerloop}, either $M(i,j)$ has been set to $\pm\infty$, or we have reduced $E$ into an expression of the form $h \cdot 2^{x_0}$. 
In the latter case, we have $\sem{\sigma}(x_i) - \sem{\sigma}(x_j) = \frac{h \cdot 2^{\sem{\sigma}(x_0)}}{d(\sigma)} = \frac{h}{d(\sigma)}$. 
If $\frac{h}{d(\sigma)}$ belongs to $[-C..C]$, the algorithm sets $M(i,j) = \frac{h}{d(\sigma)}$ (line~\ref{algo:pos:line:if-E-small-set-val}). Else, $M(i,j)$ is set to $\pm \infty$, according to the sign of~$\frac{h}{d(\sigma)}$.

To prove that~\Cref{algo:pos} decides \posileslp in polynomial time, 
the key observation is that $C$ is linear in the bit size of~$\sigma$, 
and thus so is the bit size of the integers $2^{M(\ell_k,\ell_{k-1})}$ computed in~line~\ref{algo:pos:line:sub-var-diff-leq-M}.


We now formalize the above explanation, proving correctness and polynomial 
running time of~\Cref{algo:pos}. The correctness proof centers on the 
semantics of the map $M$.

\begin{lemma}
    \label{lemma:pos-correctness}
    Given an input ILESLP \(\sigma \coloneqq (x_0 \gets \rho_0, \dots, x_n \gets \rho_n)\),
    \Cref{algo:pos} constructs a map 
    \({M:[0..n]^2 \to [-C..C] \cup \{-\infty, \infty\}}\), 
    where \(C \coloneqq 8 \cdot (\textnormal{bit size of $\sigma$}) + 8\).
    For all \(i, j \in [0..n]\), 
    this map satisfies 
    \(M(i, j) = \trunc_C(\sem{\sigma}(x_i) - \sem{\sigma}(x_j))\), where $\trunc_C$ is the truncation function
    \begin{equation*}
        \trunc_C(g) \coloneqq 
        \begin{cases}
            g & \text{if } g \in [-C..C] \\
            -\infty & \text{if } g < -C \\
            +\infty & \text{if } g > C
        \end{cases}
        \hspace{2cm} 
        \text{for every } g \in \Z.
    \end{equation*}
\end{lemma}
\begin{proof}
    \label{proof:lemma:pos-correctness}
    The map $M$ is initialized as empty in line~\ref{algo:pos:line:map-var-diffs}.
    Let $E_0,\dots,E_n$ be the expressions computed in line~\ref{algo:pos:line:def-Ei}, following~\Cref{lemma:simple-expressions}. They satisfy $\sem{\sigma}(E_i) = d(\sigma) \cdot \sem{\sigma}(x_i)$, for every $i \in [0..n]$.
    We prove by induction on $i \in \N$ that after the $(i+1)$th iteration of the outer \textbf{for} loop of line~\ref{algo:pos:line:outerloop}, the map $M$ satisfies 
    \(M(j, k) = \trunc_C(\sem{\sigma}(x_j) - \sem{\sigma}(x_k))\)
    for every $j,k \in [0..i]$.

    \begin{description}
        \item[base case: $i=0$.] In the first iteration,
        line~\ref{algo:pos:line:mapV-diag-zero} sets $M(0,0)= 0$ as required. The inner loop of line~\ref{algo:pos:line:loop-diff-var} does not execute for $i = 0$ as the range of $j$ is empty.
        \item[induction hypothesis.] For every $j,k \in [0..i-1]$, 
            \(M(j, k) = \trunc_C(\sem{\sigma}(x_j) - \sem{\sigma}(x_k))\).
        \item[induction step: $i \geq 1$.] In the $(i+1)$th iteration, the algorithm sets $M(i,j)$ and $M(j,i)$ for all $j \in [0..i]$. 
        Line~\ref{algo:pos:line:mapV-diag-zero} correctly sets $M(i,i) = 0$. For each $j \in [0..i-1]$, the inner loop of line~\ref{algo:pos:line:loop-diff-var} computes $M(i,j)$ and $M(j,i)$. As $M(j,i) = -M(j,i)$ by line~\ref{algo:pos:line:mapV-sym-diff}, it suffices to show that the computed $M(i,j) = \trunc_C(\sem{\sigma}(x_i) - \sem{\sigma}(x_j))$ for all $j \in [0..i-1]$.

        In line~\ref{algo:pos:line:def-var-index-ordering}, 
        the algorithm computes the maximal subset of indices $\{\ell_0 = 0,\dots,\ell_{m}\} \subseteq [0..i-1]$ such that $M(\ell_k,\ell_{k-1}) \geq 0$ for every $k \in [1..m]$.
        By the induction hypothesis, $M(\ell_k,\ell_{k-1}) = \trunc_C(\sem{\sigma}(x_{\ell_k}) - \sem{\sigma}(x_{\ell_{k-1}}))$, 
        which implies $0 = \sem{\sigma}(x_{\ell_0}) \leq \sem{\sigma}(x_{\ell_1}) \leq \ldots \leq \sem{\sigma}(x_{\ell_{m}})$. Moreover, by the maximality of this subset, 
        the variables $x_{\ell_0},\dots,x_{\ell_m}$ are exactly those among $x_0,\dots,x_{i-1}$ for which $\sem{\sigma}$ is non-negative.

        \vspace{3pt}
        \textit{Computation of $M(i,j)$:}
        This computation occurs at the $(j+1)$th iteration of the inner \textbf{for} loop of line~\ref{algo:pos:line:loop-diff-var}, with $j \in [0..i-1]$. 
        Let $E$ be the expression $E_i - E_j$ as in line~\ref{algo:pos:line:expression-diff}.
        By~\Cref{lemma:simple-expressions},
        $E$ is of the form $\sum_{k=0}^{i-1} a_k \cdot 2^{x_{k}}$ where 
        \begin{enumerate}[align=left]
            \item each $a_k$ is an integer whose absolute value is bounded by $2^{n+1} \cdot e(\sigma) \cdot d(\sigma)$,
            \item if $a_k \ne 0$, then $\sem{\sigma}(x_k) \geq 0$.
        \end{enumerate}

        The second property above implies that $E$ can be written as
        $\sum_{k=0}^{m} a_{\ell_k} \cdot 2^{x_{\ell_k}}$. 
        From the definition of $E_i$ and $E_j$, we also have $\sem{\sigma}(E) = d(\sigma) \cdot (\sem{\sigma}(x_i) - \sem{\sigma}(x_j))$.

        The algorithm now enters the inner \textbf{for} loop of line~\ref{algo:pos:line:innerloop}, iterating $k$ from $m$ down to $1$.
        This loop progressively rewrites the expression $E$. 
        Let $E^{(t)}$ denote for the value of $E$ after the $t$th iteration of the loop.
        During the $t$th iteration, the value of the variable $k$ 
        is $m-t+1$.
        The following two claims (proved later) define the behavior of this loop:

        \begin{claim}
            \label{claim:pos-no-break}
            If the $t$th iteration of the loop of line~\ref{algo:pos:line:innerloop}
            completes without executing the \textbf{break} statement of line~\ref{algo:pos:line:break-coeff-neq-zero},
            then $E^{(t)}$ is of the form $c_{\ell_{m-t}} \cdot 2^{x_{\ell_{m-t}}} + \sum_{k = 0}^{m-t-1} a_{\ell_k} \cdot 2^{x_{\ell_k}}$, 
            where $c_{\ell_{m-t}}$ is an integer satisfying $\abs{c_{\ell_{m-t}}} < 2^{tC + n + 2} \cdot e(\sigma) \cdot d(\sigma)$.
            Moreover, $\sem{\sigma}(E^{(t)}) = \sem{\sigma}(E)$.
        \end{claim}

        \begin{claim}
            \label{claim:pos-break}
            If the $t$th iteration of the loop of line~\ref{algo:pos:line:innerloop}
            executes the \textbf{break} statement of line~\ref{algo:pos:line:break-coeff-neq-zero},
            then $\abs{\sem{\sigma}(E)} > d(\sigma) \cdot C$, 
            and $\sem{\sigma}(E)$ and the coefficient 
            of $2^{x_{\ell_{m-t+1}}}$~in~$E^{(t-1)}$ 
            have the same~sign.
        \end{claim}
        
        Using these two claims, we can now verify that $M(i,j)$ is computed correctly. If the \textbf{break} statement of 
        line~\ref{algo:pos:line:break-coeff-neq-zero} is executed during some iteration $t$ 
        of the loop of line~\ref{algo:pos:line:innerloop}, then
        by~\Cref{claim:pos-break}, we have $\abs{\frac{\sem{\sigma}(E)}{d(\sigma)}} > C$, and $\frac{\sem{\sigma}(E)}{d(\sigma)}$ has the same sign as that of the coefficient 
        of $2^{x_{\ell_{m-t+1}}}$ in $E^{(t-1)}$.
        This coefficient is the integer $a$ referenced in 
        line~\ref{algo:pos:line:big-case}, and 
        since the \textbf{break} statement was executed, $a \neq 0$. 
        Accordingly, lines~\ref{algo:pos:line:if-lead-coeff-gt-0} and~\ref{algo:pos:line:if-lead-coeff-lt-0} 
        set the value of $M(i,j)$ to $\pm\infty$, following the sign of $a$.
        Hence, $M(i,j) = \trunc_C(\frac{\sem{\sigma}(E)}{d(\sigma)}) = \trunc_C(\sem{\sigma}(x_i) - \sem{\sigma}(x_j))$, as required. 

        Suppose now that the \textbf{break} statement of 
        line~\ref{algo:pos:line:break-coeff-neq-zero} is never executed:
        the loop of line~\ref{algo:pos:line:innerloop} terminates after $m$ iterations, 
        and the expression~$E^{(m)}$ is defined.
        By~\Cref{claim:pos-no-break}, this expression 
        is of the form $c_{\ell_0} \cdot 2^{x_{\ell_0}}$, 
        for some integer $c_{\ell_0}$ (the bound on~$c_{\ell_0}$ given in~\Cref{claim:pos-no-break} will be later used in the runtime analysis in~\Cref{theorem:pos-in-ptime}).
        Since $\ell_0 = 0$ and $\sem{\sigma}(x_{0}) = 0$, 
        we have $\sem{\sigma}(E^{(m)}) = c_{\ell_0}$.
        Recall that $\sem{\sigma}(E_m) = \sem{\sigma}(E) = d(\sigma) \cdot (\sem{\sigma}(x_i) - \sem{\sigma}(x_j))$, so
        $M(i,j)$ must be set to $\trunc_C(\frac{c_{\ell_0}}{d(\sigma)})$.
        This is exactly what the algorithm does in lines~\ref{algo:pos:line:E-sub-x0}--\ref{algo:pos:line:else-val-E-big}.

        To complete the induction step, it is now sufficient to 
        prove Claims~\ref{claim:pos-no-break} and~\ref{claim:pos-break}.

        \begin{proof}[Proof of~\Cref{claim:pos-no-break}]
            For simplicity, let $B \coloneqq 2^{n+1} \cdot e(\sigma) \cdot d(\sigma)$. 
            Note that $B \in [2..2^{3 \cdot (\text{bit size of $\sigma$})+1}]$; 
            in particular, $B < 2^C$.
            The proof is by induction on $t$.

            \begin{description}
                \item[base case: $t = 0$.] Before the first iteration of the loop, we have $E^{(0)} = E$. Recall that $E$ is an expression of the form $\sum_{k=0}^{i-1} a_k \cdot 2^{x_{k}}$ where each $a_k$ is an integer whose absolute value is bounded by $B$. Thus, $\abs{a_0} < 2 \cdot B$, as required.
                \item[induction hypothesis.] If the $(t-1)$th iteration of the loop of line~\ref{algo:pos:line:innerloop}
                completes without executing the  \textbf{break} statement of line~\ref{algo:pos:line:break-coeff-neq-zero}, 
                then the expression $E^{(t-1)}$ is of the form ${c_{\ell_{m-t+1}} \cdot 2^{x_{\ell_{m-t+1}}} + \sum_{k = 0}^{m-t} a_{\ell_k} \cdot 2^{x_{\ell_k}}}$, 
                where $c_{\ell_{m-t+1}} \in \Z$ satisfies $\abs{c_{\ell_{m-t+1}}} < 2^{(t-1)C + 1} \cdot B$.
                Moreover, $\sem{\sigma}(E^{(t-1)}) = \sem{\sigma}(E)$.


                \item[induction step: $t \geq 1$.] Suppose that the $t$th iteration of the loop of line~\ref{algo:pos:line:innerloop}
                completes without executing the \textbf{break} statement of line~\ref{algo:pos:line:break-coeff-neq-zero}. 
                Then, the same must hold for the $(t-1)$th iteration, 
                and so the induction hypothesis applies.
                Since the \textbf{break} statement is not executed, there are two options: 
                \begin{enumerate}
                    \item The condition of the \textbf{if} statement in line~\ref{algo:pos:line:sub-var-diff-leq-M} is true, i.e., $M(\ell_{m-t+1},\ell_{m-t}) \leq C$, or 
                    \item $M(\ell_{m-t+1},\ell_{m-t}) = +\infty$ and the coefficient~$c_{\ell_{m-t+1}}$ of $2^{x_{\ell_{m-t+1}}}$ in $E^{(t-1)}$ is zero.
                \end{enumerate}


                In the second case, no update is needed: $E^{(t-1)}$ is already of the form $\sum_{k = 0}^{m-t} a_{\ell_k} \cdot 2^{x_{\ell_k}}$, and we have $E^{(t)} = E^{(t-1)}$.
                In the first case, $E^{(t)}$ is constructed from $E^{(t-1)}$ by replacing 
                $2^{x_{\ell_{m-t+1}}}$ by $2^{M(\ell_{m-t+1},\ell_{m-t})} \cdot 2^{x_{\ell_{m-t}}}$;
                see line~\ref{algo:pos:line:sub-var-diff-leq-M}.
                This removes $2^{x_{\ell_{m-t+1}}}$, 
                and modifies the coefficient of $2^{x_{\ell_{m-t}}}$ 
                from $a_{\ell_{m-t}}$ to 
                $c_{\ell_{m-t}} \coloneqq (a_{\ell_{m-t}} + c_{\ell_{m-t+1}} \cdot 2^{M(\ell_{m-t+1},\ell_{m-t})})$.
                The coefficients of the terms $2^{x_{\ell_k}}$ with $k \in [0..m-t-1]$ are unchanged.
                As~${M(\ell_{m-t+1},\ell_{m-t}) \geq 0}$, we have $c_{\ell_{m-t}} \in \Z$.
                Finally, we bound the absolute value of $c_{\ell_{m-t}}$ 
                as follows:
                \begin{align*}
                    \abs{c_{\ell_{m-t}}} &= \abs{a_{\ell_{m-t}}} + \abs{c_{\ell_{m-t+1}}} \cdot 2^{M(\ell_{m-t+1},\ell_{m-t})}\\
                    &\leq B + \abs{c_{\ell_{m-t+1}}} \cdot 2^C
                    &\hspace{-0.7cm} \Lbag\text{bounds on $a_{\ell_{m-t}}$ and $M(\ell_{m-t+1},\ell_{m-t})$}\Rbag\\
                    &\leq B + (2^{(t-1)C + 1} \cdot B-1) \cdot 2^C
                    & \Lbag\text{by induction hypothesis}\Rbag\\
                    &< 2^C + (2^{(t-1)C + 1} \cdot B-1) \cdot 2^C
                    & \Lbag\text{from $B < 2^C$}\Rbag\\
                    &< 2^{tC + 1} \cdot B.
                    &&\qedhere
                \end{align*}
            \end{description}
        \end{proof}
        \begin{proof}[Proof of~\Cref{claim:pos-break}]
            If the $t$th iteration of the loop of line~\ref{algo:pos:line:innerloop}
            executes the \textbf{break} statement of line~\ref{algo:pos:line:break-coeff-neq-zero}, 
            then the first $t-1$ iterations completed without executing the \textbf{break} statement.
            By~\Cref{claim:pos-no-break},  
            $E^{(t-1)}$ is of the form $c_{\ell_{m-t+1}} \cdot 2^{x_{\ell_{m-t+1}}} + \sum_{k = 0}^{m-t} a_{\ell_k} \cdot 2^{x_{\ell_k}}$, and $\sem{\sigma}(E^{(t-1)}) = \sem{\sigma}(E)$.
            Since the $t$th iteration executes the \textbf{break} statement, we have:
            \begin{enumerate}[align=left]
                \item\label{pos:claim-2:i1} $M(\ell_{m-t+1},\ell_{m-t}) = +\infty$ (from the condition of the \textbf{if} statements of line~\ref{algo:pos:line:sub-var-diff-leq-M}),
                \item\label{pos:claim-2:i2} $c_{\ell_{m-t+1}} \neq 0$ (from the condition of the \textbf{if} statement of line~\ref{algo:pos:line:break-coeff-neq-zero}).
            \end{enumerate}

            We show that 
            \begin{equation}
                \label{eq:pos:claim-2-bound}
                2^{\sem{\sigma}(x_{\ell_{m-t+1}})} > \abs{\sum\nolimits_{k = 0}^{m-t} a_{\ell_k} \cdot 2^{\sem{\sigma}(x_{\ell_k})}} + d(\sigma) \cdot C.
            \end{equation}
            From the definition of $E^{(t-1)}$,
            and the fact that $\sem{\sigma}(E^{(t-1)}) = \sem{\sigma}(E)$ and $c_{\ell_{m-t+1}} \neq 0$, this inequality implies~\Cref{claim:pos-break}.

            To prove~\Cref{eq:pos:claim-2-bound}, we first note that $M(\ell_{m-t+1},\ell_{m-t}) = +\infty$
            implies, from the induction hypothesis in the main proof, that $\sem{\sigma}(x_{\ell_{m-t+1}}) > \sem{\sigma}(x_{\ell_{m-t}}) + C$. 
            Also, by definition of the indices $\ell_0,\dots,\ell_m$, 
            $\sem{\sigma}(x_{\ell_{m-t}}) \geq \sem{\sigma}(x_{\ell_{k}}) \geq 0$ for every $k \in [0..m-t-1]$.
            Hence, $2^{\sem{\sigma}(x_{\ell_{m-t+1}})} > 2^C \cdot 2^{\sem{\sigma}(x_{\ell_{m-t}})}$ and $\big(\sum\nolimits_{k = 0}^{m-t} \abs{a_{\ell_k}}\big) \cdot 2^{\sem{\sigma}(x_{\ell_{m-t}})} \geq \abs{\sum\nolimits_{k = 0}^{m-t} a_{\ell_k} \cdot 2^{\sem{\sigma}(x_{\ell_k})}}$, 
            which in turn implies that~\Cref{eq:pos:claim-2-bound} 
            holds as soon as we prove $2^C \geq \big(\sum\nolimits_{k = 0}^{m-t} \abs{a_{\ell_k}}\big) + d(\sigma) \cdot C$.
            To show this inequality, we establish that 
            $2^{C/2} \geq \sum\nolimits_{k = 0}^{m-t} \abs{a_{\ell_k}}$ 
            and $C \geq 4 \cdot \log_2(d(\sigma)) + 8$;
            the inequality then follows from~\Cref{lemma:pos-analysis-constant}.
            We have 
            \begin{align*}
                &2 \cdot \ceil{\log_2(\max(1, {\textstyle\sum}_{k=0}^{m-t}\abs{a_{\ell_k}}))} + 4\ceil{\log_2(d(\sigma))} + 8\\
                \leq{}& 2 \cdot \ceil{\log_2(n \cdot 2^{n+1} \cdot e(\sigma) \cdot d(\sigma))} + 4\ceil{\log_2(d(\sigma))} + 8
                &\hspace{-1.5cm} \Lbag\text{bound on each $a_{\ell_k}$}\Rbag\\
                \leq{}& 2 \cdot (\underline{n+1} + \underline{\ceil{\log_2(n)} + \ceil{\log_2(e(\sigma))} + \ceil{\log_2(d(\sigma))}}) + 4 \cdot \underline{\ceil{\log_2(d(\sigma))}} + 8\\
                \leq{}& 8 \cdot (\text{bit size of $\sigma$}) + 8 = C.
                &\hspace{-7cm} \Lbag\text{each \underline{underlined} quantity is $\leq (\text{bit size of $\sigma$})$}\Rbag
            \end{align*}
            Therefore, both $2^{C/2} \geq \sum\nolimits_{k = 0}^{m-t} \abs{a_{\ell_k}}$ 
            and $C \geq 4 \cdot \log_2(d(\sigma)) + 8$ hold.
        \end{proof}

    \end{description}

    This completes the proof of~\Cref{lemma:pos-correctness}.
    \qedhere

\end{proof}

\TheoremPosInPtime*

\begin{proof}

    Let $\sigma = (x_0 \gets \rho_0, \dots, x_n \gets \rho_n)$ be the input ILESLP of Algorithm~\ref{algo:pos}.

    \vspace{3pt}
    \textit{Correctness:}
    By~line~\ref{algo:pos:return}, the algorithm returns true if and only if ${M(n, 0)} \geq 0$.
    Recall that \(\sem{\sigma}(x_0) = 0\). It then follows from~\Cref{lemma:pos-correctness} that
    $M(n,0) = \trunc_C(\sem{\sigma}(x_n))$.
    Since $C > 0$, $\trunc_C(\sem{\sigma}(x_n)) \geq 0$ if and only if $\sem{\sigma}(x_n) \geq 0$.
    Therefore, $M(n, 0) \geq 0$ if and only if $\sem{\sigma}(x_n) \geq 0$; 
    showing the correctness of the algorithm. 

    \vspace{3pt}
    \textit{Complexity:}
    To prove that ~\Cref{algo:pos} runs in polynomial time, observe that:
    \begin{itemize}[itemsep=3pt]
        \item The bit length of $C$ is bounded logarithmically in the bit size of $\sigma$.
        \item The expressions $E_0,\dots,E_n$ are computed in polynomial time following~\Cref{lemma:simple-expressions}. 
        \item The map $M$ requires only $O(n^2\log_2{C})$ space. 
        \item All \textbf{for} loops (lines~\ref{algo:pos:line:outerloop}, \ref{algo:pos:line:loop-diff-var} and~\ref{algo:pos:line:innerloop}) iterate on intervals of size linear in the bit size of $\sigma$.
    \end{itemize}
    Following these observations, the only remaining crucial point is showing that repeated executions of line~\ref{algo:pos:line:sub-var-diff-leq-M}, locally to one iteration of the \textbf{for} loop of line~\ref{algo:pos:line:loop-diff-var}, do not cause the integers in the expression $E$ 
    to grow superpolynomially. 
    This property is already established in~Claim~\ref{claim:pos-no-break} of the proof of~\Cref{lemma:pos-correctness}. In particular, the absolute value of each integer in $E$ is bounded by $2^{nC + n + 2} \cdot e(\sigma) \cdot d(\sigma)$;
    the bit length of this number is polynomial in the bit size of $\sigma$. 
    With this key observation, it follows that all remaining operations performed by the algorithm run in polynomial time.
\end{proof} 

\section{Deciding \texorpdfstring{\modileslp}{ModILESLP} in \texorpdfstring{$\ptime^{\factoring}$}{P{\textasciicircum}Factoring}}
\label{sec:deciding-mod}
\begin{center}
    \vspace{5pt}
    {\def\arraystretch{1.2}
    \begin{tabular}{|rl|}
        \hline
        \multicolumn{2}{|c|}{{\modileslp}}\\
        \hline
        \textbf{Input:}& An ILESLP $\sigma$, and $g \in \N_{\geq 1}$ encoded in binary.\\[-2pt]
        \textbf{Question:} & Is $\semlast{\sigma}$ divisible by $g$\,?\\
        \hline
    \end{tabular}}
    \vspace{5pt}
\end{center}
We describe a procedure that, given an ILESLP $\sigma = (x_0 \gets \rho_0, \dots, x_n \gets \rho_n)$ and $g \in \N_{\geq 1}$ encoded in binary, 
outputs (the binary encoding of) the remainder of $\semlast{\sigma}$ modulo $g$. The decision problem \modileslp\ is solved by checking if the remainder in output is $0$. The pseudocode of this procedure is shown in~\Cref{algo:mod}.

We denote by $\nu_\sigma \colon \N_{\geq 1} \to \N_{\geq 1}$ the function ${\nu_\sigma(x) \defeq \totient(\odd( x \cdot d(\sigma)))}$, where $\odd(a)$ denotes the largest odd factor of $a \in \N_{\geq 1}$, and $\totient$ denotes Euler's totient function (see~\Cref{equation:compute-totient-via-factorization} on page~\pageref{equation:compute-totient-via-factorization} for the formal definition).
We denote by $\nu^k_\sigma$ the $k$th iterate of $\nu_\sigma$, that is, $\nu^0_\sigma(x) \defeq x$ and $\nu^{k+1}_\sigma(x) \defeq \nu_\sigma(\nu^k_\sigma(x))$ for every $k \in \N$.

\Cref{algo:mod} constructs a map $M$ with the following property: for every  $i,k \in [0..n]$ with $i+k \leq n$, the entry $M(i,k)$ stores the value of $\sem{\sigma}(x_i) \bmod \nu_\sigma^k(g)$.
Once the map is opportunely populated, it returns the value $M(n,0)$ 
corresponding to $\semlast{\sigma} \bmod g$ (line~\ref{algo-mod-line-return}).

  \begin{algorithm}[t]
    \caption{A $\fptime^\factoring$ algorithm for computing $\semlast{\sigma} \bmod g$.}
    \label{algo:mod}
    \begin{algorithmic}[1]
      \Require 
      ILESLP $\sigma \coloneqq (x_0 \gets \rho_0, \dots, x_n \gets \rho_n)$, and $g \in \N_{\geq 1}$ encoded in binary.%
      \vspace{3pt}
      \State $M \gets $ empty map from $[0..n]^2$ to $\mathbb{N}$ 
      \label{algo:modes:line-def-M}
      \For{$k$ form $0$ to $n$} 
        $M(0, k) \gets 0$\label{algo-mod-set-delta-for-0}
      \EndFor
      \For{$i$ from $1$ to $n$}\label{algo:mod:main-loop}
        \State \textbf{let} $a_0,\dots,a_{i-1} \in \Z$ such that~$\sem{\sigma}(x_i) = \frac{\sum_{j=0}^{i-1}a_j \cdot 2^{\sem{\sigma}(x_j)}}{d(\sigma)}$\label{algo-mod-simplify-xi} 
        \For{$k$ from $0$ to $(n-i)$}\label{algo-mod-line-internal-for-loop-start}
        \State $h \gets \nu_\sigma^k(g) \cdot d(\sigma)$ 
          \Comment{requires factorization oracle}\label{algo-mod-line-compute-psi_k_d}
          \State \textbf{let} $m,q \in \N$ such that $q = \odd(h)$ and $h = 2^m \cdot q$\label{algo-mod-line-decompose-psi_k_d}
          \For{$j$ from $0$ to $i-1$}\label{algo-mod-line-inner-j-loop-start}
          \State $b \gets $ 
          \textbf{if}
            $\sem{\sigma}(x_{j}) \ge m$ 
           \textbf{then}
            $0$ 
          \textbf{else} 
            $2^{\sem{\sigma}(x_{j})}$\label{algo-mod-line-internal-loop-compute-a-ij}
            \Comment{uses the algorithm for~\posileslp}
          \State $c \gets 2^{M(j, k+1)} \bmod q$\label{algo-mod-line-internal-loop-compute-b-ij}
          \State \textbf{let}  $r_{j}$ be the (only) value in $[0..h-1]$ such that $2^m$ divides $r_{j} - b$, and $q$ divides $r_{j} - c$%
          \label{algo-mod-line-internal-loop-compute-r-ij}
          \EndFor
          \State $M(i, k) \gets (\frac{1}{d(\sigma)}\sum_{j=0}^{i-1}a_{j} \cdot r_{j}) \bmod \nu_\sigma^k(g)$\label{algo-mod-line-internal-loop-compute-Delta-ij}
        \EndFor
      \EndFor
      \State \textbf{return} $M(n,0)$
      \label{algo-mod-line-return}
    \end{algorithmic}
  \end{algorithm}

The core of the algorithm is the \textbf{for} loop of line~\ref{algo:mod:main-loop}. During its $i$th iteration, this loop populates the entries $M(i,0),\dots,M(i,n-i)$.
(The base case of $i = 0$ is handled in line~\ref{algo-mod-set-delta-for-0}, as $\sem{\sigma}(x_0) = 0$.) 
Similarly to \Cref{algo:pos} and following~\Cref{lemma:simple-expressions}, in line~\ref{algo-mod-simplify-xi} the algorithm ``flattens'' the expression associated to $x_i$ into one of the form $\frac{1}{d(\sigma)} \sum_{j=0}^{i-1}a_j \cdot 2^{x_j}$, with $a_1,\dots,a_{i-1}$ integers.
Let $h \coloneqq \nu_\sigma^k(g) \cdot d(\sigma)$.

To compute $M(i,k)$ (during the $(k+1)$th iteration of the loop of line~\ref{algo-mod-line-internal-for-loop-start}) we reason modulo $\nu_\sigma^k(g)$:
\begin{align*}
    M(i,k) &= \textstyle\frac{1}{d(\sigma)} \cdot \sum_{j=0}^{i-1}a_j \cdot 2^{\sem{\sigma}(x_j)}\\
    &= \textstyle\frac{1}{d(\sigma)} \cdot \sum_{j=0}^{i-1}a_j \cdot (2^{\sem{\sigma}(x_j)} \bmod h),
\end{align*}
where the last number is an integer, because the sum $\sum_{j=0}^{i-1}a_j \cdot 2^{\sem{\sigma}(x_j)}$ is divisible by $d(\sigma)$ (as $\sigma$ is an ILESLP). To compute $2^{\sem{\sigma}(x_j)} \bmod h$, we appeal to the Chinese Remainder Theorem (CRT). Write $h = 2^m \cdot q$ with $m,q \in \N$ and $q$ odd. Compute ${b \coloneqq 2^{\sem{\sigma}(x_j)} \bmod 2^m}$ and ${c \coloneqq 2^{\sem{\sigma}(x_j)} \bmod q}$, 
and then use CRT to find the only $r \in [0..h-1]$ such that $r \equiv b \bmod{2^m}$, and $r \equiv c \bmod{q}$ (line~\ref{algo-mod-line-internal-loop-compute-r-ij}).

The value $b$ is computed via~\Cref{algo:pos} (line~\ref{algo-mod-line-internal-loop-compute-a-ij}). 
To compute $c$, we see that $2^{\totient(q)} \bmod q = 1$, by Euler's theorem.
Therefore, $2^{\sem{\sigma}(x_j)}$ and $2^{\sem{\sigma}(x_j) \bmod \totient(q)}$ 
have the same reminder modulo $q$.
We have $\totient(q) = \nu_\sigma^{k+1}(g)$ by definition, and so 
$M(j,k+1) = \sem{\sigma}(x_j) \bmod \totient(q)$.
Note that $j < i$, and so $M(j,k+1)$ has been populated in a previous iteration of the loop of line~\ref{algo:mod:main-loop}; 
the algorithm thus constructs $c$ by computing $2^{M(j, k+1)} \bmod q$ 
(line~\ref{algo-mod-line-internal-loop-compute-b-ij}).

Regarding the complexity of~\Cref{algo:mod}, most of its operations can be implemented in polynomial time. The map $M$ requires polynomial space, 
since $\nu_\sigma^k(g)$ is bounded by $d(\sigma)^k \cdot g$.
Moreover, while $2^{M(j,k+1)}$ can in principle be of exponential bit size, computing it modulo~$q$ can be done in polynomial time in the bit size of $M(j,k+1)$ and $q$ by relying on the exponentiation-by-squaring method~\cite[Ch.~1.4]{BressoudW08}.
The only difficulty stems from the computation 
of~$\nu_\sigma^k(g)$. 
For this we use the integer factorization oracle in order to compute Euler's totient function, following~\Cref{equation:compute-totient-via-factorization}.

We now formalize the above arguments into a full proof of correctness and runtime analysis.

\TheoremModInPFactoring*

\begin{proof}
Consider as input an ILESLP $\sigma = (x_0 \gets \rho_0, \dots, x_n \gets \rho_n)$ and $g \in \N_{\geq 1}$.
We show that~\Cref{algo:mod} computes $\semlast{\sigma} \bmod g$, and runs in polynomial time with a factoring oracle.
For simplicity of the exposition, 
let us see the map $M$ as an $(n+1) \times (n+1)$ matrix over $\mathbb{N}$, with indices for rows and columns in $[0..n]$. The matrix is initially empty, and the algorithm only populates it in a ``triangular way'',
only filling the entries $(i,k)$ such that $i + k \leq n$.
Upon completion of the algorithm, we will have $M(i,k) = (\sem{\sigma}{x_i} \bmod \nu_{\sigma}^k(g))$, for every such entry $(i,k)$. \Cref{figure:algo-mod-matrix} depicts this matrix.

Since $\totient(a) \leq a$ for all $a \in \N_{\geq 1}$, we have $\nu_{\sigma}^k(g) \leq d(\sigma)^k g$.
As a result, the number of bits required to store $M$ is in $O(n^2 \log_2(d(\sigma)^n g))$. 
We prove by induction on $i \in \N$ that, during the $i$-th iteration of the \textbf{for} loop of line~\ref{algo:mod:main-loop}, the algorithm correctly fills the matrix representing $M$ up to the $i$-th row, and 
does so in polynomial time, assuming access to a factoring oracle.

\begin{description}
\item[base case: $i = 0$.] 
  (i.e., before the loop of line~\ref{algo:mod:main-loop} starts), the $0$-th row of the matrix is already populated with all $0$s (line~\ref{algo-mod-set-delta-for-0}).

\item[induction hypothesis.] 
  The first $(i-1)$th rows of the matrix $M$ are correctly populated. That is, for each $j \in [0..i-1]$ and $k \in [0..n-j]$, we have $M(j,k) = \sem{\sigma}{x_k}\bmod \nu_{\sigma}^j(g)$. 

  \begin{figure}[t]
    \[
    \begin{blockarray}{cccccc}
    \nu_{\sigma}^0(g) = g & \nu_{\sigma}(g) & \nu_{\sigma}^2(g) & \nu_{\sigma}^3(g) & \nu_{\sigma}^4(g) \\[3pt]
    \begin{block}{(ccccc)c}
      0 & 0 & 0 & 0 & 0 & x_0 \\
      \sem{\sigma}{x_1}\bmod g   &  \sem{\sigma}{x_1}\bmod \nu_{\sigma}(g)  &  \sem{\sigma}{x_1}\bmod \nu_{\sigma}^2(g)  & \sem{\sigma}{x_1}\bmod \nu_{\sigma}^3(g)  & \textcolor{red}{\times} & x_1 \\
      \sem{\sigma}{x_2}\bmod g  &  \sem{\sigma}{x_2}\bmod \nu_{\sigma}(g)  &  \sem{\sigma}{x_2}\bmod \nu_{\sigma}^2(g)  &  \textcolor{red}{\times} &  \textcolor{red}{\times} & x_2 \\
      - & - &  \textcolor{red}{\times} &   \textcolor{red}{\times}&  \textcolor{red}{\times} & x_3 \\
      - &  \textcolor{red}{\times} &  \textcolor{red}{\times} &  \textcolor{red}{\times} &  \textcolor{red}{\times} & x_4 \\
    \end{block}
    \end{blockarray}
     \]
    
     \vspace{-10pt}
     \caption{Illustration of how the matrix $M$ looks like after $2$ iterations of the loop of line~\ref{algo:mod:main-loop} (in the case of $n = 4$). To fill the next row, we only need the values in the rows above it. Entries with $\textcolor{red}{\times}$ are left empty in the algorithm as they are not needed to perform the computation of $\sem{\sigma}(x_n) \bmod g$.}
     \label{figure:algo-mod-matrix}
    \end{figure}

    \item[induction step: $i \geq 1$.] The $i$th iteration of the loop handles the variable $x_i$. 
    Line~\ref{algo-mod-simplify-xi} computes
    integers $a_0,\dots,a_{i-1}$ such that $\sem{\sigma}(x_i) =  \frac{1}{d(\sigma)} \sum_{j=0}^{i-1} a_{j}2^{\sem{\sigma}(x_j)}$.
    Following~\Cref{lemma:simple-expressions}, this computation can be 
    performed in polynomial time.
    Next, for each $k \in [0..n-i]$, the inner \textbf{for} loop of line~\ref{algo-mod-line-internal-for-loop-start} computes $\sem{\sigma}(x_i) \bmod \nu_{\sigma}^k(g)$, storing the in $M(i,k)$ (see line~\ref{algo-mod-line-internal-loop-compute-Delta-ij}).
    In order to compute $\sem{\sigma}(x_i) \bmod \nu_{\sigma}^k(g)$, the algorithm uses the identity 
\begin{equation}
  \label{equation:mod:main-equation}
  \sem{\sigma}(x_i) \bmod \nu_{\sigma}^k(g) \ = \ {\left(\frac{1}{d(\sigma)} \cdot \sum\nolimits_{j=0}^{i-1} a_{j} \cdot \big(2^{\sem{\sigma}(x_j)} \bmod \nu_{\sigma}^k(g) \cdot d(\sigma)\big)\right)  \bmod \nu_{\sigma}^k(g)}.
\end{equation} 

We discuss the $(k+1)$th iteration of the inner \textbf{for} loop of line~\ref{algo-mod-line-internal-for-loop-start}, analyzing it line by line.

  \begin{description}
      \item[line~\ref{algo-mod-line-compute-psi_k_d}.] The algorithm computes $\nu_{\sigma}^k(g)$. The integer factorization oracle is used to factorize the arguments of Euler's totient function, to then compute the result of this function via \Cref{equation:compute-totient-via-factorization}.
      Since $\nu_{\sigma}^k(g)$ is bounded by $d(\sigma)^k g$, the computation is polynomial time except the oracle calls. The value $h = \nu_\sigma^k(g) \cdot d(\sigma)$ is then computed in polynomial time.

      \item[line~\ref{algo-mod-line-decompose-psi_k_d}.] 
      The algorithm factorizes $h$ as $h = 2^m\cdot q$, where $q \in \N$ is the largest odd divisor of $h$. This is done by repeatedly dividing $h$ by $2$  until it becomes odd to compute $q$; which takes polynomial time. Note that $2^m$ and $q$ are coprime. 

      The values $m$ and $q$ will be used to compute the value $r_j \defeq 2^{\sem{\sigma}(x_j)} \bmod h$, following~\Cref{equation:mod:main-equation}, in the $(j+1)$th iteration of the inner \textbf{for} loop of line~\ref{algo-mod-line-inner-j-loop-start}.
      To do this, we first compute $b \coloneqq 2^{\sem{\sigma}(x_j)} \bmod 2^m$ and $c \coloneqq 2^{\sem{\sigma}(x_j)} \bmod q$. Then, by the~CRT, we compute the unique $r_j \in [0..h-1]$ such that $r_j \equiv b \bmod 2^m$ and $r_j \equiv c \bmod q$. 
  \end{description}
  We analyse lines~\ref{algo-mod-line-internal-loop-compute-a-ij}--\ref{algo-mod-line-internal-loop-compute-r-ij} locally to the $(j+1)$th iteration of the \textbf{for} loop of line~\ref{algo:pos:line:innerloop} ($j \in [0..i-1]$).



  \begin{description}
      \item[line~\ref{algo-mod-line-internal-loop-compute-a-ij}.] This line computes $b = 2^{\sem{\sigma}(x_j)} \bmod 2^m$. 
      If $\sem{\sigma}(x_j) \ge m$ then clearly $b = 0$, and otherwise $b = 2^{\sem{\sigma}(x_j)}$. 
      The comparison $\sem{\sigma}(x_j) \ge m$  can be performed in polynomial time using~\Cref{algo:pos}. 
      Moreover, when we set $b = 2^{\sem{\sigma}(x_j)}$, its value is less than $2^m$, so its bit size remains polynomial in the input size.

      \item[line~\ref{algo-mod-line-internal-loop-compute-b-ij}.] This line computes $c = 2^{\sem{\sigma}(x_j)} \bmod q$. First, recall that Euler's theorem states that if two numbers $a$ and $q$ are coprime, then $a^{\totient(q)}$ is congruent to $1$ modulo $q$. Since $q$ is odd, we thus have $c = 2^{\sem{\sigma}(x_j) \bmod \totient(q)} \bmod q$. 
      From the induction hypothesis, 
      we have already set $M(j,k+1)$ to $\sem{\sigma}(x_j) \bmod \totient(q)$ in previous iterations of the loop of line~\ref{algo:mod:main-loop}.
      Therefore, to compute $c$ 
      it only remains to compute $2^{M(j, k+1)} \bmod q$ 
      in polynomial time; which can be done by relying on the 
      exponentiation-by-squaring method~\cite[Ch.~1.4]{BressoudW08}.
      \item[line~\ref{algo-mod-line-internal-loop-compute-r-ij}.] Given $b$ and $c$, the algorithm replies on (a constructive version of) the CRT to compute $r_j$ in polynomial time.
      More precisely, for this computation one can use the extended Euclidean algorithm to compute two integers $\ell_1$ and $\ell_2$ such that $\ell_1 \cdot 2^m + \ell_2 \cdot q = 1$
      ($\ell_1$ and $\ell_2$ exist by B\'ezout's identity). 
      Then, $r_j \coloneqq (b \cdot \ell_2 \cdot q + c \cdot \ell_1 \cdot 2^m) \bmod h$.
  \end{description}
  When the \textbf{for} loop of line~\ref{algo-mod-line-inner-j-loop-start} ends, the algorithm has computed $r_j$ for every ${j \in [0..i-1]}$. 
  Line~\ref{algo-mod-line-internal-loop-compute-Delta-ij} 
  executes, populating $M(i,k)$ in polynomial time following the formula in~\Cref{equation:mod:main-equation}. 
\end{description}
From the analysis above, we conclude that the algorithm is correct. 
Regarding the running time, note that each loop in the procedure iterates only over natural numbers bounded by $n$ (which is bounded by the bit size of $\sigma$). 
Furthermore, all other operations performed by the algorithm run in polynomial time, 
except for line~\ref{algo-mod-line-compute-psi_k_d}, 
which relies on the integer factorization oracle to compute~$\nu_{\sigma}^k(g)$. 
So,~\Cref{algo:mod} runs in polynomial time with a factoring oracle, and~\modileslp is in $\ptime^\factoring$. 
\end{proof}

We can avoid appealing to the factorization oracle by 
providing~\Cref{algo:mod} with the set 
\begin{equation}
    \label{equation:PPsigmag}
    \PP(\sigma,g) \coloneqq \{ p \text{ prime} : p \text{ divides either $d(\sigma)$ or $\nu_\sigma^k(g)$, for some $k \in [0..n-2]$} \}.
\end{equation}
Observe that this set only contains polynomially many primes of polynomial bit size with respect to the sizes of $\sigma$ and $g$, since $\nu_\sigma^k(g) \leq d(\sigma)^k \cdot g$.

\begin{restatable}{lemma}{LemmaModInPtime}
    \label{lemma:mod-in-ptime}
    \Cref{algo:mod} runs in polynomial time, when provided with the set~$\PP(\sigma,g)$ (or any superset of this set) as an~additional~input.
\end{restatable}


\begin{proof}
    Following the proof of~\Cref{theorem:mod-in-p-factoring}, 
    the only line that requires the factoring oracle is line~\ref{algo-mod-line-compute-psi_k_d}.
    In particular, this line asks to compute $\nu_\sigma^k(g)$, for a value of $k$ that ranges in~$[0..n-1]$.
    By induction on $k \in [0..n-1]$, we show that knowing~$\PP(\sigma,g)$ suffices to compute this number in polynomial time. 

    \begin{description}
        \item[base case: $k = 0$.] Since $\nu_\sigma^0(g) = g$, and $g$ is part of the input, this case is trivial.
        \item[induction hypothesis.] For $k \geq 1$, the positive integer $\nu_{\sigma}^{k-1}(g)$ can be computed in polynomial time in the sizes of $\sigma$ and $g$, by relying on~$\PP(\sigma,g)$. 
        \item[induction step: $k \geq 1$.] By induction hypothesis, we can compute $\nu_{\sigma}^{k-1}(g)$ in polynomial time. Observe that $k-1 \leq n-2$, and therefore 
        all prime divisors of $\nu_{\sigma}^{k-1}(g)$ occur in~$\PP(\sigma,g)$. 
        Recall that this set has size polynomial in the sizes of $\sigma$ and $g$. 
        We iterate through~$\PP(\sigma,g)$ in order to find all prime divisors of $\nu_{\sigma}^{k-1}(g)$, 
        as well as those of $d(\sigma)$. With this set of primes at hand, we can efficiently compute the prime factorization of ${\odd(\nu_\sigma^{k-1}(g) \cdot d(\sigma))}$.
        Finally, we compute $\nu_\sigma^{k}(g) = \totient(\odd(\nu_\sigma^{k-1}(g) \cdot d(\sigma)))$ using~\Cref{equation:compute-totient-via-factorization}.
        \qedhere
    \end{description}
\end{proof}

\section{Computing an ILESLP representing \texorpdfstring{$x \bmod 2^y$}{x mod 2{\textasciicircum}y}}\label{computing-ileslp-xmod2y}

\begin{center}
    \vspace{5pt}
    {\def\arraystretch{1.2}
    \begin{tabular}{|rl|}
        \hline
        \multicolumn{2}{|c|}{\textsc{Computation of $x \bmod 2^y$}}\\
        \hline
        \textbf{Input:}& An ILESLP $\sigma$ and two of its variables $x$ and $y$.\\[-2pt]
        \textbf{Output:} & An ILESLP $\xi$ such that $\semlast{\xi} = \sem{\sigma}(x) \bmod 2^{\sem{\sigma}(y)}$\,.\\
        \hline
    \end{tabular}}
    \vspace{5pt}
\end{center}
\Cref{algo:slp-mod-2y} describes a procedure for solving the above problem.
It builds on~\Cref{algo:pos,algo:mod}, 
inheriting a polynomial running time given either a factoring oracle or access to the set 
$\PP(\sigma,\nu_\sigma(1))$.

\begin{algorithm}[t]
    \caption{An $\fptime^\factoring$ algorithm for computing $x\bmod 2^y$.}
    \label{algo:slp-mod-2y}
    \begin{algorithmic}[1]
      \Require 
      ILESLP $\sigma \coloneqq (x_0 \gets \rho_0, \dots, x_n \gets \rho_n)$, and two variables $x$ and $y$ in $\sigma$.
      \vspace{3pt}
      \If{$\sem{\sigma}(y) \leq 0$} \textbf{return} the ILESLP $(x_0 \gets 0)$
      \label{algo:slp-mod-2y:line-zero}
      \EndIf
      \vspace{2pt}
      \State \textbf{let} $a_0,\dots,a_{n-1} \in \Z$ such that~$\sem{\sigma}(x) = \frac{\sum_{j=0}^{n-1}a_j \cdot 2^{\sem{\sigma}(x_j)}}{d(\sigma)}$
      \label{algo:slp-mod-2y:line-def-E}
      \State $I \gets \{ j \in [0..n-1] : \sem{\sigma}(x_j) < \sem{\sigma}(y) \}$ 
      \Comment{each comparison resolved with~\Cref{algo:pos}}
      \label{algo:slp-mod-2y:line-def-I}
      \State \textbf{let} $S = \sum_{j \in I} a_j \cdot 2^{x_j}$ and $L = \sum_{j \in [0..n-1] \setminus I }a_j \cdot 2^{x_j}$ 
      \label{algo:slp-mod-2y:line-def-S-L}
      \vspace{2pt}
      \State $A \gets \sum_{j \in I} \abs{a_j}$
      \label{algo:slp-mod-2y:line-def-A}
      \State Perform binary search to find $q \in [-A..A]$
        satisfying~$0 \le \sem{\sigma}(S) - q \cdot 2^{\sem{\sigma}(y)} < 2^{\sem{\sigma}(y)}$\label{algo:slp-mod-2y:line-binary} 
      \Statex \Comment{each iteration of binary search uses~\Cref{algo:pos}}
      \State \textbf{let} $r$ be the residue of $\sem{\sigma}(L) \cdot 2^{-\sem{\sigma}(y)} + q$ modulo $d(\sigma)$%
      \label{algo:slp-mod-2y:line-r}%
      \Comment{uses~\Cref{algo:mod}}
      \State \textbf{return} an ILESLP~$\xi$ such that $\semlast{\xi} = \frac{1}{d(\sigma)} \cdot \big(\sem{\sigma}(S)+(r-q) \cdot 2^{\sem{\sigma}(y)} \big)$
      \label{algo:slp-mod-2y:line-return}

    
  

   
    \end{algorithmic}
  \end{algorithm}

\begin{lemma}\label{lemma:slp-mod-2y-in-P-factoring}
    Given an ILESLP~$\sigma$ and two of its variables $x$ and $y$, \Cref{algo:slp-mod-2y} returns an ILESLP~$\xi$ such that $\semlast{\xi} = \sem{\sigma}(x) \bmod 2^{\sem{\sigma}(y)}$. The algorithm runs in polynomial time with a factoring oracle.
\end{lemma}

\begin{proof}
    We analyze the runtime and correctness of the algorithm line by line, 
    providing the underlying intuition throughout. 
    Below, let $\sigma \coloneqq (x_0 \gets \rho_0, \dots, x_n \gets \rho_n)$.
    Recall that $d(\sigma)$ is a positive integer.

    \begin{description} 
        \item[line~\ref{algo:slp-mod-2y:line-zero}.] Over the reals, given $a,m \in \R$ with $m \neq 0$,
        $(a \bmod m)$ is defined as $a - m \cdot \floor{\frac{a}{m}}$. In particular, $(a \bmod 2^\ell) = 0$ whenever $\ell \leq 0$. Accordingly, line~\ref{algo:slp-mod-2y:line-zero} checks whether $\sem{\sigma}(y) \leq 0$, and in that case the algorithm returns an ILESLP that encodes the number $0$. 
        This line can be implemented in polynomial time by appealing to~\Cref{algo:pos}. 
        Below, assume $\sem{\sigma}(y) \geq 1$.

        \item[line \ref{algo:slp-mod-2y:line-def-E}.] 
        As done in~\Cref{algo:pos,algo:mod}, this line computes in polynomial time (\Cref{lemma:simple-expressions}) an expression $E \defeq \sum_{j=0}^{n-1}a_j \cdot 2^{x_j}$, where $a_0, \ldots, a_{n-1} \in \mathbb{Z}$, and $\sem{\sigma}(x) = \frac{E}{d(\sigma)}$.
    
        \item[lines \ref{algo:slp-mod-2y:line-def-I} and~\ref{algo:slp-mod-2y:line-def-S-L}.] 
        The algorithm sorts the monomial $a \cdot 2^{z}$ in the expression $E$ depending on the comparison $\sem{\sigma}(z) < \sem{\sigma}(y)$. 
        This is done by computing the set $I \subseteq [0..n-1]$ of indices~$j$ of variables $x_j$ such that $\sem{\sigma}(x_j) < \sem{\sigma}(y)$. Then, $E$ can be rearranged as $L + S$, where 
        $S \coloneqq \sum_{j \in I} a_j \cdot 2^{x_j}$ contains the exponentials~$2^z$ that are ``small'' comparatively to $2^{\sem{\sigma}(y)}$, and $L \coloneqq \sum_{j \in [0..n-1] \setminus I }a_j \cdot 2^{x_j}$ 
        contains those that are ``large''.
        In particular, $2^{\sem{\sigma}(y)}$ divides $\sem{\sigma}(L)$. 
        The set $I$ can be constructed in polynomial time, by 
        appealing $n$ times to~\Cref{algo:pos}.
        
        \item[line \ref{algo:slp-mod-2y:line-def-A}.] This line computes (in polynomial time) $A \coloneqq \sum_{j \in I} \abs{a_j}$. Observe that: 
        \begin{equation}
            \label{prof:xmod2y:eq1}
                \abs{\sem{\sigma}(S)}   
            =   \sum\nolimits_{j \in I} \abs{a_j} \cdot 2^{\sem{\sigma}(x_j)} \le A \cdot 2^{\sem{\sigma}(y)}.
        \end{equation}
        \item[line \ref{algo:slp-mod-2y:line-binary}.] 
        This line computes the quotient $q$ of the division of $\sem{\sigma}(S)$ by $2^{\sem{\sigma}(y)}$; 
        formally, the only integer satisfying $0 \le \sem{\sigma}(S) - q \cdot 2^{\sem{\sigma}(y)} < 2^{\sem{\sigma}(y)}$.
        By~\Cref{prof:xmod2y:eq1}, we know that $q \in [-A..A]$. 
        Since $A$ has bit size polynomial in the size of $\sigma$, we can compute $q$ in polynomial time 
        by performing binary search on the interval $[-A..A]$, appealing to~\Cref{algo:pos}.        
        For the sake of completeness, let us briefly explain how the search is implemented.
        Suppose knowing that the required $q$ belongs to $[\ell..u]$, where $l,u \in \mathbb{Z}$. 
        Initially, $[\ell..u] = [-A..A]$. Let $v \coloneqq \ceil{\frac{\ell+u}{2}}$. Then, 
        \begin{itemize}
        \item If $\sem{\sigma}(S) - v \cdot 2^{\sem{\sigma}(y)} < 0$, then we can restrict the search to $[\ell..v]$. 
        \item If $\sem{\sigma}(S) - v \cdot 2^{\sem{\sigma}(y)} > 2^{\sem{\sigma}(y)}$, then we can restrict the search to$[v..u]$. 
        \item If none of the previous two cases hold, then $v$ is the required $q$.
        \end{itemize}
        The conditions in the first two cases above are checked in polynomial time using~\Cref{algo:pos}. 
        Specifically, by following the operations in the expression $S - v \cdot 2^y$, 
        it is simple to extend the ILESLP~$\sigma$ into a new ILESLP~$\sigma'$ such that $\semlast{\sigma'} = \sem{\sigma}(S) - v \cdot 2^{\sem{\sigma}(y)}$. One can then apply~\Cref{algo:pos} to $\sigma'$ 
        to check the first of the two cases (the second case is handled similarly).

        From the definition of the expression $E$, we have: 
        \begin{equation} 
            \label{prof:xmod2y:eq2}
            \sem{\sigma}(x) = \frac{\sem{\sigma}(E)}{d(\sigma)} = \frac{\sem{\sigma}(L) + \sem{\sigma}(S)}{d(\sigma)} 
                = \frac{(\sem{\sigma}(L) + q\cdot 2^{\sem{\sigma}(y)}) + (\sem{\sigma}(S) - q\cdot 2^{\sem{\sigma}(y)}) }{d(\sigma)}.
        \end{equation} 

        \item[line \ref{algo:slp-mod-2y:line-r}.] This line computes the residue $r$ of $\sem{\sigma}(L) \cdot 2^{-\sem{\sigma}(y)} + q$ modulo $d(\sigma)$.
        This is done by constructing, in polynomial time, an ILESLP~$\sigma'$ encoding $\sem{\sigma}(L) \cdot 2^{-\sem{\sigma}(y)} + q$, and then calling~\Cref{algo:mod} on $\sigma'$ and $d(\sigma)$. 
        Thus, this line can be implemented in polynomial time with access to the factoring oracle 
        ---this is the only line of the algorithm requiring the oracle. 
        To construct $\sigma'$, recall that 
        the expression $L = \sum_{j \in [0..n-1] \setminus I }a_j \cdot 2^{x_j}$ 
        is such that $\sem{\sigma}(x_j) \geq \sem{\sigma}(y)$ for every $j \in [0..n-1] \setminus I$. 
        By following the operations in the expression ${L' \coloneqq \sum_{j \in [0..n-1] \setminus I }a_j \cdot 2^{x_j-y}}$, we can extend $\sigma$ into an ILESLP~$\sigma''$ 
        such that $\semlast{\sigma''} = \sem{\sigma}(L') = \sem{\sigma}(L) \cdot 2^{-\sem{\sigma}(y)}$. 
        Finally, $\sigma''$ can be further extended to produce the desired $\sigma'$.
        
        Following~\Cref{prof:xmod2y:eq2}, we see that:
        {\allowdisplaybreaks
        \begin{align} 
            \sem{\sigma}(x)
                &= \frac{(\sem{\sigma}(L) + q\cdot 2^{\sem{\sigma}(y)}) + (\sem{\sigma}(S) - q\cdot 2^{\sem{\sigma}(y)}) }{d(\sigma)}\notag\\
                &= \frac{(\sem{\sigma}(L) + q\cdot 2^{\sem{\sigma}(y)} - r\cdot 2^{\sem{\sigma}(y)} ) + (\sem{\sigma}(S) - q\cdot 2^{\sem{\sigma}(y)} + r\cdot 2^{\sem{\sigma}(y)})}{d(\sigma)}\notag\\
                &= \frac{\sem{\sigma}(L) + q\cdot 2^{\sem{\sigma}(y)} - r\cdot 2^{\sem{\sigma}(y)}}{d(\sigma)} + \frac{\sem{\sigma}(S) - q\cdot 2^{\sem{\sigma}(y)} + r\cdot 2^{\sem{\sigma}(y)}}{d(\sigma)} \notag\\ 
                &=\frac{\sem{\sigma}(L') + q - r}{{d(\sigma)}} \cdot 2^{\sem{\sigma}(y)} + \frac{\sem{\sigma}(S) + (r-q)\cdot 2^{\sem{\sigma}(y)}}{d(\sigma)}.
                \label{prof:xmod2y:eq3}
        \end{align}
        }
        By definition of $r$, 
        $\frac{\sem{\sigma}(L') + q - r}{{d(\sigma)}}$ is an integer. 
        Then, since $\sem{\sigma}(x)$ and $\sem{\sigma}(y)$ are both integers, and the latter is positive, 
        \Cref{prof:xmod2y:eq3} shows that $\ell \coloneqq \frac{\sem{\sigma}(S) + (r-q)\cdot 2^{\sem{\sigma}(y)}}{d(\sigma)}$ is an integer.

        \item[line~\ref{algo:slp-mod-2y:line-return}.] 
        From~\Cref{prof:xmod2y:eq3}, we conclude that $\sem{\sigma}(x) \bmod 2^{\sem{\sigma}(y)} = \ell \bmod 2^{\sem{\sigma}(y)}$.
        We will now show that $\ell \in [0..2^{\sem{\sigma}(y)}-1]$, which implies that $\ell$
        is in fact $\sem{\sigma}(x) \bmod 2^{\sem{\sigma}(y)}$.
        Accordingly, line~\ref{algo:slp-mod-2y:line-return} 
        of the algorithm constructs (and returns) an ILESLP~$\xi$ encoding $\ell$.
        Clearly, $\xi$ can be constructed in polynomial time 
        by extending $\sigma$, following the operations in the expression $\frac{1}{d(\sigma)} \cdot \big(S+(r-q) \cdot 2^y \big)$. 
        Recall that $(\sem{\sigma}(S) - q \cdot 2^{\sem{\sigma}(y)}) \in [0..2^{\sem{\sigma}(y)}-1]$
        and $r \in [0..d(\sigma)-1]$,
        by definition of $q$ and $r$, respectively. 
        Then,
        $0 \,\leq\, \sem{\sigma}(S) - q\cdot 2^{\sem{\sigma}(y)} \,\leq\,
                \sem{\sigma}(S) + (r-q) \cdot 2^{\sem{\sigma}(y)},$
        and $\sem{\sigma}(S) + (r-q) \cdot 2^{\sem{\sigma}(y)}
            \,<\,  2^{\sem{\sigma}(y)}  + r\cdot  2^{\sem{\sigma}(y)} \,\leq\, d(\sigma) \cdot 2^{\sem{\sigma}(y)}$.
        Therefore, by definition of $\ell$, 
        we conclude that $\ell \in [0..2^{\sem{\sigma}(y)}-1]$.
        \qedhere
    \end{description}
\end{proof}

Lastly, we show that $\sem{\sigma}(x) \bmod 2^{\sem{\sigma}(y)}$ is computable in polynomial time given~$\PP(\sigma,\nu_\sigma(1))$.

\begin{lemma}\label{lemma:slp-mod-2y-in-P-with-advice}
    \Cref{algo:slp-mod-2y} runs in polynomial time 
    when provided~$\PP(\sigma,\nu_{\sigma}(1))$ as an additional input.
\end{lemma}

\begin{proof}
    As explained during the proof of~\Cref{lemma:slp-mod-2y-in-P-factoring}, 
    only line~\ref{algo:slp-mod-2y:line-r} requires the factoring oracle. 
    This line requires computing the residue of 
    $\sem{\sigma}(L) \cdot 2^{-\sem{\sigma}(y)} + q$ modulo $d(\sigma)$, 
    where $L \coloneqq \sum_{j \in [0..n-1] \setminus I }a_j \cdot 2^{x_j}$
    defined in line~\ref{algo:slp-mod-2y:line-def-S-L} is such that 
    $\sem{\sigma}(x_j) \geq \sem{\sigma(y)}$ for every $j \in [0..n-1] \setminus I$.
    Therefore, 
    \begin{align*}
        \sem{\sigma}(L) \cdot 2^{-\sem{\sigma}(y)} + q 
        \ =\ \sum\nolimits_{j \in [0..n-1] \setminus I }a_j \cdot 2^{\sem{\sigma}(x_j)-\sem{\sigma}(y)}+q.
    \end{align*}
    Since all $a_j$ and $q$ have a bit size polynomial in the size of $\sigma$, 
    it suffices to show how to compute $2^{\sem{\sigma}(x_j)-\sem{\sigma}(y)} \bmod d(\sigma)$ 
    in polynomial time. The arguments are similar as those in~\Cref{sec:deciding-mod}. 

    Let $m \in \N$ be such that $d(\sigma) = 2^m \cdot \odd(d(\sigma))$. By the~CRT, $2^{\sem{\sigma}(x_j)-\sem{\sigma}(y)} \bmod d(\sigma)$ can be computed in polynomial time given the following two values:
    \begin{center}
        $b \coloneqq 2^{\sem{\sigma}(x_j)-\sem{\sigma}(y)} \bmod 2^m$
        \quad and \quad 
        $c \coloneqq 2^{\sem{\sigma}(x_j)-\sem{\sigma}(y)} \bmod \odd(d(\sigma))$,
    \end{center}
    The value $b$ can be computed in polynomial time by appealing to~\Cref{algo:pos}. 
    This is done as for the identically named value ``$b$'' in line~\ref{algo-mod-line-internal-loop-compute-a-ij} of~\Cref{algo:mod}; see the proof 
    of~\Cref{theorem:mod-in-p-factoring}.

    By definition, the set~$\PP(\sigma,\nu_\sigma(1))$ contains all prime factors of $d(\sigma)$. Therefore, we can compute $t \coloneqq \nu_\sigma(1) = \totient(\odd(d(\sigma)))$ 
    in polynomial time using~\Cref{equation:compute-totient-via-factorization}. 
    For obtaining the value $c$, we then first derive the residue $r \coloneqq (\sem{\sigma}(x_j) \bmod t)$ and the residue $s \coloneqq (\sem{\sigma}(y) \bmod t)$ in polynomial time using~\Cref{algo:mod}. 
    Afterwards, $c = {\big(2^{(r-s) \bmod t} \bmod \odd(d(\sigma))\big)}$ 
    is computed in polynomial time using the exponentiation-by-squaring method.
\end{proof}

\clearpage
\newcommand{\PartIIITitle}{On the complexity of ILEP}
\fancyhead[R]{{\color{gray}Part I: \PartIIITitle}}
\part{\PartIIITitle}\label{part:npocmp}
\addtocontents{toc}{This part builds on the results from the previous two parts in order to prove~\Cref{corollary:ILEP-in-npocmp}.\par}

We combine the results of the two previous parts of the paper 
to show~\Cref{corollary:ILEP-in-npocmp}, i.e., 
that the optimization problem for integer linear-exponential programs is in~\npocmp.
The first section of this part of the paper introduces the class~\npocmp.
The second section proves the corollary.

\section{The complexity class~\npocmp}\label{section:npocmp}

We briefly recall the notion of an optimization problem. 
An \emph{optimization problem}~$\mathcal{P}$ is characterized by a quintuple $(I,U,\sol,m,\goal)$ where:
\begin{itemize}[itemsep=3pt]
  \item $I$ is the set of instances of~$\mathcal{P}$,
  \item $U$ is a set (or, \emph{universe}) containing all possible solutions,
  \item $\sol \colon I \to 2^U$ assigns each input instance $x \in I$ to the set of its solutions $\sol(x)$,
  \item $m \colon I \times U \rightharpoonup \Z$ is the \emph{measure function}, a partial function defined for every $x \in I$ and $y \in \sol(x)$,
  \item $\goal \in \{\min, \max\}$ specifies a minimization or a maximization objective.
\end{itemize}
For $x \in I$, 
the set of \emph{optimal solutions} of $x$ is defined as
\[\opt(x) \coloneqq \{y \in \sol(x) : m(x, y) = \goal\{m(x, z) : z \in \sol(x)\}\}.\]
The computational task associated to $\mathcal{P}$ is the following: 
\begin{center}
\begin{tabular}{rl}
  \textbf{Input:} & An instance $x \in I$.\\
  \textbf{Output:} & An element~$y \in \opt(x)$ if $\opt(x) \neq \emptyset$, otherwise \texttt{reject}.
\end{tabular}
\end{center}
Below, we assume the elements of the sets $I$ and $U$ to be endowed with a notion of size $\abs{\cdot}$.
We define~\npocmp as the class of all optimization problems~$\mathcal{P} = (I,U,\sol,m,\goal)$ such that:
\begin{enumerate}[itemsep=3pt]
  \item\label{npocmp:sets-P} The sets $I$ and $U$ are recognizable in polynomial time.
  \item\label{npocmp:member} Given in input $x \in I$ and $y \in U$, checking $y \in \sol(x)$ is in~\ptime.
  \item\label{npocmp:m} $m$ is computable, and checking $m(x,y_1) \leq m(x,y_2)$, given $x \in I$ and $y_1,y_2 \in \sol(x)$, is in~\ptime.
  \item\label{npocmp:short} There is a polynomial $q \colon \N \to \N$ such that, for all $x \in I$, $\short(x) \coloneqq \{{y \in \sol(x)} :  {\abs{y} \leq q(\abs{x})}\}$ satisfies:
  \textit{(a)}\customlabel{4a}{npocmp:sol} 
    if $\sol(x) \neq \emptyset$ then $\short(x) \neq \emptyset$, and 
  \textit{(b)}\customlabel{4b}{npocmp:opt} 
    if $\opt(x) \neq \emptyset$ then $\opt(x) \cap \short(x) \neq \emptyset$.
  \item\label{npocmp:unb} Given an instance $x \in I$, deciding ${\sol(x) \neq \emptyset \land \opt(x) = \emptyset}$ is~in~\np.
\end{enumerate}
It is worth noting that some authors prefer replacing Properties~\eqref{npocmp:short} and~\eqref{npocmp:unb} above with the simpler 
\begin{enumerate}
  \item[4'.]\customlabel{4'}{npo:poly-sol} There is a polynomial $q \colon \N \to \N$ such that  $\abs{y} \leq q(\abs{x})$ for every $x \in I$ and $y \in \sol(x)$,
\end{enumerate}
which in particular implies the finiteness of $\sol(x)$~(see, e.g., the definition of~\npo in~\cite{Ausiello99}). 
The only difference between Properties~\eqref{npocmp:short} and~\eqref{npocmp:unb} and Property~\eqref{npo:poly-sol} lies in whether only small solutions are considered: 
if an optimization problem~$(I,U,\sol,m,\goal)$ is in~\npocmp, 
then the problem $(I,U,\short,m,\goal)$, where $\short$ is the function required by Property~\eqref{npocmp:short}, is also in~\npocmp and satisfies Property~\eqref{npo:poly-sol}.
In other words, Property~\eqref{npo:poly-sol} reflects the idea that only polynomial size solutions are reasonable solutions. 
Our rationale for preferring the more wordy Properties~\eqref{npocmp:short} and~\eqref{npocmp:unb} is that they provide a nice blueprint for 
organizing the results in the previous parts of the paper.
This modest goal is indeed the main purpose behind the class~\npocmp; 
as stated in the introduction, we make no presumption on the naturality of this class in a broader context.

Aside from the differences between Properties~\eqref{npocmp:short} and~\eqref{npocmp:unb} and Property~\eqref{npo:poly-sol},
starting from the definition of~\npocmp, one obtains the class \npo by 
replacing Property~\eqref{npocmp:m} with the stronger
\begin{enumerate}
  \item[3'.]\customlabel{3'}{npo:m} $m$ is computable in polynomial time, assuming a binary encoding for the integer in output.
\end{enumerate}
Therefore, every problem in \npo belongs to \npocmp. 

Expanding on the discussion in~\Cref{subsection:intro:comparing}, 
we see that, when $\mathcal{P}$ belongs to~\npo, Properties~\eqref{npo:m} 
and~\eqref{npo:poly-sol} ensure that, for any input $x \in I$, one can compute in polynomial time two integers~$a$ and~$b$ 
such that for every (short) solution $y \in \sol(x)$ we have ${m(x,y) \in [a..b]}$. 
(Implicitly, this step assumes knowing the polynomial $q$ in Property~\eqref{npo:poly-sol}, as well as a polynomial bounding the runtime~of~$m$.)
One can then search for the optimal solution by performing binary search: at each iteration, the interval~$[a..b]$ shrinks in half following the answer to the query $\exists y \in U : y \in \sol(x) \land m(x,y) \geq \frac{b-a}{2}$. By Properties~\eqref{npocmp:member}, \eqref{npo:m} and~\eqref{npo:poly-sol}, this query is solvable in~\np. 
Using this approach, it follows that~\npo problems can be solved by polynomial-time Turing machines with access to an~$\np$ oracle, that is, $\npo \subseteq \fptime^{\np}$. (In fact, $\npo  = \fptime^{\np}$ for a suitable model of computation characterizing~\npo, see~\cite{Krentel88,CrescenziP89}.)

In the case of~\npocmp, Property~\eqref{npocmp:opt} 
ensures that ${\{m(x,y) : y \in
\short(x)\}}$ is a set of exponentially many integers containing the optimal value for~$m$ (if one exists). However, \npocmp does not fix any representation on the integers returned by $m$ (we only know that one such representation exists, since $m$ is computable). Therefore, the size of these integers is unknown, and there is no guarantee that binary search can be performed on this set.
Instead of an inclusion within~$\fptime^{\np}$,~we have ${\npocmp \subseteq \fnp^{\np}}$. Indeed: a polynomial-time non-deterministic Turing machine with access to an~$\np$ oracle 
can solve an~\npocmp problem in the following simple way:
\begin{algorithmic}[1]
  \State Check that the input $x$ belongs to $I$; if not, \texttt{reject} 
  \Comment{In~\ptime by Property~\eqref{npocmp:sets-P}.}
  \State Query the \np oracle to determine if ${\sol(x) \neq \emptyset \land \opt(x) = \emptyset}$ holds; if the answer is \emph{yes},~\texttt{reject}
  \Statex \Comment{This query can be solved in~\np by~Property~\eqref{npocmp:unb}.}
  \State Guess a string $y$ of length $q(\abs{x})$, where $q$ is the polynomial in Property~\eqref{npocmp:short}
  \State Check $y \in U$ and $y \in \sol(x)$; if not, \texttt{reject} 
  \Comment{In~\ptime by Properties~\eqref{npocmp:sets-P} and~\eqref{npocmp:member}.}
  \State Query the \np oracle to determine if there exists~$z \in \short(x)$ such that $m(x,z) > m(x,y)$ (assuming $\goal = \max$); if the answer is~\emph{yes},~\texttt{reject}
  \Statex \Comment{This query can be solved in~\np because $\abs{z} \leq q(\abs{x})$, and checking whether $z \in \sol(x)$ and\Statex \hfill $m(x,z) > m(x,y)$ can be done in polynomial time by Properties~\eqref{npocmp:member} and~\eqref{npocmp:m}.}
  \State \textbf{return} $y$
\end{algorithmic}

\RestoreHeader
\section{ILEP is in~\npocmp}
\label{section:leslps}

We now prove that the optimization problem for integer linear-exponential programs is in~\npocmp (\Cref{corollary:ILEP-in-npocmp}). Let us first define the objects $I,U,\sol$ and $m$, noting that $\goal$ is simply $\min$ or $\max$:%
\begin{itemize}
    \item $I$ is the set of all pairs $(\tau, \phi)$ where $\tau$ is a linear-exponential term (the objective function) and~$\phi$ is an integer linear-exponential program. The size $\abs{(\tau,\phi)}$ of $(\tau,\phi) \in I$ is the sum of the sizes of $\tau$ and $\phi$.
    \item\label{def-npocmp:universe} $U \coloneqq \{(\sigma,\PP(\sigma)) : \sigma \text{ is a ILESLP}\}$, where $\PP(\sigma) \coloneqq \PP(\sigma,d(\sigma) \cdot \nu_\sigma(1))$ and
    \begin{itemize}
        \item $d(\sigma)$ is the product of all denominators occurring in rational constants of scaling expressions in $\sigma$ 
        (as defined at the beginning of~\Cref{part:deciding-properties-ILESLP});
        \item $\nu_\sigma$ is the function $\nu_\sigma(x) \coloneqq \totient(\odd(x \cdot d(\sigma)))$, as defined in~\Cref{sec:deciding-mod};
        \item $\PP(\sigma,g)$ is the set of primes defined in~\Cref{equation:PPsigmag} on page~\pageref{equation:PPsigmag}.
    \end{itemize}
    The size $\abs{(\sigma,\PP(\sigma))}$ of $(\sigma,\PP(\sigma)) \in U$ 
    is the sum of the bit sizes of $\sigma$ and $\PP(\sigma)$.
    \item Given $(\tau,\phi) \in I$, we define $\sol(\tau,\phi)$ as the set of all $(\sigma,\PP(\sigma)) \in U$ with the following property. Let~${\sigma = (x_0 \gets \rho_0,\, \dots\, ,\, x_n \gets \rho_n)}$. Then,
    \begin{enumerate}[label=(\alph*)]
        \item\label{def-sol:i1} the set $\{x_0,\dots,x_n\}$ contains (at least) all variables in $\tau$ and in~$\phi$;
        \item\label{def-sol:i2} each variable $x$ occurring in $\phi$ or $\tau$ is such that $\sem{\sigma}(x)\geq 0$;
        \item\label{def-sol:i3} the map assigning to each $x$ in $\phi$ the value~$\sem{\sigma}(x_i)$ is a solution of $\phi$.
    \end{enumerate}
    \item Given $(\tau,\phi) \in I$ and $(\sigma,\PP(\sigma)) \in \sol(\tau,\phi)$, we define $m((\tau,\phi),\sigma)$ as the integer $\tau(\sigma)$ obtained by
evaluating $\tau$, replacing each variable~$x$ occurring in it with $\sem{\sigma}(x)$.
\end{itemize}
Let us prove that these objects satisfy the five properties of~\npocmp. 

\paragraph*{Property~\eqref{npocmp:sets-P}.}
The set $I$ is clearly recognizable in polynomial time. 
We show that the same is true for the set~$U$ ---this is the content of~\Cref{theorem:U-recognition} (\Cref{subsection:intro:recognizing}):

\TheoremURecognition* 

\begin{proof}
    Consider a pair $(\sigma,S)$, where $S$ is a set of positive integers, and $\sigma \coloneqq (x_0 \gets \rho_0,\, \dots\, ,\, x_n \gets \rho_n)$ is a LESLP
    (both objects are clearly recognizable in polynomial time). We first check that  $S = \PP(\sigma)$.
    Recall that, given $g \in \N_{\geq 1}$, 
    $\PP(\sigma,g) \coloneqq \{ p \text{ prime} : p \text{ divides $d(\sigma)$ or $\nu_\sigma^k(g)$, for some $k \in [0..n-2]$} \}$;
    and $\PP(\sigma) = \PP(\sigma,d(\sigma) \cdot \nu_\sigma(1))$.
    Here,~$\nu_{\sigma}^k$~stands for the $k$th iterate of the function $\nu_\sigma$.
    To check $S = \PP(\sigma)$, we first
    we use the polynomial time algorithm for primality testing~\cite{AgrawalKS04} to verify that all elements of~$S$ are primes. 
    Afterwards, we check that these primes are exactly those appearing in the prime factorization of~$d(\sigma)$ or of numbers of the form $\nu^k_\sigma(d(\sigma) \cdot \nu_\sigma(1))$, 
    with $k \in [0..n-2]$. 
    For this second step, the arguments are similar to those in the proof of~\Cref{lemma:mod-in-ptime}. 
    Below, we give the pseudocode of a polynomial time procedure preforming this step:
    \begin{algorithmic}[1]
        \State \textbf{assert} $S$ contains all prime divisors of $d(\sigma)$
        \State compute $\nu_\sigma(1) = \totient(\odd(d(\sigma)))$ by relying on the prime divisors of $d(\sigma)$ \Comment{see~\Cref{equation:compute-totient-via-factorization}}
        \State $m_0 \gets d(\sigma) \cdot \nu_\sigma(1)$ 
        \Comment{$m_i = \nu_\sigma^{i}(d(\sigma) \cdot \nu_\sigma(1))$}
        \For{$k$ from $0$ to $n-2$}
            \State \textbf{assert} $S$ contains all prime divisors of $m_{k}$ 
            \If{$k \neq n-2$}
                \State compute $\nu_\sigma(m_{k}) = \totient(\odd(m_{k} \cdot d(\sigma)))$ by relying on the prime divisors of $d(\sigma)$ and $m_{k}$
                \State $m_{k+1} \gets \nu_\sigma(m_{k})$ 
            \EndIf
        \EndFor
        \State \textbf{assert} every prime in $S$ divide $\prod_{i=0}^{k-2} m_i$
        \State \textbf{return} true
    \end{algorithmic}

    After establishing $S = \PP(\sigma)$, 
    we determine 
    whether the LESLP~$\sigma$ is actually an ILESLP. 
    We recall the snippet of code from~\Cref{subsection:intro:recognizing}  
    that solves this problem: 
    \begin{algorithmic}[1]
        \For{$i = 1$ to $n$}
            \If{$\rho_i$ is of the form $2^x$}\label{snippet-v2:line2}
                \textbf{assert} $\sem{\sigma}(x) \geq 0$
            \EndIf
            \If{$\rho_i$ is of the form $\frac{m}{g} \cdot x$}\label{snippet-v2:line3}
                \textbf{assert} $\frac{g}{\gcd(m,g)}$ divides $\sem{\sigma}(x)$ 
            \EndIf 
        \EndFor
        \vspace{-5pt}
        \State \textbf{return} true
    \end{algorithmic}   
    By~\Cref{theorem:pos-in-ptime}, the condition in the \textbf{assert} command of line~\ref{snippet-v2:line2} can be checked in polynomial time.
    To show that the same is true for the \textbf{assert} command of line~\ref{snippet-v2:line3}, first observe that $\frac{g}{\gcd(m,g)}$ 
    is a divisor of~$d(\sigma)$. Then, the statement follows from~\Cref{lemma:mod-in-ptime}, as soon as we show the following claim: 
    \begin{claim}
        \label{claim:PPa-PPb}
        for every~$a,b \in \N_{\geq 1}$ such that $a$ is a divisor of $b$, 
        $\PP(\sigma,a) \subseteq \PP(\sigma,b)$.
    \end{claim}
    \begin{proof}[Proof of~\Cref{claim:PPa-PPb}.]
        It suffices to show that $\nu_\sigma^k(a)$ divides~$\nu_\sigma^k(b)$. The proof is by induction on $k$.
    \begin{description}
        \item[base case: $k = 0$.] Since $\nu_{\sigma}^0(x) = x$, the statements follows trivially.
        \item[induction hypothesis.] Given $k \geq 1$, $\nu_\sigma^{k-1}(a)$ divides $\nu_\sigma^{k-1}(b)$.
        \item[induction step: $k \geq 1$.] 
        By definition of $\odd$,
        $\odd(a \cdot d(\sigma))$ divides $\odd(b \cdot d(\sigma))$. 
        Moreover, one of the basic properties of 
        Euler's totient function~$\totient$ is that $\totient(g)$ divides $\totient(c)$ whenever $g$ divides~$c$ (this follows directly from the definition of~$\phi$). Hence, $\nu_\sigma(a)$ divides $\nu_\sigma(b)$. 
        By induction hypothesis, $\nu_{\sigma}^{k-1}(\nu_\sigma(a))$ divides $\nu_{\sigma}^{k-1}(\nu_\sigma(b))$; 
        in other words, $\nu_\sigma^{k}(a)$ divides $\nu_\sigma^{k}(b)$.
        \qedhere
    \end{description}
    \end{proof}
    \noindent
    This concludes the proof of~\Cref{theorem:U-recognition}.
\end{proof}

\paragraph*{Property~\eqref{npocmp:member}.} We start with an auxiliary lemma which we will also use to show Property~\eqref{npocmp:m}.

\begin{lemma}
    \label{lemma:from-term-to-ileslp} 
    There is a polynomial time procedure with the following specification: 
  \begin{center}  
    {\def\arraystretch{1.3}
    \begin{tabular}{rl}
        \textbf{Input:}&  $(\sigma,\PP(\sigma)) \in U$ and a linear-exponential term $\tau$ featuring variables~$X$ from $\sigma$.\\[-2pt]
        \textbf{Output:} &  An ILESLP $\xi$.
    \end{tabular}}
  \end{center}
  Under the assumption that $\sem{\sigma}(x) \geq 0$ for all $x \in X$, 
  the algorithm ensures that $\semlast{\xi}$ is the integer obtained 
  by evaluating~$\tau$ by replacing all $x \in X$ with~$\sem{\sigma}(x)$.
\end{lemma}

\begin{proof}
    Let $X = \{y_1,\dots,y_\ell\}$, and $\tau$ be the term 
    $\sum_{i=1}^\ell \big(a_i \cdot y_i + b_i \cdot 2^{y_i} + \sum_{j=1}^\ell c_{i,j} \cdot (y_i \bmod 2^{y_j})\big) + d$.
    Here is the pseudocode of the algorithm:
    \begin{algorithmic}
        \For{$(i,j) \in [1..\ell] \times [1..\ell]$}\label{eval-terms:for}
            \State \textbf{let} $\xi_{i,j}$ be the ILESLP computed by~\Cref{algo:slp-mod-2y} on input $(\sigma,y_i,y_j,\PP(\sigma))$
            \State Rename variables in $\xi_{i,j}$ to be distinct from those in $\sigma$; 
            \State \textbf{let} $z_{i,j}$ be the variable in the last assignment of $\xi_{i,j}$ 
            \Comment{the one encoding $\semlast{\xi_{i,j}}$}
            \State Extend $\sigma$ by appending the assignments in $\xi_{i,j}$
        \EndFor
        \State \textbf{return} an ILESLP for the expression $\sum_{i=1}^\ell \big(a_i \cdot y_i + b_i \cdot 2^{y_i} + \sum_{j=1}^\ell c_{i,j} \cdot z_{i,j}\big) + d$
        \label{eval-terms:last-line}
    \end{algorithmic}
    Each call to~\Cref{algo:slp-mod-2y} runs in polynomial time (\Cref{lemma:slp-mod-2y-in-P-with-advice} and~\Cref{claim:PPa-PPb}), and by~\Cref{lemma:slp-mod-2y-in-P-factoring}
    it produces an ILESLP $\xi_{i,j}$ with 
    $\semlast{\xi_{i,j}} = (\sem{\sigma}(y_i) \bmod 2^{\sem{\sigma}(y_j)})$. 
    Upon reaching line~\ref{eval-terms:last-line}, 
    the (augmented) ILESLP $\sigma$ is of polynomial size, 
    and contains not only the initial assignments, 
    but also all those added by the \textbf{for} loop of line~\ref{eval-terms:for}
    ---specifically, assignments to variables $z_{i,j}$ 
    satisfying $\sem{\sigma}(z_{i,j}) = (\sem{\sigma}(y_i) \bmod 2^{\sem{\sigma}(y_j)})$. The expression in line~\ref{eval-terms:last-line} 
    involves $O(\ell^2)$ additions, exponentiations, and multiplications by integer constants. So, line~\ref{eval-terms:last-line} can be implemented in polynomial time by appending, 
    to $\sigma$, a suitable sequence of assignments corresponding to these operations.
\end{proof}

The following proposition (first stated in~\Cref{subsection:intro:recognizing}) implies Property~\eqref{npocmp:member}.

\TheoremFastChecking*

\begin{proof}
    Checking condition \Cref{def-sol:i1} in the definition of $\sol$ can clearly be done in polynomial time; let us write $X$ for the 
    variables occurring in $\phi$ or $\tau$. For~\Cref{def-sol:i2}, 
    given a variable $x \in X$, 
    we can decide $\sem{\sigma}(x) \geq 0$ in polynomial time 
    by appealing to~\Cref{algo:pos} (see~\Cref{theorem:pos-in-ptime}).
    For~\Cref{def-sol:i3}, we need to check whether the map~$\nu$ assigning to each $x \in X$ the value $\sem{\sigma}(x)$ satisfies $\phi$. 
    Let $\tau \leq 0$ be an inequality in $\phi$ 
    (equalities $\tau = 0$ are treated analogously by viewing them 
    as conjunctions ${\tau \leq 0 \land -\tau \leq 0}$).
    By~\Cref{lemma:from-term-to-ileslp}, we can construct, in polynomial time,
    an ILESLP $\xi$ such that $\semlast{\xi}$ is the integer~$\tau(\sigma)$ obtained 
    by evaluating $\tau$ on~$\nu$. Next, we append an assignment ${y \gets -1 \cdot x}$ to $\xi$, where $x$ is the variable in the last assignment of $\xi$ (the one encoding $\semlast{\xi}$), and $y$ is a fresh variable. Then, $\tau(\sigma) \leq 0$ if and only if $\semlast{\xi} \geq 0$, 
    and we can check whether $\semlast{\xi} \geq 0$ in polynomial time 
    via~\Cref{algo:pos}.%
\end{proof}

\paragraph*{Property~\eqref{npocmp:m}.} 
The map $m$ is clearly computable. Consider an instance $(\tau,\phi) \in I$ and 
two of its solutions $(\sigma_1,\PP(\sigma_1)),(\sigma_2,\PP(\sigma_2)) \in \sol(\tau,\phi)$. An algorithm to decide $\tau(\sigma_1) \leq  \tau(\sigma_2)$ is the following: 
\begin{algorithmic}[1]
    \State\label{npocmp:m:line1} Construct a polynomial-size ILESLP~$\xi$ such that $\semlast{\xi} = \tau(\sigma_2) - \tau(\sigma_1)$
    \State\label{npocmp:m:line2} \textbf{return} true \textbf{if} $\semlast{\xi} \geq 0$ \textbf{else} false
\end{algorithmic}
The ILESLP~$\xi$ in line~\ref{npocmp:m:line1} is computed in polynomial time by relying on the algorithm in~\Cref{lemma:from-term-to-ileslp}.
The check~$\semlast{\xi} \geq 0$ is performed in polynomial time by 
appealing to~\Cref{algo:pos}. 

\paragraph*{Property~\eqref{npocmp:short}.}
By~\Cref{theorem:small-optimum} (proven in~\Cref{part:small-ILESLP}), 
if an instance $(\tau,\phi) \in I$ has an (optimal) solution, then it has one representable with a polynomial-size ILESLP $\sigma$. 
We remark that our proof of this theorem is constructive, meaning that it allows one to explicitly derive a suitable monotonic polynomial~$h_1$, such that, 
for every $(\tau,\phi) \in I$, the corresponding ILESLP $\sigma$ 
from~\Cref{theorem:small-optimum} has (when it exists) 
size bounded by $h_1(\abs{(\tau,\phi)})$. 
Moreover, the set of primes $\PP(\sigma)$ has size polynomial in $\sigma$ (see page~\pageref{equation:PPsigmag}, and note that the definition of $\PP(\sigma)$ is constructive). 
Let $h_2$ be a monotonic polynomial bounding the size of $\PP(\sigma)$ given the size of any ILESLP $\sigma$. Setting ${q(x) \coloneqq h_1(x) + h_2(h_1(x))}$ 
results in the polynomial required by Property~\eqref{npocmp:short}: 
given an instance $(\tau,\phi) \in I$, if an (optimal) solution exists, 
then there is $(\sigma,\PP(\sigma)) \in \sol(\tau,\phi)$ such that
$\sigma$ has size~$s$ bounded by $h_1(\abs{(\tau,\phi)})$, 
and $\PP(\sigma)$ has size bounded by $h_2(s) \leq h_1(h_1(\abs{(\tau,\phi)}))$. Then, $\abs{(\sigma,\PP(\sigma))} \leq q(\abs{(\tau,\phi)})$.

\paragraph*{Property~\eqref{npocmp:unb}.}
Informally, this property asks for an~\np procedure to check whether the input instance is \emph{unbounded}, that is, it has 
an infinite sequence of solutions in which the value of the objective function 
strictly increases (assuming $\goal = \max$; the argument we give is analogous for $\goal = \min$). We reason similarly to how this property is proven in ILP.
Let $q$ and $\short$ be the polynomial and function from Property~\eqref{npocmp:short}. 
From the above discussion, we have an explicit definition for $q$, 
and this polynomial is monotonic.
Let $(\tau,\phi) \in I$. 
We show that ${\sol(\tau,\phi) \neq \emptyset \land \opt(\tau,\phi) = \emptyset}$ 
holds if and only if the following integer linear-exponential program is feasible:
\begin{equation}
    \label{eq:very-large-solutions}
    \phi \,\land\, \tau \geq \onenorm{\tau} \cdot 2^{z_s} \,\land\, z_1 = 2^s \,\land\, \bigwedge\nolimits_{i=2}^{s} z_i = 2^{z_i}\,,
\end{equation}
where $s \coloneqq q(\abs{(\tau,\phi)})+3$ and $z_1,\dots,z_s$ are variables not occurring in $\phi$ or $\tau$. The size of this linear-exponential program is polynomial in $\abs{(\tau,\phi)}$, and its feasibility can be decided in~\np
by~\Cref{theorem:small-optimum} (or, alternatively, the original algorithm from~\cite{ChistikovMS24}); Property~\eqref{npocmp:unb} follows.
Below, let $\xi$ be the ILESLP ${\xi \coloneqq (z_{-1} \gets 0,\, z_0 \gets 2^{z_{-1}},\, z_1 \gets 2^s \cdot z_0,\, z_2 \gets 2^{z_1},\, \dots,\, z_s \gets 2^{z_s},\, z_{s+1} \gets \onenorm{\tau} \cdot z_s)}$, 
and observe that in any solution to~\Cref{eq:very-large-solutions}, 
the value taken by the term $\onenorm{\tau} \cdot 2^{z_s}$ is exactly $\semlast{\xi}$.

For the left-to-right direction of the double implication, supposes ${\sol(\tau,\phi) \neq \emptyset \land \opt(\tau,\phi) = \emptyset}$. As stated above, this means that there is an infinite sequence of solutions of $\phi$ in which the value of~$\tau$ strictly increases. 
Therefore, there is a solution for which the value of~$\tau$ 
exceeds~$\semlast{\xi}$, 
and this implies the feasibility of~\Cref{eq:very-large-solutions}.

For the right-to-left direction, we consider the contrapositive. 
Assume that either $\sol(\tau,\phi) = \emptyset$ or $\opt(\tau,\phi) \neq \emptyset$. If $\sol(\tau,\phi) = \emptyset$ then $\phi$ is infeasible and therefore so is~\Cref{eq:very-large-solutions}. 
If instead $\opt(\tau,\phi) \neq \emptyset$, then by Property~\eqref{npocmp:opt} 
we have $\opt(\tau,\phi) \cap \short(\tau,\phi) \neq \emptyset$.
To show that~\Cref{eq:very-large-solutions} is infeasible, 
it suffices to show that
for every $(\sigma,\PP(\sigma)) \in \short(\tau,\phi)$, 
the integer $\tau(\sigma)$ obtained by evaluating $\tau$ on $\sigma$ 
is strictly smaller than~$\semlast{\xi}$.
Let $\vec x \coloneqq (x_1,\dots,x_n)$ be the variables occurring in~$\tau$. 
By definition, $\tau(\vec x) \leq \onenorm{\tau}\cdot 2^{\max(x_1,\dots,x_n)}$ for all values given to~$\vec x$ among the natural numbers.
Consider then $(\sigma,\PP(\sigma))\in \short(\tau,\phi)$, 
with $\sigma = (y_0 \gets \rho_0,\dots,y_m \gets \rho_m)$.
By definition of $\short$, we have $m \leq q(\abs{(\tau,\phi)})$. 
We show that, for every $i \in [1..m]$, $\abs{\sem{\sigma}(y_i)} \leq \sem{\xi}(z_i)$.
Together with the fact that $\sem{\xi}(z_{j-1}) < \sem{\xi}(z_{j})$ for every $j \in [1..s]$, this implies $\tau(\sigma) \leq \onenorm{\tau}\cdot 2^{\semlast{\xi}}$, concluding the proof.
\begin{description}
    \item[base case: $i = 1$.] The expression $\rho_1$ has one of the following forms: $0$, $a \cdot y_0$ (for some $a \in \Q$), $y_0 + y_0$, or $2^{y_0}$.
    Since $\sem{\sigma}(y_0) = 0$, we have $\sem{\sigma}(y_1) \in \{0,1\}$. 
    On the other hand, $\sem{\xi}(z_1) = 2^s > 1$.
    \item[induction hypothesis.] Given $i \geq 2$, we have $\abs{\sem{\sigma}(y_{j})} \leq \sem{\xi}(z_{j})$ for every $j \in [1..i-1]$.
    \item[induction step: $i \geq 2$.] The expression $\rho_i$ has one of the following forms: $0$, $a \cdot y_j$, $y_j + y_k$ or $2^{y_j}$, where $j,k \in [1..i-1]$. Let $\ell \coloneqq \max\{\abs{\sem{\sigma}(y_j)} : j \in [1..i-1]\}$, and $k \in [1..i-1]$ be such that $\abs{\sem{\sigma(y_k)}} = \ell$. Then, $\sem{\sigma}(y_i) \leq \max(\abs{a} \cdot \ell,2^{\ell})$.
    We show that $\max(\abs{a} \cdot \ell,2^{\ell}) \leq \sem{\xi}(z_i)$:
    \begin{itemize}
        \item $2^{\ell} \leq \sem{\xi}(z_i)$\textbf{:} 
        By induction hypothesis, ${\ell = \abs{\sem{\xi}(y_k)} \leq \sem{\xi}(z_k)}$. By definition of $\xi$, $\sem{\xi}(z_{j-1}) < \sem{\xi}(z_{j})$ for every $j \in [1..s]$, and therefore
        $\sem{\xi}(z_k) \leq \sem{\xi}(z_{i-1})$. 
        Moreover, by definition $\sem{\xi}(z_i) = 2^{\sem{\xi}(z_{i-1})}$.
        Therefore, $2^{\ell} \leq 2^{\sem{\xi}(z_k)} \leq 2^{\sem{\xi}(z_{i-1})} = \sem{\xi}(z_i)$.
        \item $\abs{a} \cdot \ell \leq \sem{\xi}(z_i)$\textbf{:} 
        Since $0 \leq \ell \leq \sem{\xi}(z_{i-1})$ (from the previous point in the proof), it suffices to show $\abs{a} \cdot \sem{\xi}(z_{i-1}) \leq 2^{\sem{\xi}(z_{i-1})}$. This inequality is trivial for $a = 0$. 
        Else, by~\Cref{lemma:pos-analysis-constant}, 
        we see that the inequality is true as soon as $\sem{\xi}(z_{i-1}) \geq 4 \cdot \log_2(\abs{a})+8$.
        Observe that the bit size of $a$ is bounded by the bit size of $\sigma$, 
        and therefore $\log_2(\abs{a}) \leq q(\abs{(\tau,\phi)})$.
        By definition of $\xi$, we also have $\sem{\xi}(z_{i-1}) \geq \sem{\xi}(z_{1}) = 2^{q(\abs{(\tau,\phi)})+3}$. Then, 
        \[
                4 \cdot \log_2(\abs{a})+8 \leq 4 \cdot q(\abs{(\tau,\phi)})+8
                \leq 2^{q(\abs{(\tau,\phi)})+3} 
                \leq \sem{\xi}(z_{i-1}).
        \]
    \end{itemize}
\end{description}
This completes the proof of~\Cref{corollary:ILEP-in-npocmp}.

\clearpage
\appendix
\renewcommand{\subsectionmark}[1]{\markright{Appendix~\Alph{section}.\arabic{subsection}: #1}{}}
\renewcommand{\sectionmark}[1]{\markright{Appendix~\Alph{section}: #1}{}}
\fancyhead[R]{{\color{gray}Appendix A: The Sequential Squaring Assumption and {ILESLPs}}}
\part{Appendices}
\addtocontents{toc}{The appendices include additional material (\Cref{appendix:ssa,section:analysis-step-i-and-iii}) as well as complete proofs of those statements whose arguments were 
omitted or only outlined in the main text.\par}
\setboolean{appendix}{true}

\section{The Sequential Squaring Assumption and ILESLPs}
\label{appendix:ssa}

This appendix contains a detour on the \emph{time-lock puzzle} introduced in~\cite{Rivest96}, 
which we use to establish
a lower bound for the problem \modileslp 
from~\Cref{subsection:intro:recognizing} 
(or, equivalently, the problem of deciding if an LESLP is an ILESLP) in terms of a well-established cryptographic assumption.

\paragraph*{Basic number theory concepts for cryptography.}
A prime $p$ is said to be \emph{safe} whenever $\frac{p-1}{2}$ is also prime.
A number $b$ is a \emph{quadratic residue} modulo~$N$ whenever it is congruent to $r^2$ modulo~$N$, for some $r \in [0..N-1]$; if $b$ is also in~$[0..N-1]$, then it is a quadratic residue \emph{of}~$N$.
Given two distinct safe primes $p$ and $q$, 
the set of all quadratic residues modulo $N \coloneqq p \cdot q$ forms a multiplicative cyclic subgroup of order $\frac{(p-1) \cdot (q-1)}{4}$; the key point being that it is then possible 
to generate all quadratic residues of $N$ starting from any of them, but 
it is impossible to do so in time polynomial in the bit sizes of $p$ and $q$.
A function $f \colon \N \to (0,1)$ is said to be \emph{negligible} if for every $c \in \N$ there is an integer $M_c$ such that
$f(n) < \frac{1}{n^c}$ for every $n > M_c$.

\paragraph*{The time-lock puzzle.}
We give a brief description of the time-lock puzzle from~\cite{Rivest96}, 
referring the reader to that paper for a full account on the problem and its applications.
The objective is to encrypt a message $M$ in a way that gives not only strong guarantees 
on the minimum amount of time any adversary must spend to decrypt it, 
but also some (mild) guarantees on the maximum time a strong adversary would take.
As usual in the computational model of cryptography, adversaries are modelled as probabilistic polynomial-time Turing machines, and we moreover assume to know a reasonably tight upper bound $S$ on the number of squaring per second that these adversaries can perform, modulo any number. 
At our disposal, we also have a pair of symmetric-key cryptographic algorithms (\textsc{Encrypt},\textsc{Decrypt});
these algorithms are known to the adversary.
Besides minor changes that we will discuss later, \cite{Rivest96} proposes the following protocol for encrypting~$M$:

\begin{algorithmic}[1]
    \State $p,q$ $\gets$ two distinct safe primes such that $\frac{p-1}{2}$ and $\frac{q-1}{2}$ are both congruent to $3$ modulo~$8$\label{rivest:line1}
    \State $T \gets S \cdot t$
    \Comment{$t$: number of seconds the puzzle must last. Given in input with $M$ and $S$}
    \State generate a secret key $K \in [0..N-1]$ for the pair of algorithms (\textsc{Encrypt},\textsc{Decrypt})
    \State $C \gets \textsc{Encrypt}(K,M)$
    \State $E \gets 2^{T} \bmod \totient(N)$\label{rivest:line5}
    \Comment{use exponentiation-by-squaring method~\cite[Ch.~1.4]{BressoudW08}...}
    \State $D \gets (K + 2^{E}) \mod N$\label{rivest:line6}
    \Comment{...twice}
    \State \textbf{return} $(N,D,C,T)$
\end{algorithmic}

To solve the time-lock puzzle, an adversary must retrieve the message $M$. 
Except for trying to compute $K$ from $C$ ---which is infeasible, since secure symmetric-key cryptographic algorithms exists (the simplest of all being one-time pad)---
the only way for the adversary to retrieve $M$ is to extract $K$ from $D$, and then run $\textsc{Decrypt}(K,C)$.
That can be done by computing $y \coloneqq 2^{2^{T}} \bmod N$ via repeated squaring, to then subtract it from $D$ (modulo~$N$). The key cryptographic assumption implying the security of the time-lock puzzle thus focuses on the computation of $y$:

\begin{conjecture}[(Sequential Squaring Assumption)]
    \label{conjecture:bssa}
    There is a polynomial 
    $P \colon \N \to \N$ 
    such that for every probabilistic polynomial-time adversary 
    $\mathcal{A}$, 
    there exists a negligible function 
    ${\text{negl} \colon \N \to (0,1)}$ such that for all $\lambda \in \N$ (in unary): 
    \[ 
        \left| 
            \,\text{Pr}\hspace{-2.5pt}\left[
                \, b = b' : 
                \begin{array}{rl}
                    (p,q,N) &\gets \textsc{GenMod}(1^\lambda)\\
                    T & \xleftarrow{\$} P(2^{\lambda})\\
                    b &\xleftarrow{\$} \{0,1\}\\
                    \text{if } b = 0 &\text{then } y \coloneqq 2^{2^{T}} \bmod N\\
                    \text{if } b = 1 &\text{then } y \xleftarrow{\$} \mathbb{QR}_N\\
                    b' &\gets \mathcal{A}(1^\lambda,N,T,y)
                \end{array}
            \right]
            - \frac{1}{2}
        \,\right| 
        \leq \text{negl}(\lambda).
    \]
\end{conjecture}

In~\Cref{conjecture:bssa}, $\lambda$ is a security parameter that governs the bit sizes of the primes $p$ and $q$, the ``time limit'' $T$ of the time-puzzle, and the runtime of the adversary $\mathcal{A}$.
The function \textsc{GenMod} is a probabilistic polynomial-time algorithm that returns a triple $(p,q,N)$ 
where $p$ and $q$ are distinct safe primes such that $\frac{p-1}{2}$ and $\frac{q-1}{2}$ are congruent to $3$ modulo~$8$ (i.e., those computed in line~\ref{rivest:line1} of the protocol), and~$N \coloneqq p \cdot q$.
The arrow $\xleftarrow{\$}$ stands for uniform sampling.
In a nutshell, \Cref{conjecture:bssa} states that adversaries can only distinguish between a $y$ computed as $(2^{2^{T}} \bmod N)$ and one randomly sampled among the quadratic residues of $N$ (denoted $\mathbb{QR}_{N}$ above) with negligible probability.  

\begin{proposition}
    \Cref{conjecture:bssa} implies that  \modileslp is not in \bpp.
\end{proposition}

\begin{proof}
    Suppose that \modileslp is in BPP. A simple adversary~$\mathcal{A}$ breaking the binary sequential squaring assumption is defined as follows. Given the input $(1^\lambda,N,T,y)$, $\mathcal{A}$ constructs the ILESLP~$\sigma$: 
    \[
    x_0 \gets 0,\ \ x_1 \gets 2^{x_0},\ \ x_2 \gets T \cdot x_1,\ \ x_3 \gets 2^{x_2},\ \ x_4 \gets 2^{x_3},\ \ x_5 \gets -y \cdot x_1,\ \ x_6 \gets x_4 + x_5,
    \]
    which evaluates to $\semlast{\sigma} = 2^{2^{T}} - y$.
    The adversary then invokes the BPP algorithm for \modileslp with inputs $\sigma$ and $N$. 
    If the algorithm returns true, $\mathcal{A}$ outputs $0$; otherwise, it outputs $1$.
\end{proof}

\paragraph*{On the security of the time-lock puzzle.}
As mentioned above, the protocol (and thus the cryptographic assumption) 
considered in~\cite{Rivest96} differ very slightly from the one reported here.
In particular, the protocol in~\cite{Rivest96} allows in line~\ref{rivest:line5} to use exponentiation~$x^E$ instead of $2^E$, where $x \in [2..N-1]$ is randomly chosen.
Correspondingly, the cryptographic assumption in~\Cref{conjecture:bssa} would sample uniformly at random $x$ in $[2..N-1]$ to then define $y \coloneqq x^{2^{T}} \bmod N$ in the case of $b = 0$.
The fact that the protocol is believed to be secure then 
stems from the fact that~\factoring is not believed to be in \bpp and that, with high probability, 
the period of the sequence $x_0,x_1,x_2,\dots$, where $x_i \coloneqq x^{2^i} \bmod N$, is large 
comparatively to $N$.
(Note that we can assume $\textsc{GenMod}(1^\lambda)$ to return an $N$ in $\Omega(2^\lambda)$, hence the period of the sequence is also large comparatively to $T$.)
As remarked in~\cite{Rivest96}, we can fix $x=2$ as long as we guarantee this period to still be large. 
This is ensured by the constraints on the primes $p$ and $q$ imposed in line~\ref{rivest:line1} of the protocol (which are absent in~\cite{Rivest96}), as we explain below.

Denote by $\lambda \colon \N_{\geq 1} \to \N_{\geq 1}$ the Carmichael function. 
Given a positive integer $n$, this function returns the smallest positive integer $m$ 
such that $a^m \equiv 1 \pmod n$ holds for every $a$ coprime with~$n$.

\begin{theorem}[\cite{BlumBlumShub86}]
    \label{theorem:blumblumshub}
    Let  $N \coloneqq p \cdot q$, with $p$ and $q$ distinct safe primes such that $\frac{p-1}{2}$ and $\frac{q-1}{2}$ are congruent to $3$ modulo~$8$. The period of the sequence 
    $x_0,x_1,\dots$, where $x_i = 2^{2^i} \bmod N$, is $\lambda(\lambda(N))$.
\end{theorem}

\begin{proof}
    Denote by $\pi$ the period of the sequence $x_0,x_1,\dots$,
    and by $\text{ord}_{M}(x)$ the (multiplicative) order of $x$ modulo $M$ (assuming $x$ and $M$ coprime).
    From the definition of the Charmichael function one can show that $\frac{\lambda(N)}{2} = \frac{p-1}{2} \cdot \frac{q-1}{2}$. 
    In particular, since both $\frac{p-1}{2}$ and $\frac{q-1}{2}$ are odd, $\text{ord}_{\lambda(N)/2}(2)$ is well-defined.
    In~\cite{BlumBlumShub86} the following results are established:
    \begin{enumerate}
        \item\label{bbs:i1} Under the sole assumptions that $p$ and $q$ are distinct primes that are congruent to $3$ modulo~$4$, 
        and that $2$ is a quadratic residue modulo $N$,
        we have $\pi \divides \lambda(\lambda(N))$. (See~\cite[Theorem~6]{BlumBlumShub86}.)
        \item\label{bbs:i2} Under the same assumptions as~Item~\ref{bbs:i1}, and further assuming $\text{ord}_{\lambda(N)/2}(2) = \lambda(\lambda(N))$ and $\text{ord}_{N}(2) = \frac{\lambda(N)}{2}$, it holds that $\lambda(\lambda(N)) \divides \pi$. (See~\cite[Theorem~7]{BlumBlumShub86}.)
        \item\label{bbs:i3} Under the sole assumption that $p$ and $q$ are distinct safe primes that are congruent to $3$ modulo~$4$, 
            and that $2$ is a quadratic residue with respect to \emph{at most} one among $\frac{p-1}{2}$ and $\frac{q-1}{2}$, 
            then $\text{ord}_{\lambda(N)/2}(2) = \lambda(\lambda(N))$.  (See~\cite[Theorem~8]{BlumBlumShub86}.)
    \end{enumerate} 
    The theorem then follows as soon as we establish that all the above assumptions are covered 
    by ``$p$ and $q$ are distinct safe primes such that $\frac{p-1}{2}$ and $\frac{q-1}{2}$ are congruent to $3$ modulo~$8$'':
    \begin{description}
        \item[Assumption: {\rm``$p$ and $q$ are distinct are congruent to $3$ modulo~$4$''}.] This assumption is equivalent to asking $\frac{p-1}{2}$ and $\frac{q-1}{2}$ to be odd, which they are, since they are congruent to $3$ modulo~$8$.
        \item[Assumption: {\rm``$2$ is a quadratic residue modulo $N$''}.]
            First, 
            remark that $2$ is a quadratic residue modulo $p \cdot q$ 
            if and only if $2$ is a quadratic residue modulo $p$ and modulo $q$.
            The left-to-right direction of this double implication is trivial. 
            For the right-to-left direction, suppose $r^2 \equiv 2 \pmod p$ and $s^2 \equiv 2 \pmod q$.
            By the Chinese remainder theorem, there is $z \in [0..p \cdot q-1]$ such that $z \equiv r \pmod p$ and $z \equiv s \pmod q$. Then, $z^2 - 2$ is divisible by both $p$ and $q$,
            and so, by coprimality of $p$ and $q$, it is also divisible by $p \cdot q$; 
            i.e., $2$ is a quadratic residue of $p \cdot q$.

            Let us then show that $2$ is a quadratic residue modulo $p$ (same arguments for $q$). 
            The second supplement to the law of quadratic reciprocity (see the entry for~{\textit{``quadratic reciprocity, law of''}} in~\cite{Nelson08})
            states that  
            $2$ is a quadratic residue modulo $p$ if and only if ${p \equiv \pm 1 \pmod 8}$.
            Note that if $p \equiv 1 \pmod 8$, then $\frac{p-1}{2} \equiv 0 \pmod 4$ 
            and therefore $p$ cannot be a safe prime.
            The congruence ${p \equiv - 1 \pmod 8}$ is instead equivalent to ${\frac{p-1}{2} \equiv 3 \pmod 4}$ 
            (and there are safe primes satisfying these pairs of constraints, take e.g.~$p = 7$).
            Since we are assuming ${\frac{p-1}{2} \equiv 3 \pmod 8}$, we also have ${\frac{p-1}{2} \equiv 3 \pmod 4}$; as required.

            Observe that we have now shown that the assumptions of Item~\ref{bbs:i1} apply. 

        \item[Assumption: $\text{\rm ord}_N(2) = \frac{\lambda(N)}{2}$.] Since $p$ and $q$ are safe primes, the set of quadratic residues modulo~$N$ forms a multiplicative cyclic group of order $\frac{(p-1) \cdot (q-1)}{4} = \frac{\lambda(N)}{2}$; meaning in particular that all quadratic residues have the same order (i.e.,~$\frac{\lambda(N)}{2}$). 
        From the previous point, $2$ is a quadratic residue modulo $N$.
        \item[Assumption: {\rm ``$2$ is a quadratic residue with respect to at most one among $\frac{p-1}{2}$ and $\frac{q-1}{2}$''}.] 

            ~Again\\ from the second supplement to the law of quadratic reciprocity, 
            $2$ is a quadratic residue modulo $\frac{p-1}{2}$ if and only if $\frac{p-1}{2} \equiv \pm 1 \pmod{8}$.
            We are however imposing $\frac{p-1}{2} \equiv 3 \pmod{8}$, 
            so $2$ is not a quadratic residue modulo $\frac{p-1}{2}$ (nor modulo $\frac{q-1}{2}$).
        
            This is the last assumption we needed to show: it completes the assumptions in Item~\ref{bbs:i3} and, following the conclusion of that item, also the assumptions in~Item~\ref{bbs:i2}.
            \qedhere
    \end{description}
\end{proof}

One last point: \Cref{conjecture:bssa} requires~\textsc{GenMod} to run in probabilistic polynomial-time, and so we also need to check that the set of numbers~$p$ that are safe primes satisfying~${\frac{p-1}{2} \equiv 3 \pmod{8}}$ 
is not only infinite, but also not too sparse. However, already whether there are infinitely many safe primes is not known.
The usual (well-corroborated) assumption in cryptography is that among the first~$n$ integers, $\Omega\big(\frac{n}{\log(n)^2}\big)$ are safe primes. Following Dirichlet's theorem, it is natural to expect such a bound to hold also for the safe primes~$p$ satisfying $\frac{p-1}{2} \equiv 3 \pmod{8}$:

\begin{conjecture}
    \label{conjecture:safe-primes}
    Among the first $n \in \N_{\geq 1}$ positive integers,
    $\Omega\big(\frac{n}{\log(n)^2}\big)$ are safe primes numbers~$p$ that satisfy $\frac{p-1}{2} \equiv 3 \pmod{8}$.
\end{conjecture}

The bound in~\Cref{conjecture:safe-primes} is also implied by the well-known Bateman-Horn conjecture.
Taken together, \Cref{theorem:blumblumshub}, \Cref{conjecture:safe-primes}, and the assumption that \factoring is not in \bpp
provide a rationale for believing that~\Cref{conjecture:bssa} holds.

\begin{proof}[Proof that the Bateman-Horn conjecture implies~\Cref{conjecture:safe-primes}.]
    Define the maps~$f_1(x) \coloneqq 8 \cdot x + 3$ and ${f_2(x) \coloneqq 2 \cdot f_1(x)+1}$.
    Note that the subset of $[1..n]$ we are interested in is 
    \[
        S(n) \coloneqq {\{ f_2(i) : i \in \N \text{ and both } f_1(i) \text{ and } f_2(i) \text{ are prime}\}} \cap [1..n].
    \]
    We first assume Dickson's conjecture and show that it implies that $S(n)$ is infinite. 
    We will later appeal to Bateman--Horn conjecture (which implies Dickson's conjecture) to obtain an estimation on $\card{S(n)}$.
    Recall that Dickson's conjecture
    states that, for any given finite family~$f_1,\dots,f_k$ 
    of univariate functions $f_i(x) \coloneqq a_i \cdot x + b_i$, where $a_i \in \N_{\geq 1}$ and $b_i \in \Z$:
    \begin{enumerate}
        \item\label{dickson-case-1} there are infinitely many $n \in \N$ such that $f_1(n),\dots,f_k(n)$ are all primes, or 
        \item\label{dickson-case-2} there is a single integer $\alpha \geq 2$ dividing, for every $m \in \N$, the product $\prod_{i=1}^k f_i(m)$.
    \end{enumerate}
    We simply have to exclude the second of the two cases above, showing that there are two integers $m_1,m_2 \in \N$ such that $f_1(m_1) \cdot f_2(m_1)$ and $f_1(m_2) \cdot f_2(m_2)$ are coprime.
    This holds already for $m_1 = 1$ and $m_2 = 2$:
    $f_1(1) \cdot f_2(1) = 11 \cdot 23$
    and 
    $f_1(2) \cdot f_2(2) = 3 \cdot 13 \cdot 19$.

    Moving to the density estimation, since the first of the two cases in Dickson's conjecture applies, 
    the Bateman--Horn conjecture implies that there is a real number $C \in \R$ 
    dependent on $f_1$ and $f_2$ (and independent on $n$)
    such that $\card S(n) \geq C \cdot \int_{2}^{n} \frac{dt}{(\log t)^2}$; 
    i.e., $\card S(n)$ is in $\Omega(\frac{n}{(\log n)^2})$
    (in fact, the Bateman--Horn conjecture gives $\Theta(\frac{n}{(\log n)^2})$, but we only need a lower bound).
\end{proof}

\RestoreHeader

\section[The algorithm for deciding ILEP: Further information on Steps~I and~III]{The algorithm for deciding ILEP:\\ Further information on Steps~I and~III}%
\label{section:analysis-step-i-and-iii}

In this appendix, we provide a further information on the Steps I and III 
of the algorithm in~\cite{ChistikovMS24}. 
In particular, we import the pseudocode of these steps, 
as well as their complexity analysis. 
When appealing to formal statements from~\cite{ChistikovMS24}, we refer to the full version of the paper (as indicated in the corresponding bibliography entry). 

\paragraph*{Some additional notation.} 
Throughout this appendix, we sometimes write a divisibility constraint $d \divides \tau$ as $\tau \equiv_d 0$.
We need a few definitions from~\cite{ChistikovMS24}. 
Below, let $\theta$ be the ordering of exponentiated variables $\theta(\vec x) \coloneqq {2^{x_n} \geq
  2^{x_{n-1}} \geq \dots \geq 2^{x_0} = 1}$, for some ${n \geq
  1}$.

\begin{definition}[Quotient System]
  A~\emph{quotient system induced by}~$\theta$ 
  is a system~$\phi(\vec x, \vec q, \vec r)$
  of equalities, inequalities, and divisibility constraints $\tau \sim
  0$, where ${\sim} \in \{{<},{\leq},=,\equiv_d : d \geq 1\}$ and
  $\tau$ is an \emph{quotient term (induced by $\theta$)}, that is,
  a term of the form
  \begin{equation*}
    a \cdot 2^{x_n} + f(\vec q) \cdot 2^{x_{n-1}} + b \cdot x_
    {n-1} + \tau'(x_{0},\dots,x_{n-2}, \vec r)\,,
  \end{equation*}
  where $a,b \in \Z$, $f(\vec q)$ is a linear term on \emph{quotient variables}~$\vec q$, and
  $\tau'$ is a linear-exponential term in which the \emph{remainder} variables~$\vec r$ do not occur exponentiated. 
  Furthermore, for every remainder variable~$r$, the 
  quotient system $\phi$ features the inequalities $0 \leq r < 
  2^{x_{n-1}}$.
\end{definition}

\noindent
\emph{A disclaimer:} Observe that quotient systems can be syntactically equal to linear-exponential programs. In particular, this happens when every linear term~$f(\vec q)$ appearing in quotient terms is an integer. However,~\cite{ChistikovMS24} keeps the two types of objects somewhat separated, as if they are distinct ``types'' (in the sense of programming languages). That is, quotient systems are not linear-exponential programs. In the algorithm from~\cite{ChistikovMS24}, quotient systems only appear at the beginning of Step~I of the algorithm. To be more precise, let us look at the pseudocode of Step~I (\Cref{algo:step-i}). In input, this step takes an ordering~$\theta$, and a linear-exponential program with divisions $\phi$.
The \textbf{foreach} loop of line~\ref{algo:step-i:inner-loop} translates $\phi$ into a quotient system (adding new quotient and remainder variables). When the procedure reaches line~\ref{line:step-i-u}, the translation is complete. 
All other systems
constructed by the algorithm (in Step I, $\gamma$ and $\psi$) are linear-exponential programs.

\medskip
The above distinction between linear-exponential programs and quotient systems is important for the definition of 
\emph{least significant part} of a term (introduced in~\Cref{sec:putting-all-together} for linear-exponential programs, but restated below to avoid confusion). In this definition, quotient terms and linear-exponential terms are treated differently, and the definition becomes ill-formed if quotient systems are mistakenly regarded as linear-exponential programs.

\begin{definition}[Least Significant Part]
  The \emph{least significant part} of a term~$\tau$, with respect to the ordering $\theta$, is defined as follows:
  \begin{enumerate}
    \item If $\tau$ is a \emph{linear-exponential term} $a \cdot 2^{x_n} 
    + b \cdot x_n + \tau'$, with $\tau'$ linear-exponential term only featuring $x_n$ in remainders $(x \bmod 2^{y})$, its least 
    significant part 
    is the term $b \cdot x_n + \tau'$. 

    \item If $\tau$ is a \emph{quotient term} $a \cdot 2^{x_n} + f
    (\vec q) \cdot 2^{x_{n-1}} + b \cdot x_{n-1} + \tau'$, with $\tau'$ not featuring $x_n$ nor $x_{n-1}$, its least significant part is the 
    term $b \cdot x_{n-1} + \tau'$.
  \end{enumerate}
  Moreover, let $\phi$ be either a linear-exponential program with divisions or a quotient system. 
  We denote by $\lst(\phi,\theta)$ the following set of least significant terms: 
  \[
    \lst(\phi,\theta) = \left\{ \pm \rho : \begin{array}{l}
      \rho \text{ is the least significant part of a term } 
      \tau \text{ appearing in an (in)equality } \tau \sim 0 \\
      \text{of } \phi, \text{ with respect to } \theta
    \end{array} \right\}.
  \]
\end{definition}

An analogous distinction arises in the definition of~\emph{linear norm}:

\begin{definition}[Linear norm]
  The \emph{linear norm} $\linnorm{\tau}$ of a term $\tau$ is defined as follows: 
  \begin{enumerate}
    \item If $\tau$ is a \emph{linear-exponential term} $\sum_{i=1}^n \big(a_i \cdot x_i + b_i \cdot 2^{x_i} + \sum\nolimits_{j=1}^n c_{i,j} \cdot (x_i \bmod 2^{x_j})\big) + d$, 
    then $\linnorm{\tau} \coloneqq \max\{\abs{a_i},\abs{c_{i,j}} : i,j \in [1..n]\}$,
    i.e., it reflects the maximum absolute value among the coefficients of the linear terms $x_i$ and the remainder terms $(x_i \bmod 2^{x_j})$.
    
    \item If $\tau$ is a quotient term induced by $\theta$, of the form $\tau = a \cdot 2^{x_n} + f
    (\vec q) \cdot 2^{x_{n-1}} + b \cdot x_{n-1} + \tau'$, 
    then $\linnorm{\tau} \coloneqq \max(\abs{b},\linnorm{\tau'},\linnorm{f})$. 
    Note that $\linnorm{\tau}$ accounts for the coefficients of the variables $\vec q$.
  \end{enumerate}
  Moreover, let $\phi$ be either a linear-exponential program with divisions or a quotient system. The linear norm of $\phi$ is defined as 
  $\linnorm{\phi} \coloneqq \max\{\linnorm{\tau} : \tau \text{ is a term appearing in an (in)equality of $\phi$} \}$.
\end{definition}

The parameters $\card{\phi}$, $\onenorm{\phi}$ and $\fmod(\phi)$ extend instead trivially to quotient systems~$\phi$: 
\begin{itemize}
\item $\card{\phi}$ for the number of constraints (inequalities, equalities and divisibility constraints) in $\phi$;
\item $\fterms(\phi)$ for the set of all terms $\tau$ occurring in inequalities $\tau \leq 0$ or equalities $\tau = 0$ of $\phi$;
\item $\onenorm{\phi} \coloneqq \max\{\onenorm{\tau} : \tau \in \fterms(\phi)\}$. 

For a quotient term~$\tau = a \cdot 2^{x_n} + f(\vec q) \cdot 2^{x_{n-1}} + b \cdot x_{n-1} + \tau'$, 
we have $\onenorm{\tau} \coloneqq \abs{a} + \abs{b} + \onenorm{f} + \onenorm{\tau'}$.
\item $\fmod(\phi)$ is the least common multiple
of the divisors $d$ of the divisibility constraints $d \divides \tau$ of $\phi$.

We recall that, in these constraints, all integers appearing in $\tau$ belong to $[0..d-1]$. 
\end{itemize}

\subsection{Step I}
\label{subsection:proof-step-i}
\label{subsection:complexity-step-i}
\Cref{algo:step-i} presents the pseudocode of Step~I. 
This pseudocode is obtained by merging (exclusively to simplify the presentation)
lines 4--14 of Algorithm~2 with lines 1--21 of Algorithm~3 from \cite{ChistikovMS24}. 
The full specification of \Cref{algo:step-i} is recalled below.

\begin{algorithm}
  \caption{Step I of the algorithm from~\cite{ChistikovMS24}. See~\Cref{lemma:CMS:first-step} for its full specification.}
  \label{algo:step-i}
  \begin{algorithmic}[1]
    \Require
      {\hspace{-17pt}\setlength{\tabcolsep}{0pt}
      \begin{tabular}[t]{rl}
        $\theta(\vec x):{}$ \ & ordering of exponentiated variables;\\
        & \textit{[Below, let~$2^x$ and $2^y$ be the largest and second-largest terms in this ordering, and}\\
        & \textit{let $\vec y$ be the vector obtained by removing $x$ from $\vec x$.]}\\
        $\phi(\vec x, \vec r):{}$ \ &
        linear-exponential program with divisions, implying $\vec r < 2^x$.\\ 
        &Variables $\vec r$ do not occur in exponentials.\\
      \end{tabular}}
    \NDBranchOutput 
    \Statex
    {\setlength{\tabcolsep}{0pt}
    \begin{tabular}[t]{rp{13.4cm}}
    $\gamma_{\beta}(q_x,\vec q, u):{}$ \ &%
    linear program with divisions; \\
    $\psi_{\beta}(\vec y, r_x,\vec r'):{}$ \ &%
      linear-exponential program with divisions, implying $r_x < 2^y \land \vec r' < 2^y$.\\
      &Variables $r_x$ and $\vec r'$ do not occur in exponentials.
    \end{tabular}}
    \medskip 
  \State\label{algo:step-i:large-mod-sub}%
    $\phi\gets\phi\sub{w}{(w \bmod 2^x):
    \text{$w$ is a variable}}$
    \label{algo:step-i:def-z}
    \ForEach{$r$ in $\vec r \cup \{x\}$ }
      \Comment{translate $\phi$ into a quotient system induced by~$\theta$}
      \label{algo:step-i:inner-loop}
      \State\label{algo:step-i:inner-loop-add-vars} 
      \textbf{let} $q_r$ and $r'$ be two fresh variables
      \State\label{algo:step-i:imposezlessy}%
      $\phi\gets\phi\land(0\leq r'< 2^y)$
      \State\label{algo:step-i:elimmod}%
        $\phi\gets\phi\sub{r'}{(r \bmod 2^y)}$
      \State\label{algo:step-i:simpmod}%
        $\phi\gets\phi\sub{(r'\bmod 2^w)}{(r \bmod 2^w):
        \text{$w$ is such that $\theta$ implies $2^w \leq 2^y$}}$
      \State $\phi\gets\phi\sub{(q_r\cdot 2^y + r')}{r}$
      \Comment{replaces only the linear occurrences of $r$}
      \label{algo:step-i:linear-substitution}
      \If{$r$ is $x$} 
        $(q_x,r_x) \gets (q_r,r')$
        \Comment{the substitution of $x$ in exponentials is delayed} 
      \EndIf\label{algo:step-i:end-inner-loop}
    \EndFor
    \State\textbf{let} $u$ be a fresh variable
    \label{line:step-i-u}
    \Comment{$u$ is an alias for $2^{x-y}$}
    \State $\gamma\gets\top;\ \psi\gets\top$
    \label{line:step-i-initialize}
    \Comment{new linear-exponential programs constructed from $\phi$}
    \tikzmark{quotient-1-elim-begins}
    \State $\Delta\gets\varnothing$
    \label{line:step-i-delta}
    \Comment{map from linear-exponential terms to $\Z$}\;
    \tikzmark{step-i-decouple-begin}
    \ForEach{ $(\tau \sim 0)$ in $\phi$, where ${\sim} \in \bigl\{=,<,\leq,\equiv_d : d \geq 1\bigr\}$}
    \label{line:step-i-foreach}
      \State \textbf{let} $\tau$ be $(a \cdot 2^x + f(\vec x') \cdot 2^y + \rho)$, where $\rho$ is the least significant part of $\tau$
      \label{line:step-i-tau}
      \If{$a=0$ and $f(\vec x')$ is an integer}
        $\psi\gets\psi\land(\tau\sim0)$ 
        \label{line:step-i-psi-trivial}
      \ElsIf{the symbol $\sim$ belongs to $\{{=},{<},{\leq}\}$}
        \label{line:step-i-second-case}
        \If{$\Delta(\rho)$ is undefined}
        \label{line:step-i-delta-undefined}
          \State \myguess $h \gets$ integer in $[-\onenorm{\rho},\onenorm{\rho}]$
          \label{line:step-i-guess-rho}
          \State $\psi\gets\psi\land((h-1)\cdot2^y<\rho)\land(\rho\leq h\cdot2^y)$
          \label{line:step-i-psi-inequality}
          \State update $\Delta$ : add the key--value pair $(\rho, h)$
          \label{line:step-i-delta-update}
        \EndIf
        \State $h \gets \Delta(\rho)$
        \label{line:step-i-delta-from}
        \If{the symbol $\sim$ is $<$}
        \label{line:step-i-if-less}
            \State \myguess ${\sim^{\prime}} \gets$ sign in $\{=,<\}$
            \label{line:step-i-guess-sign} 
            \State $\psi\gets \psi\land (\rho \sim^{\prime} h\cdot 2^y)$ 
            \label{line:step-i-strict-add-to-psi} 
            \State ${\sim} \gets\,\leq$ 
            \label{line:step-i-strict-ineq}
            \If{the symbol $\sim^{\prime}$ is $=$} $h\gets h+1$
            \EndIf
            \label{line:step-i-update-r}
        \EndIf
        \State $\gamma\gets\gamma\land(a \cdot u + f(\vec x')+h\sim0)$
        \label{line:step-i-gamma-inequality}
        \If{the symbol $\sim$ is $=$}
          $\psi\gets\psi\land(h\cdot2^y=\rho)$
          \label{line:step-i-psi-equality}
        \EndIf
      \Else
        \label{line:step-i-third-case}
        \Comment{$\sim$ is $\equiv_d$ for some $d \in \N$}
        \State \myguess $h\gets$ integer in $[1,\emph{mod}(\phi)]$
        \label{line:step-i-guess-mod}
        \State $\gamma\gets\gamma\land(a \cdot u + f(\vec x') -h\sim0)$
        \label{line:step-i-gamma-divisibility}
        \State $\psi\gets\psi\land(h\cdot2^y+\rho\sim0)$
        \label{line:step-i-psi-divisibility}
      \EndIf	
    \EndFor
    \tikzmark{step-i-decouple-end}
    \State\label{algo:step-i:return}%
              \textbf{return }$(\gamma,\psi)$
  \end{algorithmic}
  \AddNote{step-i-decouple-begin}{step-i-decouple-end}{left-margin}{Decouple $\vec q$, $q_z$ and $u$ from the other variables}
\end{algorithm}%

\CMSFirstStep*

\begin{proof}
   The correctness of \Cref{algo:step-i} follows from~\cite[Proposition 4 and Lemma 23]{ChistikovMS24}:

   \begin{itemize}
    \item \textbf{Lines~\ref{algo:step-i:large-mod-sub}--\ref{algo:step-i:end-inner-loop}.} 
    These lines are analyzed in the proof of~\cite[Proposition~4]{ChistikovMS24} (Appendix~C.3, page~54). The~\textbf{foreach} loop of line~\ref{algo:step-i:inner-loop} (deterministically) manipulates $\phi$ 
    into a quotient system $\phi'$. Let $\vec q$ and $\vec r'$ be the set of all fresh quotient and reminder variables introduced in line~\ref{algo:step-i:inner-loop-add-vars} (across all iterations of the loop). This part of the algorithm ensure that 
    \begin{equation}\label{eq:step-i-part-1-spec}
      \left[\begin{matrix}
        x\\ 
        \vec r
        \end{matrix}\right]
        = \left[\begin{matrix}
            q_x\\ 
            \vec q
        \end{matrix}\right] \cdot 2^y + \left[\begin{matrix}
            r_x\\ 
            \vec r'
        \end{matrix}\right]
        \text{ \ implies \ }
        (\phi \land \theta) \iff (\phi' \land \theta)
    \end{equation}

    In~\cite{ChistikovMS24}, lines 4--14 of Algorithm~2 
    construct $\phi'$ to then pass it to Algorithm~3, which performs (with respect to our pseudocode) the following lines~\ref{line:step-i-u}--\ref{line:step-i-psi-divisibility}.

    \item \textbf{ Lines~\ref{line:step-i-u}--\ref{line:step-i-psi-divisibility}.} 
    These lines are analyzed in the proof of~\cite[Lemma 23]{ChistikovMS24} (Appendix~C.2, page~40; see in particular the subsection titled ``Correctness of Step~(i)'').
    In a nutshell, these lines ``divide'' each quotient term 
    in $\phi'$ by $2^y$, by relying on the equivalences in~\Cref{lemma:split:inequalities}. For example, an equality 
    ${a \cdot 2^x + f(\vec q) \cdot 2^{y} + b \cdot y + \tau' = 0}$ 
    is (non-deterministically) rewritten as 
    ${a \cdot 2^{x-y} + f(\vec q) + r = 0} \land {b \cdot y + \tau' = r \cdot 2^y}$.
    Note that $b \cdot y + \tau'$ is the least significant part of the term in the initial equality.
    Constraints concerning these least significant parts (in our example, $b \cdot y + \tau' = r \cdot 2^y$) are added to the formula $\psi$ (lines~\ref{line:step-i-psi-inequality},~\ref{line:step-i-strict-add-to-psi},~\ref{line:step-i-psi-equality} and~\ref{line:step-i-psi-divisibility}), whereas the remaining constraints 
    featuring the variable $x$ (in our example, ${a \cdot 2^{x-y} + f(\vec q) + r = 0}$) are added to the formula $\gamma$ 
    (lines~\ref{line:step-i-gamma-inequality} and~\ref{line:step-i-gamma-divisibility}; note that the algorithm uses $u$ as a proxy for $2^{x-y}$).

    Since these lines of the algorithm are guided by the equivalences in~\Cref{lemma:split:inequalities}, one obtains 
    \begin{equation}\label{eq:step-i-part-2-spec}
      \left[\begin{matrix}
        x\\ 
        \vec r
        \end{matrix}\right]
        = \left[\begin{matrix}
            q_x\\ 
            \vec q
        \end{matrix}\right] \cdot 2^y + \left[\begin{matrix}
            r_x\\ 
            \vec r'
        \end{matrix}\right]
        \text{ \ implies \ }
        (\phi' \land \theta) \iff \bigvee_{\beta} \exists u\, \big(\gamma_{\beta} \land \psi_{\beta} \land (u = 2^{x-y}) \land \theta\big),
    \end{equation}
    where $\bigvee_\beta$ ranges over all non-deterministic branches $\beta$. 
    \end{itemize}
    The lemma follows by~\Cref{eq:step-i-part-1-spec,eq:step-i-part-2-spec}. 
    Consider a solution to $\phi \land \theta$.
    We can uniquely decompose the values $x$, $y$ and $\vec r$ take in this solution by following the system $\left[\begin{matrix}
      x\\ 
      \vec r
      \end{matrix}\right]
      = \left[\begin{matrix}
          q_x\\ 
          \vec q
      \end{matrix}\right] \cdot 2^y + \left[\begin{matrix}
          r_x\\ 
          \vec r'
      \end{matrix}\right]$ and the equation $u = 2^{x-y}$,
      producing a solution to $\bigvee_{\beta} \big(\gamma_{\beta} \land \psi_{\beta} \land (u = 2^{x-y}) \land (x = q_x \cdot 2^y + r_x) \land \theta\big)$ 
      thanks to~\Cref{eq:step-i-part-1-spec,eq:step-i-part-2-spec}. 
      Conversely, from a solution to $\bigvee_{\beta} \big(\gamma_{\beta} \land \psi_{\beta} \land {(u = 2^{x-y})} \land {(x = q_x \cdot 2^y + r_x)} \land \theta\big)$, 
      we can compute (unique) values for the variables $\vec r$ following the system 
      $\vec r = \vec q \cdot 2^y + \vec r'$, producing 
      a solution to $\theta \land \phi$.
\end{proof}

We now move to the complexity of~\Cref{algo:step-i}:

\LemmaComplexityOfStepI*

\begin{proof}
    We report the complexity analysis from~\cite[Lemma~6 and Lemma~37]{ChistikovMS24}:
    \begin{itemize}
      \item \textbf{Lines~\ref{algo:step-i:large-mod-sub}--\ref{algo:step-i:end-inner-loop}.} These lines 
      are analyzed in~\cite[Lemma~6]{ChistikovMS24} (Appendix~D.3, page~63).
      As done in the sketch of the proof of~\Cref{lemma:CMS:first-step}, 
      let us write $\phi'$ for the quotient system obtained from $\phi(\vec x, \vec r)$ 
      after executing the \textbf{foreach} loop of line~\ref{algo:step-i:inner-loop}. 
      \cite[Lemma~6]{ChistikovMS24} established the following bounds on $\phi'$:
      \begin{equation}
        \label{eq:complexity-step-i:t1}
        \text{if \ } 
        \begin{cases}
          \card \lst(\phi,\theta) \cand \leq \ell\\
          \card \phi              \cand \leq s\\
          \linnorm{\phi}          \cand \leq a\\ 
          \onenorm{\phi}          \cand \leq c\\
          \fmod(\phi)             \cand \hspace{3pt}\divides\hspace{2pt} d
        \end{cases}
        \text{ \ then \ }  
        \begin{cases}
          \card \lst(\phi',\theta)  \cand \leq \ell + 2 \cdot k\\
          \card \phi'               \cand \leq s + 2 \cdot k\\
          \linnorm{\phi'}           \cand \leq 3 \cdot a\\ 
          \onenorm{\phi'}           \cand \leq 2 \cdot (c + 1)\\
          \fmod(\phi')              \cand \hspace{3pt}\divides\hspace{2pt} d\\
        \end{cases}
      \end{equation}
    
    
      \item \textbf{Lines \ref{line:step-i-u}--\ref
      {line:step-i-psi-divisibility}.} These lines are analyzed in~\cite[Lemma~37]{ChistikovMS24} (Appendix D.2, page 59; see in particular the subsection titled ``Step (a)''). 
      Let $(\gamma,\psi)$ be an output of~\Cref{algo:step-i}.
      The following bounds are established:
      \begin{equation}
        \label{eq:complexity-step-i:t2}
        \hspace{-1cm}
        \text{if \ } 
        \begin{cases}
          \card \lst(\phi',\theta)  \cand \leq \ell\\
          \card \phi'               \cand \leq s\\
          \linnorm{\phi'}           \cand \leq a\\ 
          \onenorm{\phi'}           \cand \leq c\\
          \fmod(\phi')              \cand \hspace{3pt}\divides\hspace{2pt} d
        \end{cases}
        \text{ \ then \ }  
        \begin{cases}
          \card \lst(\psi,\theta')  \cand \leq \ell\\
          \card \psi                \cand \leq s + 2 \cdot \ell\\
          \linnorm{\psi}            \cand \leq a\\ 
          \onenorm{\psi}            \cand \leq 2 \cdot c + 1\\
          \fmod(\psi)               \cand \hspace{3pt}\divides\hspace{2pt} d\\
        \end{cases}\text{ \ and \ }
        \begin{cases}
          \card \gamma                  \cand \leq s\\
          \linnorm{\gamma\sub{2^u}{u}}  \cand \leq a\\ 
          \onenorm{\gamma}              \cand \leq c+1\\
          \fmod(\gamma)                 \cand \hspace{3pt}\divides\hspace{2pt} d
        \end{cases}
      \end{equation}
      Note that above (and in the statement of the lemma) $\linnorm{\gamma\sub{2^u}{u}}$ 
      is considered instead of~$\linnorm{\gamma}$. As reported in~\cite[Lemma~37]{ChistikovMS24},
      this is because only the coefficients of the variables distinct from $u$ are interesting 
      for the overall analysis of the complexity of the algorithm (in any case, we have a bound of $c+1$ on this coefficient, given by $\onenorm{\gamma}$).  
      Performing the substitution $\sub{2^u}{u}$ makes it so that the coefficients of $u$ 
      are not accounted when computing~$\linnorm{\cdot}$.

    \end{itemize}

    The bounds in the statement of the lemma are obtained by imply conjoining the bounds 
    in~\Cref{eq:complexity-step-i:t1,eq:complexity-step-i:t2}. 
    The fact that~\Cref{algo:step-i} runs in non-deterministic polynomial time 
    follows again directly from~\cite[Lemma 6 and Lemma 37]{ChistikovMS24}.
\end{proof}


\subsection{Step III}\label{subsection:proof-step-iii}\label{subsection:complexity-step-iii}

\Cref{algo:step-iii} presents the pseudocode of Step III. 
It corresponds to lines~24--34 of 
Algorithm~3 from~\cite{ChistikovMS24}. Note that line~24 of that algorithm calls a procedure named~\textsc{SolvePrimitive}; in our pseudocode, this line is replaced directly with the code of~\textsc{SolvePrimitive} 
(i.e., with lines~\mbox{1--15} of Algorithm~4 from \cite{ChistikovMS24}).

\begin{algorithm}
    \caption{Step III of the algorithm from~\cite{ChistikovMS24}. See~\Cref{lemma:CMS:third-step} for its full specification.}\label{algo:step-iii}
    \setstretch{1.1}
    \begin{algorithmic}[1]
      \Require
      {\setlength{\tabcolsep}{0pt}
      \begin{tabular}[t]{rp{10cm}}
            $\gamma'(q_x,u):$ &\ 
            linear program with divisions.
      \end{tabular}}
      \NDBranchOutput 
      \begin{minipage}[t]{0.94\linewidth}
              \begin{tabular}[t]{rcp{0.8\linewidth}}
              $\gamma_{\beta}''(q_x)$&:& linear program with divisions;\\ 
              $\psi_{\beta}''(y, r_x)$&:& linear-exponential program with divisions.
              \end{tabular}
    \end{minipage}
        
    \medskip
    \Statex \Comment{Recall: the procedure assumes both $u = 2^{x-y}$ and $x = q_x \cdot 2^{y}+r_x$ to hold (see~\Cref{lemma:CMS:third-step})}
    \State \textbf{let} $\gamma'$ be $(\chi\land\phi)$, where $\chi$ is the conjunction of all (in)equalities from $\gamma'$ containing $u$
    \label{line:step-iii-decompose}
    \State $(d,n)\gets$ pair of non-negative integers such that $\fmod(\gamma')=d\cdot2^n$ and $d$ is odd
    \label{line:step-iii-mod-factor}
    \State $C \gets \max\big\{n, 3 + 2 \cdot\bigl\lceil\log(\frac{|b|+|c|+1}{|a|})\bigr\rceil : \text{$(a\cdot u + b\cdot q_x + c\sim0)$ in $\chi$, where ${\sim} \in \{=,<,\leq\}$} \big\}$
    \label{line:step-iii-max-constant}
    \State \myguess $c\gets{}$element of $[0..C-1] \cup \{\star\}$
    \tikzmark{x2-elim-begins}
    \label{line:step-iii-guess-c}
    \Comment{$\star$ signals $x-y \geq C$}
    \label{line:step-iii-guess-const}
    \If{$c$ is not $\star$}
      \State $\chi\gets(x-y=c)$
      \label{line:step-iii-gamma-equality}
      \State $\gamma\gets\gamma'\sub{2^c}{u}$
      \label{line:step-iii-chi-equality}
      \Comment{according to $x-y=c$ and $u = 2^{x-y}$}
    \Else
      \Comment{assuming~$v \geq C$, (in)equalities in $\chi$ simplify to $\top$ or $\bot$}
      \label{line:step-iii-else}
      \Assert{$\chi$ has no equality, and in all its inequalities $u$ has a negative coefficient}
      \label{line:assert-not-bottom}
      \State \myguess $r\gets{}$integer in $[0..d-1]$
      \label{line:step-iii-guess-div}
      \Comment{remainder of $2^{x-y-n}$ modulo $d$ when $x-y \geq C \geq n$}
      \Assert{the divisibility $d \divides 2^v - 2^n \cdot r$ is satisfied by some $v \in [0..d-1]$}
      \label{line:step-iii-assert}
      \State $r'\gets$ discrete logarithm of $2^n \cdot r$ base $2$, modulo $d$
      \label{line:step-iii-discrete-log}
      \State $d'\gets$ multiplicative order of $2$ modulo $d$
      \label{line:step-iii-mult-ord}
      \State $\chi\gets(x-y \geq C)\land(d'\mid x-y-r')$
      \label{line:step-iii-chi-case-2}
      \State $\gamma \gets \phi\sub{2^n \cdot r}{u}$
      \Comment{$2^n \cdot r$ is a remainder of $2^{x-y}$ modulo $\fmod(\gamma')=d\cdot2^n$}
      \label{line:step-iii-chi-div}
    \EndIf\label{line-step-iii-first-loop-end}
    \tikzmark{x2-elim-ends}
    
    \State $\chi\gets\chi[q_x\cdot2^y+r_x\,/\,x]$
    \label{line:step-iii-delayed-subs}
    \Comment{apply substitution: $x$ is eliminated}
    \tikzmark{x1-elim-begins}
    \If{$\chi$ is $(-q_x\cdot2^y-r_x+y+c = 0)$}
    \label{line:step-iii-x1-start}
      \Comment{true if~$\chi$ was constructed in line~\ref{line:step-iii-gamma-equality}}
      \State \myguess $b\gets{}$integer in $[0..c]$
      \label{line:step-iii-guess-b}
      \State $\gamma \gets\gamma\land(q_x = b)$
      \label{line:step-iii-gamma-equality-x1}
      \State $\psi\gets b\cdot2^y=-r_x+y+c$
      \label{line:step-iii-psi-equality-x1}
    \Else 
      \Comment{true if~$\chi$ was constructed in line~\ref{line:step-iii-chi-case-2}}
      \State \textbf{let} $\chi$ be $(-q_x\cdot2^y-r_x+y+C \leq 0) \land (d' \divides q_x \cdot 2^y + r_x - y - r')$, for some~$d',r'\in\N$
      \label{line:step-iii-x1-let}
      \State \myguess $(b,g) \gets{}$pair of integers in $[0..C] \times [1..d']$
      \label{line:step-iii-guess-times}
      \State $\gamma\gets\gamma\land(q_x\geq b)\land(d' \mid q_x-g)$
      \label{line:step-iii-gamma-inequality-x1}
      \State $\psi\gets((b-1)\cdot2^y<-r_x+y+C)\land(-r_x+y+C \leq b\cdot2^y)\land(d'\mid g\cdot2^y+r_x-y-r')$
      \label{line:step-iii-psi-inequality-x1}
    \EndIf 
    \tikzmark{x1-elim-ends}
    \State \textbf{return }$(\gamma,\psi)$
    \label{line:step-iii-return}
  \end{algorithmic}
  \AddNote{x2-elim-begins}{x2-elim-ends}{left-margin}{Eliminate $x$}%
  \AddNote{x1-elim-begins}{x1-elim-ends}{left-margin}{Decoupling $q_x$}%
\end{algorithm}%

\CMSThirdStep*

\begin{proof}
    The correctness of the procedure directly follows from~\cite[Lemma~21 and Lemma~32]{ChistikovMS24} (Appendix C.1, page 36, and Appendix~C.2, page~52). 
\end{proof}

Here is the complexity of~\Cref{algo:step-iii}:

\LemmaComplexityOfStepIII*
  
  \begin{proof} 
    The proof of this lemma follows from the proofs of~\cite[Lemma~36 and~Lemma~37]{ChistikovMS24}. 
    For completeness, we give below a standalone analysis of the bounds on $\gamma''$ and $\psi''$.

    \begin{description}
      \item[\textit{Analysis on $\gamma''$:}] 
        The relevant lines are line~\ref{line:step-iii-chi-equality}, line~\ref{line:step-iii-chi-div}, line~\ref{line:step-iii-psi-equality-x1} and line~\ref{line:step-iii-psi-inequality-x1}.
        \begin{description}
          \item[bound on $\card{\gamma''}$:] The system $\gamma''$ is initialized with $\card{\gamma'}$ constraints in  lines~\ref{line:step-iii-chi-equality} or~\ref{line:step-iii-chi-div}. The subsequent lines~\ref{line:step-iii-psi-equality-x1} and~\ref{line:step-iii-psi-inequality-x1} add at most $2$ constraints. 
          \item[bound on $\linnorm{\gamma''}$:] Recall the $\gamma''$ is a linear program with divisions featuring a single variable $q_x$ (see~\Cref{lemma:CMS:third-step}). When this system is initialized in lines~\ref{line:step-iii-chi-equality} or~\ref{line:step-iii-chi-div}, the absolute values of the coefficients of $q_x$ are bounded by $\linnorm{\gamma'\sub{2^u}{u}} \leq a$. Lines~\ref{line:step-iii-psi-equality-x1} and~\ref{line:step-iii-psi-inequality-x1} only add constraints in which $q_x$ appears with coefficient $1 \leq a$. Hence, $\linnorm{\gamma''} \leq a$.
          \item[bound on~$\onenorm{\gamma''}$:] When the system $\gamma''$ is initialized in lines~\ref{line:step-iii-chi-equality} or~\ref{line:step-iii-chi-div}, 
          its norm~$\onenorm{\cdot}$ is bounded by ${\onenorm{\gamma'} \cdot \max(2^C,\fmod(\gamma'))}$, 
          where $C$ is the integer defined in line~\ref{line:step-iii-max-constant}.
          The (in)equalities added in lines~\ref{line:step-iii-psi-equality-x1} and~\ref{line:step-iii-psi-inequality-x1} feature term with norm~$\onenorm{\cdot}$ bounded by $C+1 \leq 2^C$. 
          Therefore, 
          \begin{align*}
            \onenorm{\gamma''} & \leq c \cdot \max(2^C,d) 
              &\Lbag \text{by $\onenorm{\gamma'} \leq c$ and $\fmod(\gamma') \leq d$} \Rbag\\ 
            & \leq c \cdot \max(2^{3+2 \cdot (\log(c)+1)},d)
              &\Lbag \text{by def.~of~$C$ and $2^n \leq \fmod(\gamma')$} \Rbag\\ 
              & \leq c \cdot \max(2^{5} c^2,d)
              \, \leq \, \max(2^{5} c^3,c \cdot d).
          \end{align*}
          \item[bound on~{\rm$\fmod(\gamma'')$}:] When $\gamma''$ is initialized in lines~\ref{line:step-iii-chi-equality} or~\ref{line:step-iii-chi-div}, its parameter $\fmod(\cdot)$ 
          is equal to~$\fmod(\gamma')$. Line~\ref{line:step-iii-psi-equality-x1} adds no divisibility constraints, whereas~\ref{line:step-iii-psi-inequality-x1} adds a single divisibility constraints with divisor $d'$. Here, $d'$ is the multiplicative order of $2$ modulo the largest odd factor of $\fmod(\gamma')$. That is, $d'$ is a divisor of~$\totient(\fmod(\gamma'))$, which in turn is a 
          divisor of~$\totient(d)$, where $\totient$ is Euler's totient function. 
          Therefore, $\fmod(\gamma')$ divides $\lcm(d,\totient(d))$.
        \end{description}

        \item[\textit{Analysis on $\psi''$:}] The relevant lines are line~\ref{line:step-iii-psi-equality-x1} and line~\ref{line:step-iii-psi-inequality-x1}, which define $\psi''$ depending on~$\chi$.
        
        \begin{description}
          \item[bound on $\card \psi''$:] The system $\psi''$ has a single equality when defined in line~\ref{line:step-iii-psi-equality-x1}, and three constraints when defined in line~\ref{line:step-iii-psi-inequality-x1}.  
          \item[bound on $\linnorm{\psi''}$:] The variables occurring in $\psi''$ are $r_x$ and $y$. Following lines~\ref{line:step-iii-psi-equality-x1} and~\ref{line:step-iii-psi-inequality-x1}, these variables always appear with coefficients $\pm 1$.
          \item[bound on $\onenorm{\psi''}$:] Following lines~\ref{line:step-iii-psi-equality-x1} and~\ref{line:step-iii-psi-inequality-x1}, we see that $\onenorm{\psi''} \leq 2+2 \cdot C$, 
          where $C$ is the integer defined as in line~\ref{line:step-iii-max-constant}. Therefore, 
          \begin{align*}
            \onenorm{\psi''} &\leq 2 + 2 \cdot C\\ 
            & \leq 2 + 2 \cdot \max(\ceil{\log(d)},3 + 2 \ceil{\log(c)})
            & \Lbag \text{using $\onenorm{\gamma'} \leq c$ and $2^n \leq \fmod(\gamma') \leq d$} \Rbag\\
            & \leq \max(2\log(d)+4,4 \log(c)+12)\\
            & \leq 12 + 4 \cdot \log(\max(c,d)).
          \end{align*}
          \item[bound on {\rm$\fmod(\psi'')$}:] From line~\ref{line:step-iii-psi-inequality-x1} we conclude that $\fmod(\psi'')$ is a divisor of $d'$. As discussed in the analysis of $\fmod(\gamma'')$, this implies that it also divides $\totient(d)$.
          \qedhere
        \end{description}
    \end{description}

  \end{proof}
\section{Proofs of statements from Part~\ref{part:small-ILESLP}}
\subsection{Proofs of statements from~Section~\ref{section:ILEP-in-npocmp}}%
\label{appendix:proofs-section-three}

\LemmaModPeriodicity*%
\begin{proof}\label{proof:LemmaModPeriodicity}
    For brevity, define $p \coloneqq \fmod(x,\phi)$.
    Consider two solution $\nu$ and $\nu + [x \mapsto m]$ to $\phi$, 
    with $m \geq p$. (Recall that $\nu + [x \mapsto m]$ denotes the map obtained from $\nu$ by increasing the value assigned to $x$ by $m$.)
    We prove that $\nu' \coloneqq \nu + [x \mapsto p]$ is also solution to $\phi$.
    We analyse divisibility constraints and (in)equalities separately. 
    Clearly, $\nu'$ satisfies all constraints in which $x$ does not appear, hence below we only consider constraints featuring $x$ (linearly, as per hypothesis).
    \begin{description}
        \item[divisibility constraints.] 
            Consider a divisibility constraint $\psi \coloneqq (d \divides  c \cdot x + \tau)$. By definition of~$p$, we have $d \divides p$, 
            and therefore $c \cdot \nu(x)$ and $c \cdot (\nu(x) + p)$ have the same reminder modulo $d$. Hence, $\nu'$ satisfies $\psi$.
        \item[equalities and inequalities.] We only show the case of inequalities (the proof for equalities is analogous, as they can be seen as a conjunction of two inequalities). 
        Consider an inequality $\chi \coloneqq (c \cdot x + \tau \leq 0)$.
        Both $\nu$ and $\nu + [x \mapsto m]$ satisfy this inequality, so we have ${c\cdot \nu(x) + \nu(\tau) \leq 0}$ and ${c\cdot (\nu(x) + m) + \nu(\tau) \leq 0}$.
        Recall that $m \geq p \geq 1$. 
        If $c < 0$, then $c\cdot(\nu(x)+p) + \nu(\tau) < c\cdot \nu(x) +  \nu(\tau) \leq 0$. 
        If instead $c > 0$, 
        then $c\cdot(\nu(x)+p) + \nu(\tau) < c\cdot (\nu(x)+m) +  \nu(\tau) \leq 0$.
        In both cases, 
        we conclude that $\nu'$ satisfies $\chi$.
        \qedhere
    \end{description}
\end{proof}

\LemmaMonotoneGaussianElimination*%
\begin{proof}\label{proof:LemmaMonotoneGaussianElimination}
    We prove the lemma for the case of maximization.
    The case of minimization is analogous.
    For simplicity, let $\vec x= (x_1,\dots,x_n)$.
    Given a solution $\nu$ to $\phi$, 
    we write $\Delta_x^p [f](\nu)$ as a shortcut for $\Delta_x^p(\nu(x_1),\dots,\nu(x_n))$.
    Assume that an optimal solution to $(f,\phi)$ exists.

    The lemma is trivially true when~$x$ occurs in an equality of $\phi$, 
    as all optima must satisfy that equality. Moreover, if $x$ does not occur in equalities nor inequalities of $\phi$, 
    then the existence of an optimum implies that ${\Delta_x^p [f](\nu) = 0}$ whenever $\nu$ and $\nu + [x \mapsto p]$ are solutions to $\phi$.
    We only need to make sure to satisfy the divisibility constraints, 
    and therefore it suffices to consider equations $x = r$ 
    with $r \in [0..p-1]$. These equations are among those considered in the statement of the lemma
    (since we consider $\fterms(\phi \land x \geq 0)$).
    Below, we assume that $x$ occurs in an inequality of $\phi$ 
    but in no equality. 
    We divide the proof depending on the sign of 
    $\Delta_x^p [f]$. 
    \begin{description}
        \item[case: ${\Delta_x^p [f](\nu) > 0}$ whenever $\nu$ and ${\nu + [x \mapsto p]}$ are solutions to $\phi$.] 
            Consider a solution $\nu$ to $\phi$, 
            and assume that it does not satisfy any equation $a \cdot x + \tau + r = 0$ 
            having $(a \cdot x + \tau) \in {\fterms(\phi \land x  \geq 0)}$, $a \neq 0$, 
            and $r \in [0..\abs{a} \cdot p-1]$.
            We show that then $\nu' \coloneqq \nu + [x \mapsto p]$ is still a solution 
            to $\phi$. From ${\Delta_x^p [f](\nu) > 0}$, 
            the value of the objective function~$f$ for~the solution~$\nu'$ is greater than the one for~$\nu$. In particular, this means that $\nu$ cannot be optimal.

            First, observe that $\nu'$ still 
            satisfies all inequalities in $\phi$ of the form $b \cdot x \leq \tau$ with $b \leq 0$, as well as all divisibility constraints.
            For the inequalities, this follows from 
            $b \cdot \nu'(x) \leq b \cdot \nu(x) \leq \nu(\tau) = \nu'(\tau)$. 
            For the divisibility constraints, it suffices to observe that $\nu(x)$ and $\nu'(x)$ have the same residue modulo $p = \fmod(x,\phi)$. 
            Consider then an inequality $a \cdot x \leq \tau$ with $a > 0$. 
            We have $(a \cdot x - \tau) \in \fterms(\phi)$,
            and so $\nu$ does not satisfy any equation $a \cdot x - \tau + r = 0$ with $r \in [0..a \cdot p-1]$. 
            Hence, there is a positive integer $k \geq a \cdot b$ 
            such that $a \cdot \nu(x) - \tau + k = 0$.
            We have $a \cdot \nu'(x) - \tau + k' = 0$, 
            where $k' \coloneqq k - a \cdot b \geq 0$; 
            that is, $\nu'$ satisfies $a \cdot x \leq \tau$.

        \item[case: ${\Delta_x^p [f](\nu) < 0}$ whenever $\nu$ and ${\nu + [x \mapsto p]}$ are solutions to $\phi$.] 
            Consider a solution $\nu$ that does not satisfy any equation $a \cdot x + \tau + r = 0$, 
            where $(a \cdot x + \tau) \in \fterms(\phi \land x  \geq 0)$, $a \neq 0$, 
            and $r \in [0..\abs{a} \cdot p-1]$.
            This time we show that 
            $\nu' \coloneqq \nu + [x \mapsto -p]$ is still a solution 
            to~$\phi$. From~${\Delta_x^p [f](\nu) < 0}$, 
            the value that $f$ takes for $\nu'$ is greater that the value that it takes for $\nu$; and therefore $\nu$ cannot be optimal.

            As in the previous case, it is trivial to see that $\nu'$ satisfies all inequalities of the form $b \cdot x \leq \tau$, with $b \geq 0$, as well as all divisibility constraints.
            Consider then an inequality $a \cdot x \leq \tau$ with $a < 0$. Since $\nu$ does not satisfy any equation $a \cdot x - \tau + r = 0$ with $r \in [0..\abs{a} \cdot p-1]$, 
            we conclude that $\nu(\tau) \geq a \cdot \nu(x) + \abs{a} \cdot p$. 
            Since $a < 0$, we have $a \cdot \nu(x) + \abs{a} \cdot p = a \cdot \dot (\nu(x) - p)$. Hence,~$\nu'$ satisfies $a \cdot x \leq \tau$.

        \item[case: ${\Delta_x^p [f](\nu) = 0}$ whenever $\nu$ and ${\nu + [x \mapsto p]}$ are solutions to $\phi$.] Let $\nu$ be an optimal solution.
            Recall that we are assuming that $x$ occurs in an inequality $a \cdot x \leq \tau$ of $\phi$, and in no equality. 
            Suppose that $a > 0$ (the case for $a < 0$ is analogous).
            Let $a_1 \cdot x \leq \tau_1,\,\dots\,,\,a_{j} \cdot x \leq \tau_j$ be an enumeration of all inequalities featuring $x$ and such that $a_i > 0$. 
            Let $k \in [1..j]$ satisfying
            \begin{equation}
                \label{eq:lemmaMonotoneGaussian:eq1}
                \frac{\nu(\tau_k)}{a_k} = \min\left\{\frac{\nu(\tau_i)}{a_i} : i \in [1..j]\right\}.
            \end{equation}
            Let $\ell \in \N$ be the maximum natural number 
            such that $a_k \cdot (\nu(x) + \ell \cdot p) \leq \nu(\tau_k)$, 
            and define $\nu' \coloneqq \nu + [x \mapsto \ell \cdot p]$.
            We show that $\nu'$ is still an optimal solution to $\phi$, 
            and that it satisfies an equation $a_k \cdot x - \tau_k + r = 0$, 
            for some~$r \in [0..a_k \cdot p-1]$.

            To prove that $\nu'$ is a solution to $\phi$, 
            observe first that, exactly as in the first case of the proof (${\Delta_x^p [f](\nu) > 0}$), the map $\nu'$ satisfies all inequalities in $\phi$ of the form $b \cdot x \leq \tau'$ with $b \leq 0$, as well as all divisibility constraints. 
            Given an inequality $a_i \cdot x \leq \tau_i$ with $i \in [1..j]$, 
            from~\Cref{eq:lemmaMonotoneGaussian:eq1} we conclude that 
            $\nu'(x) \leq  \frac{\nu(\tau_k)}{a_k} \leq  \frac{\nu(\tau_i)}{a_i}$, 
            and therefore $a_i \cdot \nu'(x) \leq \nu(\tau_i) = \nu'(\tau_i)$.
            Hence, $\nu'$ is a solution to $\phi$. 
            It is also an optimal solution.
            Indeed, since the set of solutions of $\phi$ is $(x,p)$-periodic (\Cref{lemma:mod-periodicity}), and ${\Delta_x^p [f](\nu'') = 0}$ whenever $\nu''$ and $\nu'' + [x \mapsto p]$ are solutions to $\phi$, 
            we conclude that $\nu$ and the optimal solution $\nu'$ have the same value with respect to the objective function~$f$.

            Lastly, $\ell$ has been defined to be such that $a_k \cdot \nu'(x) \leq \nu(\tau_k) < a_k \cdot (\nu'(x) + p)$. 
            Observe that ${a_k \cdot (\nu'(x) + p) - a_k \cdot \nu'(x)} = a_k \cdot p$. 
            Therefore, there is an $r \in [0..a_k \cdot p-1]$ 
            such that $a_k \cdot \nu'(x) + r = \nu(\tau_k)$; 
            that is, $\nu'$ satisfies~$a_k \cdot x - \tau_k + r = 0$.
            \qedhere
    \end{description}
\end{proof}

\CorrBasicFactFromPresburger*%
\begin{proof}\label{proof:CorrBasicFactFromPresburger}
    Consider the objective function~$f$ that simply returns the value assigned to~$x$. 
    This is clearly $(x,\fmod(x,\phi))$-monotone locally to the set of solutions of~$\phi$. 
    Since we are looking at non-negative solutions to~$\phi$,~$f$ ranges over the natural numbers, and so it has a minimum. 
    The corollary follows then immediately from~\Cref{lemma:monotone-gaussian-elimination}.
\end{proof}

\LemmaMonotoneOnlyXMatters*%
\begin{proof}\label{proof:LemmaMonotoneOnlyXMatters}
    By definition of monotone decomposition, 
    $\phi$ is equivalent to $\psi_1 \lor \dots \lor \psi_t$.
    Let $\nu$ be a solution to~$\phi$ that maximizes (or minimizes) the objective function $f$. 
    There is $i \in [1..t]$ such that 
    $\nu$ is a solution to  $\psi_i$.
    Since $\psi_i$ implies $\phi$, 
    the solution $\nu$ is optimal for $(f,\psi_i)$.
    By~\Cref{lemma:monotone-gaussian-elimination}, 
    there is an optimal solution~$\nu'$ to $(f,\psi_i)$ 
    that satisfies an equation $a \cdot x + \tau + r = 0$, 
    where $a \neq 0$, $(a \cdot x + \tau) \in \fterms(\psi_i \land x \geq 0)$, 
    and $r \in [0..\abs{a} \cdot \fmod(x,\psi_i)-1]$.
    The value $f$ takes with respect to the two solutions~$\nu$ and~$\nu'$ is the same, and since $\psi_i$ implies $\phi$, 
    the map~$\nu'$ is also a solution to~$\phi$.
\end{proof}
\subsection{Proofs of Lemma~\ref{lemma:pos-analysis-constant} and Claim~\ref{claim:prop:monotone-decomp:3} from~Section~\ref{section:proof-monotone-decomposition}}%
\label{appendix:proofs-section-four}

\LemmaPosAnalysisConstant*%
\begin{proof}\label{proof:LemmaPosAnalysisConstant}
    From $C \geq 2$ we get 
    \(2^C - 2^{C/2} - d \cdot C \geq 2^{C/2}(2^{C/2} - 1) - d \cdot C
    \geq 2^{C/2} - d \cdot C\). So, it suffices
    to prove \(2^{C/2} - d \cdot C \geq 0\).
    Consider the function \(f(x) \coloneqq 2^{x/2} - d\cdot x\), 
    as well as its first derivative \(f'(x) = \frac{1}{2} \cdot \ln(2) \cdot 2^{x/2} - d\) and second derivative $f''(x) = \frac{1}{4} \cdot \ln(2)^2 \cdot 2^{x/2}$.
    Note that $f''$ is positive for all $x \geq 0$, 
    and \(f'(4 \cdot \log_2(d) + 8) = 8 \cdot \ln(2) \cdot d^2 - d \geq 4 \cdot d^2 - d \geq 0\).
    Therefore, 
    \(f\) is increasing for every \(x \geq 4 \cdot \log_2(d) + 8\),  
    and 
    \(f(C) \geq f(4 \cdot \log_2(d) + 8) = 16 \cdot d^2 - d \cdot (4 \cdot \log_2(d) + 8) \geq 
    16 \cdot d^2 - 12 \cdot d^2 \geq 0\).
\end{proof}

The following claim refers to the objects defined throughout 
the proof of~\Cref{prop:monotone-decomposition}.

\ClaimPropMonotoneDecompThree*%
\begin{proof}\label{proof:ClaimPropMonotoneDecompThree}
    Let us assume that the quantifier-free part of the formula $\chi_\sigma$ 
    implies $\overline{x}_m \geq x_m$
    (a similar argument applies if it instead implies $x_m \geq \overline{x}_m$).
    We show that 
    for any two solutions 
    $\nu$ and $\nu' \coloneqq \nu + [q \mapsto p]$ of $\phi \land \phi\sub{q+p}{q} \land \chi_\sigma$, the objective function $\objfun{C}{x_m}$ evaluated at $\nu'$ is at least as large as its value at~$\nu$.
    
    Let $\nu$ and $\nu'$ be such a pair of solutions.
    We start by focusing on $\nu$.
    Define $\nu(\vec w)$ as the vector of integers 
    obtained by evaluating $C$ and $C^{+p}$ according to~$\nu$, 
    to then extract the values corresponding to the variables in $\vec w$
    (which include $\overline{x}_m$ and $x_m$). 
    Since~$\chi_{\sigma}$ contains the equations describing the assignments in $C$ and $C^{+p}$, 
    the vector 
    $\nu(\vec w)$ is the only one that satisfies
    the quantifier-free part of~$\chi_{\sigma}$ for the solution~$\nu$.
    Let $a$ and $\overline{a}$ represent the values in $\nu(\vec w)$ corresponding to $x_m$ and $\overline{x}_m$, respectively.
    Specifically, $a$ is the value taken by $\objfun{C}{x_m}$ for~$\nu$, 
    and $\overline{a}$ is the value taken by $\objfun{C^{+p}}{\overline{x}_m}$ for~$\nu$.
    Since $\nu'$ also satisfies $\phi$, 
    \Cref{claim:prop:monotone-decomp:2}
    implies that the value of $\objfun{C}{x_m}$ for~$\nu'$ 
    is the same as the value of $(\bar{x}_m,C^{+p})$ for $\nu$, 
    which is $\overline{a}$.
    Finally, because we assume that $\overline{x}_m \geq x_m$
    is implied by the quantifier-free part of $\chi_{\sigma}$, 
    we conclude that $\overline{a} \geq a$, completing the proof.
\end{proof}
\subsection{Bareiss algorithm}
\label{sub-appendix:gaussian-elimination}
\label{appendix:gaussian-elimination}
In this appendix we recall the classical Bareiss algorithm from~\cite{Bareiss68} and of some of its standard properties.
These properties are then lifted to our variation in~\Cref{sub-appendix:gaussian-elimination-twist}, where we provide the proofs of \Cref{lemma:gaussian-elimination:new:same-row,lemma:gaussian-elimination:new:below,lemma:gaussian-elimination:new:above,lemma:one-shot-replacement:new}.
Throughout this and the next appendix, we follow the notation introduced in~\Cref{subsec:variation-bareiss-body}, 
in particular when it comes to defining the subdeterminants $b^{(\ell)}_{i,j}$ and $b^{(\ell)}_{r \gets j}$ of a given matrix with entries denotes as~$b_{i,j}$.

We describe the standard Bareiss's algorithm from~\cite{Bareiss68}. 
Consider an $m \times d$ integer matrix~$B_0$:
\begin{equation*}
    \label{eq:B-matrix}
    B_0 \coloneqq \begin{pmatrix}
        b_{1,1} & \ldots  & b_{1,d}\\
        \vdots  & \ddots & \vdots\\
        b_{m,1} & \ldots  & b_{m,d}
    \end{pmatrix}.
\end{equation*}\
As done in~\Cref{subsec:variation-bareiss-body},
fix $k \in [0..\min(m,d)]$ (the number of iterations the algorithm will perform), and define $\lambda_0 \coloneqq 1$ 
and ${\lambda_\ell \coloneqq b_{\ell,\ell}^{(\ell-1)}}$, for every $\ell \in [1..k]$. 
We assume that every $\lambda_\ell$ is non-zero.

Starting from the matrix $B_0$, 
Bareiss algorithm iteratively constructs a sequence of matrices $B_1,\dots,B_{k}$ as follows.
Consider ${\ell \in [0..k-1]}$, and let $B_\ell$ be the matrix 
\[
    B_\ell \coloneqq \begin{pmatrix}
        g_{1,1} & \ldots  & g_{1,d}\\
        \vdots  & \ddots & \vdots\\
        g_{m,1} & \ldots  & g_{m,d}
    \end{pmatrix}.
\]
The matrix $B_{\ell+1}$ is constructed from~$B_\ell$ 
by applying the following transformation

{\setstretch{1.3}
\begin{algorithmic}[1]
    \For{every row $i$ except row $\ell+1$}
    \State multiply the $i$th row of $B_{\ell}$ by $g_{\ell+1,\ell+1}$ 
    \label{apx:algo:gauss:line2}
    \Comment{the entry $(i,\ell+1)$ of $B_{\ell}$ is now $g_{\ell+1,\ell+1} \cdot g_{i,\ell+1}$}
    \State subtract $g_{i,\ell+1} \cdot (g_{\ell+1,1},\dots,g_{\ell+1,d})$ from the $i$th row of $B_{\ell}$
    \label{apx:algo:gauss:line3}
    \Comment{the entry $(i,\ell+1)$ is now $0$}
    \State divide each entry of the $i$th row of $B_{\ell}$ by $\lambda_{\ell}$
    \label{apx:algo:gauss:line4}
    \Comment{these divisions are without remainder}
    \EndFor
\end{algorithmic}
}
\vspace{5pt}

The next three results from~\cite{Bareiss68} give a complete description of the entries of~$B_0,\dots,B_k$:
\Cref{lemma:gaussian-elimination:same-row} contains a trivial observation that we will use 
many times in the other two lemmas, \Cref{lemma:gaussian-elimination:below} describes 
the last $m-\ell$ rows of these matrices, 
and~\Cref{lemma:gaussian-elimination:above} describes the first $\ell$ rows.

\begin{lemma}[\cite{Bareiss68}]
    \label{lemma:gaussian-elimination:same-row}
    For all $\ell \in [0..k-1]$, 
    the $(\ell+1)$th rows of the matrices $B_\ell$ and $B_{\ell+1}$ 
    are equal.
\end{lemma}


\begin{lemma}[\cite{Bareiss68}]
    \label{lemma:gaussian-elimination:below}
    Consider $\ell \in [0..k]$. 
    For every $i \in [\ell+1..m]$ and $j \in [1..d]$, 
    the entry in position $(i,j)$ of the matrix $B_\ell$ is 
    $b^{(\ell)}_{i,j}$.
    In particular, this entry is zero whenever $j \leq \ell$.
\end{lemma}

\begin{lemma}[\cite{Bareiss68}]
    \label{lemma:gaussian-elimination:above}
    Consider $\ell \in [1..k]$. 
    For all $i \in [1..\ell]$ and $j \in [1..d]$, 
    the entry in position $(i,j)$ of $B_\ell$ is 
    $b^{(\ell)}_{i \gets j}$. 
    In particular,
    this entry is zero if $j \leq \ell$ and $i \neq j$,
    and it is instead $b^{(\ell-1)}_{\ell,\ell}$ when $i = j$.
\end{lemma}

\subsection{Analysis of the variation of Bareiss algorithm}
\label{sub-appendix:gaussian-elimination-twist}

We are now ready to analyze our variation of Bareiss algorithm, and prove~\Cref{lemma:gaussian-elimination:new:same-row,lemma:gaussian-elimination:new:below,lemma:gaussian-elimination:new:above,lemma:one-shot-replacement:new}.
Below, let $k$, $B_0$, $B_{0}'$, $\mu$, $U_g$ and $\lambda_\ell$ 
be defined as in~\Cref{subsec:variation-bareiss-body}.
We recall below how the matrices $B_1',\dots,B_k'$ are constructed from $B_0'$.
Consider ${\ell \in [0..k-1]}$, and let $B_\ell'$ be the matrix 
\[
    B_\ell' \coloneqq \begin{pmatrix}
        h_{1,1} & \ldots  & h_{1,d}\\
        \vdots  & \ddots & \vdots\\
        h_{m,1} & \ldots  & h_{m,d}
    \end{pmatrix}.
\]
The matrix $B_{\ell+1}'$ is constructed from~$B_\ell'$ 
by applying the following transformation 
{\setstretch{1.3}
\makeatletter%
\def\ALG@step%
   {%
   \addtocounter{ALG@line}{1}%
   \addtocounter{ALG@rem}{1}%
   \ifthenelse{\equal{\arabic{ALG@rem}}{\ALG@numberfreq}}%
      {\setcounter{ALG@rem}{0}\alglinenumber{0\arabic{ALG@line}}}
      {}%
   }%
\makeatletter
\begin{algorithmic}[1]
    \State \textbf{let} $\pm$ be the sign of $h_{\ell+1,\ell+1}$, and $\alpha \coloneqq \frac{\pm h_{\ell+1,\ell+1}}{\mu}$
    \customlabel{01}{apx:algo:gauss:new:sign-pivot}
    \Comment{this division is without remainder}
    \State multiply the row $\ell+1$ of $B_{\ell}$ by $\pm 1$
    \customlabel{02}{apx:algo:gauss:new:multiply-pivoting-row}
    \For{every row $i$ except row $\ell+1$}
        \customlabel{03}{apx:algo:gauss:new:loop}
        \State \textbf{let} $\beta \coloneqq \frac{h_{i,\ell+1}}{\mu}$ 
        \customlabel{04}{apx:algo:gauss:new:line1}
        \Comment{this division is without remainder}
        \State multiply the $i$th row of $B_{\ell}'$ by $\alpha$ 
        \customlabel{05}{apx:algo:gauss:new:line2}
        \Comment{$B_{\ell}'(i,\ell+1)$ is now $\alpha \cdot g_{i,\ell+1}$}
        \State subtract $\pm \beta \cdot (h_{\ell+1,1},\dots,h_{\ell+1,d})$ from the $i$th row of $B_{\ell}'$
        \customlabel{06}{apx:algo:gauss:new:line3}
        \State divide each entry of the $i$th row of $B_{\ell}'$ by $\abs{\lambda_{\ell}}$
        \customlabel{07}{apx:algo:gauss:new:line4}
        \Comment{these divisions are without remainder}
    \EndFor
\end{algorithmic}
}

\vspace{5pt}

\noindent
We now show that the main relationship between 
the matrices computed with the above transformation, and 
the sequence of matrices~$B_0,\dots,B_k$ computed with Bareiss algorithm.

\begin{lemma}
    \label{lemma:recover-bareiss}
    For every $\ell \in [0..k]$, \,$B_\ell' = \pm B_\ell \cdot U_g$, 
    where $\pm$ is the sign of $\lambda_\ell$.
\end{lemma}

\begin{proof}
    Below, we write $\pm_\ell$ for the sign of $\lambda_\ell$.
    The proof is by induction on $\ell$. 
    \begin{description}
        \item[base case: $\ell = 0$.] Trivial from the definition of $B_0'$ (recall that $\lambda_0 = 1$ in this case).
        \item[induction hypothesis.] Given $\ell \geq 1$, we have $B_{\ell-1}' = \pm_{\ell-1} B_{\ell-1} \cdot U_g$.
        \item[induction step: $\ell \geq 1$.]
            From the induction hypothesis, for every $i \in [1..m]$ the $i$th rows $\vec r_i$ and $\vec r_i'$ of $B_{\ell-1}$ and $B_{\ell-1}'$ are, respectively,
            \begin{align*} 
                \vec r_{i} = (r_{i,1},\dots,r_{i,d}),
                \qquad\qquad 
                \vec r_{i}' = \pm_{\ell-1}(\mu \cdot r_{i,1},\dots,\mu \cdot r_{\ell,g}, r_{i,g+1}, \dots, r_{i,d}),
            \end{align*}
            for some integers $r_{i,1},\dots,r_{i,d}$.
            To show that~$B_\ell' = \pm_\ell B_\ell \cdot U_g$, 
            we analyze the pseudocodes used to construct $B_\ell$ and $B_{\ell}'$, 
            considering each row separately.

            Let us start by considering the $\ell$th row. 
            In the case of $B_\ell$, this row coincides with $\vec r_\ell$. 
            In the case of $B_\ell'$, this row is instead $\pm \vec r_\ell'$, 
            where $\pm$ is the sign of $\pm_{\ell-1} \cdot \mu \cdot r_{\ell,\ell}$. By~\Cref{lemma:gaussian-elimination:below} and from the definition of $\lambda_\ell$, we have $r_{\ell,\ell} = \lambda_\ell$. Since $\mu \geq 1$, $\pm 1 = \pm_{\ell-1} \pm_{\ell} 1$, and therefore $\pm \vec r_{\ell}'$ is equal to 
            $\pm_{\ell}(\mu \cdot r_{\ell,1},\dots,\mu \cdot r_{\ell,g}, r_{\ell,g+1}, \dots, r_{\ell,d})$, as required.
 
            Consider now $i \in [1..m]$ with $i \neq \ell$.
            Following the code of Bareiss algorithm, the $i$th row of~$B_\ell$ is given by 
            $\frac{1}{\lambda_{\ell-1}} \cdot (r_{\ell,\ell} \cdot \vec r_{i} - r_{i,\ell} \cdot \vec r_{\ell})$.
            The entry of $B_{\ell}$ in position $(i,j)$, with $j \in [1..d]$,~is thus%
            \[ 
                t_{i,j} \coloneqq \frac{r_{\ell,\ell} \cdot r_{i,j} - r_{i,\ell} \cdot r_{\ell,j}}{\lambda_{\ell-1}}.
            \]
            The $i$th row of $B_{\ell}'$ is instead $\frac{1}{\abs{\lambda_{\ell-1}}} \cdot (\alpha \cdot \vec r_{i}' - \pm\beta \cdot \vec r_{\ell}')$, 
            where $\pm$ is again the sign of $\pm_{\ell-1} \mu \cdot r_{\ell,\ell}$, 
            and $\alpha \coloneqq \frac{\pm \pm_{\ell-1} \mu \cdot r_{\ell,\ell}}{\mu} = \frac{\pm_\ell \mu \cdot r_{\ell,\ell}}{\mu}$
            and $\beta \coloneqq \frac{\pm_{\ell-1}\mu \cdot r_{i,\ell}}{\mu}$ 
            are defined as in line~\ref{apx:algo:gauss:new:line1} of the pseudocode.
            For every $j \in [1..g]$, the entry of $B_{\ell}'$ in position $(i,j)$ 
            is therefore 
            \begin{align*}
                t_{i,j}' &\coloneqq \frac{\alpha \cdot (\pm_{\ell-1} \mu \cdot r_{i,j}) - \pm\beta \cdot ( \pm_{\ell-1} \mu \cdot r_{\ell,j})}{\abs{\lambda_{\ell-1}}}\\
                &= \frac{\pm_\ell r_{\ell,\ell} \cdot (\pm_{\ell-1} \mu \cdot r_{i,j}) - \pm_\ell r_{i,\ell} ( \pm_{\ell-1} \mu \cdot r_{\ell,j})}{\pm_{\ell-1} \lambda_{\ell-1}}
                \,=\,\pm_{\ell} \mu \cdot \frac{r_{\ell,\ell} \cdot r_{i,j} - r_{i,\ell} \cdot r_{\ell,j} }{\lambda_{\ell-1}}\,=\, \pm_\ell \mu \cdot t_{i,j}.
            \end{align*}
            An analogous manipulation shows $t_{i,j}' = \pm_\ell t_{i,j}$,
            for every $j \in [g+1..d]$. 
            We thus have $(t_{i,1}',\dots,t_{i,d}') = \pm_\ell(\mu \cdot t_{i,1},\dots,\mu \cdot t_{i,g}, t_{i,g+1},\dots,t_{i,d})$, 
            which concludes the proof.
            \qedhere
    \end{description}
\end{proof}

By relying on~\Cref{lemma:recover-bareiss}, we can easily rephrase the properties of Bareiss algorithm from~\Cref{sub-appendix:gaussian-elimination} to our variation,
proving~\Cref{lemma:gaussian-elimination:new:same-row,lemma:gaussian-elimination:new:below,lemma:gaussian-elimination:new:above}.

\LemmaGaussianEliminationNewSameRow*
\begin{proof}\label{proof:LemmaGaussianEliminationNewSameRow}
    This is a simple observation on the effects of lines~\ref{apx:algo:gauss:new:sign-pivot}
    and~\ref{apx:algo:gauss:new:multiply-pivoting-row} of the procedure.
\end{proof}

\LemmaGaussianEliminationNewBelow*
\begin{proof}\label{proof:LemmaGaussianEliminationNewBelow}
    Directly from~\Cref{lemma:recover-bareiss} and~\Cref{lemma:gaussian-elimination:below}. 
\end{proof}

\LemmaGaussianEliminationNewAbove*
\begin{proof}\label{proof:LemmaGaussianEliminationNewAbove}
    Directly from~\Cref{lemma:recover-bareiss} and~\Cref{lemma:gaussian-elimination:above}. 
\end{proof}

Lastly, we move to the proof of~\Cref{lemma:one-shot-replacement:new}, 
which relies on Laplace expansions. \Cref{thm:laplace-expansion} below 
recalls this notion using determinants 
of the form $a_{r \gets j}^{(\ell)}$. These determinants have two 
straightforward properties:
\begin{enumerate}
    \item If $j \leq \ell$ and $r \neq j$, then the $j$th and $r$th columns of the submatrix corresponding to this determinant are identical. Therefore, $a_{r \gets j}^{(\ell)} = 0$.
    \item If $j = r$ then $a_{j \gets j}^{(\ell)} = a_{\ell,\ell}^{(\ell-1)}$. 
\end{enumerate}

\begin{namedlemma}[Laplace Expansion]
    \label{thm:laplace-expansion}
    For every $i,j,\ell \in \N$ satisfying $\ell \geq 1$,  $\ell < i \leq m$ and $\ell < j \leq d$,
    \begin{equation}
        \label{eq:laplace-expansion}
        a^{(\ell)}_{i,j} = a^{(\ell-1)}_{\ell,\ell} \cdot a_{i,j} - \sum\nolimits_{r=1}^\ell a^{(\ell)}_{r \gets j} \cdot a_{i,r}.
    \end{equation}
\end{namedlemma}

\begin{proof}
    \Cref{eq:laplace-expansion} is not the standard way of writing the Laplace expansion, 
    but it is one that will be very natural for our purposes. 
    A more standard way of writing this identity is: 
    \begin{equation}
        \label{eq:standard-laplace-expansion}
        a^{(\ell)}_{i,j} = a^{(\ell-1)}_{\ell,\ell} \cdot a_{i,j} + \sum\nolimits_{r=1}^\ell (-1)^{r+\ell+1} m_{r} \cdot a_{i,r}\,,
    \end{equation}
    where $m_r \eqdef 
        \det
        \begin{pmatrix}
        a_{1,1} & \ldots & a_{1,r-1} & a_{1,r+1} & \dots & a_{1,\ell} & a_{1,j} \\
        \vdots  & \ddots & \vdots  & \vdots & \vdots  & \ddots & \vdots   \\
        a_{\ell,1} & \ldots & a_{\ell,r-1} & a_{\ell,r+1} & \dots & a_{\ell,\ell} & a_{\ell,j}
        \end{pmatrix}$.
    
    \vspace{6pt}
    \noindent
    The matrix~$M$ used to compute $m_r$ features the same columns as the matrix $M'$ used to compute $a_{r \gets j}^{(\ell)}$. In particular, $M'$ is obtained from $M$ by applying a cyclic permutation to the last $\ell-r+1$ columns, so that the column $(a_{1,j},\dots,a_{\ell,j})$ appears first among those. 
    Recall that swapping two rows of a matrix changes the sign of its determinant. The permutation used to compute $M'$ form $M$ can be realized with $\ell-r$ swaps, hence $(-1)^{\ell-r} m_{r} = a_{r \gets j}^{(\ell)}$.
    Therefore, $(-1)^{r+\ell+1} m_{r} = (-1)^{r+\ell+1+r-r} m_{r} = - a_{r \gets j}^{(\ell)}$, showing that~\Cref{eq:laplace-expansion,eq:standard-laplace-expansion} are equivalent.
\end{proof}

\LemmaOneShotReplacementNew*
\begin{proof}\label{proof:LemmaOneShotReplacementNew}
    Note that if $\ell = 0$, then the transformation does not modify $B_0$ (recall that $\lambda_0 = 1$), 
    and the lemma is therefore trivially true.
    Let us consider then $\ell \geq 1$,
    and write $\vec v$ for the $i$th row of $B_0'$ \emph{after} the transformation.
    Let us compute an expression for the entries in $\vec v$.
    After executing line~\ref{algo:one-shot-replacement:new:line3}, 
    the $i$th row of $B_0$ is of the form 
    $\pm b_{\ell,\ell}^{(\ell-1)} \cdot (\mu \cdot b_{i,1}, \dots, \mu \cdot b_{i,g},\, b_{i,g+1} \dots, b_{i,d})$, 
    where $\pm$ is the sign of $\lambda_{\ell} = b_{\ell,\ell}^{(\ell-1)}$.
    From~\Cref{lemma:gaussian-elimination:new:above}, 
    the vector $\vec u_r$ from line~\ref{algo:one-shot-replacement:new:line4}
    is $\vec u_r = \pm (\mu \cdot b_{r \gets 1}^{(\ell)}, \dots, \mu \cdot b_{r \gets g}^{(\ell)},\, b_{r \gets g+1}^{(\ell)}, \dots, b_{r \gets d}^{(\ell)})$.
    After performing all subtractions from line~\ref{algo:one-shot-replacement:new:line4} to the $i$th row, 
    we obtain the vector $\vec v$.
    Its $j$th entry  
    is: 
    \begin{align*}
        \pm \mu_j \cdot (b_{\ell,\ell}^{(\ell-1)} \cdot b_{i,j}
        - \sum\nolimits_{r=1}^{\ell}b_{r \gets j}^{(\ell)} \cdot  b_{i,r}),
    \end{align*}
    where $\mu_j \coloneqq \mu$ whenever $j \leq g$, and otherwise $\mu_j \coloneqq 1$. 

    We show that the expression~${(b_{\ell,\ell}^{(\ell-1)} \cdot b_{i,j}
        - \sum\nolimits_{r=1}^{\ell}b_{r \gets j}^{(\ell)} \cdot  b_{i,r})}$ 
    equals $b_{i,j}^{(\ell)}$, 
    thus establishing that $\vec v$ is the $i$th row of $B_\ell'$ 
    by~\Cref{lemma:gaussian-elimination:new:below}.
    When $j > \ell$, this results follows directly from~\Cref{thm:laplace-expansion}.
    When $j \leq \ell$ instead, observe 
    that $b_{i,j}^{(\ell)} = 0$ (\Cref{lemma:gaussian-elimination:new:below}), 
    and it suffices to show that the expression $(b_{\ell,\ell}^{(\ell-1)} \cdot b_{i,j}
        - \sum\nolimits_{r=1}^{\ell}b_{r \gets j}^{(\ell)} \cdot  b_{i,r})$ is also zero.
    For every $r \in [1..\ell]$ 
    with $r \neq j$, the determinant $b_{r \gets j}^{(\ell)}$ is zero; 
    as the $j$th and $r$th columns of the corresponding matrix are identical. 
    The expression thus simplifies
    to 
    $(b_{\ell,\ell}^{(\ell-1)} \cdot b_{i,j}
        - b_{j \gets j}^{(\ell)} \cdot  b_{i,j})
        \ =\  (b_{\ell,\ell}^{(\ell-1)} \cdot b_{i,j}
        - b_{\ell,\ell}^{(\ell-1)} \cdot  b_{i,j}) \ =\  0$.
\end{proof}

\subsection{Proofs of Lemma~\ref{lemma:what-btp-comples} from~Section~\ref{sec:efficient-variable-elimination}}%
\label{appendix:proofs-section-five}

\LemmaWhatBTPCoples*%
\begin{proof}\label{proof:LemmaWhatBTPCoples}
    Let us first observe that $\gamma$ contains an inequality $a \cdot q_{n-\ell} \geq 0$ for some $a \geq 1$, by definition of $\objcons_k^\ell$, and therefore $\fterms(\gamma)$ contains $-a \cdot q_{n-\ell}$. This implies that the guess performed in line~\ref{algo:true-tp:guess} is never on an empty set.

    Let $\rho \coloneqq (a \cdot q_{n-\ell} - \tau)$ be the guessed term. 
    By definition of~$\objcons_k^\ell$, 
    $\gamma$ is a linear program in variables $u$ and $\vec q_{[\ell,k]}$, and in which
    every inequality and equality is such that 
    all the coefficients of the variables in $\vec q_{[\ell,k]}$ 
    are divisible by $\mu_C$. 
    The statement is thus true when $\rho$
    belongs to~${\fterms(\gamma \land \gamma\sub{q_{n-\ell}+ p}{q_{n-\ell}})}$. 

    Suppose~$\rho$ to be instead computed using~\Cref{algo:additional-hyperplanes}. 
    In this case there are $b,d \in \Z$, and $q',q''$ from $\vec q_{k-1}$, 
    such that 
    $\rho$ is obtained from ${b \cdot u + \mu_C \cdot (q'-q'')+d}$ 
    by simultaneously applying two substitutions $\sub{\frac{\tau'}{\lambda}}{\mu_C \cdot q'}$ and $\sub{\frac{\tau''}{\lambda}}{\mu_C \cdot q''}$. 
    By definition of simultaneous substitution, 
    $\rho$ is thus of the form 
    $\left(\lambda \cdot b \cdot u \pm \tau' \mp \tau'' + \lambda \cdot d\right)$, 
    where the signs $\pm$ and $\mp$ depend on the sign of $\lambda$. 
    Here, $\lambda \coloneqq \frac{\eta_C}{\mu_C}$, 
    and the terms $\tau'$ and $\tau''$ are computed as described 
    in lines~\ref{algo:btp:line-tau1}--\ref{algo:btp:line-shift-tau2}. 
    From the definition of~$\objcons_{k}^\ell$, recall that $\mu_C$ divides $\eta_C$, as well as all coefficients of the variables $\vec q_{[\ell,k]}$ occurring in linear terms $\tau_{n-i}(u,\vec q_{[\ell,k]})$ featured in assignments $q_{n-i} \gets \frac{\tau_{n-i}}{\eta_C}$ of $C$, with $i \in [0..\ell-1]$. 
    Looking at 
    lines~\ref{algo:btp:line-tau1}--\ref{algo:btp:line-shift-tau2}, 
    it is then easy to see that the terms $\tau'$ and $\tau''$ 
    are linear terms in variables $u$ and $\vec q_{[\ell,k]}$, and that in these terms the coefficients of~$\vec q_{[\ell,k]}$ are all divisible by~$\mu_C$.
    Then, the same is true for the term~$\pm \tau' \mp \tau''$, and in turn also for $\rho$.
\end{proof}

\subsection{Proofs of the claims from~Lemma~\ref{lemma:key-correspondence-with-Bareiss} (Section~\ref{sec:efficient-variable-elimination})}
\label{subsec:proof-claims-lemma-correspondence-Bareiss}

The following claims refer to objects defined throughout the proof of~\Cref{lemma:key-correspondence-with-Bareiss}.

\ClaimKeyCorrespondenceRowL* 
\begin{proof}\label{proof:ClaimKeyCorrespondenceRowL}
    By induction hypothesis, the $\ell$th rows of $M_{\ell-1}$ and $B_{\ell-1}'$ are equal, and by~\Cref{claim:key-correspondence:substitution}, they contain the variable coefficients of $a \cdot q_{n-(\ell-1)}-\tau$.
    Following Bareiss algorithm (see line~\ref{apx:algo:gauss:new:multiply-pivoting-row} and~\Cref{lemma:gaussian-elimination:new:same-row}), the $\ell$th row of $B_\ell'$ 
    contains the variable coefficients of $\pm a \cdot q_{n-(\ell-1)} - (\pm \tau)$. 

    In the case of~$M_\ell$, from its definition (Item~\eqref{gaussopt-complexity:matrix-def-1}), the $\ell$th row 
    contains the variable coefficients of the term $\eta_\ell \cdot q - \tau'$ such 
    that $q \gets \frac{\tau'}{\eta_\ell}$ is the first assignment in $C_\ell$.
    From line~\ref{algo:sub-disc:make-a-positive} and the last item in line~\ref{algo:sub-disc:assert-circuit} of~\Cref{algo:sub-disc}, 
    we conclude that this assignment is $q_{n-(\ell-1)} \gets \frac{\pm\tau}{\pm a}$. Therefore, the $\ell$th rows of $M_\ell$ and $B_{\ell}'$ are equal. 
    
    Again from~\ref{algo:sub-disc:assert-circuit} of~\Cref{algo:sub-disc}, we see that in $C_\ell$ all assignments to variables in $\vec q_k$ have $\pm a$ as a denominator, that is, $\eta_\ell = \pm a$. 
    Furthermore, from~\Cref{lemma:gaussian-elimination:new:above}, 
    the entry of $B_\ell'$ in position $(\ell,\ell)$ is $\pm' \mu \cdot b_{\ell,\ell}^{(\ell-1)}$, where $\pm'$ is the sign of $\lambda_\ell = b_{\ell,\ell}^{(\ell-1)}$. By definition of $M_\ell$, the $\ell$th column contains the coefficients of $q_{n-(\ell-1)}$. 
    This means that the entry of $M_{\ell}$ in position $(\ell,\ell)$ 
    is $\pm a$, and therefore~$\pm' \mu \cdot \lambda_\ell = \pm a$.
    It follows that $\alpha = \frac{\eta_\ell}{\mu} = \abs{\lambda_\ell} \neq 0$.
\end{proof}

\ClaimKeyCorrespondenceRowGamma*
\begin{proof}\label{proof:ClaimKeyCorrespondenceRowGamma}
    This proof is similar to the one of~\Cref{claim:key-correspondence:row-circuit}. 
    By definition, the $i$th rows of $M_{\ell-1}$ and $M_{\ell}$ 
    contain the variable coefficients of the terms in 
    the (in)equalities ${\Lambda_{\ell-1}(\rho_{i-j} \sim_{i-j} 0)}$ and $\Lambda_{\ell}(\rho_{i-j} \sim_{i-j} 0)$, respectively. 
    By definition of $\Lambda_{\ell-1}$ and $\Lambda_{\ell}$, 
    $\Lambda_{\ell}(\rho_{i-j} \sim_{i-j} 0)$ is the (in)equality obtained from $\Lambda_{\ell}(\rho_{i-j} \sim_{i-j} 0)$ when 
    running lines~\ref{algo:sub-disc:eliminate}--\ref{algo:sub-disc:simplify-2} of~\Cref{algo:sub-disc}. 
    Below we analyze the updates performed in these lines of the algorithm, and compare them to those Bareiss algorithm performs on the $i$th row of $B_{\ell-1}'$ in order to produce $B_{\ell}'$.

    Let us write~$\beta \cdot \mu \cdot q_{n-(\ell-1)} + \tau'$ 
    for the term in the (in)equality $\Lambda_{\ell-1}(\rho_{i-j} \sim_{i-j} 0)$. Line~\ref{algo:sub-disc:eliminate} applies the substitution~$\sub{\frac{\pm\tau}{\alpha}}{\mu \cdot q_{n-(\ell-1)}}$ on this term, obtaining the term 
    $\pm \beta \cdot \tau + \alpha \cdot \tau'$.
    Then, following lines~\ref{algo:sub-disc:simplify} and~\ref{algo:sub-disc:simplify-2}, we see that the $i$th row 
    of $M_{\ell}$ holds the variable coefficients 
    of the term $(\pm \beta \cdot \tau + \alpha \cdot \tau')/\frac{\eta_{\ell-1}}{\mu}$ obtained by dividing each integer in $\pm \beta \cdot \tau + \alpha \cdot \tau'$ by $\frac{\eta_{\ell-1}}{\mu}$. As in the proof of~\Cref{claim:key-correspondence:row-circuit}, we will see below 
    that all variable coefficients of $\pm \beta \cdot \tau + \alpha \cdot \tau'$ are divisible by~$\frac{\eta_{\ell-1}}{\mu}$.%

    Let us look at Bareiss algorithm. 
    By induction hypothesis, the $i$th row of $B_{\ell-1}'$ 
    contains the variable coefficients of~$\beta \cdot \mu \cdot q_{n-(\ell-1)} + \tau'$. 
    The algorithm first multiplies this row by $\alpha$, and then subtracts $\pm\beta \cdot \vec r_{\ell}$, where $\vec r_{\ell}$ is the $\ell$th row of $B_{\ell-1}'$.
    By~\Cref{claim:key-correspondence:substitution}, 
    $\vec r_\ell$ holds the variable coefficients of $a \cdot q_{n-(\ell-1)}-\tau$. Hence, after this subtraction, the $i$th row of $B_{\ell-1}'$ is updated to contain the variable coefficients of the term $\pm \beta \cdot \tau+\alpha \cdot \tau'$. Lastly, each entry of the $i$th row is divided by $\abs{\lambda_{\ell-1}} = \frac{\eta_{\ell-1}}{\mu}$. 
    From~\Cref{lemma:gaussian-elimination:new:below}, these divisions are exact. This means that every variable coefficient in $\pm \beta \cdot \tau+\alpha \cdot \tau'$ is divisible by $\frac{\eta_{\ell-1}}{\mu}$; so the divisions performed in 
    lines~\ref{algo:sub-disc:simplify} and~\ref{algo:sub-disc:simplify-2}
    of~\Cref{algo:sub-disc} are also without remainder.
    The divisions performed by Bareiss algorithm 
    completes the construction of the $i$th row of $B_{\ell}$, 
    which thus contain the variable coefficients of $(\pm \beta \cdot \tau + \alpha \cdot \tau')/\frac{\eta_{\ell-1}}{\mu}$. Hence, the $i$th rows of $M_\ell$ and $B_{\ell}'$ coincide.

    Let us discuss the second statement of the claim.
    Every (in)equality in~${\gamma_{\ell-1}\sub{\frac{\pm\tau}{\alpha}}{\mu \cdot q_{n-{\ell-1}}}}$ 
    is obtained by applying the substitution $\sub{\frac{\pm\tau}{\alpha}}{\mu \cdot q_{n-{\ell-1}}}$ to an (in)equality~$\Lambda_{\ell-1}(\rho_{i-j} \sim_{i-j} 0)$, with $i \in [j+1..j+t]$.
    Following the notation above, let $\pm \beta \cdot \tau + \alpha \cdot \tau'$ be the term resulting from one such substitution.
    We have already shown that all variable coefficients of this term are divisible by $\frac{\eta_{\ell-1}}{\mu}$. 
    Then, the coefficient of $u$ is divisible by~$\frac{\eta_{\ell-1}}{\mu}$. As for the remaining variables,~\Cref{lemma:gaussian-elimination:new:below} 
    guarantee that, once divided by~$\frac{\eta_{\ell-1}}{\mu}$, their coefficients are still divisible by $\mu$;
    hence before divisions these coefficients are divisible by~$\eta_{\ell-1}$.
\end{proof}

\ClaimKeyCorrespondenceRowEquations*
\begin{proof}\label{proof:ClaimKeyCorrespondenceRowEquations}
    By definition, the contents of the $i$th rows of $M_{\ell-1}$ and $M_{\ell}$
    depend on the type of the equality $e_{i-1}$. We divide the proof in two cases, depending on this type.

    If $e_{i-1}$ is of Type~\ref{gaussopt-connection:typeI}, 
    then by definition the $i$th rows of $M_{\ell-1}$ and $M_{\ell}$ contain the variable coefficients of the terms in
    the (in)equalities~${\Lambda_{\ell-1}(g_{i-1})}$ and $\Lambda_{\ell}(g_{i-1})$, respectively. Moreover, 
    the generator~$g_{i-1}$ of~$e_{i-1}$ is an (in)equality of~$\gamma_0$.
    Therefore, there is $r \in [j+1..j+t]$ 
    such that the $i$th and $r$th rows of $M_{\ell-1}$ (resp.~$M_{\ell}$) are equal. Since Bareiss algorithm performs the same updates on every row different from~$\ell$, 
    the claim then follows from~\Cref{claim:key-correspondence:row-gamma}.
    
    Suppose now $e_{i-1}$ to be of Type~\ref{gaussopt-connection:typeII}. 
    Let $g_{i-1} \coloneqq (\rho = 0)$ be the generator of $e_{i-1}$. Recall that $\rho$ is of the form ${b \cdot u + \mu \cdot (q' - q'') + d}$, for some $b,d \in \Z$ and $q',q''$ from $\vec q_k$.
    By Item~\eqref{gaussopt-complexity:matrix-def-3} in the definition of $M_\ell$, 
    the $i$th row of $M_\ell$ contains the coefficients of the variables $\vec q_k$ and $u$ from the term obtained from $\rho$ as follows:
    \begin{algorithmic}[1]
        \State multiply every integer in $\rho$ by the quotient of the division of $\eta_\ell$ by $\mu$
        \For{$r$ in $[1..\ell]$}
            $\rho \gets \rho\sub{\tau_r}{\eta_\ell \cdot q}$, 
            where $q \gets \frac{\tau_r}{\eta_\ell}$ is the $r$th assignment in $C_\ell$
        \EndFor
    \end{algorithmic}
    Moreover, again by definition of $M_\ell$, 
    the variable coefficients of the term $\eta_\ell \cdot q - \tau_r$ corresponding to the $r$th assignment in $C_\ell$ is 
    are stored in the $r$th row of $M_\ell$. 
    By Claims~\ref{claim:key-correspondence:row-l}
    and~\ref{claim:key-correspondence:row-circuit}, 
    this row is equal to the $r$th row of $B_\ell'$.
    Therefore, from~\Cref{lemma:gaussian-elimination:new:above}, 
    we conclude that~$q$ is the variable~$q_{n-(r-1)}$, 
    and that $\tau_r$ is a linear term in variables $u$ and $\vec q_{[\ell,k]}$. 
    Note that then, the terms $\tau_1,\dots,\tau_\ell$ 
    do not feature any of the variables $q_{n-(\ell-1)},\dots,q_n$, 
    which means that the code above is in fact \emph{simultaneously} applying the substitutions 
    ${\sub{\frac{\tau_1}{\alpha}}{\mu \cdot q_n},\dots,\sub{\frac{\tau_\ell}{\alpha}}{\mu \cdot q_{n-(\ell-1)}}}$ to~$\rho$ 
    (by~\Cref{claim:key-correspondence:row-l}, $\alpha = \frac{\eta_\ell}{\mu}$).
    We conclude that the $i$th row of $M_\ell$ 
    holds the variable coefficients of 
    \begin{equation}
        \label{eq:result-of-updating-typeII}
        \alpha \cdot b \cdot u + \tau' - \tau'' + \alpha \cdot d,
    \end{equation} 
    where $\tau'$ stands for $\eta_\ell \cdot q'$ if $C_\ell$ assigns no expression to~$q'$, and otherwise it is the term such that~$q' \gets \frac{\tau'}{\eta_\ell}$ occurs in $C_\ell$;
    and similarly, $\tau''$ stands for $\eta_\ell \cdot q''$ if $C_\ell$ assigns no expression to~$q'$, and otherwise it is the term such that~$q'' \gets \frac{\tau''}{\eta_\ell}$ occurs in $C_\ell$.

    Let us look at Bareiss algorithm. 
    By~\Cref{lemma:one-shot-replacement:new}, 
    the $i$th row of $B_{\ell}'$ can be computed as follows:
    {\setstretch{1.3}
    \begin{algorithmic}[1]
        \State multiply the $i$th row of $B_0'$ by $\abs{\lambda_{\ell}}$
        \For{$r$ in $[1..\ell]$}
            \ subtract $b_{i,r} \cdot \vec u_r$ to the $i$th row of $B_0'$, where $\vec u_r$ is the $r$th row of $B_{\ell}'$
        \EndFor
    \end{algorithmic}
    }
    Above $b_{i,r}$ is the entry of $B_0$ in position $(i,r)$ ---whereas the entry of $B_0'$ in that position is~${\mu \cdot b_{i,r}}$. Also, by~\Cref{claim:key-correspondence:row-l}, $\abs{\lambda_\ell} = \alpha$. Following~Claims~\ref{claim:key-correspondence:row-l}
    and~\ref{claim:key-correspondence:row-circuit}, the $i$th row of $B_\ell'$ 
    contains the variable coefficients of the term 
    \begin{equation}
        \label{eq:result-of-updating-typeII-BR}
        \alpha \cdot b \cdot u + \eta_\ell \cdot (q' - q'') + \alpha \cdot d -\sum\nolimits_{r=1}^\ell b_{i,r} \cdot (\eta_\ell \cdot q_{n-(r-1)} - \tau_r)
    \end{equation}
    Now, if $q' = q''$, then all entries $b_{i,s}$ of $B_0$, where ranges in~$s \in [1..k+1]$, are zero, and so~\Cref{eq:result-of-updating-typeII-BR} simplifies to~$\alpha \cdot b \cdot u + \alpha \cdot d$; which
    is equal to the term in~\Cref{eq:result-of-updating-typeII}, since $q' = q''$ implies $\tau' = \tau''$.
    If instead $q' \neq q''$, then all entries $b_{i,s}$ of $B_0$ (with $s \in [1..k+1]$) are equal to zero, 
    except for the entry in the position corresponding to $q'$, 
    which is equal to~$1$, and the entry in the position corresponding to $q''$, which is equal to $-1$.
    \Cref{eq:result-of-updating-typeII-BR} can be rewritten as
    \begin{equation}
        \label{eq:result-of-updating-typeII-BR-2}
        \alpha \cdot b \cdot u + \eta_\ell \cdot (q' - q'') + \alpha \cdot d - \rho' + \rho'',
    \end{equation}
    where $\rho' \coloneqq 0$ if the column corresponding to $q'$ is not among the first $\ell$, 
    and otherwise $\rho' \coloneqq (\eta_\ell \cdot q' - \tau')$, 
    with $q' \gets \frac{\tau'}{\eta_\ell}$ occurring in $C_\ell$; and similarly, 
    $\rho'' \coloneqq 0$ if the column corresponding to $q''$ is not among the first $\ell$, 
    and otherwise $\rho'' \coloneqq (\eta_\ell \cdot q'' - \tau'')$, with $q'' \gets \frac{\tau''}{\eta_\ell}$ occurring in $C_\ell$.  
    The terms in~\Cref{eq:result-of-updating-typeII-BR-2,eq:result-of-updating-typeII} 
    are thus equal, proving the claim.
\end{proof}

\subsection{Proof of Claim~\ref{claim:divisions-without-remainder} from Section~\ref{sec:efficient-variable-elimination}}

\ClaimDivisionsWithoutRemainder*
\begin{proof}\label{proof:ClaimDivisionsWithoutRemainder}
Following the explanation provided as the start of the induction step of the proof of~\Cref{lemma:key-correspondence-with-Bareiss}, 
this lemma implies that $(C_{\ell},\inst{\gamma_{\ell}}{\psi}) \in \objcons_{k}^{\ell}$ for every $\ell \in [0..j]$.
In particular, this ensures that the \textbf{while} loop of~\GaussQE iterates at most $k$ times (possibly fewer, if an~\textbf{assert} command in~\Cref{algo:sub-disc} fails). Consequently, by examining all truncations of the non-deterministic branches in~\Cref{eq:sequence-of-configurations} of length up to~$k$, we have in fact accounted for all possible non-deterministic executions of~\GaussQE. Then, the statement from~\Cref{lemma:key-correspondence-with-Bareiss} ``if $\ell \geq 1$, then~\Cref{claim:divisions-without-remainder} holds when
restricted to~\Cref{algo:sub-disc} having as input 
$(C_{\ell-1}[x_m], \inst{\gamma_{\ell-1}}{\psi})$ and the equality $e_{\ell-1}$'' generalizes to all inputs of~\Cref{algo:sub-disc}; that is,~\Cref{claim:divisions-without-remainder} holds.  
\end{proof}

\subsection{Proof of~Lemma~\ref{lemma:ILEP:GaussOptBoundsNew} from~Section~\ref{sec:efficient-variable-elimination}}

\ILEPGaussOptBoundsNew*
\begin{proof}\label{proof:ILEPGaussOptBoundsNew}
Let $j \in [0..k]$.
Throughout the proof, we refer to the pair $(C_j,\gamma_j)$ and the matrices $M_j$, $B_j'$ and~$B_0$ defined in~\Cref{subsec:evolution-integers-elim-var}.
(Recall that $M_j$ is encoding the coefficients that the variables $\vec q_k$ and $u$ have in $C_j$ and $\gamma_j$; see~\Cref{lemma:matrices-all-constraints-encoded}.)
We let $\mu \coloneqq \mu_C$, and write $\pm$ for the sign of the determinant $\lambda_j \coloneqq b_{j,j}^{(j-1)}$ (postulating $\lambda_0 \coloneqq 1$).
We recall that, directly from the Leibniz formula for determinants, one obtains $\abs{\det(A)} \leq d^d \cdot \prod_{i=1}^{d} \alpha_i$ 
for any $d \times d$ integer matrix $A$ in which the entries of the $i$th column are bounded, in absolute value, by $\alpha_i \in \N$.
Let us start with a simple observation on the entries of $B_0$, which follows directly from the definition of this matrix:
\begin{claim}\label{gaussopt-complexity:easy-bounds}
    For every $i \in [1..j+t]$, 
    $\abs{b_{i,k+2}} \leq U$
    and, for every $j \in [1..k+1]$, $\abs{b_{i,j}} \leq \frac{Q}{\mu}$. 
\end{claim}

Next, we bound the coefficients of the variables~$\vec q_{k}$ and $u$ occurring in $C_j$.

\begin{claim}\label{gaussopt-complexity:circuit-bounds}
    Let $i \in [0..j-1]$, and
    $q_{n-i} \gets \frac{\tau_{n-i}}{\eta_j}$ be an assignment in $C_j$. 
    In the term~$\tau_{n-i}$: 
    \begin{itemize}[itemsep=0pt]
        \item 
        The coefficients of the variables~$\vec q_{[0,j-1]}$ are zero. 
        \item  Each coefficient of a variable in $\vec q_{[j,k]}$ is $\mu \cdot d$, for some $d \in \Z$ such that $\abs{d} \leq j^j \big(\frac{Q}{\mu}\big)^{j}$.
        \item The coefficient of the variable $u$ is bounded, in absolute value, by $j^j \big(\frac{Q}{\mu}\big)^{j-1} U$.
    \end{itemize}
    Moreover, $\eta_j = \mu \cdot g$ for some positive integer $g \leq j^j \big(\frac{Q}{\mu}\big)^{j}$.
\end{claim}

\begin{proof}
    For $j = 0$ the circuit $C_0$ features no assignment to variables in $\vec q_k$, and thus the claim is trivially true. Assume then $j \geq 1$.
    Let $q_{n-i} \gets \frac{\tau_{n-i}}{\eta_j}$ be an assignment in $C_j$.
    From the proof of correctness of~\GaussOpt (\Cref{lemma:second-step-opt}), 
    $C_j$ is a $(k,j)$-LEAC, 
    and so the coefficients of the variables~$\vec q_{[0,j-1]}$ in the term $\tau_{n-i}$ are zero.
    By definition of $M_j$, the coefficients of the term $\eta_j \cdot q_{n-i} - \tau_{n-i}$
    are found in the $(i+1)$th rows of~$M_j$. By~\Cref{lemma:key-correspondence-with-Bareiss}, they are also found in the $(i+1)$th row of~$B_j'$.

    Given $\ell \in [j..k]$, let us first consider the coefficient of the variable $q_{n-\ell}$, which is located in position $(i+1,\ell+1)$ of $B_j'$.
    From~\Cref{lemma:gaussian-elimination:new:above}.\ref{lemma:gaussian-elimination:new:above:i1},
    the entry of $B_j'$ at that position is~$\pm\mu \cdot b^{(j)}_{i+1 \gets \ell+1}$. 
    Since $b^{(j)}_{i+1 \gets \ell+1}$ 
    is a determinant of a $j \times j$ sub-matrix of $B_0$ not involving its $(k+2)$th column,
    from~\Cref{gaussopt-complexity:easy-bounds} we obtain $\abs{b^{(j)}_{i+1 \gets \ell+1}} \leq j^j (\frac{Q}{\mu})^j$.
    Similarly, $\eta_j$ (i.e., the coefficient of $q_{n-i}$) is located in position~$(i+1,i+1)$ of $B_j'$. 
    Then, by~\Cref{lemma:gaussian-elimination:new:above}.\ref{lemma:gaussian-elimination:new:above:i1},
    $\eta_j = \pm \mu \cdot b_{j,j}^{(j-1)}$, 
    and  $0 \leq \pm b^{(j-1)}_{j,j} \leq j^j (\frac{Q}{\mu})^j$.

    Lastly, we consider the coefficient of the variable $u$, 
    which is located in position $(i+1,k+2)$ of~$B_j'$.
    From~\Cref{lemma:gaussian-elimination:new:above}.\ref{lemma:gaussian-elimination:new:above:i2}, 
    the entry of $B_j'$ at that position is~$\pm b^{(j)}_{i+1 \gets k+2}$. 
    This $j \times j$ sub-determinant of $B_0$ involves the $(k+2)$th column.
    By~\Cref{gaussopt-complexity:easy-bounds}, $\abs{b^{(j)}_{i+1 \gets k+2}} \leq j^j (\frac{Q}{\mu})^{j-1} U$.
\end{proof}

\begin{claim}\label{gaussopt-complexity:formula-bounds:coefficients}
    In every equality or inequality of the formula~$\gamma_j$:
    \begin{itemize}[itemsep=0pt]
        \item The variables~$\vec q_{[0,j-1]}$ do not appear (their coefficients are zero). 
        \item Each coefficient of a variable in~$\vec q_{[j,k]}$ is
            $\mu \cdot d$, for some $d \in \Z$ such that $\abs{d} \leq (j+1)^{j+1} \big(\frac{Q}{\mu}\big)^{j+1}$.
        \item The coefficient of the variable $u$ is bounded, in absolute value, by $(j+1)^{j+1} \big(\frac{Q}{\mu}\big)^{j} U$.
    \end{itemize}
\end{claim}

\begin{proof}
    The proof is similar to the one of~\Cref{gaussopt-complexity:circuit-bounds}, but we now appeal to~\Cref{lemma:gaussian-elimination:new:below} instead of~\Cref{lemma:gaussian-elimination:new:above}. 
    By definition, the coefficients of the (in)equalities of $\gamma_j$ are located in the last $t$ rows of $M_j$, or alternatively of $B_j'$, by~\Cref{lemma:key-correspondence-with-Bareiss}. 
    Consider the $i$th row of $B_j'$, 
    with $i \in [j+1..j+t]$.
    From~\Cref{lemma:gaussian-elimination:new:below}.\ref{lemma:gaussian-elimination:new:below:i1}, 
    given $r \in [1..k+1]$, 
    the entry of the matrix $B_j'$ in position $(i,r)$ 
    is $\pm \mu \cdot b^{(j)}_{i, r}$; 
    and moreover $b^{(j)}_{i,r} = 0$ whenever $r \leq j$.
    The first two statements of the claim then follow from the 
    fact that the first $k+1$ columns of $B_j'$ contain 
    coefficients of the variables 
    $q_n,\dots,q_{n-k}$.
    In particular, for the second statement, note that $b^{(j)}_{i, r}$ 
    is the determinant of a $(j+1) \times (j+1)$ sub-matrix of $B_0$ 
    not involving its $(k+2)$th column. 
    From~\Cref{gaussopt-complexity:easy-bounds},
    $\abs{b^{(j)}_{i, r}} \leq (j+1)^{j+1} (\frac{Q}{\mu})^{j+1}$.
    The coefficient of the variable $u$ is located instead in 
    positions $(i,k+2)$ of $B_j'$.
    From~\Cref{lemma:gaussian-elimination:new:below}.\ref{lemma:gaussian-elimination:new:below:i2}, 
    the entry in this position 
    is $\pm b_{i,k+2}^{(j)}$.
    This $(j+1) \times (j+1)$ sub-determinant of $B_0$ involves the $(k+2)$th column. 
    From~\Cref{gaussopt-complexity:easy-bounds}, we have $\abs{b^{(j)}_{i,k+2}} \leq (j+1)^{j+1} (\frac{Q}{\mu})^{j} U$.
\end{proof}

Next, we consider the divisibility constraints in $\gamma_j$. We recall that 
these of the form~$d \divides \tau$ where all integers in the term~$\tau$ belong to $[0..d-1]$.
It thus suffices to give a bound on $\fmod(\psi_k)$ in order to bound all integers in these constraints.

\begin{claim}
    \label{claim:gaussopt-complexity:moduli}
    $\fmod(\gamma_j)$ divides $c \cdot \fmod(\gamma)$ for some positive integer $c \leq \frac{(j \cdot Q)^{j^2}}{\mu^{j (j-1)}}$.
\end{claim}

\begin{proof}
    For a given $\ell \in [1..j]$,
    we show that $\fmod(\gamma_{\ell}) = \mu \cdot \abs{b^{(\ell-1)}_{\ell,\ell}} \cdot \fmod(\gamma_{\ell-1})$. The bound then follows by recalling that 
    $\abs{b^{(\ell-1)}_{\ell,\ell}} \leq \ell^\ell \big(\frac{Q}{\mu}\big)^{\ell}$.
    Remark the divisibility constraints are only updated in line~\ref{algo:sub-disc:eliminate} of~\Cref{algo:sub-disc}. 
    In this line, the substitution updates 
    each divisibility constraint $d \divides \rho$ into a constraint of the form $(\abs{a} \cdot d) \divides \rho'$ (where $a$ is the coefficient of $q_{n-\ell+1}$ in the equality $a \cdot q_{n-\ell+1} = \tau$ returned by~\Cref{algo:btp}). Line~\ref{algo:sub-disc:eliminate} also adds a divisibility constraint with divisor $\abs{a}$.
    From~\Cref{claim:key-correspondence:substitution},
    $a$ is the value of the entry in position $(\ell,\ell)$ of the matrix $M_{\ell-1}$.
    From~\Cref{lemma:key-correspondence-with-Bareiss} and~\Cref{lemma:gaussian-elimination:new:below}.\ref{lemma:gaussian-elimination:new:below:i1},
    $\abs{a} = \mu \cdot \abs{b_{\ell,\ell}^{(\ell-1)}}$. 
    Therefore,
    $\fmod(\gamma_{\ell}) = \mu \cdot \abs{b^{(\ell-1)}_{\ell,\ell}} \cdot \fmod(\gamma_{\ell-1})$, as required.
\end{proof}

Lastly, we bound all constants occurring in equalities and inequalities of $\gamma_j$, and in terms~$\tau$ from the assignments $q_{n-i} \gets \frac{\tau_{n-i}}{\eta_j}$ occurring in~$C_j$, with $i \in [0..j-1]$.
Observe that in $\gamma$ and $C$, these constants are bounded by $R$.

\begin{claim}
    \label{claim:gaussopt-complexity:constants}
    Each constant in terms from $\fterms(\gamma_j)$ 
    and in terms $\tau$ from assignments $q_{n-i} \gets \frac{\tau}{\eta_j}$ in~$C_j$ (with~$i \in [0..j-1]$), 
    is bounded, in absolute value, by ${(j+1)^{2 (j+2)^2} \cdot \frac{Q^{2j(j+2)}}{\mu^{2j^2}} \cdot \fmod(\gamma) \cdot R}$.
\end{claim}

\begin{proof}
    For simplicity, let us write:
    \begin{itemize}[itemsep=0pt]
        \item $R_\ell$ for the maximum, in absolute value, of all constants occurring in 
        (in)equalities of $\gamma_\ell$ as well as in terms $\tau$ from assignments $q_{n-i} \gets \frac{\tau}{\gamma_\ell}$ in $C_\ell$, with $\ell \in [0..j]$.
        \item $S_\ell$ for the absolute value of the constant in the equality $a_{\ell} \cdot q_{n-\ell} = \tau_{\ell}$ returned by~\Cref{algo:btp} during the $(\ell+1)$th iteration of \GaussOpt, with $\ell \in [0..j-1]$.
        \item $g$ for the positive integer $\mu \cdot (j+1)^{j+1} \big(\frac{Q}{\mu}\big)^{j+1}$. 
            By Claims~\ref{claim:key-correspondence:substitution},~\ref{gaussopt-complexity:circuit-bounds} and~\ref{gaussopt-complexity:formula-bounds:coefficients}, this is an upper bound to
            $\eta_0,\dots,\eta_j$ and to
            the absolute values of all the coefficients of variables $\vec q_k$, in all formulae $\gamma_0,\dots,\gamma_j$, all circuits $C_0,\dots,C_j$, 
            and all equalities $a_0 \cdot q_n = \tau_0,\,\dots\,,\,a_{j-1} \cdot q_{n-(j-1)} = \tau_{j-1}$.
        \item $h$ for the positive integer $\frac{(j \cdot Q)^{j^2}}{\mu^{j (j-1)}} \cdot \fmod(\gamma)$. 
            This is an upper bound to $\fmod(\gamma_0),\dots,\fmod(\gamma_j)$.
    \end{itemize}
    
    We now bound $R_\ell$ and $S_\ell$ (first in terms of $R_{\ell-1}$ and $S_{\ell-1}$) by analyzing~\Cref{algo:btp,algo:sub-disc}.
    \begin{description}
        \item[bound on $R_0$:] The circuit~$C_0$ has no assignment on the variables $\vec q_{k-1}$, so $R_0 = 0 \leq R$.
        \item[bound on $R_{\ell}$ for $\ell \geq 1$.] We show that $R_{\ell} \leq g\cdot (R_{\ell-1} + S_{\ell-1})$. 
            Looking at~\Cref{algo:sub-disc}, we see that after the substitution in line~\ref{algo:sub-disc:eliminate} 
            takes place, the constants in the equalities and inequalities in $\gamma$ 
            are bounded by $\frac{g}{\mu} (R_{\ell-1}+ S_{\ell-1})$ (in particular, note that $\alpha$ from line~\ref{algo:sub-disc:lambda} is bounded by $\frac{g}{\mu}$). Afterwards, lines~\ref{algo:sub-disc:simplify}
            and~\ref{algo:sub-disc:simplify-2} divide these constants by $\frac{\eta_{\ell-1}}{\mu}$ 
            (line~\ref{algo:sub-disc:simplify-2} takes the ceiling of this division). 
            By definition of $\objcons_{k}^{\ell-1}$, 
            $\mu$ divides $\eta_{\ell-1}$. 
            Hence, the constants 
            appearing in the (in)equalities of~$\gamma_\ell$ are bounded by $\ceil{\frac{g}{\eta_{\ell-1}} (R_{\ell-1}+ S_{\ell-1})} \leq g \cdot (R_{\ell-1} + S_{\ell-1})$. 
            A similar analysis applies to the constants in the terms~$\tau$ from assignments $q_{n-i} \gets \frac{\tau}{\gamma_\ell}$ in $C_\ell$, since the corresponding terms in $C_{\ell-1}$ are updated in the same way as those in equalities of~$\gamma_{\ell-1}$ (line~\ref{algo:sub-disc:update-C}).
        \item[bound on $S_{\ell}$ for $\ell \geq 0$.] 
            We show that $S_{\ell} \leq 2 \cdot R_{\ell} + g \cdot (R + 3 \cdot h)$. 
            If in line~\ref{algo:true-tp:guess}~\Cref{algo:btp} guesses 
            a term from~$\fterms(\gamma \land \gamma\sub{q_{n-\ell}+ p}{q_{n-\ell}})$,
            then $S_{\ell} \leq R_{\ell} + g \cdot h$ and we are done.
            Otherwise,~\Cref{algo:additional-hyperplanes} is invoked, which returns a 
            term obtained by simultaneously applying two substitutions $\nu_1$ and $\nu_2$ 
            to a term of the form $a \cdot u + \mu \cdot (q' - q'') + d$,
            with $a,d \in [-L..L]$ and $q',q''$ variables in $\vec q_k$. 
            As already discussed during the proof of~\Cref{claim:key-correspondence:substitution},
            the substitutions~$\nu_1$ and~$\nu_2$ are of the form $\sub{\frac{\tau}{\lambda}}{\mu \cdot q}$
            where, $q$ is among $q'$ and $q''$, $\lambda = \frac{\eta_\ell}{\mu}$, and the term $\tau$ is \textit{(i)}~$\eta_{\ell} \cdot q$, or \textit{(ii)}~$\eta_{\ell} \cdot q + \eta_{\ell} \cdot p$ with $p \coloneqq \fmod(q_{n-\ell},\gamma_\ell)$, or \textit{(iii)}~such that $q \gets \frac{\tau}{\eta_{\ell}}$ occurs in $C_{\ell}$, 
            or \textit{(iv)}~of the form $\tau'\sub{q_{n-\ell}+p}{q_{n-\ell}}$ with $q \gets \frac{\tau'}{\eta_\ell}$ occurring in $C_{\ell}$.
            Therefore, the constant of~$\tau$ is bounded by $R_\ell + g \cdot h$ 
            (where $g \cdot h$ accounts for the constant in Case~\textit{(ii)} 
            and for the increase that the substitution~$\sub{q_{n-\ell}+p}{q_{n-\ell}}$ may cause).
            The constant of the term computed in line~\ref{algo:btp:term-before-shift}
            of~\Cref{algo:additional-hyperplanes} is thus bounded, in absolute value, by $2 \cdot (R_{\ell}+ g \cdot h) + \frac{\eta_\ell}{\mu} \cdot \abs{d}$. 
            This constant is then shifted by at most $g \cdot h$ in line~\ref{algo:true-tp:guess-2} 
            of~\Cref{algo:btp}.
            Recall that $\abs{d} \leq L \leq R$ and,
            from~\Cref{lemma:key-correspondence-with-Bareiss}, 
            $\frac{\eta_\ell}{\mu} = \abs{b_{\ell,\ell}^{(\ell-1)}} \leq g$.
            We conclude that $S_\ell \leq 2 \cdot R_{\ell} + g \cdot (R + 3 \cdot h)$.
    \end{description}
    By conjoining the above inequalities for $R_\ell$ and $S_{\ell}$, we derive the following recurrence relation:
    \[ 
        R_0 \leq R, 
        \qquad\qquad
        R_{\ell} \leq 3 \cdot g \cdot R_{\ell-1} + g^2 \cdot (R + 3 \cdot h)
        \qquad\text{for $\ell \in [1..j]$}.
    \]
    A simple induction shows $R_j \leq (3 \cdot g)^j R + (g^2 (R + 3 \cdot h)) \cdot \sum_{i=0}^{j-1} (3 \cdot g)^i$. For $j \geq 1$, we have:
    \begin{align*}
        R_j 
        &\leq (3 \cdot g)^j + j \cdot (3 \cdot g)^{j-1} g^2 (R + 3 \cdot h)\\
        &\leq (j+1) \cdot 3^{j+1} g^{j+1} h \cdot R 
        &\hspace{-20pt}\Lbag \text{we have $R + 3 \cdot h \leq R \cdot 3 \cdot h$} \Rbag\\
        &\leq 3^{j+1} \mu^{j+1} (j+1)^{(j+1)^2+j^2+1} \cdot \frac{Q^{(j+1)^2+j^2}}{\mu^{(j+1)^2+ j  (j-1)}} \cdot \fmod(\gamma) \cdot R 
        & \Lbag \text{def.~of $g$ and $h$} \Rbag\\
        &\leq \frac{((j+1) \cdot Q)^{2(j+2)^2}}{\mu^{2j^2}} \cdot \fmod(\gamma) \cdot R.
        & \Lbag \text{using $j \geq 1$} \Rbag
    \end{align*}
    Note that for $j = 0$ the last expression reduces to $Q^{8} \fmod(\gamma) \cdot R$, which is an upper bound~to~$R_0$.
\end{proof}

Since we have let $j$ range arbitrarily in $[0..k]$, 
Claims~\ref{gaussopt-complexity:circuit-bounds}--\ref{claim:gaussopt-complexity:constants} 
establish that, throughout its execution, 
\GaussOpt only constructs objects whose sizes polynomial in the sizes of $C$ and $\gamma$. 
Since~$k$ bounds the number of iterations of~\GaussOpt along any non-deterministic branch, we conclude that it runs in non-deterministic polynomial time.
This completes the proof of the lemma:
Items~\ref{lemma:ILEP:GaussOptBounds:i1}--\ref{lemma:ILEP:GaussOptBounds:i4} 
follow from the fact that~\Cref{algo:sub-disc} does not update $\psi$ nor any of the expressions in $C$ featuring $x_{n-k},\dots,x_n$,
whereas Items~\ref{lemma:ILEP:GaussOptBounds:i5}--\ref{lemma:ILEP:GaussOptBounds:i7} follow directly from Claims~\ref{gaussopt-complexity:circuit-bounds}--\ref{claim:gaussopt-complexity:constants}.
\end{proof}

\subsection{Proof of Lemma~\ref{lemma:putting-all-together-k-iterations} from~Section~\ref{sec:putting-all-together}}%
\label{appendix:proofs-section-six}

\LemmaPuttingAllTogetherKIterations*
\begin{proof}\label{proof:LemmaPuttingAllTogetherKIterations}
  The proof is by induction on $k$. 
  
  \begin{description}
  \item[base case: $k = 0$.] 
  In this case, $\phi_0$ is equal to $\phi$, 
  and $C_0$ is the empty $0$-\preleac 
  (hence, by definition~$\mu_{C_0} = 1$ and $\xi_{C_0} = 0$). All bounds in the statement trivially follows.

  \item[induction hypothesis:] 
    For $k \geq 0$, 
    the bounds in the statement hold
    for the~$k$th loop iteration.

  \item[induction step:] 
    Given $k \geq 0$, 
    consider a triple $(\phi_k,\theta_k,C_k)$ 
    obtained at the end of the $k$th iteration 
    of the loop, and $(\phi_{k+1},\theta_{k+1},C_{k+1})$ be obtained  
    by applying the body of the loop 
    to $(\phi_k,\theta_k,C_k)$.
    We bound the parameters of $\phi_{k+1}$ 
    and $C_{k+1}$.
    For brevity, we write $\xi_{k}$ and $\mu_k$ 
    for $\xi_{C_{k}}$ and $\mu_{C_k}$, 
    respectively (and use similar notation for $\xi_{C_{k+1}}$ and $\mu_{C_{k+1}}$). 

  \begin{description}
    \item[least significant terms of~$\phi_{k+1}$:] 
      \begin{align*}
        \card \lst(\phi_{k+1}, \theta_{k+1}) 
        &\leq \max(\card \lst(\phi_{k}, \theta_{k}),1) + 2 \cdot k + 3 
        &\Lbag\text{by \Cref{lemma:putting-all-together-one-iteration}}\Rbag\\
        &\leq (\boundLstk[k]) + 2 \cdot k + 3 
        &\Lbag \text{by I.H.}\Rbag\\
        &\leq \boundLstk[(k+1)]
      \end{align*}
  \end{description}
  In the following cases, for simplicity we assume that all parameters of $\phi_k$ and $C_k$ are greater than or equal to $1$. This assumption is made solely to avoid repeatedly writing expressions involving~$\max(\cdot,1)$, as we did earlier for $\card \lst(\phi_{k}, \theta_{k})$.
  \begin{description}
    \item[number of constraints in~$\phi_{k+1}$:]
      \begin{align*}
        \card{\phi_{k+1}}
        &\leq \card{\phi_{k}} + 6 \cdot k + 2 \cdot \card{\lst(\phi_{k}, \theta_{k})} + 3 
        &\Lbag\text{by \Cref{lemma:putting-all-together-one-iteration}}\Rbag\\
        &\leq (\boundCardPhik[k]) + 6 \cdot k + 2 \cdot (\boundLstk[k]) + 3 
        &\Lbag \text{by I.H.}\Rbag\\
        &\leq \boundCardPhik[(k+1)].
      \end{align*}

    \item[linear norm of $\phi_{k+1}$:]
    \begin{align*}
      \linnorm{\phi_{k+1}} 
      &\leq 3 \cdot \linnorm{\phi_{k}} 
      &\Lbag\text{by \Cref{lemma:putting-all-together-one-iteration}}\Rbag\\
      &\leq \boundLinNormPhik[k+1].
      &\Lbag \text{by I.H.}\Rbag
    \end{align*}

    \item[denominator $\mu_{k+1}$:] 
    \begin{align*}
      \mu_{k+1} 
      &\le \mu_k (3 \cdot k \cdot \linnorm{\phi_k})^k 
      &\Lbag\text{by \Cref{lemma:putting-all-together-one-iteration}}\Rbag\\
      &\le \boundMuk[k] \big(3 \cdot k \cdot (3^k a)\big)^k
      &\Lbag \text{by I.H.}\Rbag\\
      &\leq \boundMuk[(k+1)].
    \end{align*}

    \item[modulus of $\phi_{k+1}$:]
    Below, $(\phi_0,C_0),\dots,(\phi_{k-1},C_{k-1})$ denote the formulae and~\preleac{s}
    constructed by the algorithm during the first $k-1$ iterations of 
    the \textbf{while} loop; $(\phi_k,C_k)$ are obtained from $(\phi_{k-1},C_{k-1})$ by performing a further iteration.
    In particular, $\phi_0$ is the linear-exponential program given as input to~\OptILEP, and $C_0$ is the empty $0$-\preleac.
    By~\Cref{lemma:putting-all-together-one-iteration},
    for every~$i \in [0..k]$,
    there is~$\alpha_{i+1} \in {[1..(3\cdot i \cdot \mu_{i} \cdot \linnorm{\phi_i})^{i^2}]}$
    such that $\fmod(\phi_{i+1})$ is a divisor of~$\lcm{(\fmod(\phi_{i}), \totient(\alpha_{i+1} \cdot \fmod(\phi_{i})))}$.
    Let us define $\alpha^* \coloneqq \lcm{(\alpha_1, \alpha_2, \ldots, \alpha_{k+1})}$, and consider the integers $c_0,\dots,c_{k+1}$ given by 
    \[
      \begin{cases}
        c_0 \coloneqq 1\\ 
        c_{i+1} \coloneqq \lcm(c_i,\totient(\alpha^* \cdot c_i))
        &\text{for $i \in [0..k]$}
      \end{cases}
    \]
    \begin{claim}\label{claim:bk-divides-ck}
      For every $j \in [0..k+1]$, $\fmod(\phi_{j})$ divides $c_{j}$.
    \end{claim}
    \begin{proof}
    The proof is by induction on $j$.
    \begin{description}
      \item[base case: $j=0$.] We have $\fmod(\phi_0) = 1 = c_0$.
      \item[induction step:] Assume that the claim holds for $j \in [0..k]$. Then,
    \begin{align*}
      \fmod(\phi_{j+1}) \coloneqq{}& 
          \lcm{(\fmod(\phi_j), \totient(\alpha_{j+1} \cdot \fmod(\phi_j)))} \\
          \divides{}& \lcm{(c_j, \totient(\alpha_{j+1} \cdot \fmod(\phi_j)))} 
          &\Lbag\text{by I.H., $\fmod(\phi_i) \divides c_i$}\Rbag\\
          \divides{}& \lcm{(c_j, \totient(\alpha^* \cdot c_j))} 
          &\Lbag \text{$q \divides r$ implies $\totient(q) \divides \totient(r)$}\Rbag\\
          ={}& c_{j+1}. &&\qedhere
    \end{align*}
    \end{description} 
    \end{proof}
    Given~\Cref{claim:bk-divides-ck}, 
    in order to bound $\fmod(\phi_{k+1})$ it suffices to bound $c_{k+1}$. 
    The next lemma from~\cite{ChistikovMS24} 
    will help us analyze this integer.
    \LemmaBoundLCMTotient*

    First, observe that
    \begin{align*}
    \alpha^*
    &\leq \prod\nolimits_{i=0}^{k} (3 \cdot i \cdot \mu_{i} \cdot \linnorm{\phi_i})^{i^2} \\
    &\leq (3 \cdot k \cdot (\boundMuk[k]) \cdot (\boundLinNormPhik[k]))^{k^2(k+1)}  &\Lbag\text{by~I.H.}\Rbag\\
    &\leq 3^{(k+1)^6} a^{(k+1)^5}.
    \end{align*}
    Then, $c_{k+1}$ is bounded as follows:
    \begin{align*}
      c_{k+1} &\le (\alpha^*)^{2(k+1)^2}  
      &\Lbag{\text{by~\Cref{remark:bound-on-lcm-totient-growth}}}\Rbag\\
      &\le 3^{2 {(k+1)}^8} a^{2 {(k+1)}^7}.
    \end{align*}
  \end{description}

  Note that in~\Cref{lemma:putting-all-together-one-iteration} 
  the $1$-norm of $\phi'$ and the parameter $\xi_{C'}$ 
  are bounded in terms of a relatively complex quantity denoted as~$\beta$. 
  To simplify the upcoming calculations, we first derive a more 
  manageable upper bound for $\beta$, with respect to $\phi_{k+1}$ and $C_{k+1}$. We start by simplifying the subexpression $\log(\xi_k + \mu_k)$:
    \begin{align*}
      \log(\xi_k + \mu_k) &\le \log(\boundXiCk[k] + \boundMuk[k])
      &\Lbag \text{by I.H.}\Rbag\\
      &\le 1+ \log\big(3^{8 (k+2)^8} c^{8 (k+2)^7}\big) 
      &\Lbag \text{as $a \leq c$}\Rbag\\
      &\le 17 \cdot (k+2)^8 c.
    \end{align*}
    The quantity $\beta$ can then be simplified as follows: 
    \begin{align*}
      \beta 
      &\coloneqq \fmod{(\phi_k)} \big( 2^7 (k+1) \cdot \mu_k \cdot \max(\onenorm{\phi_{k}}, \log(\xi_k + \mu_k))\big)^{3(k+2)^2} \\
      &\le \boundModPhik[k] \big(2^7 (k+1) \cdot \boundMuk[k] \max(\boundOneNormPhik[k],  \log(\xi_k + \mu_k))\big)^{3(k+2)^2}
      &\Lbag\text{by~I.H}\Rbag\\
      &\le \boundModPhik[k] \big(2^7 (k+1) \cdot \boundMuk[k] \boundOneNormPhik[k]\big)^{3(k+2)^2}
      &\hspace{-2.2cm}\Lbag\text{as }3^{8(k+1)} \geq 17 \cdot (k+2)^8\Rbag\\
      &\le 3^{2 \cdot k^8+3 \cdot (k+2)^2(k^3+9k+13)} c^{2 \cdot k^7+3 (k+2)^2 (k^2+1)}
      &\Lbag\text{as $a \leq c$}\Rbag\\
      &\le \boundBetak[k].
    \end{align*}

  \begin{description}
    \item[$1$-norm of $\phi_{k+1}$:]
    \begin{align*}
      \onenorm{\phi_{k+1}}
      &\leq 12 + 4 \cdot 
        \max{(\onenorm{\phi_k}, \log{\beta})}  
      &\Lbag\text{by \Cref{lemma:putting-all-together-one-iteration}}\Rbag\\
      &\leq 12 + 4 \cdot \max(\boundOneNormPhik[k], \log \beta)
      &\Lbag\text{by~I.H.}\Rbag
      \\
      &\leq 12 + 4 \cdot \max(\boundOneNormPhik[k], 4 \cdot (k+2)^8 c)
      &\Lbag\text{from bound on $\beta$}\Rbag
      \\
      &\leq 12 + 4 \cdot \boundOneNormPhik[k] \,\leq\, 3^{8 \cdot (k+2)} c.
    \end{align*}
  \item[parameter $\xi_{k+1}$:]
  \begin{align*}
  \xi_{k+1} 
  &\le \xi_{k} \cdot (3\cdot k \cdot \linnorm{\phi_k})^k + 2^6 (k+1) \cdot \beta^4&
  \hspace{-1cm}\Lbag\text{by \Cref{lemma:putting-all-together-one-iteration}}\Rbag\\
  &\le (\boundXiCk[k]) \cdot (3 \cdot k \cdot (\boundLinNormPhik[k]))^k  + 2^6 (k+1) \cdot (\boundBetak[k])^4
  &\Lbag\text{by~I.H.}\Rbag\\
  &\le \boundXiCk[k] ((3^{k+1} k \cdot c)^k + 2^6 (k+1))
  &\Lbag a \leq c \Rbag\\ 
  &\le 3^{8 (k+2)^8}c^{8 (k+2)^7}(c^k 3^{2k^2+k+5}) \,\le\, 3^{8 (k+3)^8}c^{8 (k+3)^7}.
  \end{align*}
  \end{description}
  \end{description}
  Given the bounds we have just established, it is simple to see that~\OptILEP runs in non-deterministic polynomial time. Indeed, these bounds ensure that, each time the execution reaches line~\ref{optilep:line:while}, both the formula $\phi$ and circuit $C$ manipulated by the algorithm 
  are of size polynomial in the input. 
  The \textbf{while} loop of line~\ref{optilep:line:while} iterates $n$ times, 
  and by~\Cref{lemma:putting-all-together-one-iteration} each iteration runs in non-deterministic 
  polynomial time. 
  It follows that~\OptILEP runs in non-deterministic polynomial time. 
\end{proof}

\section{Proofs of statements from~Part~\ref{part:deciding-properties-ILESLP}}%
\label{appendix:proofs-intro-part-ii}

\LemmaSimpleExpressions* 
\begin{proof}\label{proof:LemmaSimpleExpressions}
   Given $i \in [0..n]$,
    let $\sigma_i$ denote the ILESLP $(x_0 \gets \rho_0, \dots, x_i \gets \rho_i)$ obtained by truncating $\sigma$ after $i+1$ assignments.
    We remark that $d(\sigma_i)$ divides $d(\sigma_{j})$ for every $i \leq j$.

    We show by induction on~$i$ how to compute 
    a vector of rational numbers $\vec b_i = (b_{i,0},\dots,b_{i,i-1}) \in \Q^i$ satisfying
    \(\sem{\sigma}(x_i) = \sum_{j = 0}^{i-1}b_{i,j} \cdot 2^{\sem{\sigma}(x_j)}\). 
    Moreover, each $b_{i,j}$ is of the form $\frac{m}{d(\sigma_i)}$ for some $m \in \Z$ satisfying $\abs{m} \leq 2^{i} \cdot e(\sigma_i) \cdot d(\sigma_i)$, 
    and $m \neq 0$ only if $\sem{\sigma}(x_j) \geq 0$.
    With this result at hand, the expression $E_i$ in the statement of the lemma is computed by multiplying all these rational numbers by $d(\sigma)$ to make them integers. In particular, if $b_{i,j} = \frac{m}{d(\sigma_i)}$, then in~$E_i$ the coefficient of $2^{x_j}$ is $a_{i,j} \coloneqq m \cdot \frac{d(\sigma)}{d(\sigma_i)}$. 
    We then conclude that $\abs{a_{i,j}} \leq 2^{i} \cdot e(\sigma_i) \cdot d(\sigma_i) \cdot \frac{d(\sigma)}{d(\sigma_i)} \leq 2^i \cdot e(\sigma) \cdot d(\sigma)$.
    Note that the bit size of each $a_{i,j}$ is thus polynomial in the size of $\sigma$.
    With this in mind, the fact that the whole computation can be performed in polynomial time will be immediate from the inductive~proof.

    \begin{description}
        \item[base case: $i = 0$.] In this case we simply have $\rho_0 = 0$, and we take $\vec b_0$ to be empty vector.

        \item[induction hypothesis.] We have computed the vector $\vec b_j \in \Q^j$, for every $j \in [0..i-1]$. Given ${k \in [0..j-1]}$, the $k$th entry of $\vec b_j$ is a rational of the form $\frac{m}{d(\sigma_j)}$,
        for some $m \in \Z$ satisfying ${\abs{m} \leq 2^{j} \cdot e(\sigma_j) \cdot d(\sigma_j)}$, and $m \neq 0$ only if $\sem{\sigma}(x_k) \geq 0$.

        \item[induction step: $i \geq 1$.] We reason by cases, depending on $\rho_i$.
        \begin{description}
            \item[case: $\rho_i = 0$.] We define $\vec b_i$ to be the zero vector of length $i$ (encoded as the rational $\frac{0}{d(\sigma_i)}$). 
            \item[case: $\rho_i = 2^{x_j}$.] We define $\vec b_i$ by setting $b_{i,j} = \frac{d(\sigma_i)}{d(\sigma_i)}$, and $b_{i,\ell} = 0$ for all $\ell \neq j$. Since $\sigma$ is an ILESLP, we must have $\sem{\sigma}(x_i) \in \Z$.
            Therefore, $\sem{\sigma}(x_j) \geq 0$; which allows us to set a non-zero value to $b_{i,j}$.
            \item[case: $\rho_i = x_j + x_k$.] Following the induction hypothesis, consider the already computed vectors $\vec b_j = (b_{j,0},\dots,b_{j,j-1})$ and $\vec b_k = (b_{k,0},\dots,b_{k,k-1})$. Let the vectors $(b_{j,0},\dots,b_{j,i-1})$ and $(b_{k,0},\dots,b_{k,i-1})$ be obtained from $\vec b_j$ and $\vec b_k$ by appending a suitable amount of $0$s (encoded as $\frac{0}{d(\sigma_j)}$ and $\frac{0}{d(\sigma_k)}$, respectively).
            The vector $\vec b_i$ is defined as follows: 
            for every $\ell \in [0..i-1]$, 
            if $b_{j,\ell} = \frac{m}{d(\sigma_j)}$ and $b_{k,\ell} = \frac{r}{d(\sigma_k)}$, 
            then we define $b_{i,\ell} \coloneqq \frac{m \cdot \frac{d(\sigma_i)}{d(\sigma_j)} + r \cdot \frac{d(\sigma_i)}{d(\sigma_k)}}{d(\sigma_i)}$.
            Clearly, $b_{i,\ell} = b_{j,\ell} + b_{k,\ell}$, 
            and the numerator of $b_{i,\ell}$ is an integer (because $d(\sigma_i)$ is divided by both $d(\sigma_j)$ and $d(\sigma_k)$). 
            For the numerator, we have: 
            \begin{align*} 
                &\abs{m \cdot \frac{d(\sigma_i)}{d(\sigma_j)} + r \cdot \frac{d(\sigma_i)}{d(\sigma_k)}}\\
                \leq{}& 
                \abs{m \cdot \frac{d(\sigma_i)}{d(\sigma_j)}} + \abs{r \cdot \frac{d(\sigma_i)}{d(\sigma_k)}}\\
                \leq{}&
                2^{j} \cdot e(\sigma_j) \cdot d(\sigma_j) \cdot \frac{d(\sigma_i)}{d(\sigma_j)}
                + 
                2^{k} \cdot e(\sigma_k) \cdot d(\sigma_k) \cdot \frac{d(\sigma_i)}{d(\sigma_k)}
                & \Lbag\text{by induction hypothesis}\Rbag\\
                \leq{}& 2 \cdot 2^{i-1} \cdot e(\sigma_{i-1}) \cdot d(\sigma_{i})\\
                \leq{}& 2^{i} \cdot e(\sigma_{i}) \cdot d(\sigma_{i}).
            \end{align*}
            Lastly, observe that for a variable $x$ among $x_0,\dots,x_{i-1}$ satisfying $\sem{\sigma}(x) < 0$, (the numerators of) both corresponding rationals in $\vec b_j$ and $\vec b_k$ are zero (by induction hypothesis). 
            Therefore, the same holds for $\vec b_i$.

            \item[case: $\rho_i = \frac{m}{g} \cdot x_j$.] 
            Similarly to the previous case, consider the vector 
            $(b_{j,0},\dots,b_{j,i-1})$ obtained from $\vec b_j$ by appending~$0$s. 
            The vector $\vec b_i$ is defined as follows: 
            for every $\ell \in [0..i-1]$, 
            if $b_{j,\ell} = \frac{r}{d(\sigma_j)}$
            then we define $b_{i,\ell} \coloneqq \frac{m \cdot r \cdot \frac{d(\sigma_{i-1})}{d(\sigma_j)}}{d(\sigma_i)}$. 
            Note that, by definition, $d(\sigma_i) = g \cdot d(\sigma_{i-1})$ 
            and $b_{j,\ell} = \frac{r \cdot \frac{d(\sigma_{i-1})}{d(\sigma_j)}}{d(\sigma_{i-1})}$; 
            and thus $b_{i,\ell} = \frac{m}{g} \cdot b_{i,\ell}$.
            The numerator is bounded as follows: 
               \begin{align*}
                \abs{m \cdot r \cdot \frac{d(\sigma_{i-1})}{d(\sigma_j)}} 
                &\leq 2^{j} \cdot e(\sigma_j) \cdot d(\sigma_j) \cdot \abs{m} \cdot \frac{d(\sigma_{i-1})}{d(\sigma_j)}
                & \Lbag\text{by induction hypothesis}\Rbag\\
                & \leq 2^{j} \cdot e(\sigma_i) \cdot d(\sigma_{i-1})
                & \Lbag\text{because $e(\sigma_j) \cdot \abs{m} \leq e(\sigma_i)$}\Rbag\\
                &\leq 2^{i} \cdot e(\sigma_{i}) \cdot d(\sigma_{i}).
            \end{align*}
            Lastly, note that if a variable $x$ among $x_0,\dots,x_{i-1}$ satisfies $\sem{\sigma}(x) < 0$, then the corresponding rational in $\vec b_j$ is zero (by induction hypothesis), and so the same holds for~$\vec b_j$.
            \qedhere
        \end{description} 
    \end{description}
\end{proof}

\clearpage
\fancyhead[R]{}
\renewcommand{\headrulewidth}{0pt}
\bibliographystyle{alpha}
\bibliography{bibliography.bib}

\end{document}